\documentclass[12pt]{article}

\usepackage{header}

\begin{document}

\title{\sc A Taxonomy of Non-dictatorial Unidimensional Domains\thanks{We would like to thank the Editor Felix Brandt and the Associate Editor
for their constructive comments and suggestions. We are particularly grateful to two anonymous referees for their  very detailed and insightful reports that have substantially improved the paper.
We are grateful to Clemens Puppe and John Weymark for their helpful comments and suggestions. The research reported here was supported by
the Program for Professor of Special Appointment (Eastern Scholar) at Shanghai Institutions of Higher Learning (No.~2019140015), and
the Fundamental Research Funds for the Central Universities (No.~2018110153). This paper was previously circulated under the title ``\emph{A Taxonomy of Non-dictatorial Domains}'' \href{https://ink.library.smu.edu.sg/soe_research/2408/}{\textcolor[rgb]{0.00,0.00,1.00}{SMU Economics and Statistics Working Paper Series, Paper No.~22-2020}.}}}
\author{Shurojit Chatterji\thanks{School of Economics, Singapore Management University, Singapore.}~ and
Huaxia Zeng\thanks{School of Economics, Shanghai University of Finance and Economics, and
the Key Laboratory of Mathematical Economics (SUFE), Ministry of Education, Shanghai 200433, China.}}
\date{\today}
\maketitle

\begin{abstract}
\noindent
Non-dictatorial preference domains allow the design of unanimous social choice functions (henceforth, rules) that are non-dictatorial and strategy-proof. On a class of preference domains called unidimensional domains, we show that the unique seconds property
(introduced by \citealp*{ACS2003}) characterizes all non-dictatorial domains. Subsequently, we provide an exhaustive classification of all non-dictatorial, unidimensional domains, based on a simple property of two-voter rules called invariance. The domains constituting the classification are semi-single-peaked domains (introduced by \citealp*{CSS2013}) and semi-hybrid domains (introduced here) which are two appropriate weakenings of single-peaked domains and shown to allow strategy-proof rules to depend on non-peak information of voters' preferences; the canonical strategy-proof rules for these domains are projection rules and hybrid rules respectively. As a refinement of the classification, single-peaked domains and hybrid domains emerge as the only unidimensional domains that force strategy-proof rules to be determined completely by preference peaks.


\end{abstract}

\section{Introduction}\label{sec:introduction}

An overarching theme in the theory of incentives is that unanimous social choice functions
(henceforth, \emph{rules})
that are non-manipulable are dictatorial (and hence unsuitable for social decisions),
unless preferences of voters are restricted in particular ways so as to yield \textit{non-dictatorial domains} that allow the design of non-dictatorial, strategy-proof rules.
Indeed, the domain of single-peaked preferences in the classical voting model \citep{M1980} and the domain of quasi-linear preferences in models with monetary compensations \citep{R1979}
are leading instances of non-dictatorial domains.
In this paper, we restrict attention to the voting model, where the large literature
notwithstanding\footnote{The literature is taken up in Section \ref{sec:literature}.},
a comprehensive classification of \emph{all} non-dictatorial domains in terms of the design opportunities they afford has remained elusive,
and where in particular, not much is known about the structure of preference domains
that allow strategy-proof rules to vary with non-peak information on preferences.

Single-peaked domains are the most prominent instance of non-dictatorial domains in the voting model and are widely applied to models in public good provision, electoral competition, location theory among others.
A natural question that recent literature has made progress on is whether the full force of single-peakedness is needed to guarantee the existence of strategy-proof rules.
It is known that a particular weakening of single-peakedness, \emph{semi-single-peakedness} introduced by \citet{CSS2013}, is compatible with strategy-proofness and two other attractive axioms.
The first of these is anonymity, an axiom that spreads power evenly across voters and is hence in a sense the polar opposite of dictatorship. The second is the tops-only property which asserts that the rule is completely determined by the peaks of the voters preferences.\footnote{Rules that possess the tops-only property are easier to describe and operationalize.
}
But there are presumably other ways of relaxing single-peakedness than semi-single-peakedness, which may remain compatible with strategy-proofness but not with anonymity.
Moreover, are there other instances of non-dictatorial domains that are significantly different from single-peaked domains in their underlying description, but that allow strategy-proof rules to vary with non-peak information on voters' preferences and consequently afford very different design opportunities? These questions lie at the heart of the theory of incentives since they follow naturally from the Gibbard-Satterthwaite Theorem \citep{G1973,S1975}, and are evidently important from the practical standpoint of applying the theory to the actual design of mechanisms. These are the sort of questions this paper is concerned with. We identify three variants of single-peakedness which along with single-peakedness constitute an exhaustive classification of non-dictatorial domains (that satisfy a condition called unidimensionality) and specify properties of canonical strategy-proof rules that they admit.

\begin{figure}[t]
\centering
\subfigure[\label{fig:a}]{
\includegraphics[width=0.3\textwidth]{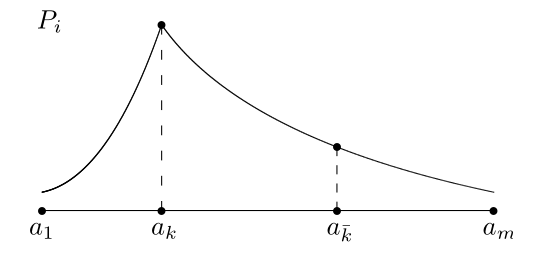}}
\subfigure[\label{fig:b}]{
\includegraphics[width=0.3\textwidth]{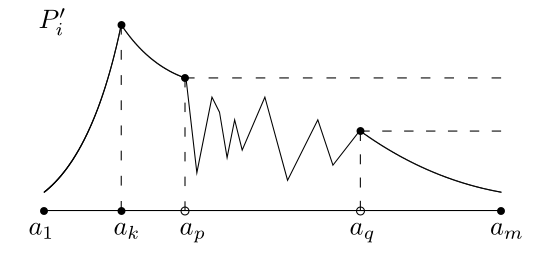}}
\subfigure[\label{fig:c}]{
\includegraphics[width=0.3\textwidth]{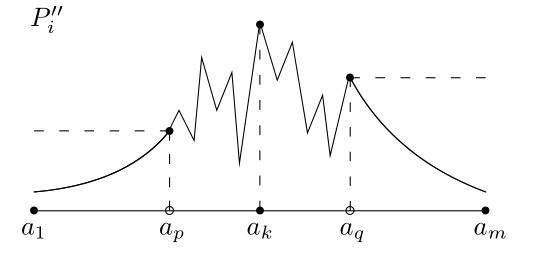}}\\[-0.5em]
\subfigure[\label{fig:d}]{
\includegraphics[width=0.3\textwidth]{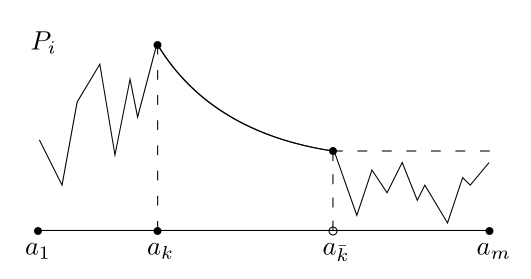}}
\subfigure[\label{fig:e}]{
\includegraphics[width=0.3\textwidth]{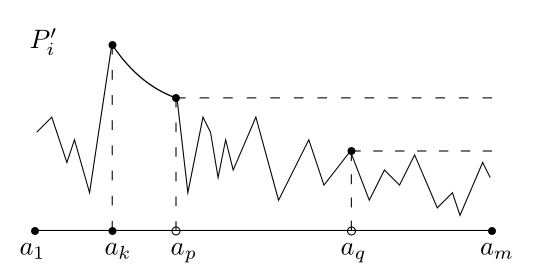}}
\subfigure[\label{fig:f}]{
\includegraphics[width=0.3\textwidth]{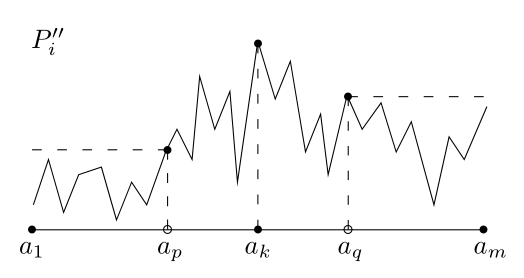}}
\caption{A single-peaked preference (a), two hybrid preferences (b) - (c),
a semi-single-peaked preference (d) and two semi-hybrid preferences (e) - (f) on the line $\mathcal{L}^A$}\label{fig:heuristic}
\end{figure}

\subsection{Preference domains}
We briefly introduce single-peaked preferences and three variants here.
These require an underlying tree on the alternatives.
However, for ease of presentation, we assume henceforth that all alternatives are exogenously located on a \emph{line}, i.e., a tree that linearly orders all alternatives, say $\mathcal{L}^A = (a_1, \dots, a_m)$.\footnote{The original definition of single-peaked preferences hypothesized a linear order on alternatives, which was subsequently generalized to single-peakedness on a tree by \citet{D1982}. Trees expand the class of models to which the theory can be applied. More importantly, the versions formulated on trees are the ones that arise as necessary conditions in our formulation.}
A preference is single-peaked on the line $\mathcal{L}^A$, if there exists a top-ranked alternative, called the peak, such that the preference declines on each side of the peak when alternatives moves farther away from the peak (see Figure \ref{fig:a}).
A hybrid preference weakens the single-peakedness restriction by letting the relative rankings of alternatives clustered in a fixed interval of the line be more permissive.
First, two distinct alternatives, called \emph{dual-thresholds}, are exogenously fixed on the line $\mathcal{L}^A$, say $a_p$, $a_q $ and $1 \leq p< q \leq m$, which separate the line $\mathcal{L}^A$ into the left interval between $a_1$ and $a_p$, the middle interval between $a_p$ and $a_q$, and
the right interval between $a_q$ and $a_m$.
In a hybrid preference, the conventional single-peakedness restriction only prevails on the left and right intervals (see Figure \ref{fig:b} and Figure \ref{fig:c}),
while all alternatives in the middle interval are arbitrarily ranked subject to an additional condition:
if the preference peak is located to the left of $a_p$ (respectively, the right of $a_q$), then $a_p$ (respectively, $a_q$) is the best alternative within the middle interval (see Figure \ref{fig:b}).
On the other hand, the notion of semi-single-peakedness weakens single-peakedness differently.
Instead of dual-thresholds, one alternative $a_{\bar{k}}$, called a \emph{threshold}, is exogenously fixed.
In a semi-single-peaked preference, the conventional single-peakedness restriction only prevails on the ``middle'' interval between the peak and the threshold, while any alternative not in this interval is ranked below its projection on the interval in an arbitrary manner (see Figure \ref{fig:d}).
Finally, we generate a semi-hybrid preference by weakening the single-peakedness restriction imposed on a hybrid preference to semi-single-peakedness.
According to the dual-thresholds $a_p$ and $a_q$ on the line $\mathcal{L}^A$, a semi-hybrid preference whose preference peak is located to the left of $a_p$, is first semi-single-peaked on the line $\mathcal{L}^A$ w.r.t.~the threshold $a_p$, and moreover ranks $a_q$
above all alternatives in the right interval (see Figure \ref{fig:e}). A symmetric condition holds when the peak is located to the right of $a_q$.
A semi-hybrid preference that has the peak located in the middle interval between $a_p$ and $a_q$ is significantly less restrictive, as it only requires the dual-thresholds $a_p$ and $a_q$ to be top-ranked within the left and right intervals respectively (see Figure \ref{fig:f}).

Single-peaked preferences were initially introduced by \citet{B1948} towards resolving the Condorcet paradox and have since become the cornerstone of modern political economy, aggregation theory, social choice theory and mechanism design theory.
Semi-single-peakedness significantly weakens the  restriction of single-peakedness, but suffices for sustaining a strategy-proof rule that satisfies anonymity and the tops-only property in the aforementioned formulations.
Hybrid preferences can naturally arise in a multidimensional strategic voting model, when a voter aggregates her multidimensional assessments over the same set of candidates to formulate a one-dimensional preference \citep[see Section 4 of][]{R2015}, or
when a voter's multidimensional preference is reduced to a one-dimensional preference under voting constraints of \citet{BMN1997}. Hybrid preferences may also emerge naturally in the public good allocation environment. This paper introduces semi-hybrid preferences as a significant weakening of the hybridness requirement and thereby expands the scope of design in these formulations.

We recall and adapt the intuitive introduction of hybrid preferences in the introduction of \citet{CRSSZ2022} to
illustrate the relevance of semi-hybrid preferences in the public good allocation environment.
Imagine a region where an urban zone stands in the center and is surrounded by a large suburban area.
A railway connects two towns in the suburban area, and goes through the central urban zone.
Each station represents a possible location for allocating a good public facility, like a sports complex.
Two particular stations (see for instance $a_p$ and $a_q$ in Figure \ref{fig:e}) in the central urban zone separate the whole railway into three intervals;
all stations lying in the middle interval (e.g., the interval between $a_p$ and $a_q$ in Figure \ref{fig:e}) are in the central urban zone and hence called urban stations, while each station in one of the two other intervals (e.g., the interval left to $a_p$ and the interval right to $a_q$ in Figure \ref{fig:e}) is in the suburban area and hence called a suburban station.
The central urban zone also possesses a modern metro transportation system that fully connects all urban stations and complements the railway, whereas
 all suburban stations are only connected via the railway.
Thus, the two particular stations can be viewed as two transportation hubs that serve as gates to the respective suburban zones.
A citizen living nearby a suburban station formulates her preference over all locations according to the following two principles:
(i) any suburban station beyond the distant transportation hub is not acceptable, and
(ii) all stations between her own location and the distant transportation hub are compared according to the distance measured by both the railway and the metro.
Thus, the citizen's preference must be semi-hybrid on the locations along the railway w.r.t.~the two transportation hubs:
single-peakedness prevails on the interval between her location and the proximate transportation hub,
the proximate transportation hub is ranked above any other urban stations, and
the distant transportation hub is better than all suburban stations beyond (see the preference in Figure \ref{fig:e}).
Similarly, a citizen living in the central urban zone does not accept any suburban station, and has arbitrary preferences on all urban stations that differ from her peak. This indicates that the citizen also has a semi-hybrid preference (see the preference in Figure \ref{fig:f}).

It is well known from the seminal work \citet{M1980} that all anonymous, tops-only and strategy-proof rules on the domain of all single-peaked preferences are characterized to be \emph{phantom voter rules}.
The domain of all semi-single-peaked preferences admits an anonymous, tops-only and strategy-proof rule,
called a \emph{projection rule} (see Definition \ref{def:projection} in Section \ref{sec:rules}), which is indeed  the phantom voter rule that fixes all phantom voters' ballots to be the threshold.
Recently, \citet{CSS2013}  showed that semi-single-peakedness is also necessary for the existence of an anonymous, tops-only and strategy-proof rule on a class of rich domains. We introduce here a fairness property of two-voter rules called \emph{invariance}
which is substantially weaker than anonymity in that it pertains to the behavior of the rule at exactly two preference profiles -
invariance requires the SCF to choose the same social outcome at two test preference profiles where the two voters are endowed with two completely reversed preferences.
We note that on both the domain of all hybrid preferences and the domain of all semi-hybrid preferences, invariance is incompatible with any two-voter, tops-only and strategy-proof rule. Indeed, the relative rankings of alternatives in the middle part are too permissive, and consequently every two-voter, tops-only and strategy-proof rule behaves like a dictatorship on the middle interval (see for instance a \emph{hybrid rule} in Definition \ref{def:hybridrule} of Section \ref{sec:rules}) and is accordingly incompatible with invariance.\footnote{More specifically, consider the domain of all hyrid/semi-hybrid preferences on the line $\mathcal{L}^A$ w.r.t.~the dual-thresholds $a_p$ and $a_q$ in Figure \ref{fig:b} or \ref{fig:e}. It is clear that the domain contains two completely reversed preferences where $a_1$ and $a_m$ are respectively top-ranked. Given a two-voter, tops-only and strategy-proof rule, assume w.l.o.g.~that voter 1 dictates on the middle interval. Then, at the test profile where voter 1's preference peak is $a_1$ and voter 2's peak is $a_m$,
in view of voter 1's dictatorship on the middle interval and strategy-proofness,
the social outcome must lie in the interval between $a_1$ and $a_p$,
while at the other test profile, the social outcome lies in the interval between $a_q$ and $a_m$.
Therefore, the SCF violates invariance.}
We accordingly conclude that preserving single-peakedness  on the middle region  (as is the case with semi-single peaked preferences) is critical for the existence of invariant, tops-only and strategy-proof rules, and provide a characterization of all such rules as projection rules.

We now turn to the second theme.
Note first that the domain of all single-peaked preferences is strictly included in the domain of all semi-single-peaked preferences.
We introduce a notion called a \emph{critical spot} to formally address the flexibility embedded in a domain of semi-single-peaked preferences, and show that its presence is necessary and sufficient in avoiding the reduction of a domain of semi-single-peaked preferences to a domain of single-peaked preferences (see Section \ref{sec:PNT}).
Now, starting with a domain of single-peaked preferences, as we add non-single-peaked preferences to move towards a domain of semi-single-peaked preferences, the set of anonymous, tops-only and strategy-proof rules of course shrinks\footnote{In most cases, only the projection rule in the class of phantom voter rules survives according to the characterization result of \citet{BM2020}.},
but more importantly these non-single-peaked preferences create space for the emergence of strategy-proof rules that utilize non-peak information on voters' preferences.
We show that critical spots turn out to be sufficient for the existence of non-tops-only and strategy-proof rules.
Analogously, critical spots are also embedded in a domain of semi-hybrid preferences, distinguish the domain from a domain of hybrid preferences, and support the design of non-tops-only and strategy-proof rules.
Thus, while semi-single-peaked and semi-hybrid preferences are more permissive (which is desirable for applications of mechanism design),
they do however admit non-tops-only and strategy-proof rules which make the task  confronting the designer, namely that of  characterizing and choosing among rules, correspondingly harder, as a full characterization of strategy-proof rules then depends delicately on how the domain is expanded.

\subsection{An outline of the results}

We conclude with an overview of our methodology and results.

We classify all non-dictatorial domains in a class of unidimensional domains.
These domains satisfy \emph{path-connectedness}\footnote{This is a ``richness'' condition that is formulated as a connectedness property of a graph on the set of alternatives; it is used here and in \citet{CSS2013} as a way of incorporating sufficiently many preferences so as to give  bite to the axiom of strategy-proofness.},
a condition called \emph{diversity} which requires the presence of two completely reversed preferences, and
a technical condition called \emph{leaf symmetry}.\footnote{This is introduced to handle some preferences whose peaks are the leaves of the graph generated according to the condition of path-connectedness.}
We first show that a unidimensional domain is a non-dictatorial domain if and only if it satisfies \emph{the unique seconds property} of \citet{ACS2003} (see the Auxiliary Proposition in Section \ref{sec:auxiliary});
in particular, for any domain satisfying the unique seconds property, we construct a strategy-proof rule that is ``almost'' dictatorial (see Section \ref{sec:auxiliary}).

We subsequently focus on the non-dictatorial domains identified above and explore more meaningful non-dictatorial strategy-proof rules than almost dictatorships.
This indeed requires us to investigate the details of the structure of non-dictatorial domains from a global perspective that goes well beyond the unique seconds property. We proceed by investigating the existence and non-existence of an invariant, tops-only and strategy-proof rule on a non-dictatorial, unidimensional domain.
Theorem \ref{thm:invariance} in Section \ref{sec:domainclassification} provides the following classification of non-dictatorial, unidimensional domains: semi-single-peakedness restriction is necessary and sufficient for the existence of an invariant, tops-only and strategy-proof rule, and
furthermore all such rules are two-voter projection rules,
while semi-hybridness is necessary and sufficient for the non-existence of an invariant, tops-only and strategy-proof rule, and
furthermore every two-voter, tops-only and strategy-proof rule is a hybrid rule that restricted to the middle interval behaves like a dictatorship.
We further use this result to extend the classification to the case of $n$-voter SCFs by replacing invariance by anonymity (see Corollary \ref{cor:anonymity} in Section \ref{sec:domainclassification}). To sum up, the existence of a two-voter, tops-only and strategy-proof rule satisfying invariance on unidimensional domains leads to a comprehensive classification of the design possibilities for such domains.
The resulting classification may be seen as reinforcing the view
that appropriate weakenings of single-peakedness characterize non-dictatorial domains, addressing thereby a long standing conjecture in this field
(see Section 6.5.2 in \citealp{B2011} and Section 4.5 in \citealp{BBM2020}).

Finally, we specialize to tops-only domains, i.e., domains where all strategy-proof rules are endogenously tops-only rules.
It is evident that a tops-only domain can never accommodate a critical spot (which ensures the existence of a non-tops-only and strategy-proof rule). Accordingly, we refine the aforementioned classification by showing that on a non-dictatorial, tops-only, unidimensional domain,
the existence of an anonymous and strategy-proof rule leads us to a classical domain of single-peaked preferences, while its non-existence, to a domain of hybrid preferences (see Corollary \ref{cor:anonymity} in Section \ref{sec:refinement}).

The paper is organized as follows.
In Section \ref{sec:preliminaries}, we specify the model.
In Section \ref{sec:domains}, we formally introduce all preferences domains.
All results are presented in Section \ref{sec:results}, while Section \ref{sec:conclusion} contains a review of the literature and some final remarks.
All proofs are gathered in an Appendix.

\section{Preliminaries}\label{sec:preliminaries}

Let $A=\{a, b, c, \dots\}$ be a finite set of alternatives with $|A| = m \geq 3$.
Let $N=\{1, \dots, n\}$ be a finite set of voters with $|N|=n \geq 2$.
Each voter $i$ has a (strict) preference order $P_{i}$ over $A$ which is a linear order.
For any $a, b \in A$, $a\mathrel{P_{i}}b$ is interpreted as ``$a$ is strictly preferred to $b$ according to $P_{i}$".\footnote{In a table,
we specify a preference ``vertically''. In a sentence, we specify a preference ``horizontally''.
For instance, $P_i=(a\,b\,c\,\cdots)$ represents a preference where $a$ is the top, $b$ is the second best,
$c$ is the third ranked alternative while the rest of rankings in $P_{i}$ are arbitrary.}
Let $r_{k}(P_{i})$ denote the $k$th ranked alternative in $P_{i}$ for all $k \in \{1, \dots, m\}$.
Given a subset $B \subset A$,\footnote{Throughout the paper, $\subset$ and $\subseteq$ denote the strict and weak inclusions respectively.}
let $\max^{P_i}(B)$ and $\min^{P_i}(B)$ respectively denote the most and the least preferred alternatives in $B$ according to $P_i$.
Two preferences $P_i$ and $P_i'$ are \textbf{completely reversed} if for all $a, b \in A$, $[a\mathrel{P_i}b] \Leftrightarrow [b\mathrel{P_i'}a]$.
Let $\mathbb{P}$ denote the set containing all linear orders over $A$.
The set of all admissible orders is a set $\mathbb{D} \subseteq \mathbb{P}$,
referred to as the \textbf{preference domain}.
In particular, we call $\mathbb{P}$ \textbf{the universal domain}. When $\mathbb{D} \neq \mathbb{P}$, $\mathbb{D}$ is referred to as a restricted domain.
For notational convenience,
let $\mathbb{D}^{a} = \{P_{i}\in \mathbb{D}:r_{1}(P_{i}) = a\}$
denote the set of preferences with the peak $a$.
Accordingly, a domain $\mathbb{D}$ is \textbf{minimally rich} if $\mathbb{D}^{a} \neq \emptyset$ for all $a \in A$.
A preference profile $P = (P_1, \ldots, P_n) = (P_{i}, P_{-i}) \in {\mathbb D}^n$ is an $n$-tuple of orders where
$P_{-i}$ represents a collection of $n-1$ voters' preferences without considering voter $i$.

Fixing a domain $\mathbb{D}$,
two alternatives $a, b \in A$ are said \textbf{adjacent}, denoted $a \sim b$, if there exist $P_i, P_i' \in \mathbb{D}$ such that
$r_1(P_i) =r_2(P_i') = a$, $r_1(P_i') = r_2(P_i) = b$ and $r_k(P_i) = r_k(P_i')$ for all $k \in \{3, \dots, m\}$.
Domain $\mathbb{D}$ is called a \textbf{path-connected} domain if for all distinct $a,b \in A$,
there exists a sequence of non-repeated alternatives $(x_1, \dots, x_v)$ such that $x_1 = a$, $x_v = b$ and $x_k \sim x_{k+1}$ for all $k \in \{1, \dots, v-1\}$.
It is evident that the universal domain $\mathbb{P}$ is a path-connected domain as any two distinct alternatives are adjacent.
Clearly, path-connectedness implies minimal richness.
Moreover, $\mathbb{D}$ is said to satisfy \textbf{diversity} if it contains two completely reversed preferences.\footnote{Diversity has been widely presumed in the Condorcet domain literature,
e.g., \citet{M2009} and \citet{P2018}, where it plays a key role in pinning down maximal Condorcet domains.}
Throughout the paper, we fix $\underline{P}_i$ and $\overline{P}_i$ as two completely reversed preferences in a domain satisfying diversity, and moreover let $\underline{P}_i$ and $\overline{P}_i$ be such that
$a_k\mathrel{\underline{P}_i}a_{k+1}$ and $a_{k+1}\mathrel{\overline{P}_i}a_k$ for all $k \in \{ 1, \dots, m-1\}$,
by relabelling alternatives as necessary.

\subsection{Social Choice Functions}

A \textbf{Social Choice Function} (or SCF) is a map $f: \mathbb{D}^{n} \rightarrow A$.
At every profile $P \in \mathbb{D}^{n}$, $f(P)$ is referred to as the ``socially desirable'' outcome associated to this preference profile.
An SCF $f: \mathbb{D}^{n} \rightarrow A$ is \textbf{unanimous}
if for all $a \in A$ and $P \in \mathbb{D}^{n}$, we have $[r_{1}(P_{i}) = a$ for all $i \in N] \Rightarrow [f(P) = a]$.
Henceforth, for simplicity, we call a unanimous SCF a \textbf{rule}.
An SCF $f: \mathbb{D}^{n} \rightarrow A$ is \textbf{strategy-proof}
if for all $i \in N$, $P_{i}, P_{i}' \in \mathbb{D}$ and $P_{-i} \in \mathbb{D}^{n-1}$,
we have either $f(P_i, P_{-i}) = f(P_i', P_{-i})$ or $f(P_i, P_{-i})\mathrel{P_i}f(P_i', P_{-i})$.
A prominent class of SCFs is the class of tops-only SCFs.
The value of these SCFs at every preference profile depends only on voters' peaks.
Formally, an SCF $f: \mathbb{D}^{n} \rightarrow A$ satisfies {\bf the tops-only property}
if for all $P, P' \in \mathbb{D}^{n}$, we have $[r_{1}(P_{i}) = r_{1}(P_{i}')$ for all $i \in N]\Rightarrow [f(P) = f(P')]$.
Last, an SCF $f: \mathbb{D}^{n} \rightarrow A$ is \textbf{anonymous} if for all preference profiles $(P_1, \dots, P_n) \in \mathbb{D}^n$ and
permutations $\sigma: N \rightarrow N$, we have $f(P_1, \dots, P_n) = f\big(P_{\sigma(1)}, \dots, P_{\sigma(n)}\big)$.
In addition, we weaken anonymity to a new axiom called invariance on a two-voter SCF,
which requires that on a domain satisfying diversity,
the SCF chooses the same alternative at the two profiles where the voters are endowed with the two completely reversed preferences.
Formally, given the two completely reversed preferences $\underline{P}_i,\overline{P}_i \in \mathbb{D}$,
a two-voter SCF $f: \mathbb{D}^2 \rightarrow A$ is \textbf{invariant} if we have $f(\underline{P}_1, \overline{P}_2) = f(\overline{P}_1, \underline{P}_2)$.

Dictatorships are rules that are tops-only and strategy-proof on  arbitrary domains.
Formally, an SCF $f: \mathbb{D}^n \rightarrow A$ is a \textbf{dictatorship} if there exists $i \in N$ such that $f(P) = r_1(P_i)$ for all $P \in \mathbb{D}^n$. It is clear that anonymity is polar opposite to a dictatorship, and invariance contrasts dictatorships as well.
Given a nonempty subset $B \subseteq A$, we say that an SCF $f: \mathbb{D}^n \rightarrow A$ \textbf{behaves like a dictatorship on $B$} if there exists $i \in N$ such that $f(P_1, \dots, P_n) = r_1(P_i)$ for all $(P_1, \dots, P_n) \in \mathbb{D}^n$ with $r_1(P_1), \dots, r_1(P_n) \in B$.
The Gibbard-Satterthewaite Theorem shows that on the universal domain $\mathbb{P}$,
an SCF is a strategy-proof rule if and only if it is a dictatorship.
The same dictatorship characterization result also holds on some restricted domains (see the literature listed in Section \ref{sec:literature}).
We call a domain $\mathbb{D}$ a \textbf{dictatorial domain} if every strategy-proof rule $f: \mathbb{D}^n \rightarrow A$, $n\geq 2$, is a dictatorship, and call any domain that admits a non-dictatorial, strategy-proof rule a \textbf{non-dictatorial domain}.
Clearly, a domain that admits an anonymous/invariant and strategy-proof rule is a non-dictatorial domain.
Conversely, a non-dictatorial domain may not admit an anonymous/invariant and strategy-proof rule.

\subsection{Graphs}\label{sec:graph}

Let $G^A = \langle A, \mathcal{E}^A\rangle$ denote an \textbf{undirected graph} where $A$ is the vertex set and $\mathcal{E}^A$ is the set of edges.\footnote{If $(a, b) \in \mathcal{E}^A$, then $a \neq b$ and $(b,a)\in \mathcal{E}^A$.}
Given $x, y \in A$, a \textbf{path} in $G^A$ connecting $x$ and $y$ is a sequence of non-repeated vertices $(x_1, \dots, x_t)$ such that $x_1 = x$, $x_t = y$ and $(x_k, x_{k+1}) \in \mathcal{E}^A$ for all $k \in \{1, \dots, t-1\}$.
The graph $G^A$ is \textbf{connected} if for every pair of distinct vertices, there exists a path connecting them.
In particular, the graph $G^A$ is called a \textbf{complete graph} if any two distinct vertices form an edge, i.e., $(a, b) \in \mathcal{E}^A$ for all distinct $a, b\in A$.
Given $a \in A$, let $\mathcal{N}^{A}(a) = \{b \in A: (a,b) \in \mathcal{E}^A\}$ denote the set of alternatives that are neighbor to $a$ in $G^A$.
Given a graph $G^A$, a vertex $a \in A$ is called a \textbf{leaf} if it has a unique neighbor, i.e., $|\mathcal{N}^A(a)| = 1$.
Accordingly, let $\textrm{Leaf}(G^A) = \big\{x \in A: |\mathcal{N}^{A}(x)| = 1\big\}$ collect all leaves in $G^A$.
Given a subset $B \subset A$, let $G^B = \langle B, \mathcal{E}^{B}\rangle$ denote the subgraph of $G^A$ where the vertex set is $B$ and
the edge set is $\mathcal{E}^B = \{(a, b) \in \mathcal{E}^A: a, b \in B\}$.

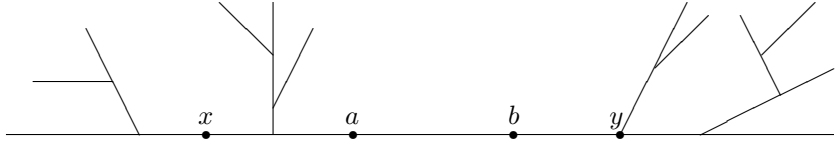
\begin{figure}[t]
\begin{tikzpicture}
\node at (-0.7,1.8) {};

		\put(50, 0){\line(1,0){310}}

		\put(125, 0){\circle*{3}}
		\put(180, 0){\circle*{3}}
        \put(240, 0){\circle*{3}}
		\put(280, 0){\circle*{3}}
		
		\put(122, 4){\footnotesize{$x$}}
		\put(177, 4){\footnotesize{$a$}}
		\put(238, 4){\footnotesize{$b$}}
		\put(276, 4){\footnotesize{$y$}}

        \put(100, 0){\line(-1,2){20}}
        \put(90, 20){\line(-1,0){30}}
        \put(150, 0){\line(0,1){50}}
        \put(150, 10){\line(1,2){15}}
        \put(150, 30){\line(-1,1){20}}

        \put(280, 0){\line(1,2){25}}
        \put(293,25){\line(1,1){20}}
        \put(310, 0){\line(2,1){50}}
        \put(340,15){\line(-1,2){15}}
        \put(333,30){\line(1,1){20}}
\end{tikzpicture}
\bigskip\bigskip\bigskip\bigskip
\caption{A tree $\mathcal{T}^A$}\label{fig:tree}
\end{figure}

A \textbf{tree} $\mathcal{T}^A = \langle A, \mathcal{E}^A\rangle$ is a connected graph where each pair of distinct vertices is connected by a unique path (see Figure \ref{fig:tree}).
A \textbf{line} is a particular tree which has exactly two leaves.
Throughout the paper, we fix $\mathcal{L}^A = (a_1,\dots, a_m)$ to be the line where $a_k$ and $a_{k+1}$ form an edge for all $k \in \{1, \dots, m-1\}$.
Fix a tree $\mathcal{T}^A$. Given $x, y \in A$, let $\langle x, y |\mathcal{T}^A\rangle$ denote the unique path connecting $x$ and $y$ in $\mathcal{T}^A$.\footnote{For notational convenience, we also use $\langle x, y |\mathcal{T}^A\rangle$ to denote the set of alternatives in the path between $x$ and $y$. We also call $\langle x, y |\mathcal{T}^A\rangle$ the interval between $x$ and $y$ in $\mathcal{T}^A$.}
Given a subset $B \subset A$ such that the path between any two alternatives of $B$ is also included in $B$,
i.e., $[a, b \in B]\Rightarrow \big[\langle a, b |\mathcal{T}^A\rangle \subseteq B\big]$, the subgraph $\mathcal{T}^B = \langle B, \mathcal{E}^B\rangle$ is also a tree.
Furthermore, given $a \in A$, if $a \in B$, it is natural to let $\mathop{\textrm{Proj}}(a, \mathcal{T}^B) = a$ denote the \textbf{projection} of $a$ on the subtree $\mathcal{T}^B$; otherwise,
there exists a unique $a' \in B$ such that $a' \in \langle a, b|\mathcal{T}^A\rangle$ for all $b \in B$, and then let $\mathop{\textrm{Proj}}(a, \mathcal{T}^B) = a'$ denote the projection of $a$ on $\mathcal{T}^B$.
Given a preference profile $P$, we construct the set
$\Gamma(P) = \big\{a \in A: a \in \langle r_1(P_i), r_1(P_j)|\mathcal{T}^A\rangle\; \textrm{for some}\; i, j \in N\big\}$
which includes all voters' preference peaks and alternatives that are located between voters' preference peaks.
Thus, $\mathcal{T}^{\Gamma(P)}$ is the \textbf{minimal subtree} nested in $\mathcal{T}^A$ that covers all voters' preference peaks.
Given two distinct alternatives $x, y \in A$,
we fix the set $A^{x \rightharpoonup y} = \big\{z \in A: x \in \langle z, y|\mathcal{T}^A\rangle\big\}$ to include every alternative whose path to $y$ always goes through $x$.
Therefore, $\mathcal{T}^{A^{x \rightharpoonup y}}$ is a subtree nested in $\mathcal{T}^A$.
Given a tree $\mathcal{T}^A$, we consider two distinct alternatives $a$ and $b$ that completely separate $\mathcal{T}^A$
into the middle interval $\langle a, b|\mathcal{T}^A\rangle$ and the two subtrees $\mathcal{T}^{A^{a \rightharpoonup b}}$ and $\mathcal{T}^{A^{b \rightharpoonup a}}$ (see Figure \ref{fig:tree}).
Thus, $\mathop{\textrm{Proj}}(c, \langle a, b|\mathcal{T}^A\rangle) \in \{a, b\}$ for all $c \in A\backslash \langle a, b|\mathcal{T}^A\rangle$.
We call $a$ and $b$ \textbf{dual-thresholds} in $\mathcal{T}^A$ (see Figure \ref{fig:tree}).\footnote{In particular, if $a$ and $b$ form an edge in $\mathcal{T}^A$, they are naturally dual-thresholds.}
According to the vertices $x$ and $y$ in Figure \ref{fig:tree}, we have
the subtree $\mathcal{T}^{A^{x\rightharpoonup y}}$, the interval $\langle x, y|\mathcal{T}^A\rangle$ and the subtree $\mathcal{T}^{A^{y\rightharpoonup x}}$,
which in combination however does not recover the tree $\mathcal{T}^A$ as the branch attached to the interior of the interval $\langle x, y|\mathcal{T}^A\rangle$ is not covered. Hence, $x$ and $y$ are not dual-thresholds.
Also note that $x$ is a leaf of the subtree $\mathcal{T}^{A^{x \rightharpoonup y}}$, whereas $y$ is not a leaf of the subtree $\mathcal{T}^{A^{y \rightharpoonup x}}$.

We conclude this section by adopting the terminology of a connected graph to represent a connected domain.
Given a domain $\mathbb{D}$, we construct a graph $G_{\sim}^A = \langle A, \mathcal{E}_{\sim}^A\rangle$, called an \textbf{adjacency graph},
where the vertex set is $A$, and two alternatives form an edge if and only if they are adjacent, i.e., $\mathcal{E}_{\sim}^A = \{(a,b)\in A^2: a \sim b\}$.
Then, it is clear that $\mathbb{D}$ is a path-connected domain if and only if $G_{\sim}^A$ is a connected graph.
Given $x \in A$, let $\mathcal{N}_{\sim}^A(x) = \{y \in A: (x, y) \in \mathcal{E}_{\sim}^A\}$ collect all neighbors of $x$ in the adjacency graph $G_{\sim}^A$.

\section{Preference Domains and Tops-only Rules}\label{sec:domains}

In this section, we introduce preference domains and rules that constitute our classification.

\subsection{Single-peaked domain and hybrid domain}

We first introduce the classical single-peaked domain and its recently introduced variant called the hybrid domain.\footnote{The idea of the hybrid domain originates from the multiple single-peaked domain of \citet{R2015}. \citet{AR2020} and \citet{CRSSZ2022} establish the formal definition of a hybrid preference on a line, and study strategy-proof rules and random SCFs.
We generalize the notion to trees.}

\begin{definition}[\citet{D1982}]
Fixing a tree $\mathcal{T}^A$, a preference $P_i$ is \textbf{single-peaked} on $\mathcal{T}^A$ if
for all distinct $a, b \in A$, we have $\big[a \in \langle r_1(P_i), b|\mathcal{T}^A\rangle\big]\Rightarrow [a\mathrel{P_i}b]$.
Let $\mathbb{D}_{\emph{SP}}(\mathcal{T}^A)$ denote \textbf{the single-peaked domain} of all single-peaked preferences on $\mathcal{T}^A$.
\end{definition}

The single-peaked domain $\mathbb{D}_{\textrm{SP}}(\mathcal{T}^A)$ is naturally a path-connected domain as its adjacency graph is identical to $\mathcal{T}^A$, and
it satisfies diversity if and only if $\mathcal{T}^A$ is a line.



\newpage

\begin{definition}\label{def:h}
Fixing a tree $\mathcal{T}^A$ and dual-thresholds $a, b \in A$,
a preference $P_i$ is \textbf{$\bm{(a, b)}$-hybrid} on $\mathcal{T}^A$ if it satisfies the following two conditions:
\begin{itemize}
\item[\rm (i)]
for all distinct $y, z \in A^{a\rightharpoonup b}$ or $y, z \in A^{b \rightharpoonup a}$,
$\big[y \in \langle r_1(P_i), z|\mathcal{T}^A\rangle\big] \Rightarrow [y\mathrel{P_i}z]$, and

\item[\rm (ii)] $[r_1(P_i) \in A^{a\rightharpoonup b}\backslash \{a\}] \Rightarrow
[\max^{P_i}(\langle a, b|\mathcal{T}^A\rangle) = a]$ and\\
$[r_1(P_i) \in A^{b\rightharpoonup a}\backslash \{b\}] \Rightarrow [\max^{P_i}(\langle a, b|\mathcal{T}^A\rangle) = b]$.
\end{itemize}
Let $\mathbb{D}_{\emph{H}}(\mathcal{T}^A, a,b)$ denote \textbf{the hybrid domain}
of all $(a, b)$-hybrid preferences on $\mathcal{T}^A$.
\end{definition}

The hybrid domain $\mathbb{D}_{\textrm{H}}(\mathcal{T}^A, a, b)$ is naturally a path-connected domain as its adjacency graph includes $\mathcal{T}^A$ as a subgraph,
and it satisfies diversity if and only if $\mathcal{T}^A$ is a line.
Note that $\mathbb{D}_{\textrm{SP}}(\mathcal{T}^A) \subseteq \mathbb{D}_{\textrm{H}}(\mathcal{T}^A, a, b)$, where the equality holds
when the dual-thresholds $a$ and $b$ form an edge in $\mathcal{T}^A$.
In another extreme circumstance, when $\langle a, b|\mathcal{T}^A\rangle = A$, we have $\mathbb{D}_{\textrm{H}}(\mathcal{T}^A, a, b)= \mathbb{P}$.


\begin{definition}\label{def:asph}
A domain $\mathbb{D}$ is called \textbf{a single-peaked domain} if there exists a tree $\mathcal{T}^A$ such that $\mathbb{D}\subseteq \mathbb{D}_{\emph{SP}}(\mathcal{T}^A)$, and $G_{\sim}^A$ is a connected graph.
A domain $\mathbb{D}$ is called \textbf{a hybrid domain} if the following three conditions are satisfied:
\begin{itemize}
\item[\rm (i)] there exist a tree $\mathcal{T}^A$ and dual-thresholds $a,b\in A$ such that $\mathbb{D}\subseteq \mathbb{D}_{\emph{H}}(\mathcal{T}^A,a,b)$, and $G_{\sim}^A$ is a connected graph,

\item[\rm (ii)] there exist no tree $\widehat{\mathcal{T}}^A$ and dual-thresholds $\hat{a},\hat{b} \in A$ such that
$\mathbb{D} \subseteq \mathbb{D}_{\emph{H}}(\widehat{\mathcal{T}}^A, \hat{a}, \hat{b})$ and
$\langle \hat{a}, \hat{b}\big|\widehat{\mathcal{T}}^A\rangle \subset \langle a,b|\mathcal{T}^A\rangle$\footnote{The notation $\langle \hat{a}, \hat{b}\big|\widehat{\mathcal{T}}^A\rangle \subset \langle a,b|\mathcal{T}^A\rangle$ here only concerns the inclusion relation between the two subsets of alternatives, not the inclusion relation between the two graphs of intervals.}, and

\item[\rm (iii)] $\big|\langle a, b|\mathcal{T}^A\rangle\big| \geq 3$.\footnote{We impose $\big|\langle a, b|\mathcal{T}^A\rangle\big| \geq 3$ to avoid the case that a hybrid domain reduces to a single-peaked domain.}
\end{itemize}
We call the interval $\langle a, b|\mathcal{T}^A\rangle$ here the \textbf{free zone}.
To highlight the tree $\mathcal{T}^A$ and the dual-thresholds $a$ and $b$,
we further call $\mathbb{D}$ \textbf{an $\bm{(a,b)}$-hybrid domain on $\bm{\mathcal{T}^A}$}.
In particular, $\mathbb{D}$ is said to be
\textbf{non-degenerate} if $A^{a \rightharpoonup b} \neq \{a\}$ or $A^{b \rightharpoonup a} \neq \{b\}$, and \textbf{degenerate} otherwise.
\end{definition}

\subsection{Semi-single-peaked domain and semi-hybrid domain}\label{sec:SSP-SH}

Next, we weaken single-peakedness and hybridness to the notions of semi-single-peakedness and semi-hybridness respectively.

We fix an alternative $\bar{x}$ in a tree $\mathcal{T}^{A}$, and call it a \textbf{threshold} in establishing a semi-single-peaked preference.
In a semi-single-peaked preference, the full force of single-peakedness  only prevails on the relative rankings of alternatives located in the interval between the preference peak and the threshold $\bar{x}$,
while an alternative elsewhere is only required to be ranked below its projection on the interval between the preference peak and the threshold $\bar{x}$ (e.g., recall Figure \ref{fig:d}).

\newpage
\begin{definition}[\citet*{CSS2013}]\label{def:ssp}
Fixing a tree $\mathcal{T}^A$ and a threshold $\bar{x} \in A$,
a preference $P_i$ is \textbf{semi-single-peaked} on $\mathcal{T}^A$ w.r.t.~$\bar{x}$ if it satisfies the following two conditions:
\begin{itemize}
\item[\rm (i)] for all distinct $a, b \in \langle r_1(P_i), \bar{x} |\mathcal{T}^A\rangle$, $\big[a \in \langle r_1(P_i), b |\mathcal{T}^A\rangle\big] \Rightarrow [a\mathrel{P_i}b]$, and

\item[\rm (ii)] for all $a \notin \langle r_1(P_i), \bar{x} |\mathcal{T}^A\rangle$,
$\big[\mathop{\emph{Proj}}\big(a, \langle r_1(P_i), \bar{x}|\mathcal{T}^A\rangle\big)=a'\big] \Rightarrow [a'\mathrel{P_i}a]$.
\end{itemize}
Let $\mathbb{D}_{\emph{SSP}}(\mathcal{T}^A, \bar{x})$ denote \textbf{the semi-single-peaked domain}
of all semi-single-peaked preferences on $\mathcal{T}^A$ w.r.t.~$\bar{x}$.
\end{definition}

The semi-single-peaked domain $\mathbb{D}_{\textrm{SSP}}(\mathcal{T}^A, \bar{x})$ is path-connected as its adjacency graph coincides with $\mathcal{T}^A$,
and it satisfies diversity if and only if $|\mathcal{N}^A(\bar{x})| \leq 2$ (see Clarification \ref{cla:ssp} in Appendix \ref{app:clarification}).
Moreover, it is clear that $\mathbb{D}_{\textrm{SP}}(\mathcal{T}^A) = \cap_{\bar{x} \in A}\mathbb{D}_{\textrm{SSP}}(\mathcal{T}^A, \bar{x})$.
\medskip


To establish a semi-hybrid preference, we fix dual-thresholds $a$ and $b$ in a tree $\mathcal{T}^A$.
A semi-hybrid preference, whose preference peak is located in the subtree $\mathcal{T}^{A^{a \rightharpoonup b}}$, is semi-single-peaked on $\mathcal{T}^A$ w.r.t.~$a$, and in addition ranks $b$ above all other alternatives of $A^{b \rightharpoonup a}$ (e.g., recall Figure \ref{fig:e}).
Analogous conditions are imposed on a semi-hybrid preference with the peak located in $\mathcal{T}^{A^{b \rightharpoonup a}}$. Otherwise, the preference peak is located between $a$ and $b$, and then $a$ and $b$ are required to be top-ranked within $A^{a \rightharpoonup b}$ and
$A^{b \rightharpoonup a}$ respectively (e.g., recall Figure \ref{fig:f}).
Thus, a semi-hybrid preference is significantly more permissive than its counterpart hybrid preference which follows the full restriction of single-peakedness  on both subtrees $\mathcal{T}^{A^{a \rightharpoonup b}}$ and $\mathcal{T}^{A^{b \rightharpoonup a}}$ .

\begin{definition}\label{def:sh}
Fixing a tree $\mathcal{T}^A$ and dual-thresholds $a, b \in A$,
a preference $P_i$ is \textbf{$\bm{(a, b)}$-semi-hybrid} on $\mathcal{T}^A$ if it satisfies the following three conditions:
\begin{itemize}
\item[\rm (i)] $\big[r_1(P_i) \in A^{a \rightharpoonup b}\backslash \{a\}\big] \Rightarrow
\left[
\begin{array}{l}
\!P_i \;\textrm{is semi-single-peaked on}\; \mathcal{T}^A\; \textrm{w.r.t.}~a
\; \textrm{and}\\
\!\max^{P_i}(A^{b \rightharpoonup a}) = b
\end{array}
\right]$,

\item[\rm (ii)] $\big[r_1(P_i) \in A^{b \rightharpoonup a}\backslash \{b\}\big] \Rightarrow
\left[
\begin{array}{l}
\!P_i \;\textrm{is semi-single-peaked on}\; \mathcal{T}^A\; \textrm{w.r.t.}~b
\; \textrm{and}\\
\!\max^{P_i}(A^{a \rightharpoonup b}) = a
\end{array}
\right]$, and

\item[\rm (iii)] $\big[r_1(P_i) \in \langle a, b |\mathcal{T}^A\rangle\big]\Rightarrow
\big[\max^{P_i}(A^{a \rightharpoonup b}) = a\; \textrm{and}\; \max^{P_i}(A^{b \rightharpoonup a}) = b\big]$.
\end{itemize}
Let $\mathbb{D}_{\emph{SH}}(\mathcal{T}^A,a,b)$ denote \textbf{the semi-hybrid domain}
of all $(a, b)$-semi-hybrid preferences on $\mathcal{T}^A$.
\end{definition}

The semi-hybrid domain $\mathbb{D}_{\textrm{SH}}(\mathcal{T}^A, a, b)$ is a path-connected domain as its adjacency graph includes $\mathcal{T}^A$ as a subgraph.
More specifically, the adjacency graph of $\mathbb{D}_{\textrm{SH}}(\mathcal{T}^A, a,b)$
is a combination of the adjacency subgraph $G_{\sim}^{A^{a \rightharpoonup b}}$, which coincides with the subtree $\mathcal{T}^{A^{a \rightharpoonup b}}$, the adjacency subgraph over the set $\langle a, b|\mathcal{T}^A\rangle$, denoted $G_{\sim}^{\langle a, b|\mathcal{T}^A\rangle}$,
which is a complete subgraph, and
the adjacency subgraph $G_{\sim}^{A^{b \rightharpoonup a}}$, which coincides with the subtree $\mathcal{T}^{A^{b \rightharpoonup a}}$
(see Clarification \ref{cla:sh} in Appendix \ref{app:clarification}).
The semi-hybrid domain $\mathbb{D}_{\textrm{SH}}(\mathcal{T}^A, a, b)$ satisfies diversity if and only if we have
$[A^{a \rightharpoonup b} \neq \{a\}] \Rightarrow [a \in \textrm{Leaf}(\mathcal{T}^{A^{a \rightharpoonup b}})]$ and
$[A^{b \rightharpoonup a} \neq \{b\}] \Rightarrow [b \in \textrm{Leaf}(\mathcal{T}^{A^{b \rightharpoonup a}})]$ (see Clarification \ref{cla:sh} in Appendix \ref{app:clarification}).
Note that when $\big|\langle a, b|\mathcal{T}^A\rangle\big| = 2$, we have $\mathbb{D}_{\textrm{SH}}(\mathcal{T}^A, a, b)=\mathbb{D}_{\textrm{SSP}}(\mathcal{T}^A, a)\cap \mathbb{D}_{\textrm{SSP}}(\mathcal{T}^A, b)$, and
when $\langle a, b|\mathcal{T}^A\rangle = A$, all three conditions in Definition \ref{def:sh} become ineffective and impose no restriction on the preference, and consequently the semi-hybrid domain expands to the universal domain, i.e., $\mathbb{D}_{\textrm{SH}}(\mathcal{T}^A, a, b)= \mathbb{P}$.\footnote{This indicates that not all the semi-hybrid domains are non-dictatorial domains.}
Moreover, it is obvious that $\mathbb{D}_{\textrm{H}}(\mathcal{T}^A, a, b) \subseteq \mathbb{D}_{\textrm{SH}}(\mathcal{T}^A, a, b)$, where the equality holds if and only if
$|A^{a \rightharpoonup b}| \leq 2$ and $|A^{b \rightharpoonup a}| \leq 2$.

\begin{definition}\label{def:asspsh}
A domain $\mathbb{D}$ is called \textbf{a semi-single-peaked domain} if there exist a tree $\mathcal{T}^A$ and a threshold $\bar{x} \in A$ such that
$\mathbb{D} \subseteq \mathbb{D}_{\emph{SSP}}(\mathcal{T}^A, \bar{x})$, and $G_{\sim}^A$ is a connected graph.
A domain $\mathbb{D}$ is called \textbf{a semi-hybrid domain} if
the following three conditions are satisfied:
\begin{itemize}
\item[\rm (i)] there exist a tree $\mathcal{T}^A$ and dual-thresholds $a,b\in A$ such that
$\mathbb{D}\subseteq \mathbb{D}_{\emph{SH}}(\mathcal{T}^A,a,b)$, and $G_{\sim}^A$ is a connected graph,\footnote{This condition implies that the adjacency graph $G_{\sim}^A$ is a combination of $G_{\sim}^{A^{a \rightharpoonup b}}$ that coincides with the subtree $\mathcal{T}^{A^{a \rightharpoonup b}}$,
the connected adjacency subgraph $G_{\sim}^{\langle a, b |\mathcal{T}^A\rangle}$ that may be different from the interval $\langle a, b|\mathcal{T}^A\rangle$ in $\mathcal{T}^A$ (see Example \ref{exm:non-trivialness}),
and $G_{\sim}^{A^{b \rightharpoonup a}}$ that coincides with the subtree $\mathcal{T}^{A^{b \rightharpoonup a}}$.}

\item[\rm (ii)] there exist no tree $\widehat{\mathcal{T}}^A$ and dual-thresholds $\hat{a},\hat{b} \in A$ such that
$\mathbb{D} \subseteq \mathbb{D}_{\emph{SH}}(\widehat{\mathcal{T}}^A, \hat{a}, \hat{b})$ and
$\langle \hat{a}, \hat{b}\big|\widehat{\mathcal{T}}^A\rangle \subset \langle a,b|\mathcal{T}^A\rangle$,\footnote{Note that any arbitrary domain is contained in $\mathbb{D}_{\textrm{SH}}(\mathcal{L}^A, a_1, a_m)$.
Condition (ii) is not content with information delivered by $(a_1,a_m)$-semi-hybridness on $\mathcal{L}^A$, but seeks to
push the middle interval to its minimal form to reveal the key preference restrictions via the notion of semi-hybridness, which in return guides the design of strategy-proof rules.} and

\item[\rm (iii)] if $G_{\sim}^A$ is a tree, then for each $x \in \emph{Leaf}\big(G_{\sim}^{\langle a, b|\mathcal{T}^A\rangle}\big)$,
there exists a preference $P_i\in \mathbb{D}$ such that $P_i$ is not semi-single-peaked on $G_{\sim}^A$ w.r.t.~$x$.\footnote{Note that if $G_{\sim}^A$ is a tree, condition (i) implies that the subgraph $G_{\sim}^{\langle a, b|\mathcal{T}^A\rangle}$ is also a tree.
This condition implies $|\langle a, b|\mathcal{T}^A\rangle| \geq 3$. More importantly, it in conjunction with condition (ii) ensures that a semi-hybrid domain satisfying diversity is never a semi-single-peaked domain (see Lemma \ref{lem:nssp} in Appendix \ref{app:Theorem}).}
\end{itemize}
We call the interval $\langle a, b|\mathcal{T}^A\rangle$ here the \textbf{free zone}.
To highlight the tree $\mathcal{T}^A$ and the dual-thresholds $a$ and $b$,
we further call $\mathbb{D}$ \textbf{an $\bm{(a,b)}$-semi-hybrid domain on $\mathcal{T}^A$}.
In particular, $\mathbb{D}$ is said to be \textbf{non-degenerate} if $A^{a \rightharpoonup b} \neq \{a\}$ or $A^{b \rightharpoonup a} \neq \{b\}$, and \textbf{degenerate} otherwise.
\end{definition}

We provide an example to illustrate a semi-hybrid domain.

\begin{table}[t]
\centering
\begin{tabular}{ccccccccccccccc}
$P_1$ & $P_2$ & $P_3$ & $P_4$ & $P_5$ & $P_6$ &  $P_7$ & $P_8$ & $P_9$ & $P_{10}$ & $P_{11}$ & $P_{12}$ \\
$a_1$ & $a_1$ & $a_1$ & $a_2$ & $a_2$ & $a_3$ &  $a_4$ & $a_4$ & $a_4$ & $a_5$    & $a_5$    & $a_6$    \\[-0.2em]
$a_2$ & $a_2$ & $a_2$ & $a_1$ & $a_4$ & $a_4$ &  $a_2$ & $a_3$ & $a_5$ & $a_4$    & $a_6$    & $a_5$    \\[-0.2em]
$a_3$ & $a_5$ & $a_4$ & $a_4$ & $a_1$ & $a_2$ &  $a_1$ & $a_2$ & $a_3$ & $a_3$    & $a_4$    & $a_4$    \\[-0.2em]
$a_4$ & $a_4$ & $a_3$ & $a_3$ & $a_3$ & $a_1$ &  $a_3$ & $a_1$ & $a_2$ & $a_2$    & $a_3$    & $a_3$    \\[-0.2em]
$a_5$ & $a_3$ & $a_5$ & $a_5$ & $a_5$ & $a_5$ &  $a_5$ & $a_5$ & $a_6$ & $a_6$    & $a_2$    & $a_2$    \\[-0.2em]
$a_6$ & $a_6$ & $a_6$ & $a_6$ & $a_6$ & $a_6$ &  $a_6$ & $a_6$ & $a_1$ & $a_1$    & $a_1$    & $a_1$
\end{tabular}
\caption{Domain $\mathbb{D}_1$}\label{tab:non-trivial}
\end{table}

\begin{figure}[t]
\begin{tikzpicture}

\node at (-2.9, 3) {};

\put(-68, 100){\line(1,0){175}}
\put(-68, 100){\circle*{3}}
\put(-33, 100){\circle*{3}}
\put(2, 100){\circle*{3}}
\put(37, 100){\circle*{3}}
\put(72, 100){\circle*{3}}
\put(107, 100){\circle*{3}}

\put(-72, 91){\footnotesize{$a_1$}}
\put(-37, 91){\footnotesize{$a_2$}}
\put(-2, 91){\footnotesize{$a_3$}}
\put(33, 91){\footnotesize{$a_4$}}
\put(68, 91){\footnotesize{$a_5$}}
\put(103, 91){\footnotesize{$a_6$}}

\node at (4.8, 3.5) {\small{$\mathcal{L}^A$}};

\put(195, 100){\line(1,0){105}}
\put(195, 100){\circle*{3}}
\put(230, 100){\circle*{3}}
\put(265, 100){\circle*{3}}
\put(300, 100){\circle*{3}}

\put(191, 91){\footnotesize{$a_2$}}
\put(226, 91){\footnotesize{$a_3$}}
\put(261, 91){\footnotesize{$a_4$}}
\put(296, 91){\footnotesize{$a_5$}}

\node at (12.2, 3.5) {\small{$\langle a_2, a_5|\mathcal{L}^A\rangle$}};

		\put(-61, 50){\line(1,0){160}}
		\put(-61, 50){\circle*{3}}
		\put(-21, 50){\circle*{3}}
        \put(19,  50){\circle*{3}}
        \put(59,  50){\circle*{3}}
		\put(99,  50){\circle*{3}}
		
		\put(-65, 41){\footnotesize{$a_1$}}
		\put(-25, 41){\footnotesize{$a_2$}}
		\put(15,   41){\footnotesize{$a_4$}}
		\put(55,  41){\footnotesize{$a_5$}}
		\put(95,  41){\footnotesize{$a_6$}}

        \put(19, 70){\circle*{3}}
		\put(22, 68){\footnotesize{$a_3$}}
        \put(19, 70){\line(0,-1){20}}

        \node at (4.8, 1.8) {\small{$G_{\sim}^A$}};

        \put(209, 50){\line(1,0){80}}
		\put(209, 50){\circle*{3}}
        \put(249, 50){\circle*{3}}
        \put(289, 50){\circle*{3}}
		
		\put(205, 41){\footnotesize{$a_2$}}
		\put(245, 41){\footnotesize{$a_4$}}
		\put(285, 41){\footnotesize{$a_5$}}

        \put(249, 70){\circle*{3}}
		\put(252, 68){\footnotesize{$a_3$}}
        \put(249, 70){\line(0,-1){20}}

        \node at (12.3, 1.8) {\small{$G_{\sim}^{\langle a_2, a_5|\mathcal{L}^A\rangle}$}};
\end{tikzpicture}
\caption{Line, interval, adjacency graph and adjacency subgraph}\label{fig:connectedness}
\end{figure}
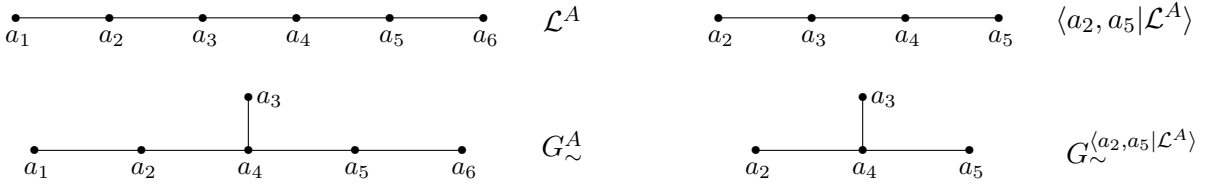

\begin{example}\label{exm:non-trivialness}\rm
Let $A = \{a_1, a_2, a_3, a_4, a_5, a_6\}$.
We specify a domain $\mathbb{D}_1$ of 12 preferences in Table \ref{tab:non-trivial} that are all $(a_2, a_5)$-semi-hybrid on the line $\mathcal{L}^A$.
The line $\mathcal{L}^A$, interval $\langle a_2, a_5|\mathcal{L}^A\rangle$, adjacency graph $G_{\sim}^A$  and adjacency subgraph $G_{\sim}^{\langle a_2, a_5|\mathcal{L}^A\rangle}$ are all specified in Figure \ref{fig:connectedness}, respectively.
One can immediately notice the difference between the
interval $\langle a_2, a_5|\mathcal{L}^A\rangle$ and the adjacency subgraph $G_{\sim}^{\langle a_2, a_5|\mathcal{L}^A\rangle}$ in Figure \ref{fig:connectedness}.

Domain $\mathbb{D}_1$ is a path-connected domain as indicated by the adjacency graph $G_{\sim}^A$ of Figure \ref{fig:connectedness}, and
satisfies diversity according to the preferences $P_1$ and $P_{12}$ in Table \ref{tab:non-trivial}.
Next, we check whether $\mathbb{D}_1$ is an $(a_2, a_5)$-semi-hybrid domain on $\mathcal{L}^A$.
Clearly, condition (i) of Definition \ref{def:sh} is satisfied: $\mathbb{D}_1 \subset \mathbb{D}_{\textrm{SH}}(\mathcal{L}^A,a_2, a_5)$, and $G_{\sim}^A$ is a connected graph.
We next claim that condition (ii) of Definition \ref{def:sh} holds.
Suppose by contradiction that there exist a tree $\widehat{\mathcal{T}}^A$ and dual-thresholds $\hat{a}, \hat{b} \in A$ such that
$\mathbb{D}_1 \subseteq \mathbb{D}_{\textrm{SH}}(\widehat{\mathcal{T}}^A, \hat{a}, \hat{b})$ and $\langle \hat{a}, \hat{b}|\widehat{\mathcal{T}}^A\rangle \subset \langle a_2, a_5|\mathcal{L}^A\rangle$.
Let $\hat{A}^{\hat{a} \rightharpoonup \hat{b}} = \big\{a \in A: \hat{a} \in \langle a, \hat{b}|\widehat{\mathcal{T}}^A\rangle\big\}$ and
$\hat{A}^{\hat{b} \rightharpoonup \hat{a}} = \big\{a \in A: \hat{b} \in \langle a, \hat{a}|\widehat{\mathcal{T}}^A\rangle\big\}$.
According to the connected graph $G_{\sim}^A$ in Figure \ref{fig:connectedness} and the contradictory hypothesis,
we know that $G_{\sim}^{\langle \hat{a}, \hat{b}|\widehat{\mathcal{T}}^A\rangle}$ must be a connected graph and strictly nested in $G_{\sim}^{\langle a_2, a_5|\mathcal{L}^A\rangle}$ in Figure \ref{fig:connectedness}.
Then, one of the following five cases must occur:\\
(1) $\hat{a} = a_2$, $\hat{b} = a_4$ and $G_{\sim}^{\langle \hat{a}, \hat{b}|\widehat{\mathcal{T}}^A\rangle} = (a_2, a_4, a_3)$,\footnote{The case $\hat{a} = a_4$, $\hat{b} = a_2$ and $G_{\sim}^{\langle \hat{a}, \hat{b}|\widehat{\mathcal{T}}^A\rangle} = (a_2, a_4, a_3)$ is symmetric, and therefore is omitted.}\\
(2) $\hat{a} = a_4$, $\hat{b} = a_5$ and $G_{\sim}^{\langle \hat{a}, \hat{b}|\widehat{\mathcal{T}}^A\rangle} = (a_3, a_4, a_5)$,\\
(3) $\hat{a} = a_2$, $\hat{b} = a_4$ and $G_{\sim}^{\langle \hat{a}, \hat{b}|\widehat{\mathcal{T}}^A\rangle} = (a_2, a_4)$,\\
(4) $\hat{a} = a_4$, $\hat{b} = a_5$ and $G_{\sim}^{\langle \hat{a}, \hat{b}|\widehat{\mathcal{T}}^A\rangle} = (a_4, a_5)$, and\\
(5) $\hat{a} = a_4$, $\hat{b} = a_3$ and $G_{\sim}^{\langle \hat{a}, \hat{b}|\widehat{\mathcal{T}}^A\rangle} = (a_4, a_3)$.\\
In each case, note that $a_1 \in \hat{A}^{\hat{a} \rightharpoonup \hat{b}}\backslash \{\hat{a}\}$.
Then,  by $(\hat{a}, \hat{b})$-semi-hybridness on $\widehat{\mathcal{T}}^A$, we know that in each one of the first four cases, $a_4$ must be ranked above $a_5$ in every preference with the peak $a_1$, which contradicts $P_2$ in Table \ref{tab:non-trivial},
while in the fifth case, $a_4$ must be ranked above $a_3$ in every preference with the peak $a_1$, which contradicts $P_1$ in Table \ref{tab:non-trivial}.
Last, we notice that condition (iii) of Definition \ref{def:sh} is violated, since $G_{\sim}^A$ is a tree,
$a_2$ is a leaf of $G_{\sim}^{\langle a_2, a_5|\mathcal{L}^A\rangle}$, and
all preferences in Table \ref{tab:non-trivial} are semi-single-peaked on $G_{\sim}^A$ w.r.t.~$a_2$.
Note that if we add a preference $P_{13} = (a_5\,a_3\,a_2\,a_1\,a_4\,a_6)$,
which is also $(a_2, a_5)$-semi-hybrid on $\mathcal{L}^A$,
the new domain $\widehat{\mathbb{D}}_1 = \mathbb{D}_1\cup \{P_{13}\}$ satisfies condition (iii) of Definition \ref{def:sh}, and hence
becomes an $(a_2, a_5)$-semi-hybrid domain on $\mathcal{L}^A$.\footnote{The adjacency graph of $\widehat{\mathbb{D}}_1$ is also the tree $G_{\sim}^A$ in Figure \ref{fig:connectedness}.
Conditions (i) and (ii) of Definition \ref{def:sh} continue to hold for domain $\widehat{\mathbb{D}}_1$,
while condition (iii) becomes valid, i.e., given $\textrm{Leaf}(G_{\sim}^{\langle a_2, a_5|\mathcal{L}^A\rangle}) = \{a_2, a_3, a_5\}$,
$P_{13}$ is not semi-single-peaked on $G_{\sim}^A$ w.r.t.~$a_2$, and
$P_1$ is not semi-single-peaked on $G_{\sim}^A$ w.r.t.~$a_3$ or $a_5$.}
\hfill$\Box$
\end{example}

\subsection{Projection rule and hybrid rule}\label{sec:rules}

In this section, we introduce two specific tops-only rules that are related to our investigation.

The first tops-only rule is called the projection rule.
Given a tree $\mathcal{T}^A$, we fix an alternative $\bar{x} \in A$.
Taking an arbitrary preference profile $P$ as an example, we first identify the minimal subtree $\mathcal{T}^{\Gamma(P)}$ that covers all preference peaks. Then, the projection of $\bar{x}$ on $\mathcal{T}^{\Gamma(P)}$ is selected by the projection rule as the social outcome.

\begin{definition}\label{def:projection}
An SCF $f: \mathbb{D}^n \rightarrow A$ is \textbf{a projection rule} if
there exist a tree $\mathcal{T}^A$ and an alternative $\bar{x} \in A$ such that  for all $P \in \mathbb{D}^n$,
\begin{align*}
f(P) = \mathop{\emph{Proj}}(\bar{x}, \mathcal{T}^{\Gamma(P)}).
\end{align*}
To highlight the tree $\mathcal{T}^A$ and the alternative $\bar{x}$,
we further call $f$ \textbf{the projection rule} on $\mathcal{T}^A$ w.r.t.~$\bar{x}$.
\end{definition}

By the sufficiency part of the Theorem of \citet{CSS2013}, we know that
given a tree $\mathcal{T}^A$ and a threshold $\bar{x} \in A$,
the semi-single-peaked domain $\mathbb{D}_{\textrm{SSP}}(\mathcal{T}^A, \bar{x})$ admits
the projection rule on $\mathcal{T}^A$ w.r.t.~$\bar{x}$ as an anonymous, tops-only and strategy-proof rule.
Furthermore, Corollary 1 of \citet{BM2020} implies that
when an additional condition is imposed on the location of the threshold, that is, $\bar{x}$ is never a neighbor to a leaf of $\mathcal{T}^A$, i.e., $\mathcal{N}^A(\bar{x})\cap \textrm{Leaf}(\mathcal{T}^A) = \emptyset$,
the projection rule on $\mathcal{T}^A$ w.r.t.~$\bar{x}$ is the \emph{unique} anonymous, tops-only and strategy-proof SCF admitted by the semi-single-peaked domain $\mathbb{D}_{\textrm{SSP}}(\mathcal{T}^A, \bar{x})$.\footnote{Further investigation on projection rules can be found in \citet{T1993} and \citet{V1999}.}

The second tops-only rule is the hybrid rule, which can be viewed as a variant of the projection rule.
Given a tree $\mathcal{T}^A$, we fix some dual-thresholds $a, b \in A$.
Moreover, a particular voter $i \in N$ is fixed in advance.
Taking an arbitrary preference profile $P$ as an example, the hybrid rule first detects whether voter $i$'s preference peak $r_1(P_i)$ is located in the interval $\langle a, b|\mathcal{T}^A\rangle$, or the subtree $\mathcal{T}^{A^{a \rightharpoonup b}}$, or the subtree $\mathcal{T}^{A^{b \rightharpoonup a}}$.
Then, in the first case, $r_1(P_i)$ is selected by the hybrid rule as the social outcome, while
in the second case (respectively, the third case), the social outcome changes to the projection of $a$ (respectively, $b$) on the minimal subtree $\mathcal{T}^{\Gamma(P)}$.

\begin{definition}\label{def:hybridrule}
An SCF $f: \mathbb{D}^n \rightarrow A$ is \textbf{a hybrid rule} if
there exist a tree $\mathcal{T}^A$, dual-thresholds $a,b \in A$ with $|\langle a, b|\mathcal{T}^A\rangle| \geq 3$, and a voter $i \in N$ such that
for all $P \in \mathbb{D}^n$,
\begin{align*}
f(P) = \left\{
\begin{array}{ll}
r_1(P_i) & \emph{if}\; r_1(P_i) \in \langle a, b|\mathcal{T}^A\rangle,\\[0.4em]
\mathop{\emph{Proj}}\big(a, \mathcal{T}^{\Gamma(P)}\big) & \emph{if}\; r_1(P_i) \in A^{a \rightharpoonup b}\backslash \{a\},\;\textrm{and}\\[0.4em]
\mathop{\emph{Proj}}\big(b, \mathcal{T}^{\Gamma(P)}\big) & \emph{if}\; r_1(P_i) \in A^{b \rightharpoonup a}\backslash \{b\}.
\end{array}
\right.
\end{align*}
To highlight the tree $\mathcal{T}^A$ and the dual-thresholds $a$ and $b$,
we further call $f$ \textbf{an $\bm{(a,b)}$-hybrid rule} on $\mathcal{T}^A$.
\end{definition}

It is easy to show that given a tree $\mathcal{T}^A$ and dual-thresholds $a,b\in A$,
an $(a,b)$-hybrid rule on $\mathcal{T}^A$ is a tops-only and strategy-proof rule on a domain
$\mathbb{D} \subseteq \mathbb{D}_{\textrm{SH}}(\mathcal{T}^A, a, b)$ (see Clarification \ref{cla:hybridrule} in Appendix \ref{app:clarification}).\footnote{Proposition 3 of \citet{CZ2020} characterizes all tops-only and strategy-proof rules on the semi-hybrid domain $\mathbb{D}_{\textrm{SH}}(\mathcal{T}^A, a, b)$, given $|A^{a \rightharpoonup b}| \neq 2$ and $|A^{b \rightharpoonup a}| \neq 2$.}
Clearly, an $(a,b)$-hybrid rule on a tree $\mathcal{T}^A$ behaves like a dictatorship on the interval $\langle a, b|\mathcal{T}^A\rangle$, and
it becomes a dictatorship when the interval $\langle a, b|\mathcal{T}^A\rangle$ expands to the whole alternative set.
Therefore, an $(a,b)$-hybrid rule on $\mathcal{T}^A$ is never anonymous.
Moreover, note that the two peaks of the completely reversed preferences, by the restriction of $(a,b)$-semi-hybridness on a tree $\mathcal{T}^A$, can never be both in $A^{a \rightharpoonup b}$ (or both in $A^{b \rightharpoonup a}$).
Consequently, a two-voter $(a,b)$-hybrid rule on $\mathcal{T}^A$ must choose distinct social outcomes at the two profiles where the two voters are endowed with the two completely reversed preferences, and hence violates invariance.

\section{Results}\label{sec:results}

\subsection{Non-dictatorial unidimensional domains}\label{sec:auxiliary}

In this section, we provide a complete characterization of non-dictatorial domains.
We first introduce some notation and an additional richness condition.

Fix a domain $\mathbb{D}$.
Given $a \in A$,
let $\mathcal{S}(\mathbb{D}^a) = \{b \in A: b=r_2(P_i)\; \textrm{for some}\; P_i \in \mathbb{D}^a\}$
collect all alternatives that are second ranked in the preferences of $\mathbb{D}^a$.
Given a leaf of the adjacency graph $G_{\sim}^A$, say $x \in \textrm{Leaf}(G_{\sim}^A)$,
it is evident that $|\mathcal{S}(\mathbb{D}^x)| \geq 1$.
More specifically, we know that either $\mathcal{S}(\mathbb{D}^x)$ is a singleton set of $x$'s unique neighbor in $G_{\sim}^A$, which implies $|\mathcal{S}(\mathbb{D}^x)| =1$,
or $\mathcal{S}(\mathbb{D}^x)$ contains some alternative other than the unique neighbor of $x$, which implies  $|\mathcal{S}(\mathbb{D}^x)| > 1$.
We then introduce a technical richness condition called \textit{leaf symmetry} to handle each leaf in the second case.
Formally, domain $\mathbb{D}$ is said to satisfy \textbf{leaf symmetry}
if for each $x \in \textrm{Leaf}(G_{\sim}^{A})$ with $|\mathcal{S}(\mathbb{D}^x)| > 1$,
there exists $z \in \mathcal{S}(\mathbb{D}^x)$ such that $z \notin \mathcal{N}_{\sim}^A(x)$ and $x \in \mathcal{S}(\mathbb{D}^z)$.\footnote{If $\textrm{Leaf}(G_{\sim}^A) = \emptyset$, or $\textrm{Leaf}(G_{\sim}^A) \neq \emptyset$ and $|\mathcal{S}(\mathbb{D}^x)| = 1$ for all $x \in \textrm{Leaf}(G_{\sim}^A)$,
domain $\mathbb{D}$ satisfies leaf symmetry vacuously.
Under leaf symmetry, given $x \in \textrm{Leaf}(G_{\sim}^{A})$,
$\mathcal{N}_{\sim}^A(x) = \{y\}$ and $|\mathcal{S}(\mathbb{D}^x)|>1$,
we have some $z \in A\backslash \{x, y\}$ and $P_i, P_i' \in \mathbb{D}$ such that
$r_1(P_i) = r_2(P_i) = x$ and $r_1(P_i') = r_2(P_i) = z$.
This indicates that $x$ and $z$ form an edge analogous to, but weaker than the edge of adjacency.
Furthermore, by path-connectedness, we have a path $(x_1, \dots, x_v)$ in $G_{\sim}^A$ connecting $x$ and $z$.
Since $\mathcal{N}_{\sim}^A(x) = \{y\}$, it must be the case that $y = x_2$.
Thus, by combining the path $(x_1, \dots, x_v)$ and the ``weaker edge'' between $x$ and $z$,
we formulate a cycle, which, analogous to a circular domain of \citet{S2010}, ensures that all strategy-proof rules behave like dictatorships on the set $\{x, y, z\}$ (see Lemma \ref{lem:dictatorship1} in Appendix \ref{app:AP}).}
Henceforth, we call a domain a \textbf{unidimensional domain} if it satisfies path-connectedness, diversity and leaf symmetry.

\begin{remark}\rm
A single-peaked domain satisfies leaf symmetry vacuously.\footnote{Given a single-peaked domain $\mathbb{D} \subseteq \mathbb{D}_{\textrm{SP}}(\mathcal{T}^A)$, since $G_{\sim}^A = \mathcal{T}^A$, the restriction of single-peakedness implies $|\mathcal{S}(\mathbb{D}^x)| = 1$ for all $x \in \textrm{Leaf}(G_{\sim}^A)$.}
A hybrid/semi-hybrid domain satisfies leaf symmetry vacuously if its adjacency subgraph on the free zone has no leaf.\footnote{Given an $(a,b)$-hybrid/semi-hybrid domain $\mathbb{D}$ on a tree $\mathcal{T}^A$, let $\textrm{Leaf}(G_{\sim}^A) \neq \emptyset$ and $\textrm{Leaf}\big(G_{\sim}^{\langle a, b|\mathcal{T}^A\rangle}\big)=\emptyset$. Clearly, $A \neq \langle a, b|\mathcal{T}^A\rangle$.
Thus, we have one of the following three cases: (i) $A^{a \rightharpoonup b} \neq \{a\}$ and $A^{b \rightharpoonup a} = \{b\}$ which imply $\textrm{Leaf}(G_{\sim}^A) = \textrm{Leaf}(\mathcal{T}^A)\backslash \{b\}$,
(ii) $A^{a \rightharpoonup b} = \{a\}$ and $A^{b \rightharpoonup a} \neq \{b\}$ which imply $\textrm{Leaf}(G_{\sim}^A) = \textrm{Leaf}(\mathcal{T}^A)\backslash \{a\}$,
or (iii) $A^{a \rightharpoonup b} \neq \{a\}$ and $A^{b \rightharpoonup a} \neq \{b\}$ which imply $\textrm{Leaf}(G_{\sim}^A) = \textrm{Leaf}(\mathcal{T}^A)$.
Then, by the restriction of $(a,b)$-hybridness/semi-hybridness on $\mathcal{T}^A$,
we have $|\mathcal{S}(\mathbb{D}^x)| = 1$ for all $x \in \textrm{Leaf}(G_{\sim}^A)$.
In Example \ref{exm:indispensability} behind, we provide an example of a semi-hybrid domain that violates leaf symmetry.}
A semi-single-peaked domain $\mathbb{D}\subseteq \mathbb{D}_{\textrm{SSP}}(\mathcal{T}^A, \bar{x})$
satisfies leaf symmetry if and only if either $\bar{x} \notin \textrm{Leaf}(\mathcal{T}^A)$,
or $\bar{x} \in \textrm{Leaf}(\mathcal{T}^A)$ and $\mathbb{D}\subseteq \mathbb{D}_{\textrm{SSP}}(\mathcal{T}^A, \bar{x})\cap \mathbb{D}_{\textrm{SSP}}(\mathcal{T}^A, x)$ where $\mathcal{N}^A(\bar{x}) = \{x\}$
(see Clarification \ref{cla:leafsymmetry} in Appendix \ref{app:clarification}).
\end{remark}

\begin{remark}\label{rem:richness}\rm
Many preference domains studied in the literature are unidimensional domains, e.g., the universal domain, some linked domains of \citet{ACS2003}\footnote{\citet{ACS2003} introduced a notion between two alternatives,
which we call \emph{weak adjacency}. Formally, two alternatives $a$ and $b$ are said \textbf{weakly adjacent}, denoted $a \nsim b$, if there exist $P_i, P_i' \in \mathbb{D}$ such that $r_1(P_i) = r_2(P_i') = a$ and $r_1(P_i') = r_2(P_i) = b$. It is clear that weak adjacency is significantly less demanding than the notion of adjacency. Accordingly, a domain $\mathbb{D}$ is said \textbf{weakly path-connected}, if the weak adjacency graph $G_{\nsim}^A = \langle A, \mathcal{E}_{\nsim}^A\rangle$, where two alternatives formulate an edge if and only if they are weakly adjacent, is a connected graph. It is clear that $G_{\sim}^A \subseteq G_{\nsim}^A$. \citet{ACS2003} showed that a \emph{linked domain}, i.e., all alternatives are able to be relabeled as $a_1, \dots, a_m$ such that $a_1 \nsim a_2$ and for each $k \in \{3, \dots, m\}$, $a_k \nsim a_s$ and $a_k \nsim a_t$ for some distinct $s, t \in \{1, \dots, k-1\}$, is a dictatorial domain. Indeed, the weak adjacency graph of a linked domain is a connected graph, has no leaf, and contains at least $2m-3$ edges.\label{footnote:weakadjacency}},
the single-peaked domain of \citet{B1948} and \citet{M1980},
the single-crossing domain of \citet{S2009} and the multiple single-peaked domains of \citet{R2015}.
More generally,  the class of no-restoration domains of \citet{S2013} that satisfies minimal richness and diversity are unidimensional domains (implied by Theorem 1 of \citet{CRSSZ2022}).\footnote{No-restoration is not only concerned with the richness of a domain, but also ensures that all preferences in a domain are well organized: given two preferences and two alternatives, one given preference is transformed to the other via a sequence of preferences in the domain that switches two contiguously ranked alternative across each pair of consecutive preferences, and moreover the relative ranking of the two given alternatives is switched at most once. \citet{CRSSZ2022} show that a domain satisfying minimal richness, no-restoration and diversity is either a single-peaked domain on $\mathcal{L}^A$, or a hybrid domain $\mathcal{L}^A$ such that the adjacency subgraph on the free zone has no leaf. Path-connectedness significantly weakens no-restoration, so as to accommodate more permissive preference restrictions that are the main concern of this paper.}
The single-peaked domain on a tree (not a line) of \citet{D1982} satisfies path-connectedness and leaf symmetry, but is excluded by the class of unidimensional domains due to the violation of diversity.
The class of unidimensional domains also excludes multidimensional domains, like the (inclusion/exclusion) separable domain of \citet{BSZ1991}, the multidimensional single-peaked domain of \citet{BGS1993}, the separable domain of \citet{LS1999} and the top-separable domain of \citet{LW1999},
	as they all fail to satisfy path-connectedness.\footnote{More specifically, the notion of adjacency is not applicable on each one of these multidimensional domains, since no two alternatives are adjacent. We intensionally adopt the notion of adjacency to introduce the terminology ``unidimensionality'' so as to exclude these multidimensional domains from our analysis.
In fact, the notion of weak adjacency is applicable on these multidimensional domains. However, weak adjacency is too flexible for us to explore our analysis in trackable way. For instance, give a domain $\mathbb{D}$, if the adjacency graph $G_{\sim}^A$ is a connected graph and has no leaf, then it is a dictatorial domain (see Observation \ref{obs:cycle} in Appendix \ref{app:AP}). However, this important result fails if we replace $G_{\sim}^A$ by $G_{\nsim}^A$, since each one of these multidimensional domain has a no-leaf connected weak adjacency graph and is a non-dictatorial domain.}
\end{remark}

\citet{ACS2003} introduced \textbf{the unique seconds property} on a domain $\mathbb{D}$,
which says that there exists $x \in A$ such that $|\mathcal{S}(\mathbb{D}^x)| = 1$,
and showed that it is sufficient for $\mathbb{D}$ to be a non-dictatorial domain (also see the \emph{inseparable top-pair property} of \citet{KR1980}).
The Auxiliary Proposition shows that the unique seconds property is also necessary, provided that path-connectedness and leaf symmetry hold.
The proof is contained in Appendix \ref{app:AP}.

\begin{AP}\label{AP}
Let a domain $\mathbb{D}$ satisfy path-connectedness and leaf symmetry.
Then, $\mathbb{D}$ is a non-dictatorial domain if and only if it satisfies the unique seconds property.
\end{AP}

Under the unique seconds property of a domain $\mathbb{D}$, say $\mathcal{S}(\mathbb{D}^x) = \{y\}$, we can construct the following non-dictatorial, strategy-proof rule that is loosely speaking called an ``almost dictatorship'' (it follows a dictatorship \emph{almost everywhere} and avoids dictatorship at few particular preference profiles): fixing two distinct voters $i,j \in N$, for all $P\in \mathbb{D}^n$, let
\begin{align*}
f(P) = \left\{
\begin{array}{ll}
r_1(P_i) & \textrm{if}\; r_1(P_i) \neq x,\; \textrm{and}\\
\max^{P_j}(\{x,y\}) & \textrm{otherwise.}
\end{array}
\right.
\end{align*}
The unique seconds property only addresses a  preference restriction that is \emph{locally} embedded in a domain - it only concerns one common second best alternative in preferences with one common peak, and hence cannot be further utilized for the social planner's task of designing meaningful non-dictatorial, strategy-proof rules beyond almost dictatorships.
In order to explore the scope of designing non-dictatorial, strategy-proof rules on non-dictatorial domains,
we go beyond the unique seconds property, and uncover more information on preference restrictions that are \emph{globally} obeyed by the rankings of all alternatives in all preferences in a domain, like single-peakedness/semi-single-peakedness and hybridness/semi-hybridness introduced in Section \ref{sec:domains}.

\begin{remark}\rm
All single-peaked domains, semi-single-peaked domains, non-degenerate hybrid domains and non-degenerate semi-hybrid domains satisfy the unique seconds property, and therefore are non-dictatorial domains.
Indeed, a degenerate hybrid/semi-hybrid domain sometimes violates the unique second property, e.g., the adjacency graph has no leaf,
and hence by the Auxiliary Proposition, is a dictatorial domain;
when it satisfies the unique seconds property (see the example provided in Clarification \ref{cla:example} of Appendix \ref{app:clarification}), it is a non-dictatorial domain.
\end{remark}

\begin{remark}\label{rem:usp}\rm
\citet{RS2019} provide three other domain richness conditions, and show that the unique seconds property is necessary and sufficient for non-dictatorial domains.
Path-connectedness strengthens their first condition, and is more transparent and easier to verify than their third condition, while
leaf symmetry significantly weakens their second condition as it only concerns the leaves of the adjacency graph.
This weakening is meaningful and critical to our analysis because it accommodates semi-single-peaked domains which however are ruled out by their second condition.
See the detailed explanation in Clarification \ref{cla:RS2019} of Appendix \ref{app:clarification}.
\end{remark}

We conclude this section by addressing the indispensability of the two richness conditions in establishing the Auxiliary Proposition.
First, all multidimensional domains mentioned in Remark \ref{rem:richness} are non-dictatorial domains, but are not covered by the Auxiliary Proposition since they all fail to satisfy path-connectedness.
We also provide another example to show the indispensability of path-connectedness,
which would provide a better understand on the role played by path-connectedness in establishing the Auxiliary Proposition.

\begin{table}[t]
\centering
\begin{tabular}{ccccccccccccccc}
$P_1$ & $P_2$ & $P_3$ & $P_4$ & $P_5$ & $P_6$ & $P_7$ & $P_8$ & $P_9$ & $P_{10}$ & $P_{11}$ & $P_{12}$\\
$a  $ & $a  $ & $b  $ & $b  $ & $c  $ & $c  $ & $x  $ & $x  $ & $y  $ & $y  $    & $z  $    & $z  $   \\[-0.2em]
$b  $ & $c  $ & $a  $ & $c  $ & $a  $ & $b  $ & $y  $ & $z  $ & $x  $ & $z  $    & $x  $    & $y  $   \\[-0.2em]
$c  $ & $b  $ & $c  $ & $a  $ & $b  $ & $a  $ & $z  $ & $y  $ & $z  $ & $x  $    & $y  $    & $x  $   \\[-0.2em]
$x  $ & $x  $ & $x  $ & $x  $ & $x  $ & $x  $ & $a  $ & $a  $ & $a  $ & $a  $    & $a  $    & $a  $   \\[-0.2em]
$y  $ & $y  $ & $y  $ & $y  $ & $y  $ & $y  $ & $b  $ & $b  $ & $b  $ & $b  $    & $b  $    & $b  $   \\[-0.2em]
$z  $ & $z  $ & $z  $ & $z  $ & $z  $ & $z  $ & $c  $ & $c  $ & $c  $ & $c  $    & $c  $    & $c  $
\end{tabular}
\caption{Domain $\mathbb{D}_2$}\label{tab:violatepath-connectedness}
\end{table}

\begin{example}\label{exm:violatepath-connectedness}\rm
Let $A = \{a, b, c, x, y, z\}$.
We specify a domain $\mathbb{D}_2$ of 12 preferences in Table \ref{tab:violatepath-connectedness}.
Let $B = \{a,b,c\}$.
One can easily observe that $G_{\sim}^A$ consists of two isolated triangles $G_{\sim}^{B}$ and $G_{\sim}^{A\backslash B}$.
Therefore, $\mathbb{D}_2$ violates path-connectedness and satisfies
leaf symmetry vacuously.
Next, we construct the following SCF: for all $P_i, P_j \in \mathbb{D}_2$,
\begin{align*}
f(P_i, P_j) = \left\{
\begin{array}{ll}
\max^{P_i}(B) & \textrm{if}\; r_1(P_j) \in B,\; \textrm{and}\\
r_1(P_j) & \textrm{otherwise}.
\end{array}
\right.
\end{align*}

\noindent
It is clear that $f$ satisfies unanimity, and is non-dictatorial.
The verification of strategy-proofness is put in Clarification \ref{cla:violatepath-connectedness} of Appendix \ref{app:clarification}.
Note that the SCF $f$, supported by the isolation of $G_{\sim}^{B}$ and $G_{\sim}^{A\backslash B}$, accommodates two distinct dictatorships on $B$ and $A\backslash B$ respectively, i.e., voter $i$ dictates on $B$, while voter $j$ dictates on $A\backslash B$.
This contrasts the unified dictatorship established in Lemma \ref{lem:union} of Appendix \ref{app:AP}.
\hfill$\Box$
\end{example}

\begin{table}[t]
\centering
\begin{tabular}{ccccccccccccccc}
$P_1$ & $P_2$ & $P_3$ & $P_4$ & $P_5$ & $P_6$ & $P_7$ & $P_8$ & $P_9$ \\
$a  $ & $b  $ & $b  $ & $b  $ & $c  $ & $d  $ & $a  $ & $c  $ & $d  $ \\[-0.2em]
$b  $ & $a  $ & $c  $ & $d  $ & $b  $ & $b  $ & $d  $ & $a  $ & $c  $ \\[-0.2em]
$c  $ & $c  $ & $d  $ & $a  $ & $d  $ & $a  $ & $b  $ & $b  $ & $b  $ \\[-0.2em]
$d  $ & $d  $ & $a  $ & $c  $ & $a  $ & $c  $ & $c  $ & $d  $ & $a  $
\end{tabular}
\caption{Domain $\mathbb{D}_3$}\label{tab:special}
\end{table}

\begin{figure}[t]
\begin{tikzpicture}
\node at (-0.9,1) {};

		\put(180, 0){\line(1,0){50}}
        \put(205, 20){\line(0,-1){20}}

		\put(180, 0){\circle*{3}}
        \put(205, 0){\circle*{3}}
        \put(205, 20){\circle*{3}}
        \put(230, 0){\circle*{3}}
		
		\put(177, -10){\footnotesize{$a$}}
		\put(203, -10){\footnotesize{$b$}}
		\put(227, -10){\footnotesize{$c$}}
		\put(202, 23){\footnotesize{$d$}}

\end{tikzpicture}
\vspace{2.5em}
\caption{The star-shape tree $\mathcal{T}^A$}\label{fig:sptree}
\end{figure}
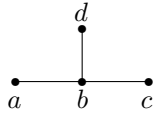

Last, we provide a specific example to illustrate the indispensability of leaf symmetry.

\begin{example}\label{exm:indispensability}\rm
Let $A = \{a, b, c, d\}$ be allocated on a star-shape tree $\mathcal{T}^A$ of Figure \ref{fig:sptree}.
We specify a domain $\mathbb{D}_3$ of 9 preferences in Table \ref{tab:special}.
It is easy to verify $G_{\sim}^A = \mathcal{T}^A$. Therefore, $\mathbb{D}_3$ is a path-connected domain.
We claim that $\mathbb{D}_3$ violates leaf symmetry.
Suppose by contradiction that $\mathbb{D}_3$ satisfies leaf symmetry.
Given $a \in \textrm{Leaf}(G_{\sim}^A)$,
since $\mathcal{S}(\mathbb{D}_3^a) = \{b, d\}$ and $\mathcal{N}_{\sim}^A(a) = \{b\}$,
by leaf symmetry, we must induce a contradiction: $a \in \mathcal{S}(\mathbb{D}_3^d) = \{b, c\}$.
Next, we construct the following SCF to illustrate that $\mathbb{D}_3$ is a non-dictatorial domain:
for all $P_i, P_j \in \mathbb{D}_3$,
\begin{align*}
f(P_i, P_j) =\left\{
\begin{array}{l}
d ~~~~~~~~~~~~~~~~~~~~~~\, \textrm{if}\; P_i = P_7\; \textrm{and}\; P_j \in \mathbb{D}_3^d,\; \textrm{or}\; P_i \in \mathbb{D}_3^d\; \textrm{and}\; P_j = P_7, \\
a ~~~~~~~~~~~~~~~~~~~~~~\, \textrm{if}\; P_i = P_8\; \textrm{and}\; P_j \in \mathbb{D}_3^a,\; \textrm{or}\; P_i \in \mathbb{D}_3^a\; \textrm{and}\; P_j = P_8, \\
c ~~~~~~~~~~~~~~~~~~~~~~\,\, \textrm{if}\; P_i = P_9\; \textrm{and}\; P_j \in \mathbb{D}_3^c,\; \textrm{or}\; P_i \in \mathbb{D}_3^c\; \textrm{and}\; P_j = P_9, \; \textrm{and}\\
\mathop{\textrm{Proj}}\big(b, \langle r_1(P_i), r_1(P_j)|\mathcal{T}^A\rangle\big) \quad \textrm{otherwise}.
\end{array}
\right.
\end{align*}
It is clear that $f$ satisfies unanimity and anonymity, and violates the tops-only property.
The verification of strategy-proofness is put in Clarification \ref{cla:strategy-proofness} of Appendix \ref{app:clarification}.
Last, we observe $|\mathcal{S}(\mathbb{D}_3^x)|\geq 2$ for all $x \in A$
which suggests the violation of the unique seconds property.
\hfill$\Box$
\end{example}

\subsection{A classification of non-dictatorial, unidimensional domains}\label{sec:domainclassification}

In this section, we establish a classification of non-dictatorial, unidimensional domains using the notions of semi-single-peakedness and semi-hybridness via the existence and non-existence respectively of an invariant, tops-only and strategy-proof rule.
The following is the main result of the paper.

\newpage
\begin{theorem}\label{thm:invariance}
Let $\mathbb{D}$ be a non-dictatorial, unidimensional domain.
Then, the following two statements hold:
\begin{itemize}
	\item[\rm (i)]
There exists an invariant, tops-only and strategy-proof rule,
\emph{(a)} if and only if $\mathbb{D}$ is a semi-single-peaked domain, and
\emph{(b)} if and only if $\mathbb{D}$ admits a two-voter, strategy-proof projection rule.
Furthermore, every invariant, tops-only and strategy-proof rule is a projection rule.

%
	
\item[\rm (ii)] There exists no invariant, tops-only and strategy-proof rule
if and only if $\mathbb{D}$ is a semi-hybrid domain satisfying the unique seconds property.
Furthermore, every two-voter, tops-only and strategy-proof rule is a hybrid rule that behaves like a dictatorship on a weak superset of the free zone.

\end{itemize}
\end{theorem}

The proof of Theorem \ref{thm:invariance} is contained in Appendix \ref{app:Theorem}.

\begin{remark}\rm
By Theorem \ref{thm:invariance} and its proof, we know that on a non-dictatorial, unidimensional domain $\mathbb{D}$,
the set of all two-voter, tops-only and strategy-proof rules consists of a set of projection rules which can be either an empty set or not, and
a set of hybrid rules, which at least includes a dictatorship (equivalently, the $(a_1, a_m)$-hybrid rule on the line $\mathcal{L}^A$).\footnote{Even if the domain $\mathbb{D}$ turns out to be a semi-single-peaked domain, the proof of Statement (ii) still implies that every two-voter, tops-only and strategy-proof rule that violates invariance is a hybrid rule.}
Therefore, by Statement (i) of Theorem \ref{thm:invariance}, whether the domain $\mathbb{D}$ is a semi-single-peaked domain completely depends on the existence of a two-voter, strategy-proof projection rule.
Furthermore, if domain $\mathbb{D}$ turns out to be an $(a,b)$-semi-hybrid domain on a tree $\mathcal{T}^A$,
among all two-voter, strategy-proof hybrid rules, by Statement (ii) of Theorem \ref{thm:invariance},
the $(a,b)$-hybrid rule on $\mathcal{T}^A$ is \emph{the most desirable one},
as it minimizes the set of alternatives on which a dictatorship inevitably prevails.
\end{remark}

%
%
%
%
%
%

By Statement (ii) of Theorem \ref{thm:invariance},
we know that once a unidimensional domain is revealed to be a semi-hybrid domain, all two-voter, tops-only and strategy-proof rules behave like dictatorships on the free zone.
We then further show via a Ramification Theorem (see the statement in Appendix \ref{app:anonymity}) that
every tops-only and strategy-proof rule with an arbitrary number of voters also behaves like a dictatorship on the free zone.
This helps us strengthen the classification provided in Theorem \ref{thm:invariance} which concentrates on two-voter rules and the axiom of invariance, by showing in the following corollary that the same classification emerges when we expand to $n$-voter rules and replace invariance by anonymity.
This also suggests that for the purpose of domain classification, there is no loss of generality in restricting attention to the class of two-voter, strategy-proof rules.

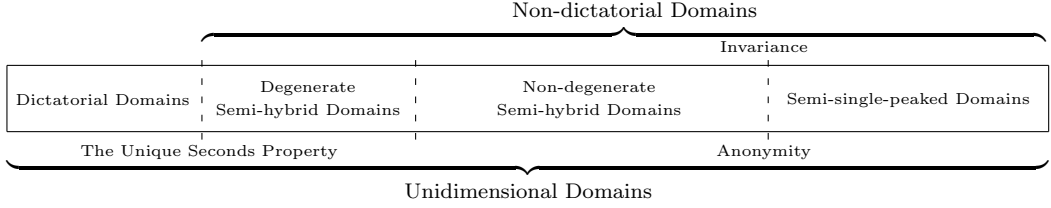
\begin{figure}[t]
\begin{tikzpicture}
\node at (2,0) {};

\node at (11.6,1.8) {\scriptsize{Non-dictatorial Domains}};
\node at (10.2,-0.6) {\scriptsize{Unidimensional Domains}};
\node at (6,-0.1) {\tiny{The Unique Seconds Property}};
\node at (13.3,1.3) {\tiny{Invariance}};
\node at (13.3,-0.1) {\tiny{Anonymity}};

\node at (4.6,0.6) {\tiny{Dictatorial Domains}};
\dashline{3}(168, 32)(168, 2)

\node at (7.3,0.75) {\tiny{Degenerate}};
\node at (7.3,0.45) {\tiny{Semi-hybrid Domains}};
\dashline{3}(248, 32)(248, 2)

\node at (11,0.75) {\tiny{Non-degenerate}};
\node at (11,0.45) {\tiny{Semi-hybrid Domains}};
\dashline{3}(380, 32)(380, 2)

\node at (15.2,0.6) {\tiny{Semi-single-peaked Domains}};

\put(95, 5){\line(1,0){390}}
\put(95, 30){\line(1,0){390}}
\put(95, 5){\line(0,1){25}}
\put(485, 5){\line(0,1){25}}

\put(168,38){\footnotesize{$\overbrace{\rule[0mm]{11.1cm}{0mm}}$}}
\put(95,-5){\footnotesize{$\underbrace{\rule[0mm]{13.7cm}{0mm}}$}}
\end{tikzpicture}
\vspace{-0.5em}
\caption{A classification of non-dictatorial, unidimensional domains}\label{fig:classification}
\end{figure}

\begin{corollary}\label{cor:anonymity}
Let $\mathbb{D}$ be a non-dictatorial, unidimensional domain.
Then, the following two statements hold:
\begin{itemize}
\item[\rm (i)]
There exists an anonymous, tops-only and strategy-proof rule,
\emph{(a)} if and only if $\mathbb{D}$ is a semi-single-peaked domain, and
\emph{(b)} if and only if $\mathbb{D}$ admits a strategy-proof projection rule.
	
	\item[\rm (ii)] There exists no anonymous, tops-only and strategy-proof rule,
if and only if $\mathbb{D}$ is a semi-hybrid domain satisfying the unique seconds property.
Furthermore, every tops-only and strategy-proof rule behaves like a dictatorship on a weak superset of the free zone.
\end{itemize}
\end{corollary}

The proof of Corollary \ref{cor:anonymity} is contained in Appendix \ref{app:anonymity}.

\begin{remark}\label{rem:classification}\rm
Theorem \ref{thm:invariance} and Corollary \ref{cor:anonymity} refine the characterization of non-dictatorial domains obtained in the Auxiliary Proposition by showing that all non-dictatorial, unidimensional domains can be classified into one of the three variants illustrated in Figure \ref{fig:classification}.
First, semi-single-peaked domains are sorted out as the unique ones that admit an invariant/anonymous, tops-only and strategy-proof rule (indeed, a strategy-proof projection rule), while every other non-dictatorial, unidimensional domain is shown to be a semi-hybrid domain, which of course is either non-degenerate or degenerate.
Next, as the free zones of semi-hybrid domains expand towards the whole alternative set,
semi-hybrid domains expand ``closer to'' dictatorial domains
since every tops-only and strategy-proof rule has to behave like a dictatorship on the free zone and consequently gradually degenerates to a dictatorship.
Furthermore, if a non-dictatorial, unidimensional domain turns out to be a degenerate semi-hybrid domain (see the example in Clarification \ref{cla:example} of Appendix \ref{app:clarification}),
since all tops-only and strategy-proof rules are dictatorships,
in order to meet the non-dictatorial-domain hypothesis, the Auxiliary Proposition must mandate the imposition of the unique seconds property, and
consequently the corresponding admissible non-dictatorial, strategy-proof rule (recall the almost dictatorship specified in Section \ref{sec:auxiliary}) must violate the tops-only property.
\end{remark}

\begin{remark}\rm
\citet{CSS2013} have shown that semi-single-peakedness on a path-connected domain is implied by the existence of an anonymous, tops-only and strategy-proof rule with an even number of voters.
By mildly strengthening their richness condition,
we obtain a significantly sharper result:
statement (i) of Corollary \ref{cor:anonymity} dispenses with their requirement on the number of voters, while statement (ii) describes the complementary configuration implied by the non-existence of an anonymous, tops-only and strategy-proof rule.
Last, note that when the number of voters increases more than two in Corollary \ref{cor:anonymity},
the set of tops-only and strategy-proof rules expands significantly beyond projection rules and hybrid rules characterized in Theorem \ref{thm:invariance} under the two-voter condition.
Indeed, the full characterization of tops-only and strategy-proof rules in the case of more than two voters depends delicately on the richness of a semi-single-peaked domain (see Theorem 1 of \citet{BM2020}) and the richness of a semi-hybrid domain (see Proposition 3 of \citet{CZ2020}) respectively.
\end{remark}

\subsection{Non-tops-only rules}\label{sec:PNT}

We now introduce a rule that on a semi-single-peaked (respectively, semi-hybrid) domain can extract non-peak information from some preference profiles
while remaining strategy-proof, and identify ``critical spots'' as configurations that allow such rules to arise.
These critical spots vanish if and only if the domain is refined to be single-peaked (respectively, hybrid).

We fix a tree $\mathcal{T}^A$, an edge $(x,y)$ which separates $\mathcal{T}^A$ into two subtrees $\mathcal{T}^{A^{x \rightharpoonup y}}$ and $\mathcal{T}^{A^{y \rightharpoonup x}}$, and two distinct voters $i, j \in N$.
The construction of a non-tops-only rule consists of three steps.
First, at each preference profile, the social outcome equals voter $i$'s most preferred alternative if it belongs to $A^{y \rightharpoonup x}$.
Next, if both voters $i$ and $j$ have their preference peaks in $A^{x \rightharpoonup y}$, the social outcome is the projection of $x$ on the minimal subtree of the preference profile.
Last, when the two most preferred alternatives of voters $i$ and $j$ lie respectively in $A^{x \rightharpoonup y}$ and $A^{y \rightharpoonup x}$, the social outcome varies according to voter $j$'s preference over $x$ and $y$.

\begin{definition}\label{def:PNT}
	An SCF $f:\mathbb{D}^n \rightarrow A$ is a \textbf{Possibly Non-Tops-only} (or \textbf{PNT}) SCF on a tree $\mathcal{T}^A$ w.r.t.~an edge $(x, y)$
	if there exist distinct $i, j \in N$ such that
	\begin{align*}
	f(P)= \left\{
	\begin{array}{ll}
	r_1(P_i) & \emph{if}\; r_1(P_i) \in A^{y \rightharpoonup x},\\
	\mathop{\emph{Proj}}\big(x, \mathcal{T}^{\Gamma(P)}\big) & \emph{if}\; r_1(P_i) \in A^{x\rightharpoonup y}\; \emph{and}\; r_1(P_j) \in A^{x\rightharpoonup y},\; \textrm{and}\\
	\max^{P_j}(\{x,y\}) & \emph{if}\; r_1(P_i) \in A^{x\rightharpoonup y}\; \emph{and}\; r_1(P_j) \in A^{y \rightharpoonup x}.
	\end{array}
	\right.
	\end{align*}
\end{definition}

By construction, a PNT SCF is unanimous and will henceforth be referred to as a PNT rule.
A PNT rule defined on a minimally rich domain is by definition non-dictatorial.\footnote{One can easily observe that the PNT rule generalizes the almost dictatorship specified in Section \ref{sec:auxiliary}. Indeed, given a line $\widehat{\mathcal{L}}^A = (x, y, \dots)$, where $x$ is a leaf and $y$ is the unique neighbor of $x$, the almost dictatorship is a PNT rule on $\widehat{\mathcal{L}}^A$ w.r.t.~the edge $(x, y)$.}
Moreover, the following fact pins down the necessary and sufficient condition for PNT rules to be strategy-proof and non-tops-only.
The proof is contained in Appendix \ref{app:Fact2}.

\begin{fact}\label{fact:PNTrule}
Fix a minimally rich domain $\mathbb{D}$, a tree $\mathcal{T}^A$ and an edge $(x,y)$.
For all $n \geq 2$,
the PNT rule $f: \mathbb{D}^n \rightarrow A$ on $\mathcal{T}^A$ w.r.t.~$(x, y)$ is strategy-proof if and only if the following two conditions are satisfied: for all $P_i \in \mathbb{D}$,
\begin{itemize}
\item[\rm (i)] if $r_1(P_i) \in A^{x \rightharpoonup y}$, then $P_i$ is semi-single-peaked on $\mathcal{T}^A$ w.r.t.~$y$, and
		
\item[\rm (ii)] if $r_1(P_i) \in A^{y \rightharpoonup x}$, then $\max^{P_i}(A^{x \rightharpoonup y}) = x$.
\end{itemize}
Moreover, $f$ violates the tops-only property if and only if an additional condition is satisfied:
\begin{itemize}
\item[\rm (iii)] there exist $P_i, P_i' \in \mathbb{D}$ such that $r_1(P_i) = r_1(P_i') \in A^{y \rightharpoonup x}$, $y\mathrel{P_i}x$ and $x\mathrel{P_i'}y$.
\end{itemize}
\end{fact}

Given a domain $\mathbb{D}$ and a tree $\mathcal{T}^A$, we call an edge $(x, y)$ a \textbf{critical spot} in $\mathcal{T}^A$,
if all conditions (i), (ii) and (iii) of Fact \ref{fact:PNTrule} are satisfied.
Proposition \ref{prop:nontopsonly} below shows that the existence of a critical spot is necessary and sufficient for distinguishing a semi-single-peaked domain from a single-peaked domain (respectively distinguishing a semi-hybrid domain from a hybrid domain), and
therefore by Fact \ref{fact:PNTrule} supports a strategy-proof PNT rule that violates the tops-only property.
The proof is contained in Appendix \ref{app:Proposition1}.

\newpage
\begin{proposition}\label{prop:nontopsonly}
The following two statements hold:
\begin{itemize}
\item[\rm (i)] Given a semi-single-peaked domain $\mathbb{D}$ on a tree $\mathcal{T}^A$,
we have $\mathbb{D} \nsubseteq \mathbb{D}_{\emph{SP}}(\mathcal{T}^A)$ if and only if there exists a critical spot $(x, y)$ in $\mathcal{T}^A$.

\item[\rm (ii)] Given an $(a,b)$-semi-hybrid domain $\mathbb{D}$ on a tree $\mathcal{T}^A$,
we have $\mathbb{D} \nsubseteq \mathbb{D}_{\emph{H}}(\mathcal{T}^A, a, b)$
if and only if there exists a critical spot $(x, y)$ in $\mathcal{T}^A$ such that
either $x, y \in A^{a \rightharpoonup b}\backslash \{a\}$ or $x, y \in A^{b\rightharpoonup a}\backslash \{b\}$.
\end{itemize}
\end{proposition}

\subsection{A refinement of the classification}\label{sec:refinement}

Last, we further restrict the unidimensional domains in question to be tops-only domains, where strategy-proof rules are endogenously completely determined by voters' preference peaks.
Formally, a domain $\mathbb{D}$ is a \textbf{tops-only domain} if every strategy-proof rule $f: \mathbb{D}^n \rightarrow A$, $n\geq 2$, satisfies the tops-only property.

Clearly, as non-tops-only and strategy-proof rules are ruled out, Fact \ref{fact:PNTrule} implies that no critical spot is embedded in a tops-only domain.
Therefore,
by applying Proposition \ref{prop:nontopsonly},
our classification result established in Section \ref{sec:domainclassification} restricted to tops-only domains is refined in the following three ways:
(i) degenerate semi-hybrid domains that exogenously satisfy the unique seconds property are explicitly excluded from the classification
as they admit a non-tops-only and strategy-proof rule (recall Remark \ref{rem:classification}),
(ii) non-degenerate semi-hybrid domains are refined to non-degenerate hybrid domains on the line $\mathcal{L}^A$ (by statement (ii) of Proposition \ref{prop:nontopsonly} and diversity), and
(iii) semi-single-peaked domains are refined to be single-peaked on $\mathcal{L}^A$ (by statement (i) of Proposition \ref{prop:nontopsonly} and diversity).
We use Figure \ref{fig:refinement} to illustrate the refined classification.

\begin{figure}[t]
\begin{tikzpicture}
\node at (2.1,0) {};

\node at (11.8,1.6) {\scriptsize{Non-dictatorial Domains}};
\node at (10.4,-0.6) {\scriptsize{Tops-only, Unidimensional Domains}};
\node at (6.4,-0.1) {\tiny{The Unique Seconds Property}};
\node at (13.4,-0.1) {\tiny{Anonymity}};

\node at (5,0.6) {\tiny{Dictatorial Domains}};
\dashline{3}(180, 32)(180, 2)

\node at (9.8,0.6) {\tiny{Non-degenerate Hybrid Domains on $\mathcal{L}^A$}};
\dashline{3}(379, 32)(379, 2)

\node at (15.2,0.6) {\tiny{Single-peaked Domains on $\mathcal{L}^A$}};

\put(105, 5){\line(1,0){380}}
\put(105, 30){\line(1,0){380}}
\put(105, 5){\line(0,1){25}}
\put(485, 5){\line(0,1){25}}
\put(180,33){\footnotesize{$\overbrace{\rule[0mm]{10.65cm}{0mm}}$}}
\put(105,-5){\footnotesize{$\underbrace{\rule[0mm]{13.35cm}{0mm}}$}}
\end{tikzpicture}
\vspace{-0.5em}
\caption{A classification of non-dictatorial, tops-only, unidimensional domains}\label{fig:refinement}
\end{figure}
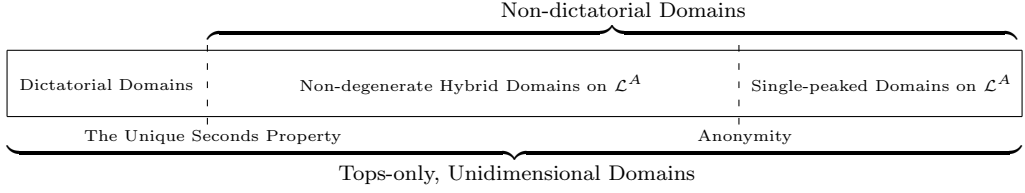

\begin{corollary}\label{cor:topsonly}
Let $\mathbb{D}$ be a non-dictatorial, tops-only, unidimensional domain.
Then, the following two statements hold:
\begin{itemize}
\item[\rm (i)]
There exists an anonymous and strategy-proof rule if and only if $\mathbb{D}$ is a single-peaked domain on $\mathcal{L}^A$.
	
\item[\rm (ii)] There exists no anonymous and strategy-proof rule if and only if $\mathbb{D}$ is a non-degenerate hybrid domain on $\mathcal{L}^A$.
\end{itemize}
\end{corollary}

The proof of Corollary \ref{cor:topsonly} is contained in Appendix \ref{app:topsonly}.

\begin{remark}\rm
By Theorem 2 of \citet{CRSSZ2022} and its proof,
one can further decode all strategy-proof rules on a non-dictatorial, tops-only, unidimensional domain $\mathbb{D}$:
when $\mathbb{D}$ is a single-peaked domain on $\mathcal{L}^A$, an SCF is a strategy-proof rule if and only if it is a \emph{fixed ballot rule} (introduced in Proposition 3 of \citealp{M1980});
when $\mathbb{D}$ is a non-degenerate hybrid domain on $\mathcal{L}^A$,
an SCF is a strategy-proof rule if and only if it is a fixed ballot rule that in addition behaves like a dictatorship on a weak superset of the free zone.
\end{remark}

\section{Literature Review and Final Remarks}\label{sec:conclusion}

\subsection{A  review of the literature}\label{sec:literature}

Following the seminal Gibbard-Satterthwaite Theorem,
domain restrictions have received much attention in the literature on strategic voting.
One stream of the literature examines the robustness of the Gibbard-Satterthwaite Theorem by showing that some sparse restricted domains
(see for instance, linked domains of \citet{ACS2003}, circular domains of \citet{S2010}, FPT (Free Pair at the Top) domains of \citet{CS2011}, and
the $\bm{\beta}$ and $\bm{\gamma}$ domains of \citet{P2015}) are in fact dictatorial domains.
These papers use richness assumptions on the domain variously to construct  connectedness relations between alternatives,
while the violation of these richness conditions appear, somewhat surprisingly, to lead to the  unique seconds property in the sense that if the unique seconds property holds, all the aforementioned richness conditions that precipitate dictatorship are violated.\footnote{This assertion can be made more precise by observing that in the case $|A| = 3$, any domain other than the universal domain satisfies the unique seconds property.}
Recently, \citet{RS2019} have shown the role of the unique seconds property in characterizing non-dictatorial domains.
The Auxiliary Proposition here is in the same vein but uses different richness conditions (recall Remark \ref{rem:usp}). More importantly, our focus on the classification of non-dictatorial domains uncovers more meaningful non-dictatorial, strategy-proof rules - projection rules and hybrid rules,
compared to the almost dictatorship associated with the unique seconds property.

Another stream of the literature starts with a specific restricted domain that not only helps escape the Gibbard-Satterthwaite impossibility,
but also accommodates the design of various well-behaved strategy-proof rules.
Almost all such domains are variants of the notion of single-peakedness.
On the single-peaked domain,
the seminal paper \citet{M1980} characterized all anonymous, tops-only and strategy-proof rules as phantom voter rules,
and all tops-only and strategy-proof rules as fixed ballot rules.
In the past four decades, several key variants of single-peakedness have been developed, and non-dictatorial, strategy-proof rules have been explored. \citet{D1982} introduced single-peakedness on a tree and \citet{SV2002} extended Moulin's fixed ballot rules;
	\citet{BGS1993} generalized single-peakedness from a unidimensional underlying line to a multidimensional grid, and discovered an important class of strategy-proof rules: \emph{multidimensional generalized median voter rules};
	\citet{NP2007} adopted a ternary relation to generally address the geometric relation among alternatives,
	invented the notion of \emph{generalized single-peakedness},\footnote{Using the terminology of \citet{NP2007}, the (inclusion/exclusion) separable domain of \citet{BSZ1991}, the multidimensional single-peaked domain of \citet{BGS1993} and the separable domain of \citet{LS1999} can be equivalently translated to generalized single-peaked domains according to three analogous ternary relations respectively.} and characterized all strategy-proof rules: \emph{voting by issues}; and
	recently, \citet{R2015} provided a transition from the single-peaked domain to the universal domain by taking unions of multiple single-peaked domains that are constructed according to different underlying lines, established the notion of a multiple single-peaked domain and characterized all strategy-proof rules as a specific subset of fixed ballot rules which simultaneously preserve features of a dictatorship and of a median voter rule.
	Two comprehensive survey papers, \citet{S1995} and \citet{B2011}, provide more detailed discussions on the development of single-peakedness restrictions and non-dictatorial, strategy-proof rules.
All preference domains considered in this literature are in fact tops-only domains.
We depart from this literature by considering non-tops-only rules; Proposition \ref{prop:nontopsonly}  identifies a critical spot that supports a non-tops-only and strategy-proof rule.

A third stream of the literature poses the following natural ``converse'' question:
is single-peakedness a consequence of the existence of a well-behaved strategy-proof rule?
Earlier literature \citet{BGS1993} showed that if a minimally rich domain admits the median voter rule as a strategy-proof rule,
the domain must be single-peaked.
	Instead of considering a specific rule,
	\citet{CSS2013} established that on a path-connected domain, semi-single-peakedness, rather than single-peakedness,
	is necessary for the existence of an anonymous, tops-only and strategy-proof rule, and
	\citet{CM2018} showed that semilattice single-peakedness, a generalization of semi-single-peakedness, arises as a consequence of the existence of an anonymous, tops-only and strategy-proof rule on a rich domain (where the richness condition is formulated relative to the particular rule that is assumed to exist).
	Recently, \citet{BBM2020} provide an insightful survey that covers these and other related issues.
This literature too restricts attention to the class of strategy-proof rules that in addition satisfy the tops-only property and anonymity,
and is therefore silent on domains that admit tops-only and strategy-proof rules that violate anonymity but remain non-dictatorial.
Our classification theorem
essentially demonstrates that appropriate weakenings of single-peakedness\footnote{\citet{BM2011} introduced another approach of weakening single-peakedness, called \emph{top-monotonicity}.
We briefly introduce the definition of top-monotonicity using our model here.
A preference profile $P$ is said to satisfy \textbf{top-monotonicity} if there exists a line $\widehat{\mathcal{L}}^A$ such that
for each $i \in N$, say $r_1(P_i) = a$, and for all distinct alternatives $b, c \in A$ where $b = r_1(P_j)$ for some $j \in N\backslash \{i\}$,
we have $\big[b \in \langle a, c|\widehat{\mathcal{L}}^A\rangle\big] \Rightarrow [b\mathrel{P}_ic]$.
It is clear that top-monotonicity allows some flexibility in the ranking of an alternative that is never top-ranked at any preference in the profile, and hence weakens the single-peakedness restriction. The restrictions of semi-single-peakedness and semi-hybridness investigated in our paper are independent. For instance, one can construct a profile of semi-single-peaked preferences that fails to meet top-monotonicity. The detailed example is available on request.} characterizes all non-dictatorial domains, and in particular uncovers domains that allow the design of non-tops-only and non-anonymous rules.

Our refinement of the classification of non-dictatorial domains provided in Corollary \ref{cor:topsonly} is also related to the literature on tops-only domains.
In this literature, various restricted domains have been shown to be tops-only domains
\citep[see for instance][]{BSZ1991,BGS1993,Ching1997,LS1999,LW1999,NP2007,W2008,R2015}.
\citet{CS2011} provided two general sufficient conditions for tops-only domains.
Corollary \ref{cor:topsonly}, to our knowledge, is the first result that characterizes necessary conditions for tops-only domains, and therefore reveals the important role of the full single-peakedness requirement, imposed on either the whole line, \newpage

\noindent
or
both the left and right parts of the line,
in establishing a tops-only domain.\footnote{Though Corollary \ref{cor:topsonly} is concerned with non-dictatorial domains,
its proof can also be adopted to show that a dictatorial, unidimensional domain, which of course is a tops-only domain, is an $(a_1, a_m)$-hybrid domain on the line $\mathcal{L}^A$. Conversely, Appendix E of \citet{CZ2020} shows that both single-peakedness and hybridness in conjunction with an additional technical condition on the free zone called \emph{non-trivialness}, are sufficient for a unidimensional domain to be a tops-only domain.}
In a model with single-peaked preferences on the real line that accommodates indifference relations,
\citet{BJ1994} established that a strategy-proof SCF that has a \emph{non-connected} range can violate the tops-only property,
while \citet{W2011} showed that the non-connected range is necessary for the violation of the tops-only property.
The non-tops-only and strategy-proof rules investigated in our paper are independent of this literature as all strategy-proof rules studied here have a full range.

Lastly, we relate our approach to the literature characterizing Condorcet domains.
Given minimal richness and the presence of two completely reversed preferences, \citet{P2018} showed that a ``connected'' domain (i.e., any two distinct preferences of the domain are connected via a path of preferences in the domain where across each consecutive pair on the path, exactly two contiguously ranked alternatives are switched)
is a maximal Condorcet domain if and only if it is single-peaked on a line.
It is clear that a Condorcet domain is a non-dictatorial domain as it supports majority voting as a strategy-proof rule.
However, non-dictatorial, strategy-proof rules  can obtain in settings where the acyclicity of the majority relation  (that is ensured in a Condorcet domain) does not hold.
Our classification theorem indicates that expanding the search for strategy-proof rules to settings where the acyclicity of the majority relation need not hold significantly enlarges the class of domains that admit non-dictatorial, strategy-proof rules.


\subsection{Final remarks}

To conclude, this paper has introduced a methodology based on the analysis of two-voter rules and a simple axiom (invariance) on unidimensional domains in the voting model, using which
we exhaustively classify all unidimensional, non-dictatorial domains as either semi-single-peaked domains or semi-hybrid domains,
which are respectively two weakenings of single-peaked domains that complement each other.
This expands the possibilities for design to models where the restriction of single-peakedness is too demanding (see for instance, multidimensional voting under constraints and allocation of public goods on a transportation network).
We provide some preliminary observations on multidimensional voting in Appendix \ref{app:extension}, and leave its detailed exploration to future work.
This methodology may also be useful beyond the specific issue of classification of non-dictatorial domains that we studied here; for instance, within the voting model it could be used to explore the structure of locally strategy-proof rules or ordinally Bayesian incentive compatible rules. It would be particularly interesting to extend this methodology beyond the voting model to more general setups that include private goods and/or monetary compensation.

\newpage


\newpage
{\small
\appendix
\section*{Appendix}

\section{Proof of the Auxiliary Proposition}\label{app:AP}

By Theorem 5.1 of \citet{ACS2003}, it is clear that if a domain satisfies the unique seconds property,
it is a non-dictatorial domain.
We henceforth focus on showing that given a domain $\mathbb{D}$ satisfying path-connectedness and leaf symmetry, if it is a non-dictatorial domain, it satisfies the unique seconds property.
We first provide 4 important independent lemmas (Lemmas \ref{lem:step1} - \ref{lem:union2}).
For Lemmas \ref{lem:step1} - \ref{lem:union2},
we fix $N = \{1,2\}$, a path-connected domain $\mathbb{D}$ and a strategy-proof rule $f: \mathbb{D}^2 \rightarrow A$.
For ease of presentation, let $\big((x\cdots), (y\cdots)\big)$ denote a profile
where voter $1$ reports an \emph{arbitrary} preference with the peak $x$ and voter 2 reports an \emph{arbitrary} preference with the peak $y$.
More importantly, let $f\big((x\cdots), (y\cdots)\big) = a$ denote ``$f(P_1, P_2) = a$ for all $P_1 \in \mathbb{D}^{x}$ and $P_2 \in \mathbb{D}^y$.''
For notational convenience,
given distinct $x, y \in A$, let $P_i^{\,x,y}$ denote a preference in the domain such that $x$ is top-ranked and $y$ is second ranked.

\begin{lemma}\label{lem:step1}
Given a path $\pi = (x_1, \dots, x_v)$ in $G_{\sim}^A$, the following statements hold:
\begin{itemize}
\item[\rm (i)] if $f\big((x_1\cdots), (x_2\cdots)\big) = x_1$,
$f\big((x_k\cdots), (x_{k'}\cdots)\big) = x_k$ for all $1 \leq k \leq k' \leq v$,

\item[\rm (ii)] if $f\big((x_2\cdots), (x_1\cdots)\big) = x_1$,
$f\big((x_{k'}\cdots), (x_k\cdots)\big) = x_k$ for all $1 \leq k \leq k' \leq v$.

\item[\rm (iii)] if $f\big((x_v\cdots), (x_{v-1}\cdots)\big) = x_v$, $f\big((x_k\cdots), (x_{k'}\cdots)\big) = x_k$ for all $1 \leq k' \leq k \leq v$, and

\item[\rm (iv)] if $f\big((x_{v-1}\cdots), (x_v\cdots)\big) = x_v$, $f\big((x_{k'}\cdots), (x_k\cdots)\big) = x_k$ for all $1 \leq k' \leq k \leq v$.
\end{itemize}
\end{lemma}

\begin{proof}
Note that the first two statements are symmetric, and
the last two statements are symmetric.
Moreover, the third statement is analogous to the first one: given statement (i), after relabeling the path $\pi$ such that
$y_k = x_{v+1-k}$ for all $k \in \{1, \dots, v\}$, we modify the hypothesis of statement (iii) to $f\big((y_1\cdots), (y_2\cdots)\big) = f\big((x_v\cdots), (x_{v-1}\cdots)\big) = x_v = y_1$, and then apply statement (i) on the path $(y_1, \dots, y_v)$ to obtain
$f\big((y_k\cdots), (y_{k'}\cdots)\big) = y_k$ for all $1 \leq k \leq k' \leq v$,
which by the relabeling implies $f\big((x_{k}\cdots), (x_{k'}\cdots)\big) = x_{k}$ for all $1 \leq k' \leq k \leq v$.
Therefore, we in the rest of the proof focus on the verification of the first statement.

Since $x_2 \sim x_3$, by Claims A and B of \citet{S2001} and their proofs,\footnote{Claim A of \citet{S2001} and its proof show that given $a, b \in A$ and $P_i^{\,a,b}, P_i^{\, b, a} \in \mathbb{D}$, we have $f(P_1, P_2) \in \{a,b\}$ for all $P_1 \in \mathbb{D}^a$ and $P_2\in \mathbb{D}^b$. Claim B of \citet{S2001} and its proof show that given $a, b \in A$ and $P_i^{\,a,b}, P_i^{\, b, a} \in \mathbb{D}$, if $f(\hat{P}_1, \hat{P}_2) = a$ (respectively, $f(\hat{P}_1, \hat{P}_2) = b$) for some $\hat{P}_1 \in \mathbb{D}^a$ and
$\hat{P}_2 \in \mathbb{D}^b$, then $f\big((a\cdots), (b\cdots)\big) = a$ (respectively, $f\big((a\cdots), (b\cdots)\big) = b$).
Therefore, given $a, b \in A$ with $a \sim b$, either $f\big((a\cdots),(b\cdots)\big) =a$
or $f\big((a\cdots),(b\cdots)\big) = b$ holds.}
we know that either $f\big((x_2\cdots),(x_3\cdots)\big) =x_2$
or $f\big((x_2\cdots),(x_3\cdots)\big) = x_3$ holds.
Suppose $f\big((x_2\cdots),(x_3\cdots)\big) =x_3$.
Since $x_1 \sim x_2$ and $x_3 \sim x_2$,
we have $P_1^{\,x_1,x_2}, P_1^{\,x_2,x_1}\in \mathbb{D}$ such that
$r_k(P_1^{\,x_1,x_2})=r_k(P_1^{\,x_2,x_1})$ for all $k \in \{3, \dots, m\}$, and
$P_2^{\,x_2, x_3},P_2^{\,x_3,x_2} \in \mathbb{D}$ such that $r_k(P_2^{\,x_2,x_3}) = r_k(P_2^{\,x_3, x_2})$
for all $k \in \{3, \dots, m\}$.
Thus, $f(P_1^{\,x_1,x_2}, P_2^{\,x_2, x_3}) = x_1$ and $f(P_1^{\,x_2,x_1}, P_2^{\,x_3,x_2}) = x_3$.
Then, by strategy-proofness, $f(P_1^{\,x_2,x_1}, P_2^{\,x_3,x_2}) = x_3$ implies $f(P_1^{\,x_1,x_2}, P_2^{\,x_3,x_2}) = x_3$.
Consequently,
voter 2 will manipulate at $(P_1^{\,x_1,x_2}, P_2^{\,x_2, x_3})$ via $P_2^{\,x_3,x_2}$.
Therefore, $f\big((x_2\cdots),(x_3\cdots)\big) =x_2$.
Applying the same argument from $x_3$ to $x_v$ step by step,
we eventually have $f\big((x_k\cdots),(x_{k+1}\cdots)\big) =x_k$ for all $k \in \{1, \dots, v-1\}$.

Now, we show statement (i).
We pick an arbitrary $l \in \{3, \dots, v\}$ and provide an induction hypothesis: $f\big((x_k\cdots),(x_{k'}\cdots)\big) = x_k$ for all $1 \leq k \leq k' < l$.
To verify the induction hypothesis, we show $f\big((x_k\cdots),(x_l\cdots)\big) = x_k$ for all $k \in \{1, \dots, l\}$.
Fix arbitrary $P_2 \in \mathbb{D}^{x_{l}}$.
We first know $f\big((x_k\cdots), P_2\big) = x_k$ for both $k \in \{l-1, l\}$.
We next show $f\big((x_{l-2}\cdots), P_2\big) = x_{l-2}$.
Since $x_{l-2} \sim x_{l-1}$, we have $P_1^{\,x_{l-1},x_{l-2}}, P_1^{\,x_{l-2},x_{l-1}} \in \mathbb{D}$ such that
$r_k(P_1^{\,x_{l-1},x_{l-2}}) = r_k(P_1^{\,x_{l-2},x_{l-1}})$ for all $k \in \{3, \dots, m\}$.
Since $f(P_1^{\,x_{l-1},x_{l-2}}, P_2) = x_{l-1}$, strategy-proofness implies $f(P_1^{\,x_{l-2},x_{l-1}}, P_2) \in \{x_{l-1}, x_{l-2}\}$.
If $f(P_1^{\,x_{l-2},x_{l-1}}, P_2) = x_{l-1}$, strategy-proofness implies $f\big(P_1^{\,x_{l-2},x_{l-1}}, (x_{l-1}\cdots)\big) = x_{l-1}$,
which contradicts the induced fact $f\big((x_{l-2}\cdots), (x_{l-1}\cdots)\big) = x_{l-2}$.
Hence, $f(P_1^{\,x_{l-2},x_{l-1}}, P_2) = x_{l-2}$. Then, strategy-proofness implies $f\big((x_{l-2}\cdots), P_2\big) = x_{l-2}$.
Applying the same argument from $x_{l-2}$ to $x_1$ step by step,
we eventually have $f\big((x_k\cdots), P_2) = x_k$ for all $k \in \{1,\dots, l\}$.
This completes the verification of the induction hypothesis, and hence proves the lemma.
\end{proof}

\vspace{-1.5em}
\begin{lemma}\label{lem:union}
Given two subsets $\bar{A},\hat{A}\subseteq A$ with $|\bar{A}|>1$ and $|\hat{A}|>1$,
let $G_{\sim}^{\bar{A}}$ and $G_{\sim}^{\hat{A}}$ be two connected graphs.
Given a path $\pi = (x_1, \dots, x_v)$ in $G_{\sim}^A$, let $x_1 \in \bar{A}$ and $x_v \in \hat{A}$.
If $f$ behaves like a dictatorship on $\bar{A}$ and $\hat{A}$ respectively, then
$f$ behaves like a dictatorship on $\bar{A} \cup \pi\cup \hat{A}$.
\end{lemma}

\begin{proof}
We assume w.l.o.g.~that voter $1$ dictates on $\bar{A}$,
i.e.,
$f(P_1, P_2) = r_1(P_1)$ for all $P_1, P_2 \in \mathbb{D}$ with $r_1(P_1), r_1(P_2) \in \bar{A}$.
\medskip

\noindent
\textsc{Claim 1}:
Voter $1$ also dictates on $\hat{A}$, i.e.,
$[r_1(P_1), r_1(P_2) \in \hat{A}]\Rightarrow [f(P_1, P_2 ) = r_1(P_1)]$.\medskip

We first consider the case $\bar{A}\cap \hat{A} \neq \emptyset$.
Let $x \in \bar{A}\cap \hat{A}$.
Since $G_{\sim}^{\bar{A}}$ is a connected graph and $|\bar{A}| > 1$,
there exists $y \in \bar{A}$ such that $y \sim x$.
Clearly, $f\big((y\cdots), (x\cdots)\big) = y$.
Symmetrically, according to $G_{\sim}^{\hat{A}}$,
there exists $z \in \hat{A}$ such that $z \sim x$.
Clearly, either $y = z$ or $y \neq z$ holds.
If $y = z$, we have $f\big((z\cdots), (x\cdots)\big) = z$,
which immediately implies that voter 1 dictates on $\hat{A}$, since $f$ is assumed to behave like a dictatorship on $\hat{A}$.
If $y\neq z$, we consider the path $(y, x, z)$.
By statement (i) of Lemma \ref{lem:step1},
$f\big((y\cdots), (x\cdots)\big) = y$ implies $f\big((x\cdots), (z\cdots)\big) = x$.
Last, since $f$ behaves like a dictatorship on $\hat{A}$,
we infer that voter 1 dictates on $\hat{A}$.

Next, we assume $\bar{A}\cap \hat{A} = \emptyset$.
Given $x_1 \in \bar{A}$ and $x_v \in \hat{A}$, we can identify $1 \leq k \leq k' \leq v$ such that
$x_k \in \bar{A}$, $x_{k'} \in \hat{A}$ and $x_l \notin \bar{A}\cup \hat{A}$ for all $l \in \{k+1, \dots, k'-1\}$.
Since $G_{\sim}^{\bar{A}}$ is a connected graph and $|\bar{A}|>1$,
there exists $x \in \bar{A}$ such that $x \sim x_k$.
Symmetrically, there exists $y \in \hat{A}$ such that $y \sim x_{k'}$.
Thus, we have a path $\pi' = (x, x_k, \dots, x_{k'}, y)$.
Since voter 1 dictates on $\bar{A}$, we have $f\big((x\cdots),(x_k\cdots)\big) = x$.
Then, according to $\pi'$,
statement (i) of Lemma \ref{lem:step1} implies $f\big((x_{k'}\cdots),(y\cdots)\big) = x_{k'}$.
Moreover, since $f$ behaves like a dictatorship on $\hat{A}$,
we infer that voter 1 dictates on $\hat{A}$.
This completes the verification of the claim.\medskip

\noindent
\textsc{Claim 2}: Voter 1 dictates on $\pi$, i.e., $[r_1(P_1), r_1(P_2) \in \pi]\Rightarrow [f(P_1,P_2)= r_1(P_1)]$.\medskip

If $x_2 \in \bar{A}$, we have $f\big((x_1\cdots),(x_2\cdots)\big) = x_1$ by voter 1's dictatorship on $\bar{A}$.
If $x_2 \notin \bar{A}$, we identify $x_0 \in \bar{A}$ such that $x_0 \sim x_1$.
Clearly, $x_0 \neq x_2$.
Thus, we have $f\big((x_0\cdots),(x_1\cdots)\big) = x_0$ by voter 1's dictatorship on $\bar{A}$.
Then, according to the path $(x_0, x_1, x_2)$,
statement (i) of Lemma \ref{lem:step1} implies $f\big((x_1\cdots),(x_2\cdots)\big) = x_1$.
Overall, $f\big((x_1\cdots),(x_2\cdots)\big) = x_1$.
Then, according to $\pi$,
statement (i) of Lemma \ref{lem:step1} implies $f\big((x_k\cdots),(x_{k'}\cdots)\big) = x_k$ for all $1 \leq k \leq k' \leq v$.
Symmetrically, by voter 1's dictatorship on $\hat{A}$ and statement (iii) of Lemma \ref{lem:step1} on $\pi$,
we also induce $f\big((x_k\cdots),(x_{k'}\cdots)\big) = x_k$ for all $1 \leq k' \leq k \leq v$.
This completes the verification of the claim.\medskip

Last, we show that voter 1 dictates on $\bar{A}\cup \pi\cup \hat{A}$.
We first show that voter 1 dictates on $\bar{A}\cup \pi$.
Given arbitrary preferences $P_1, P_2 \in \mathbb{D}$, let $r_1(P_1) = x \in \bar{A}\cup \pi$ and $r_1(P_2) = y \in \bar{A}\cup \pi$.
If $x=y$, unanimity implies $f(P_1,P_2) = x = r_1(P_1)$.
Next, assume $x \neq y$.
Evidently, if $x, y \in \bar{A}$ or $x, y \in \pi$, we have $f(P_1,P_2) = x$ by voter 1's dictatorship on $\bar{A}$ and $\pi$ respectively.
Last, we consider two cases: (i) $x \in \bar{A}\backslash \pi$ and $y \in \pi\backslash \bar{A}$, and
(ii) $x \in \pi\backslash \bar{A}$ and $y \in \bar{A}\backslash \pi$.
The two cases are symmetric, and we hence focus on the first one.
In the first case, $x \in \bar{A}\backslash \{x_1\}$ and $y = x_k$ for some $1 < k \leq v$.
Since $G_{\sim}^{\bar{A}}$ is a connected graph,
we have a path $(z_1, \dots, z_l)$ in $G_{\sim}^{\bar{A}}$ connecting $x$ and $x_1$.
Now, according to the paths $(z_1, \dots, z_l)$ and $(x_1, \dots, x_k)$,
since $z_l = x_1$, $z_1=x \in \bar{A}\backslash \pi$ and $x_k = y \in \pi\backslash \bar{A}$,
we can identify $1< s \leq l$ and $1 \leq t < k$ such that $z_s = x_t$ and $\{z_1, \dots, z_{s-1}\}\cap \{x_{t+1}, \dots, x_k\} = \emptyset$.
Then, the concatenated path $\hat{\pi} = (z_1, \dots, z_s = x_t, \dots, x_k)$ connects $x$ and $y$.
First, we have $f\big((z_1\cdots), (z_2\cdots)\big) = z_1$ by voter 1's dictatorship on $\bar{A}$.
Next, according to $\hat{\pi}$, statement (i) of Lemma \ref{lem:step1} implies
$f\big((z_1\cdots), (x_k\cdots)\big) = z_1$, and hence, $f(P_1,P_2) = x$, as required.
Therefore, voter 1 dictates on $\bar{A}\cup \pi$.
Last, note that both $G_{\sim}^{\bar{A}\cup \pi}$ and $G_{\sim}^{\hat{A}}$ are connected graphs,
$[\bar{A}\cup \pi]\cap \hat{A}\neq \emptyset$ and voter 1 dictates on $\bar{A}\cup \pi$ and $\hat{A}$ respectively.
By the same argument,
we can infer that voter 1 dictates on $\big[\bar{A} \cup \pi\big]\cup \hat{A} = \bar{A} \cup \pi\cup \hat{A}$.
\end{proof}

\begin{lemma}\label{lem:dictatorship1}
Given a path $\pi = (x_1,\dots, x_v)$, $v \geq 3$, in $G_{\sim}^A$ and two preferences
$P_i^{\,x_1,x_v},P_i^{\,x_v,x_1} \in \mathbb{D}$,
SCF $f$ behaves like a dictatorship on $\pi$.
\end{lemma}

\begin{proof}
We first show that if $f$ behaves like a dictatorship on $\{x_1, x_2\}$, it behaves like a dictatorship on $\pi$.
We assume w.l.o.g. that voter 1 dictates on $\{x_1, x_2\}$, i.e., $f(P_1, P_2) = r_1(P_1)$ for all $P_1, P_2 \in \mathbb{D}$ with $r_1(P_1), r_1(P_2) \in \{x_1, x_2\}$.
Thus, $f\big((x_1\cdots), (x_2\cdots)\big) = x_1$.
By Lemma \ref{lem:union}, it suffices to show that $f$ behaves like a dictatorship on $\{x_{v-1}, x_v\}$.
Since $x_{v-1} \sim x_v$, one of the following three cases occurs:\\
(1) $f$ behaves like a dictatorship on $\{x_{v-1}, x_v\}$,\\
(2) $f\big((x_{v-1}\cdots), (x_v\cdots)\big) = f\big((x_v\cdots), (x_{v-1}\cdots)\big) = x_{v}$, and\\
(3) $f\big((x_{v-1}\cdots), (x_v\cdots)\big) = f\big((x_v\cdots), (x_{v-1}\cdots)\big) = x_{v-1}$.\\
To complete the verification, we rule out the last two cases.
By statement (i) of Lemma \ref{lem:step1} on the path $\pi$,
$f\big((x_1\cdots), (x_2\cdots)\big) = x_1$  implies
$f\big((x_{v-1}\cdots), (x_v\cdots)\big) = x_{v-1}$, which rules out case (2).
Suppose that case (3) occurs.
Since $f\big((x_2\cdots), (x_1\cdots)\big) = x_2$ and $f\big((x_v\cdots), (x_{v-1}\cdots)\big) = x_{v-1} \neq x_v$,
searching on the path $\pi$ from $x_2$ towards $x_{v-1}$, we can identify $1<\bar{k}<v$ such that
$f\big((x_k\cdots), (x_{k-1}\cdots)\big) = x_k$ for all $k \in \{2, \dots, \bar{k}\}$
and $f\big((x_{\bar{k}+1}\cdots), (x_{\bar{k}}\cdots)\big) \neq x_{\bar{k}+1}$.
Thus, by statement (iii) of Lemma \ref{lem:step1} on the subpath $(x_1, \dots, x_{\bar{k}})$,
$f\big((x_{\bar{k}}\cdots), (x_{\bar{k}-1}\cdots)\big) = x_{\bar{k}}$ implies $f\big((x_k\cdots),(x_{k'}\cdots)\big) = x_k$ for all $1 \leq k' \leq k \leq \overline{k}$.
Meanwhile, since $x_{\bar{k}}\sim x_{\bar{k}+1}$ and $f\big((x_{\bar{k}+1}\cdots), (x_{\bar{k}}\cdots)\big) \neq x_{\bar{k}+1}$,
it is easy to show that $f\big((x_{\bar{k}+1}\cdots), (x_{\bar{k}}\cdots)\big) = x_{\bar{k}}$.
Consequently, by statement (ii) of Lemma \ref{lem:step1} on the subpath $(x_{\bar{k}}, \dots, x_v)$,
$f\big((x_{\bar{k}+1}\cdots), (x_{\bar{k}}\cdots)\big) = x_{\bar{k}}$ implies
$f\big((x_s\cdots),(x_{s'}\cdots)\big) = x_{s'}$ for all $\overline{k} \leq s' \leq s \leq v$.
Therefore, we have $f\big(P_1^{\,x_v,x_1}, (x_{\overline{k}}\cdots)\big) = x_{\overline{k}}$.
Consequently, since $f\big((x_1\cdots), (x_{\overline{k}}\cdots)\big) = x_{1}$ and $x_1\mathrel{P_1^{\,x_v,x_1}}x_{\bar{k}}$,
voter 1 will manipulate at $\big(P_1^{\,x_v,x_1}, (x_{\overline{k}}\cdots)\big)$ via
some $P_1 \in \mathbb{D}^{x_1}$.
Hence, case (3) is ruled out, as required.

Symmetrically, we can show that if $f$ behaves like a dictatorship on $\{x_{v-1}, x_v\}$, it behaves like a dictatorship on $\pi$.

Last, to prove the lemma, we show that $f$ behaves like a dictatorship on either $\{x_1, x_2\}$ or $\{x_{v-1}, x_v\}$.
Suppose that it is not true.\medskip

\noindent
\textsc{Claim 1}:
We have
$f\big((x_1\cdots), (x_2\cdots)\big) = f\big((x_2\cdots), (x_1\cdots)\big) = x_2$ and
$f\big((x_{v-1}\cdots), (x_v\cdots)\big) = f\big((x_v\cdots), (x_{v-1}\cdots)\big) = x_{v-1}$.\medskip

Since $x_1 \sim x_2$ and the contradictory hypothesis rules out dictatorships on $\{x_1, x_2\}$,
by strategy-proofness, we have either
$f\big((x_1\cdots), (x_2\cdots)\big) = f\big((x_2\cdots), (x_1\cdots)\big) = x_1$, or
$f\big((x_1\cdots), (x_2\cdots)\big) = f\big((x_2\cdots), (x_1\cdots)\big) = x_2$.
Suppose $f\big((x_1\cdots), (x_2\cdots)\big) = f\big((x_2\cdots), (x_1\cdots)\big) = x_1$.
Thus, we know $f\big(P_1^{\,x_1,x_v}, (x_2\cdots)\big) = x_1$.
Moreover, by statement(ii) of Lemma \ref{lem:step1} on $\pi$,
$f\big((x_2\cdots), (x_1\cdots)\big) = x_1$ implies $f\big(P_1^{\,x_v,x_1}, (x_2\cdots)\big) = x_2$.
Consequently, voter $1$ will manipulate at $\big(P_1^{\,x_v,x_1}, (x_2\cdots)\big)$ via some $P_1^{\,x_1,x_v}$.
Therefore, $f\big((x_1\cdots), (x_2\cdots)\big) = f\big((x_2\cdots), (x_1\cdots)\big) = x_2$.
Symmetrically, we can show
$f\big((x_{v-1}\cdots), (x_v\cdots)\big) = f\big((x_v\cdots), (x_{v-1}\cdots)\big) = x_{v-1}$.
This proves the claim.\medskip

Now, by Claims A and B of \citet{S2001}, according to preference $P_i^{\,x_1,x_v}$ and $P_i^{\,x_v,x_1}$,
we know that either $f\big((x_1\cdots), (x_v\cdots)\big) = x_1$ or
$f\big((x_1\cdots), (x_v\cdots)\big) = x_v$ holds.
To complete the proof, we will induce a contradiction in each case.
First, let $f\big((x_1\cdots), (x_v\cdots)\big) = x_1$.
Since $x_v \sim x_{v-1}$, we have $P_2^{\,x_v,x_{v-1}}, P_2^{\,x_{v-1},x_v} \in \mathbb{D}$ such that
$r_k(P_2^{\,x_v,x_{v-1}}) = r_k(P_2^{\,x_{v-1},x_v})$ for all $k \in \{3, \dots, m\}$.
Thus, $f(P_1^{\,x_1,x_v}, P_2^{\,x_v,x_{v-1}}) = x_1$, which by strategy-proofness implies $f(P_1^{\,x_1,x_v}, P_2^{\,x_{v-1},x_v}) = x_1$.
Furthermore, by strategy-proofness, $f(P_1^{\,x_1,x_v}, P_2^{\,x_{v-1},x_v}) = x_1$ implies
$f(P_1^{\,x_v,x_1}, P_2^{\,x_{v-1},x_v}) \in \{x_1,x_v\}$.
Since $x_v \sim x_{v-1}$, by Claim A of \citet{S2001}, we have $f(P_1^{\,x_v,x_1}, P_2^{\,x_{v-1},x_v}) \in \{x_v, x_{v-1}\}$.
Therefore, $f(P_1^{\,x_v,x_1}, P_2^{\,x_{v-1},x_v}) \in \{x_1,x_v\}\cap \{x_v, x_{v-1}\} = \{x_v\}$, which contradicts Claim 1, as required.
Last, let $f\big((x_1\cdots), (x_v\cdots)\big) = x_v$.
Since $x_1 \sim x_2$, we have $P_1^{\,x_1,x_2}, P_1^{\,x_2,x_1} \in \mathbb{D}$ such that
$r_k(P_1^{\,x_1,x_2}) = r_k(P_1^{\,x_2,x_1})$ for all $k \in \{3, \dots, m\}$.
Thus, $f(P_1^{\,x_1,x_2}, P_2^{\,x_v,x_1}) = x_v$, which by strategy-proofness implies $f(P_1^{\,x_2,x_1}, P_2^{\,x_v,x_1}) = x_v$.
Furthermore, by strategy-proofness, $f(P_1^{\,x_2,x_1}, P_2^{\,x_v,x_1}) = x_v$ implies $f(P_1^{\,x_2,x_1}, P_2^{\,x_1,x_v}) \in \{x_1,x_v\}$.
Since $x_1 \sim x_2$, by Claim A of \citet{S2001}, we have $f(P_1^{\,x_2,x_1}, P_2^{\,x_1,x_v}) \in \{x_1, x_2\}$.
Therefore, $f(P_1^{\,x_2,x_1}, P_2^{\,x_1,x_v}) \in \{x_1,x_v\}\cap \{x_1, x_2\} = \{x_1\}$, which contradicts Claim 1, as required.
This proves the lemma.
\end{proof}

\vspace{-0.5em}
\begin{observation}\label{obs:cycle}\rm
According to Lemma \ref{lem:dictatorship1}, one would observe that
given a cycle $\mathcal{C} = (x_1, \dots, x_v, x_1)$ in $G_{\sim}^A$, i.e., $v \geq 3$, $x_1, \dots, x_v$ are pairwise distinct, and
$x_k \sim x_{k+1}$ for all $k \in \{1, \dots, v\}$, where $x_{v+1} = x_1$,
each two-voter, strategy-proof rule $f: \mathbb{D}^2 \rightarrow A$ must behave like a dictatorship on $\mathcal{C}$.
\end{observation}

\begin{lemma}\label{lem:union2}
Fixing a subset $B\subseteq A$ with $|B| \geq 3$, let $G_{\sim}^B$ be a connected graph.
Then, the following two statements hold:
\begin{itemize}
\item[\rm (i)] if $\emph{Leaf}(G_{\sim}^B) = \emptyset$, then $f$ behaves likes a dictatorship on $B$, and

\item[\rm (ii)] given $\emph{Leaf}(G_{\sim}^B) \neq \emptyset$,
if $f$ behaves like a dictatorship on $\{x, y\}$ for all $x \in \emph{Leaf}(G_{\sim}^B)$ and $(x,y) \in \mathcal{E}_{\sim}^B$,
then $f$ behaves like a dictatorship on $B$.
\end{itemize}
\end{lemma}

\begin{proof}
First, let $\textrm{Leaf}(G_{\sim}^B) = \emptyset$.
Note that for each $x \in B$, $x$ is included in either a cycle or a path that connects two distinct cycles.
Therefore, by Observation \ref{obs:cycle} and Lemma \ref{lem:union}, we infer that $f$ behaves like a dictatorship on $B$.

Next, let $\textrm{Leaf}(G_{\sim}^B) \neq \emptyset$ and $f$ behave like a dictatorship on $\{x, y\}$ for all $x \in \textrm{Leaf}(G_{\sim}^B)$ and $(x,y) \in \mathcal{E}_{\sim}^B$.
For notational convenience, let $\textrm{Leaf}(G_{\sim}^B) = \{x_1, \dots, x_t\}$ and $(x_k, y_k) \in \mathcal{E}_{\sim}^B$ for all $k \in \{1, \dots, t\}$.
Thus, $f$ behaves like a dictatorship on $\{x_k, y_k\}$ for all $k \in \{1, \dots, t\}$.
We consider two cases: $G_{\sim}^B$ is not a tree and $G_{\sim}^B$ is a tree.

In the first case, $G_{\sim}^B$ must include a cycle $\mathcal{C}$.
Then, we can identify a subset $\bar{B} \subset B$ such that $G_{\sim}^{\bar{B}}$ is a connected graph, $\textrm{Leaf}(G_{\sim}^{\bar{B}}) = \emptyset$ and
$[\bar{B} \subset \hat{B} \subseteq B] \Rightarrow \big[\textrm{Leaf}(G_{\sim}^{\hat{B}}) \neq \emptyset\big]$.
Then, by statement (i), $f$ behaves like a dictatorship on $\bar{B}$.
For each $1 \leq k \leq t$, since $G_{\sim}^B$ is a connected graph,
there exist $z_k \in \bar{B}$ and a path $\pi_k = (x_1, \dots, x_{v-1}, x_v)$ in $G_{\sim}^B$ that connects $z_k$ and $x_k$.
Clearly, $x_{v-1} = y_k$.
Then, Lemma \ref{lem:union} implies that $f$ behaves like a dictatorship on $\bar{B}\cup \pi_k$.
Last, since $G_{\sim}^{B}$ in fact is a combination of $G_{\sim}^{\bar{B}}$ and paths $\pi_1, \dots, \pi_t$,
by repeatedly applying Lemma \ref{lem:union}, we conclude that $f$ behaves like a dictatorship on $B$.

Last, we assume that $G_{\sim}^B$ is a tree.
Evidently, $G_{\sim}^B$ has at least two leaves, i.e., $t \geq 2$.
Note that for any two distinct $x_p, x_q \in \textrm{Leaf}(G_{\sim}^B)$, there exists a unique path $\pi_{p,q} = (z_1, z_2, \dots, z_{v-1}, z_v)$ in $G_{\sim}^B$ connecting $x_p$ and $x_q$. Clearly, $z_2 = y_p$ and $z_{v-1} = y_q$ (it is possible that $y_p= y_q$). Then, Lemma \ref{lem:union} implies that $f$ behaves like a dictatorship on $\pi$.
Last, since $G_{\sim}^{B}$ in fact is a combination of all paths $\{\pi_{p,q}: 1 \leq p< q \leq t\}$,
by repeatedly applying Lemma \ref{lem:union}, we conclude that $f$ behaves like a dictatorship on $B$.
\end{proof}

Now, we are ready to show that if the domain $\mathbb{D}$ satisfying path-connectedness and leaf symmetry is non-dictatorial domain,
then $\mathbb{D}$ satisfies the unique seconds property.
Suppose by contradiction that $\mathbb{D}$ violates the unique seconds property.
We show that $\mathbb{D}$ is a dictatorial domain.
By the ramification theorem of \citet{ACS2003}, it suffices to show that every two-voter, strategy-proof rule is a dictatorship.
Henceforth, we fix $N = \{1,2\}$ and a strategy-proof rule $f: \mathbb{D}^2 \rightarrow A$.

By statement (i) of Lemma \ref{lem:union2}, if $\textrm{Leaf}(G_{\sim}^A) =\emptyset$, then $f$ is a dictatorship.
Last, we assume $\textrm{Leaf}(G_{\sim}^A) \neq \emptyset$.
By statement (ii) of Lemma \ref{lem:union2},
it suffices to show that for each $x \in \textrm{Leaf}(G_{\sim}^A)$, $f$ behaves like a dictatorship on $x$ and its unique neighbor in $G_{\sim}^A$.
Fix arbitrary $x \in \textrm{Leaf}(G_{\sim}^A)$, let $(x, y) \in \mathcal{E}_{\sim}^A$.
we show that $f$ behaves like a dictatorship on $\{x, y\}$.
Clearly, $y \in \mathcal{S}(\mathbb{D}^x)$ and the violation of the unique seconds property implies $|\mathcal{S}(\mathbb{D}^x)|>1$.
Then, by leaf symmetry, we have $z \in \mathcal{S}(\mathbb{D}^x)\backslash \{y\}$ such that $x \in \mathcal{S}(\mathbb{D}^z)$.
Hence, we have  $P_i^{\,x, z},P_i^{\,z,x} \in \mathbb{D}$.
Since $G_{\sim}^A$ is a connected graph, there exists a path $\pi = (x_1,\dots, x_v)$ connecting $x$ and $z$.
Consequently, Lemma \ref{lem:dictatorship1} implies that $f$ behaves like a dictatorship on $\pi$.
Last, since $x \in \textrm{Leaf}(G_{\sim}^A)$ and $(x, y) \in \mathcal{E}_{\sim}^A$, it must be the case that $x_2 = y$.
Therefore, $f$ behaves like a dictatorship on $\{x, y\}$, as required.
This proves the Auxiliary Proposition.

\section{Proof of Theorem \ref{thm:invariance}}\label{app:Theorem}
We first introduce two independent lemmas (Lemmas \ref{lem:transitivity} and \ref{lem:RFBR}) which will be repeatedly applied in the proof of both statements of Theorem \ref{thm:invariance}.

For Lemmas \ref{lem:transitivity} and \ref{lem:RFBR},
we fix $N = \{1, 2\}$ and a two-voter, tops-only and strategy-proof rule $f: \mathbb{D}^2 \rightarrow A$.
Since $f$ satisfies the tops-only property, by abuse of notation $f(a, b)$ will represent the social outcome at a preference profile where voter 1 reports a preference with the peak $a$ and voter 2 reports a preference with the peak $b$. Also, $f(a, P_2)$ represents the social outcome at a profile where voter 1 reports a preference with the peak $a$ and voter 2 reports preference $P_2$.

\begin{lemma}\label{lem:transitivity}
Fixing a path $\pi =(x_1,\dots, x_v)$ in $G_{\sim}^A$,
the following statements hold:
\begin{itemize}
\item[\rm (i)] $f(x_s, x_t) \in \{x_s, \dots, x_t\}$ and $f(x_t, x_s) \in \{x_s, \dots, x_t\}$ for all $1 \leq s < t \leq v$,

\item[\rm (ii)] given $i \in N$, $P_i \in \mathbb{D}$ and $x_s \in \pi$,
if $f(P_i, x_s) = x \notin \pi$, then $f(P_i, x_k) = x$ for all $k \in \{1, \dots, v\}$, and

\item[\rm (iii)] given $i \in N$, $P_i \in \mathbb{D}$ and $x_s \in \pi$,
if $f(P_i, x_s) = x_s$, then $f(P_i, x_p) \in \{x_p , \dots, x_s\}$ for all $p \in \{1, \dots, s-1\}$, and
$f(P_i, x_q) \in \{x_s , \dots, x_q\}$ for all $q \in \{s+1, \dots, v\}$.
\end{itemize}
\end{lemma}

\begin{proof}
Fix $1 \leq s < t \leq v$.
Since $x_s \sim x_{s+1}$, it is easy to show $f(x_s, x_{s+1}) \in \{x_s, x_{s+1}\}$.
Next, we pick an arbitrary integer $l \in \{s+2, \dots, t\}$, and
provide an induction hypothesis: $f(x_s, x_{l'}) \in \{x_s, \dots, x_{l'}\}$ for all $l' \in \{s+1, \dots, l-1\}$.
We show $f(x_s, x_l) \in \{x_s, \dots, x_l\}$.
First, the induction hypothesis implies $f(x_s, x_{l-1}) \in \{x_s, \dots,x_{l-1}\}$.
Next, since $x_{l-1} \sim x_l$, we have $P_2, P_2' \in \mathbb{D}$ such that
$r_1(P_2) = r_2(P_2') = x_{l-1}$, $r_1(P_2') = r_2(P_2) = x_l$ and $r_k(P_2) = r_k(P_2')$ for all $k \in \{3, \dots, m\}$.
If $f(x_s, P_2)= f(x_s, x_{l-1}) \in \{x_s, \dots, x_{l-2}\}$, then strategy-proofness implies
$f(x_s, x_{l}) = f(x_s, P_2') = f(x_s, P_2) \in \{x_s, \dots, x_{l-2}\}$.
If $f(x_s, P_2)= f(x_s, x_{l-1}) = x_{l-1}$, then strategy-proofness implies
$f(x_s, x_{l}) = f(x_s, P_2') \in \{x_{l-1}, x_l\}$.
Overall, $f(x_s, x_{l}) \in \{x_s, \dots, x_{l}\}$.
This completes the verification of the induction hypothesis.
Therefore, $f(x_s, x_{t}) \in \{x_s, \dots, x_{t}\}$.
Symmetrically, $f(x_{t}, x_s) \in \{x_s, \dots, x_{t}\}$.
This completes the verification of statement (i).

Next, we show statement (ii).
By symmetry, we assume w.l.o.g.~that $i = 1$.
Given $1 \leq k \leq v$, either $1 \leq k \leq s$ or $s< k \leq v$ holds.
First, given $s< k \leq v$, we consider the path $(x_s, \dots, x_k)$.
Since $x_s \sim x_{s+1}$, we have $P_2, P_2' \in \mathbb{D}$ such that
$r_1(P_2) = r_2(P_2') = x_{s}$, $r_1(P_2') = r_2(P_2) = x_{s+1}$ and $r_k(P_2) = r_k(P_2')$ for all $k \in \{3, \dots, m\}$.
Since $x \notin \{x_s, x_{s+1}\}$, strategy-proofness implies $f(P_1, x_{s+1}) = f(P_1, P_2') = f(P_1, P_2) = f(P_1, x_s) =x$.
According to the path $(x_s, \dots, x_k)$ from $x_{s+1}$ to $x_k$, by repeatedly applying the same argument step by step,
we eventually have $f(P_1, x_k) = x$.
Symmetrically, if $1 \leq k \leq s$, we also induce $f(P_1, x_k) = x$.
This completes the verification of statement (ii).

Last, we prove statement (iii).
By symmetry, we assume w.l.o.g.~that $i = 1$.
Given $p \in \{1, \dots, s-1\}$,
suppose $f(P_1, x_p) = x \notin \{x_p, \dots, x_s\}$.
Then, according to the path $(x_p, \dots, x_s)$, statement (ii) implies $f(P_1, x_s) = x \neq x_s$ - a contradiction.
Therefore, $f(P_1, x_p) \in \{x_p, \dots, x_s\}$.
Similarly, $f(P_1, x_q) \in \{x_s, \dots, x_q\}$ for all $q \in \{s+1, \dots, v\}$.
This completes the verification of statement (iii).
\end{proof}

\vspace{-0.5em}
\begin{lemma}\label{lem:RFBR}
Fixing a path $\pi = (x_1, \dots, x_v)$, $v \geq 3$, in $G_{\sim}^A$,
let $f(x_1, x_v) = x_{\underline{k}}$ and $f(x_v, x_1) =x_{\overline{k}}$.
The following three statements hold: given $x_s, x_t \in \pi$,
\begin{itemize}

\item[\rm (i)]
$\left[\underline{k}< \overline{k}\,\right] \Rightarrow\;
\left[
f(x_s, x_t) = \left\{
\begin{array}{ll}
x_s & \emph{if}\; \underline{k}\leq s \leq \overline{k},\\
x_{\mathop{\emph{med}}(s,\, t, \,\underline{k})} & \emph{if}\; s< \underline{k},\; \emph{and}\\
x_{\mathop{\emph{med}}(s, \,t,\, \overline{k})} & \emph{if}\; s> \overline{k}.
\end{array}
\right.\right]$,

\item[\rm (ii)]
$\left[\underline{k}> \overline{k}\,\right] \Rightarrow\;
\left[f(x_s, x_t) = \left\{
\begin{array}{ll}
x_t & \emph{if}\; \overline{k}\leq t \leq \underline{k},\\
x_{\mathop{\emph{med}}(s, \,t, \,\overline{k})} & \emph{if}\; t< \overline{k},\;\emph{and}\\
x_{\mathop{\emph{med}}(s, \,t, \,\underline{k})} & \emph{if}\; t> \underline{k}.
\end{array}
\right.\right]$, and

\item[\rm (iii)]
$\left[\underline{k}= \overline{k} = k^{\ast}\,\right] \Rightarrow\;
\left[f(x_s, x_t) = x_{\mathop{\emph{med}}(s,\,t,\,k^{\ast})}\; \textrm{for all}\; s, t \in \{1, \dots, v\}\right]$.
\end{itemize}
\end{lemma}

\begin{proof}
First, according to $f(x_1, x_v) = x_{\underline{k}}$, we establish the following claim.\medskip

\noindent
\textsc{Claim 1}: We have
$f(x_k, x_{k'}) = \left\{
\begin{array}{ll}
x_{k'} & \textrm{if}\; 1 \leq k\leq k' \leq \underline{k},\\
x_k & \textrm{if}\; \underline{k} \leq k\leq k' \leq v,\; \textrm{and}\\
x_{\underline{k}} & \textrm{if}\; 1 \leq k < \underline{k} < k' \leq v.
\end{array}
\right.$\medskip

If $\underline{k} = 1$, the first part here follows immediately from unanimity.
Next, we assume $\underline{k}>1$ and show the first part.
According to statement (iv) of Lemma \ref{lem:step1} on the subpath $(x_1, \dots, x_{\underline{k}})$,
we know that if $f(x_{\underline{k}-1}, x_{\underline{k}}) = x_{\underline{k}}$, then we have $f(x_k, x_{k'}) = x_{k'}$ for all $1 \leq k\leq k' \leq \underline{k}$, as required.
Therefore, we focus on showing $f(x_{\underline{k}-1}, x_{\underline{k}}) = x_{\underline{k}}$.
Since $x_{\underline{k}-1}\sim x_{\underline{k}}$, it is evident that $f(x_{\underline{k}-1}, x_{\underline{k}}) \in \{x_{\underline{k}-1}, x_{\underline{k}}\}$.
Suppose by contradiction that $f(x_{\underline{k}-1}, x_{\underline{k}}) = x_{\underline{k}-1}$.
Then, according to the subpath $(x_1, \dots, x_{\underline{k}-1})$, statement (iii) of Lemma \ref{lem:transitivity} implies
$f(x_1, x_{\underline{k}}) \in \{x_1, \dots, x_{\underline{k}-1}\}$.
However, $f(x_1, x_v) = x_{\underline{k}}$ implies $f(x_1, x_{\underline{k}}) = x_{\underline{k}}$ by strategy-proofness - a contradiction.
In conclusion, we have $f(x_k, x_{k'}) = x_{k'}$ for all $1 \leq k\leq k' \leq \underline{k}$.

Symmetrically,
according to the subpath $(x_{\underline{k}}, \dots, x_v)$ (no matter $\underline{k} = v$ or $\underline{k}< v$ holds),
we can show $f(x_k, x_{k'})=x_k$ for all $\underline{k} \leq k\leq k' \leq v$, as required by the second part.

Last, we show the third part. Given $1 \leq k < \underline{k} < k' \leq v$, according to the subpath $(x_1, \dots, x_k)$,
by statement (ii) of Lemma \ref{lem:transitivity}, $f(x_1, x_v) = x_{\underline{k}}$ implies $f(x_k, x_v) = x_{\underline{k}}$.
Furthermore, according to the subpath $(x_{k'}, \dots, x_v)$,
by statement (ii) of Lemma \ref{lem:transitivity}, $f(x_k, x_v) = x_{\underline{k}}$ implies $f(x_k, x_{k'}) = x_{\underline{k}}$, as required by the third part.
This completes the verification of the claim.\medskip

Symmetrically, according to $f(x_v, x_1) =x_{\overline{k}}$, we can establish the claim below.\medskip

\noindent
\textsc{Claim 2}: We have
$f(x_{k'}, x_k) = \left\{
\begin{array}{ll}
x_{k'} & \textrm{if}\; 1 \leq k\leq k' \leq \overline{k},\\
x_k & \textrm{if}\; \overline{k} \leq k\leq k' \leq v,\;\textrm{and}\\
x_{\overline{k}} & \textrm{if}\; 1 \leq k < \overline{k} < k' \leq v.
\end{array}
\right.$\medskip

Last, we combine the two claims to prove the lemma.
Note that the verifications of the three statements are symmetric.
We focus on the verification of statement (i).
Let $\underline{k}< \overline{k}$ and fix an arbitrary profile $(x_s, x_t)$.

First, let $\underline{k}\leq s \leq \overline{k}$.
If $s \leq t$, we have $\underline{k}\leq s < t \leq v$ and $f(x_s, x_t) = x_s$ by Claim 1.
If $t < s$, we have $1 \leq t \leq s \leq \overline{k}$ and $f(x_s, x_t) = x_s$ by Claim 2.

Second, let $s < \underline{k}$.
If $t \leq s$, we have $1 \leq t < s < \underline{k} < \overline{k}$ and $f(x_s, x_t) = x_s = x_{\textrm{med}(s,\, t, \,\underline{k})}$ by Claim 2.
If $s < t \leq \underline{k}$, we have $1 \leq s < t \leq \underline{k}$ and
$f(x_s, x_t) = x_t = x_{\textrm{med}(s, \,t,\, \underline{k})}$ by Claim 1.
If $\underline{k}< t$, we have $1 \leq s < \underline{k} < t \leq v$ and $f(x_s, x_t) = x_{\underline{k}} = x_{\textrm{med}(s, \,t, \,\underline{k})}$ by Claim 1.

Last, let $s > \overline{k}$.
If $t < \overline{k}$, we have $1 \leq t < \overline{k}< s \leq v$ and $f(x_s, x_t) = x_{\overline{k}} = x_{\textrm{med}(s, \,t, \,\overline{k})}$ by Claim 2.
If $\overline{k}\leq t\leq s$, we have $\overline{k}\leq t < s \leq v$ and $f(x_s, x_t) = x_t = x_{\textrm{med}(s, t, \overline{k})}$ by Claim 2.
If $s< t$, we have $\underline{k}< \overline{k} < s < t \leq v$ and $f(x_s, x_t) = x_s = x_{\textrm{med}(s, t, \overline{k})}$ by Claim 1.
This proves statement (i).
\end{proof}

Now, we are ready to prove the two statements of Theorem \ref{thm:invariance}.\medskip

\noindent
\textbf{Proof of Statement (i)}.
Let $\mathbb{D}$ be a non-dictatorial, unidimensional domain.
To prove \textbf{Statement (i)}, we show the equivalence of the following three sub-statements:
\begin{itemize}
\item[\rm (1)] there exists an invariant, tops-only and strategy-proof rule,

\item[\rm (2)] domain $\mathbb{D}$ is a semi-single-peaked domain, and

\item[\rm (3)] domain $\mathbb{D}$ admits a two-voter, strategy-proof projection rule.
\end{itemize}

We first show the direction: (2) $\Rightarrow$ (3) $\Rightarrow$ (1).
By the proof of the sufficiency part of the Theorem of \citet{CSS2013}, we know that
if $\mathbb{D}$ is a semi-single-peaked domain, it admits a two-voter, strategy-proof projection rule,
which of course is an invariant, tops-only and strategy-proof rule.
In the rest of the proof, we focus on showing (1) $\Rightarrow$ (2).
More specifically, we complete the proof via the following three steps:\medskip

\noindent
\textbf{Step 1.} We show that the existence of invariant, tops-only and strategy-proof rule implies that the adjacency graph $G_{\sim}^A$ is a tree (see Lemmas \ref{lem:uniquepath} and \ref{lem:tree}).

\noindent
\textbf{Step 2.} Given an invariant, tops-only and strategy-proof rule $f: \mathbb{D}^2 \rightarrow A$ admitted by $\mathbb{D}$,
we characterize the SCF $f$ to be a projection rule on the tree $G_{\sim}^A$ w.r.t.~the threshold which equals the same social outcome at the two test profiles $(\underline{P}_1, \overline{P}_2)$ and $(\overline{P}_1, \underline{P}_2)$ (see Lemma \ref{lem:projection}).
This of course implies that every invariant, tops-only and strategy-proof rule defined on $\mathbb{D}$ is a projection rule.

\noindent
\textbf{Step 3.} By adopting strategy-proofness of the projection rule $f$ characterized in \textbf{Step 2}, we show that $\mathbb{D}$ a semi-single-peaked domain (see Lemma \ref{lem:semisinglepeaked}). This proves (1) $\Rightarrow$ (2).

\begin{lemma}\label{lem:uniquepath}
Recall the two completely reversed preferences $\underline{P}_i$ and $\overline{P}_i$, and their peaks
$r_1(\underline{P}_i) = a_1$ and $r_1(\overline{P}_i) = a_m$.
There exists a unique path in $G_{\sim}^A$ connecting $a_1$ and $a_m$.
\end{lemma}

\begin{proof}
Since $G_{\sim}^A$ is a connected graph, there exists a path in $G_{\sim}^A$ connecting $a_1$ and $a_m$.
Suppose by contradiction that there are two distinct paths
$\pi = (x_1, \dots, x_p)$ and $\pi' = (y_1, \dots, y_q)$ in $G_{\sim}^A$ connecting $a_1$ and $a_m$.
Then, we can identify $1 \leq s < t \leq p$ and $1 \leq s' < t' \leq q$ with either $t-s>1$ or $t'-s' >1$
such that $x_s = y_{s'}$, $x_{t} = y_{t'}$ and $\{x_{s+1}, \dots, x_{t-1}\}\cap \{y_{s'+1}, \dots, y_{t'-1}\} = \emptyset$.
Consequently, we construct a cycle $\mathcal{C} = (x_s, \dots, x_t = y_{t'}, \dots, y_{s'+1}, y_{s'} = x_s)$.
By Observation \ref{obs:cycle} in Appendix \ref{app:AP}, every two-voter, strategy-proof rule behaves like a dictatorship on $\mathcal{C}$.
Fixing an arbitrary two-voter, tops-only and strategy-proof rule $g: \mathbb{D}^2 \rightarrow A$,
we assume w.l.o.g.~that voter 1 dictates on $\mathcal{C}$, i.e.,
$g(P_1, P_2) = r_1(P_1)$ for all $P_1, P_2 \in \mathbb{D}$ with $r_1(P_1), r_1(P_2) \in \mathcal{C}$.
Thus, $g(x_s, x_t) = x_s$ and $g(x_t, x_s) = x_t$.
According to subpaths $(x_t, \dots, x_p)$ and $(x_s, \dots, x_1)$,
by statements (ii) of Lemma \ref{lem:transitivity}, $g(x_s, x_t) = x_s$ implies $g(x_s, x_p) = x_s$, and $g(x_t, x_s) = x_t$ implies $g(x_t, x_1) = x_t$.
Furthermore, according to subpaths $(x_s, \dots, x_1)$ and $(x_t, \dots, x_p)$,
by statements (iii) of Lemma \ref{lem:transitivity}, $g(x_s, x_p) = x_s$ implies $g(x_1, x_p) \in \{x_1, \dots, x_s\}$ and $g(x_t, x_1) = x_t$ implies $g(x_p, x_1) \in \{x_t, \dots, x_p\}$.
Consequently, we have $g(\underline{P}_1, \overline{P}_2) = g(x_1, x_p) \neq g(x_p, x_1)=g(\overline{P}_1, \underline{P}_2)$, which
implies that $g$ violates invariance.
In conclusion, all two-voter, tops-only and strategy-proof rules violate invariance.
This contradicts the hypothesis that there exists an invariant, tops-only and strategy-proof rule, and hence proves the lemma.
\end{proof}

Let $\pi^{\ast} = (x_1, \dots, x_p)$ be the unique path connecting $a_1$ and $a_m$ in $G_{\sim}^A$.
Note that this path may not include all alternatives of $A$ (recall the adjacency graph $G_{\sim}^A$ in Figure \ref{fig:connectedness}).

\begin{lemma}\label{lem:tree}
The graph $G_{\sim}^A$ is a tree.
\end{lemma}

\vspace{-0.5em}
\begin{proof}
Suppose not, i.e., there exists a cycle $\mathcal{C} = (b_1, \dots, b_v, b_1)$, $v \geq 3$.
By Observation \ref{obs:cycle} in Appendix \ref{app:AP}, every two-voter, strategy-proof rule behaves like a dictatorship on $\mathcal{C}$.
Fixing an arbitrary two-voter, tops-only and strategy-proof rule $g: \mathbb{D}^2 \rightarrow A$,
we assume w.l.o.g.~that voter 1 dictates on $\mathcal{C}$, i.e.,
$g(P_1, P_2) = r_1(P_1)$ for all $P_1, P_2 \in \mathbb{D}$ with $r_1(P_1), r_1(P_2) \in \mathcal{C}$.

We know either $\mathcal{C}\cap \pi^{\ast} = \emptyset$ or $\mathcal{C}\cap \pi^{\ast} \neq \emptyset$.
If $\mathcal{C}\cap \pi^{\ast} = \emptyset$, we can identify $b_s \in \mathcal{C}$, $x_k \in \pi^{\ast}$ and a path
$(y_1, \dots, y_u)$ in $G_{\sim}^A$ connecting $b_s$ and $x_k$ such that
$y_p \notin \mathcal{C}\cup \pi^{\ast}$ for all $p \in \{2, \dots, u-1\}$ (see the first diagram of Figure \ref{fig:cycle}).
If $\mathcal{C}\cap \pi^{\ast} \neq \emptyset$,
we must identify \emph{a unique} alternative of $\pi^{\ast}$ that is contained in $\mathcal{C}$, say $b_s = x_k \in\pi^{\ast}$ (see the second diagram of Figure \ref{fig:cycle}); for otherwise, we can identify two distinct paths in $G_{\sim}^A$ connecting $a_1$ and $a_m$, which contradicts Lemma \ref{lem:uniquepath}.
Overall, we have the cycle $\mathcal{C} = (b_1, \dots, b_v, b_1)$, the path $\pi^{\ast } = (x_1, \dots, x_p)$ and the path
$(y_1, \dots, y_u)$ which may be a null path when $b_s = x_k$.
We consider three cases of $x_k$ on $\pi^{\ast}$: $1< k < p$, $k = 1$ and $k = p$.
In each case, we show that $g$ violates invariance.\bigskip

\hspace{0.5em}
\begin{figurehere}
\begin{tikzpicture}

\begin{scope}[thin]
\draw (0,-1) ellipse (1.2 and 0.4);
\end{scope}

\node at (-3,-1.7) {{\footnotesize$\bullet$}};
\node at (0,-1.7) {{\footnotesize$\bullet$}};
\node at (3,-1.7) {{\footnotesize$\bullet$}};
\node at (0,-1.4) {{\footnotesize$\bullet$}};
\node at (1.2,-1) {{\footnotesize$\bullet$}};

\node at (-3,-2) {$x_1$};
\node at (3,-2) {$x_p$};
\node at (0,-1.95) {$x_{k}$};
\node at (0,-1.1) {$b_s$};
\node at (1.5,-1) {$b_t$};

\draw (-3,-1.7) -- (3,-1.7);
\draw (0,-1.7) -- (0,-1.4);


\begin{scope}[thin]
\draw (7.93,-1.18) ellipse (1 and 0.5);
\end{scope}

\node at (5,-1.7) {{\footnotesize$\bullet$}};
\node at (8,-1.7) {{\footnotesize$\bullet$}};
\node at (11,-1.7) {{\footnotesize$\bullet$}};
\node at (8.92,-1.17) {{\footnotesize$\bullet$}};

\node at (5,-2) {$x_1$};
\node at (11,-2) {$x_p$};
\node at (8,-1.95) {$x_k$};
\node at (8,-1.4) {$b_s$};
\node at (9.2,-1.15) {$b_t$};

\draw (5,-1.7) -- (11,-1.7);
%
\end{tikzpicture}
\vspace{-1em}
\caption{The relation between the cycle $\mathcal{C}$ and the path $\pi^{\ast}$}\label{fig:cycle}
\end{figurehere}\medskip

In the first case $1< k < p$,
fixing $b_t \in \mathcal{C}\backslash  \{b_s\}$, we have $g(b_t, b_s) = b_t$ and $g(b_s, b_t) = b_s$ by voter 1's dictatorship on $\mathcal{C}$.
According to paths $(b_s=y_1, \dots, y_u=x_k,\dots, x_1)$ and $(b_s=y_1, \dots, y_u=x_k, \dots, x_p)$, by statement (ii) of Lemma \ref{lem:transitivity},
$g(b_t, b_s) = b_t$ implies $g(b_t, x_1) = b_t$ and $g(b_t, x_p) = b_t$.
Furthermore, according to paths $(b_t, \dots, b_s=y_1, \dots, y_u=x_k, \dots, x_p)$ and $(b_t, \dots, b_s=y_1, \dots, y_u=x_k,\dots, x_1)$,
by statement (iii) of Lemma \ref{lem:transitivity},
$g(b_t, x_1) = b_t$ implies $g(x_p, x_1) \in \{b_t, \dots, b_s=y_1, \dots, y_u=x_k, \dots, x_p\}$,
and $g(b_t, x_p) = b_t$ implies $g(x_1, x_p) \in \{b_t, \dots, b_s=y_1, \dots, y_u=x_k,\dots, x_1\}$.
Furthermore, according to $\pi^{\ast}$, statement (i) of Lemma \ref{lem:transitivity} implies
$g(x_p, x_1) \in \{x_1, \dots, x_p\}$ and $g(x_1, x_p) \in \{x_1, \dots, x_p\}$.
Therefore, we have $g(x_p, x_1) \in \{b_t, \dots, b_s=y_1, \dots, y_u=x_k, \dots, x_p\} \cap \{x_1, \dots, x_p\} = \{x_k, \dots, x_p\}$ and
$g(x_1, x_p) \in \{b_t, \dots, b_s=y_1, \dots, y_u=x_k,\dots, x_1\} \cap \{x_1, \dots, x_p\} = \{x_1, \dots, x_k\}$.
Thus, $g(x_p, x_1) = x_{\overline{k}}$ for some $\overline{k} \in \{k, \dots, p\}$ and
$g(x_1, x_p) = x_{\underline{k}}$ for some $\underline{k} \in \{1, \dots, k\}$.
Now, given $\overline{P}_1 = \overline{P}_i$ and $\underline{P}_1 = \underline{P}_i$,
according to $g(b_t, x_1) = b_t$ and $g(\overline{P}_1, x_1) = g(x_p, x_1) = x_{\overline{k}}$, and
$g(b_t, x_p) = b_t$ and $g(\underline{P}_1, x_p) = g(x_1, x_p) = x_{\underline{k}}$,
strategy-proofness implies $x_{\overline{k}}\mathrel{\overline{P}_1}b_t$ and $x_{\underline{k}}\mathrel{\underline{P}_1}b_t$.
Consequently, by the fact that $\overline{P}_1$ and $\underline{P}_1$ are complete reversals,
it must be the case that $x_{\overline{k}} \neq x_{\underline{k}}$.
Therefore, $g(\underline{P}_1, \overline{P}_2) = g(x_1, x_p) \neq g(x_p, x_1) =g(\overline{P}_1, \underline{P}_2)$,
which indicates that $g$ violates invariance.
This completes the verification of the first case.

The second and third cases are symmetric.
We focus on the verification of the second case $k  = 1$.
Fixing $b_t \in \mathcal{C}\backslash \{b_s\}$, we have $g(b_t, b_s) = b_t$.
According to the path $(b_s=y_1, \dots, y_u=x_1, \dots, x_p)$,
by statement (ii) of Lemma \ref{lem:transitivity},
$g(b_t, b_s) = b_t$ implies $g(b_t, x_1) = b_t$ and $g(b_t, x_p) = b_t$.
Furthermore, according to the path $(b_t, \dots, b_s = y_1, \dots, y_u = x_1)$,
by statement (iii) of Lemma \ref{lem:transitivity}, $g(b_t, x_p) = b_t$ implies $g(x_1, x_p) \in \{b_t, \dots, b_s = y_1, \dots, y_u = x_1\}$.
Meanwhile, according to $\pi^{\ast}$,
statement (i) of Lemma \ref{lem:transitivity} implies $g(x_1, x_p)\in \{x_1, \dots, x_p\}$.
Hence, it must be the case that $g(x_1, x_p) = x_1$.
Similarly, according to $\pi^{\ast}$,
statement (i) of Lemma \ref{lem:transitivity} implies $g(x_p, x_1)\in \{x_1, \dots, x_p\}$.
Thus, given $\overline{P}_1 = \overline{P}_i$, we have $g(\overline{P}_1, x_1)=g(x_p, x_1) =x_q$ for some $q \in \{1, \dots, p\}$.
Consequently, according to $g(b_t, x_1) = b_t$ and $g(\overline{P}_1, x_1) =x_q$,
strategy-proofness implies $x_q\mathrel{\overline{P}_1}b_t$, which further implies that $x_q$ is never bottom-ranked in $\overline{P}_i$.
Since $\overline{P}_i$ and $\underline{P}_i$ are complete reversals and $r_1(\underline{P}_i) = a_1=x_1$,
$x_1$ must be the bottom-ranked alternative in $\overline{P}_i$.
Therefore, $g(\underline{P}_1, \overline{P}_2) = g(x_1, x_p) = x_1 \neq x_q= g(x_p, x_1) =g(\overline{P}_1, \underline{P}_2)$,
which indicates that $g$ violates invariance.
This completes the verification of the second case.

In conclusion, all two-voter, tops-only and strategy-proof rules violate invariance.
This contradicts the hypothesis that there exists an invariant, tops-only and strategy-proof rule.
This proves the lemma, and completes the verification in \textbf{Step 1}.
\end{proof}

Now, we start the proof in \textbf{Step 2}.
We fix an arbitrary invariant, tops-only and strategy-proof rule $f: \mathbb{D}^2 \rightarrow A$.
According to the unique path $\pi^{\ast} = (x_1, \dots, x_p)$ in the tree $G_{\sim}^A$ connecting $a_1$ and $a_m$,
by the tops-only property, invariance and statement (i) of Lemma \ref{lem:transitivity}, we know $f(x_1, x_p) =f(\underline{P}_1, \overline{P}_2) =f(\overline{P}_1, \underline{P}_2) = f(x_p, x_1) = x_{\bar{k}}$ for some $\bar{k} \in \{1, \dots, p\}$.

\begin{lemma}\label{lem:projection}
SCF $f$ is a projection rule on the tree $G_{\sim}^A$ w.r.t.~$x_{\bar{k}}$, i.e., $f(y, z) = \mathop{\emph{Proj}}\big(x_{\bar{k}}, \langle y, z |G_{\sim}^A\rangle\big)$ for all $y, z \in A$.
\end{lemma}

\begin{proof}
The proof consists of the following three claims.
\medskip

\noindent
\textsc{Claim 1:} Given $y \in A$,
we have $f(y, x_{\bar{k}}) = f(x_{\bar{k}}, y) = x_{\bar{k}}$.

\medskip
Recall the unique path $\pi^{\ast} = (x_1, \dots, x_p)$ connecting $a_1$ and $a_m$ in $G_{\sim}^A$.
There are two cases: $y \in \pi^{\ast}$ and $y \notin \pi^{\ast}$.
If $y \in \pi^{\ast}$,
then $y = x_k$ for some $k \in \{1, \dots, p\}$, and we hence have
$f(y, x_{\bar{k}}) = f(x_k, x_{\bar{k}}) = \mathop{\textrm{Proj}}\big(x_{\bar{k}}, \langle x_k, x_{\bar{k}}|G_{\sim}^A\rangle\big) = x_{\bar{k}}$ and
$f(x_{\bar{k}}, y) = f(x_{\bar{k}}, x_k) = \mathop{\textrm{Proj}}\big(x_{\bar{k}}, \langle x_{\bar{k}}, x_k|G_{\sim}^A\rangle\big) = x_{\bar{k}}$.

Henceforth, we assume $y \notin \pi^{\ast}$.
We focus on showing $f(y, x_{\bar{k}}) = x_{\bar{k}}$.
By a symmetric proof, one would immediately conclude $f(x_{\bar{k}}, y) = x_{\bar{k}}$.
Let $\langle x_{\bar{k}}, y|G_{\sim}^A\rangle = (z_1, \dots, z_v)$ be the path connecting $x_{\bar{k}}$ and $y$ in the tree $G_{\sim}^A$.
We first show $f(z_2, z_1) = z_1$.
If $z_2 \in \pi^{\ast}$, then $z_2 = x_k$ for some $k \in \{1, \dots, p\}$ and hence we have
$f(z_2, z_1) = f(x_k, x_{\bar{k}}) = \mathop{\textrm{Proj}}\big(x_{\bar{k}}, \langle x_k, x_{\bar{k}}|G_{\sim}^A\rangle\big) = x_{\bar{k}} = z_1$.
Next, let $z_2 \notin \pi^{\ast}$.
Since $z_1 \sim z_2$, it is evident that $f(z_2, z_1) \in \{z_1, z_2\}$.
Suppose by contradiction that $f(z_2, z_1) = z_2$.
We have three cases: $\bar{k} = 1$, $\bar{k} = p$ and $1< \bar{k}< p$.
In each case, we show $f(z_2, x_1) = z_2$ and $f(z_2, x_p) = z_2$.
Note that the first two cases are symmetric, and we hence omit the verification in the second case.
In the first case $\bar{k} = 1$,
the contradictory hypothesis immediately implies $f(z_2, x_1) = f(z_2, z_1) =z_2$.
Furthermore, according to the path $(z_1 = x_{\bar{k}}, \dots, x_p)$, by statement (ii) of Lemma \ref{lem:transitivity}, $f(z_2, z_1) = z_2$ implies $f(z_2, x_p) = z_2$.
In the third case $1< \bar{k}< p$,
according to paths $(z_1 = x_{\bar{k}}, \dots, x_1)$ and $(z_1 = x_{\bar{k}}, \dots, x_p)$,
statement (ii) of Lemma \ref{lem:transitivity} implies $f(z_2, x_1) = z_2$ and $f(z_2, x_p) = z_2$.
Overall, we have $f(z_2, x_1) = z_2$ and $f(z_2, x_p) = z_2$.
Given $\overline{P}_1 = \overline{P}_i$ and $\underline{P}_1 = \underline{P}_i$,
we have $f(\overline{P}_1, x_1) = f(x_p, x_1) = x_{\bar{k}} = z_1$ and $f(\underline{P}_1, x_p) = f(x_1, x_p) = x_{\bar{k}} = z_1$.
Then, strategy-proofness implies $z_1\mathrel{\overline{P}_1}z_2$ (according to $(\overline{P}_1, x_1)$ and $(z_2, x_1)$) and
$z_1\mathrel{\underline{P}_1}z_2$ (according to $(\underline{P}_1, x_p)$ and $(z_2, x_p)$).
This contradicts the fact that $\overline{P}_1$ and $\underline{P}_1$ are complete reversals.
Therefore, $f(z_2, z_1) = z_1$.
Then, according to the path $(z_2, \dots, z_v)$, statement (ii) of Lemma \ref{lem:transitivity} implies
$f(y, x_{\bar{k}}) = f(z_v, z_1) = z_1 = x_{\bar{k}}$. This proves the claim.\medskip

Henceforth, we fix arbitrary $y, z \in A$ and
let $\langle y, z|G_{\sim}^A\rangle = (y_1, \dots, y_u)$.
There are three cases:
(1) $\mathop{\textrm{Proj}}\big(x_{\bar{k}}, \langle y, z|G_{\sim}^A\rangle\big) = y$,
(2) $\mathop{\textrm{Proj}}\big(x_{\bar{k}}, \langle y, z|G_{\sim}^A\rangle\big) = z$ and
(3) $\mathop{\textrm{Proj}}\big(x_{\bar{k}}, \langle y, z|G_{\sim}^A\rangle\big) = y_l$ for some $1 < l < u$.
In each case, we show $f(y, z) = \mathop{\textrm{Proj}}\big(x_{\bar{k}}, \langle y, z |G_{\sim}^A\rangle\big)$.
Note that the first two cases are symmetric, and we hence omit the verification in case (2).
\medskip

\noindent
\textsc{Claim 2}: In case (1), $f(y, z) = y = \mathop{\textrm{Proj}}\big(x_{\bar{k}}, \langle y, z |G_{\sim}^A\rangle\big)$.\medskip

If $y = x_{\bar{k}}$, this claim follows from Claim 1.
Next, assume $y \neq x_{\bar{k}}$.
Let $(b_1, \dots, b_v)$ denote the path in $G_{\sim}^A$ connecting $y$ and $x_{\bar{k}}$.
Since $\mathop{\textrm{Proj}}\big(x_{\bar{k}}, \langle y, z|G_{\sim}^A\rangle\big)=y$,
we have a concatenated path $\pi = (z = y_u, \dots, y_1 = y = b_1, \dots, b_v = x_{\bar{k}})$.
Since $f(b_v, b_{v-1}) = f(x_{\bar{k}}, b_{v-1}) = x_{\bar{k}}=  b_v$,
according to the path $\pi$, statement (iii) of Lemma \ref{lem:step1} in Appendix \ref{app:AP} implies $f(y, z) = y = \mathop{\textrm{Proj}}\big(x_{\bar{k}}, \langle y, z |G_{\sim}^A\rangle\big)$.
This completes the verification of the claim.

\medskip
\noindent
\textsc{Claim 3}: In case (3), $f(y, z) = y_l = \mathop{\textrm{Proj}}\big(x_{\bar{k}}, \langle y, z |G_{\sim}^A\rangle\big)$.\medskip

First, let $(c_1, \dots, c_v)$ denote the path connecting $x_{\bar{k}}$ and $y_l$ in $G_{\sim}^A$,
which may be a null path if $x_{\bar{k}} = y_l$.
According to the path $\langle y, z|G_{\sim}^A\rangle = (y_1, \dots, y_u)$, statement (i) of Lemma \ref{lem:transitivity} implies
$f(y, z) = y_k$ for some $ k \in \{1, \dots, u\}$.
Suppose $k \neq l$. Thus, either $1 \leq k < l$ or $l < k \leq u$ holds.
Moreover, since $\mathop{\textrm{Proj}}\big(x_{\bar{k}}, \langle y, z|G_{\sim}^A\rangle\big) = y_l$,
$k \neq l$ also implies $y_k \neq x_{\bar{k}}$.
If $1 \leq k < l$, according to the concatenated path $\pi' = (z = y_u, \dots, y_l = c_v, \dots, c_1 = x_{\bar{k}})$,
by statement (ii) of Lemma \ref{lem:transitivity}, $f(y, z) = y_k \notin \pi'$ implies $f(y, x_{\bar{k}}) = y_k\neq x_{\bar{k}}$,
which contradicts Claim 1.
Symmetrically, if $l < k \leq u$, according to the concatenated path $\pi'' =(y = y_1, \dots, y_l=c_v, \dots, c_1=x_{\bar{k}})$,
by statement (ii) of Lemma \ref{lem:transitivity}, $f(y, z) = y_k \notin \pi''$ implies $f(x_{\bar{k}}, z) = y_k \neq x_{\bar{k}}$,
which contradicts Claim 1.
Therefore, $f(y, z) = y_l = \mathop{\textrm{Proj}}\big(x_{\bar{k}}, \langle y, z |G_{\sim}^A\rangle\big)$.
This completes the verification of the claim, and hence proves the lemma.
This completes the proof in \textbf{Step 2}.
\end{proof}

Last, we move to \textbf{Step 3} and show that $\mathbb{D}$ is a semi-single-peaked domain.

\begin{lemma}\label{lem:semisinglepeaked}
Domain $\mathbb{D}$ is a semi-single-peaked domain.
\end{lemma}

\begin{proof}
Since $G_{\sim}^A$ is a tree, it suffices to show $\mathbb{D} \subseteq \mathbb{D}_{\textrm{SSP}}(G_{\sim}^A, x_{\bar{k}})$.
Fixing arbitrary $P_1 \in \mathbb{D}$, let $r_1(P_1) = x$.
First, given distinct $a, b \in \langle x, x_{\bar{k}}|G_{\sim}^A\rangle$ and $a \in \langle x, b|G_{\sim}^A\rangle$, we show $a\mathrel{P_1}b$.
By Lemma \ref{lem:projection},
we have $f(P_1, a) = \mathop{\textrm{Proj}}\big(x_{\bar{k}}, \langle x, a|G_{\sim}^A\rangle\big) = a$ and
$f(b, a) = \mathop{\textrm{Proj}}\big(x_{\bar{k}}, \langle b, a|G_{\sim}^A\rangle\big) = b$.
Then, strategy-proofness implies $a\mathrel{P_1}b$, as required.
Next, given $a \notin \langle x, x_{\bar{k}}|G_{\sim}^A\rangle$ and $\mathop{\textrm{Proj}}\big(a, \langle x, x_{\bar{k}}|G_{\sim}^A\rangle\big) = a'$,
we show $a'\mathrel{P_1}a$.
By Lemma \ref{lem:projection}, we have $f(P_1, a) = \mathop{\textrm{Proj}}\big(x_{\bar{k}}, \langle x, a|G_{\sim}^A\rangle\big)=
\mathop{\textrm{Proj}}\big(a, \langle x, x_{\bar{k}}|G_{\sim}^A\rangle\big) = a'$ and $f(a, a) = a$.
Then, strategy-proofness implies $a'\mathrel{P_1}a$, as required.
This completes the proof in \textbf{Step 3}, and proves \textbf{Statement (i)} of Theorem \ref{thm:invariance}.
\end{proof}

Now, we will focus on the proof of \textbf{Statement (ii)} of Theorem \ref{thm:invariance}.\medskip

\noindent
\textbf{Proof of Statement (ii)}:
Let $\mathbb{D}$ be a non-dictatorial, unidimensional domain.
We first show the ``if part'' of \textbf{Statement (ii)}:
``There exists no invariant, tops-only and strategy-proof rule.'' $\Leftarrow$ ``Domain $\mathbb{D}$ is a semi-hybrid domain.''
More specifically, let $\mathbb{D}$ be an $(a,b)$-semi-hybrid domain on a tree $\mathcal{T}^A$.
The proof consists of the following two lemmas.
Recall the two completely reversed preferences $\underline{P}_i$ and $\overline{P}_i$ included in $\mathbb{D}$ by diversity, and their peaks $r_1(\underline{P}_i) = a_1$ and $r_1(\overline{P}_i) =a_m$.
The following lemma indicates that we can assume w.l.o.g.~that $a_1 \in A^{a \rightharpoonup b}$ and $a_m \in A^{b \rightharpoonup a}$.

\begin{lemma}\label{lem:wlog}
There exist a tree $\widehat{\mathcal{T}}^A$ and dual-thresholds $\hat{a}, \hat{b} \in A$ such that
$\mathbb{D}$ is an $(\hat{a}, \hat{b})$-semi-hybrid domain on $\widehat{\mathcal{T}}^A$, $a_1 \in \widehat{A}^{\hat{a} \rightharpoonup \hat{b}} = \big\{x \in A: \hat{a} \in \langle x, \hat{b}|\widehat{\mathcal{T}}^A\rangle\big\}$ and
$a_m \in \widehat{A}^{\hat{b} \rightharpoonup \hat{a}} = \big\{x \in A: \hat{b} \in \langle x, \hat{a}|\widehat{\mathcal{T}}^A\rangle\big\}$.
\end{lemma}

\begin{proof}
There are four situations:
(1) $A^{a \rightharpoonup b} \neq \{a\}$ and $A^{b \rightharpoonup a} \neq \{b\}$,
(2) $A^{a \rightharpoonup b} = \{a\}$ and $A^{b \rightharpoonup a} = \{b\}$,
(3) $A^{a \rightharpoonup b} = \{a\}$ and $A^{b \rightharpoonup a} \neq \{b\}$, and
(4) $A^{a \rightharpoonup b} \neq \{a\}$ and $A^{b \rightharpoonup a} = \{b\}$.
In each case, we construct a tree $\widehat{\mathcal{T}}^A$ and identify dual-thresholds $\hat{a}, \hat{b} \in A$ such that
$\mathbb{D}$ is an $(\hat{a}, \hat{b})$-semi-hybrid domain on $\widehat{\mathcal{T}}^A$, $a_1 \in \widehat{A}^{\hat{a} \rightharpoonup \hat{b}}$ and $a_m \in \widehat{A}^{\hat{b} \rightharpoonup \hat{a}}$.

In Situation (1), we have $c \in A^{a \rightharpoonup b} \backslash \{a\}$ and $d \in A^{b \rightharpoonup a} \backslash \{b\}$.
We first claim $a_1 \notin \langle a,b|\mathcal{T}^A\rangle$.
Suppose not, i.e., $a_1 \in \langle a,b|\mathcal{T}^A\rangle$.
Then, by $(a,b)$-semi-hybridness on $\mathcal{T}^A$, we have $b\mathrel{\underline{P}_i}d$ and $a\mathrel{\underline{P}_i}c$.
Note that either $a_m \in A^{a \rightharpoonup b}\backslash \{a\}$, or $a_m \in \langle a,b|\mathcal{T}^A\rangle\cup A^{b \rightharpoonup a}$ holds, which respectively by $(a,b)$-semi-hybridness on $\mathcal{T}^A$ implies $b\mathrel{\overline{P}_i}d$ and $a\mathrel{\overline{P}_i}c$.
This contradicts the fact that $\underline{P}_i$ and $\overline{P}_i$ are complete reversals.
Symmetrically, $a_m \notin \langle a,b|\mathcal{T}^A\rangle$.
Thus, there are four cases:
(i) $a_1, a_m \in A^{a \rightharpoonup b}\backslash \{a\}$,
(ii) $a_1, a_m \in A^{b \rightharpoonup a}\backslash \{b\}$,
(iii) $a_1 \in A^{a \rightharpoonup b}\backslash \{a\}$ and $a_m \in A^{b \rightharpoonup a}\backslash \{b\}$, and
(iv) $a_1 \in A^{b \rightharpoonup a}\backslash \{b\}$ and $a_m \in A^{a \rightharpoonup b}\backslash \{a\}$.
We first rule out case (i).
In case (i), by $(a,b)$-semi-hybridness on $\mathcal{T}^A$, we have $b\mathrel{\underline{P}_i}d$ and $b\mathrel{\overline{P}_i}d$,
which contradict the fact that $\underline{P}_i$ and $\overline{P}_i$ are complete reversals.
Symmetrically, we can rule out case (ii).
In case (iii), it is evident that $\mathbb{D}$ is an $(a,b)$-semi-hybrid domain on $\mathcal{T}^A$, $a_1 \in A^{a \rightharpoonup b}$ and $a_m \in A^{b \rightharpoonup a}$, as required.
In case (iv), let $\hat{a} = b$ and $\hat{b} = a$.
Thus, $a_1 \in A^{b \rightharpoonup a} = A^{\hat{a} \rightharpoonup \hat{b}}$ and
$a_m \in A^{a \rightharpoonup b} = A^{\hat{b} \rightharpoonup \hat{a}}$.
Evidently, since $\mathbb{D}$ is an $(a,b)$-semi-hybrid domain on $\mathcal{T}^A$,
it is true that $\mathbb{D}$ is an $(\hat{a},\hat{b})$-semi-hybrid domain on $\mathcal{T}^A$, as required.

In Situation (2), we refer to the line $\mathcal{L}^A = (a_1, \dots, a_m)$ and the dual-thresholds $a_1$ and $a_m$, and
show that $\mathbb{D}$ is an $(a_1, a_m)$-semi-hybrid domain on $\mathcal{L}^A$.
First, it is clear that $\mathbb{D} \subseteq \mathbb{P} = \mathbb{D}_{\textrm{SH}}(\mathcal{L}^A, a_1, a_m)$.
Next, since $\mathbb{D}$ is an $(a,b)$-semi-hybrid domain on $\mathcal{T}^A$ and $\langle a, b|\mathcal{T}^A\rangle = A$,
we know by condition (ii) of Definition \ref{def:asspsh} that there exist no tree $\widetilde{\mathcal{T}}^A$ and dual-thresholds $\tilde{a}, \tilde{b} \in A$ such that $\mathbb{D} \subseteq \mathbb{D}_{\textrm{SH}}(\widetilde{\mathcal{T}}^A, \tilde{a}, \tilde{b})$ and
$\langle \tilde{a}, \tilde{b}|\widetilde{\mathcal{T}}^A\rangle \subset \langle a, b|\mathcal{T}^A\rangle = A =\langle a_1, a_m|\mathcal{L}^A\rangle$.
Last, by condition (iii) of Definition \ref{def:asspsh},
we know that if $G_{\sim}^A$ is a tree,
then for each $x \in \textrm{Leaf}\big(G_{\sim}^{\langle a, b|\mathcal{T}^A\rangle}\big) = \textrm{Leaf}\big(G_{\sim}^{A}\big)
=\textrm{Leaf}\big(G_{\sim}^{\langle a_1, a_m|\mathcal{L}^A\rangle}\big)$, there exists a preference $P_i \in \mathbb{D}$ such that
$P_i$ is not semi-single-peaked on $G_{\sim}^A$ w.r.t.~$x$.
Therefore, $\mathbb{D}$ is an $(a_1, a_m)$-semi-hybrid domain on $\mathcal{L}^A$, as required.

Since Situations (3) and (4) are symmetric,
we focus on verifying Situation (2).
We have two cases: (i) $a_1 \in A^{a \rightharpoonup b}\backslash \{a\}$ and (ii) $a_1 \in \langle a, b|\mathcal{T}^A\rangle$.
In case (i), since $\underline{P}_i$ and $\overline{P}_i$ are complete reversals,
by $(a, b)$-semi-hybridness on $\mathcal{T}^A$,
it is easy to show that $a_m \in \langle a, b|\mathcal{T}^{b}\rangle \backslash \{a\}$.
Now, we construct a line $(x_1, \dots, x_s)$ over all alternatives of $\langle a, b|\mathcal{T}^A\rangle$ such that
$s = |\langle a, b|\mathcal{T}^A\rangle|$, $x_1 = a$, $x_s = a_m$, and all alternatives of $\langle a, b|\mathcal{T}^A\rangle\backslash \{a, a_m\}$ are arbitrarily arranged in the interior of the line.
Then, by combining the subtree $\mathcal{T}^{a \rightharpoonup b}$ and the line $(x_1, \dots, x_s)$,
we generate a tree $\widehat{\mathcal{T}}^A$. Clearly, $a$ and $a_m$ are dual-thresholds in $\widehat{\mathcal{T}}^A$.
Let $\hat{a} = a$ and $\hat{b} = a_m$.
Thus, $\langle a, b|\mathcal{T}^A\rangle = \langle \hat{a}, \hat{b}|\widehat{\mathcal{T}}^A\rangle$,
$a_1 \in \widehat{A}^{\hat{a} \rightharpoonup \hat{b}}$ and $a_m \in \widehat{A}^{\hat{b} \rightharpoonup \hat{a}}$.
We last show that $\mathbb{D}$ is an $(\hat{a},\hat{b})$-semi-hybrid domain on $\widehat{\mathcal{T}}^A$.
By $(a, b)$-semi-hybridness on $\mathcal{T}^A$, one can easily show $\mathbb{D} \subseteq \mathbb{D}_{\textrm{SH}}(\widehat{\mathcal{T}}^A, \hat{a}, \hat{b})$.
Furthermore, by condition (ii) of Definition \ref{def:asspsh},
we know that there exist no tree $\widetilde{\mathcal{T}}^A$ and dual-thresholds $\tilde{a}, \tilde{b} \in A$ such that $\mathbb{D} \subseteq \mathbb{D}_{\textrm{SH}}(\widetilde{\mathcal{T}}^A, \tilde{a}, \tilde{b})$ and
$\langle \tilde{a}, \tilde{b}|\widetilde{\mathcal{T}}^A\rangle \subset \langle a, b|\widehat{\mathcal{T}}^A\rangle=\langle \hat{a},\hat{b}|\mathcal{T}^A\rangle$.
Last, by condition (iii) of Definition \ref{def:asspsh},
we know that if $G_{\sim}^A$ is a tree, then for each $x \in \textrm{Leaf}\big(G_{\sim}^{\langle a, b|\mathcal{T}^A\rangle}\big) =
\textrm{Leaf}\big(G_{\sim}^{\langle \hat{a}, \hat{b}|\widehat{\mathcal{T}}^A\rangle}\big)$, there exists a preference $P_i \in \mathbb{D}$ such that $P_i$ is semi-single-peaked on $G_{\sim}^A$ w.r.t.~$x$.
Therefore, $\mathbb{D}$ is an $(\hat{a}, \hat{b})$-semi-hybrid domain on $\widehat{\mathcal{T}}^A$, as required.
In case (ii), fixing an alternative $c \in A^{a \rightharpoonup b}\backslash \{a\}$,
$(a, b)$-semi-hybridness on $\mathcal{T}^A$ implies $a\mathrel{\underline{P}_i}c$.
Since $\underline{P}_i$ and $\overline{P}_i$ are complete reversals, we have $c\mathrel{\overline{P}_i} a$.
Consequently, to meet $(a, b)$-semi-hybridness on $\mathcal{T}^A$, it must be the case that $a_m \in  A^{a \rightharpoonup b}\backslash \{a\}$.
Then, $(a, b)$-semi-hybridness on $\mathcal{T}^A$ implies $a\mathrel{\overline{P}_i}b$.
We further claim $a_1 \neq a$.
Otherwise, $r_1(\underline{P}_i) = a_1 = a \neq b$ implies $a\mathrel{\underline{P}_i}b$, which contradicts the fact that $\underline{P}_i$ and $\overline{P}_i$ are complete reversals.
Thus, we have $a_m \in A^{a \rightharpoonup b}\backslash \{a\}$ and $a_1 \in \langle a, b|\mathcal{T}^A\rangle\backslash \{a\}$,
which are analogous to $a_1 \in A^{a \rightharpoonup b}\backslash \{a\}$ and $a_m \in \langle a, b|\mathcal{T}^A\rangle\backslash \{a\}$ in the verification of case (i). Then, by a symmetric argument, we can construct a tree $\widehat{\mathcal{T}}^A$ and identify dual-thresholds $\hat{a}, \hat{b} \in A$ such that
$\mathbb{D}$ is an $(\hat{a}, \hat{b})$-semi-hybrid domain on $\widehat{\mathcal{T}}^A$, $a_1 \in \widehat{A}^{\hat{a} \rightharpoonup \hat{b}}$ and $a_m \in \widehat{A}^{\hat{b} \rightharpoonup \hat{a}}$. This proves the lemma.
\end{proof}

Henceforth, let $a_1 \in A^{a \rightharpoonup b}$ and $a_m \in A^{b \rightharpoonup a}$.
By \textbf{Statement (i)},
to complete the verification,
it suffices to show that $\mathbb{D}$ is not a semi-single-peaked domain.

\begin{lemma}\label{lem:nssp}
Domain $\mathbb{D}$ is not a semi-single-peaked domain.
\end{lemma}

\begin{proof}
Suppose not, i.e., there exist a tree $\widetilde{\mathcal{T}}^A$ and a threshold $\bar{x} \in A$ such that $\mathbb{D} \subseteq \mathbb{D}_{\textrm{SSP}}(\widetilde{\mathcal{T}}^A, \bar{x})$, and $G_{\sim}^A$ is a connected graph.
Then, by Clarification \ref{cla:ssp} and its proof, we know that $G_{\sim}^A$ is tree that coincides to $\widetilde{\mathcal{T}}^A$, i.e., $G_{\sim}^A = \widetilde{\mathcal{T}}^A$, $\bar{x}$ has at most two neighbors in $G_{\sim}^A$, i.e., $|\mathcal{N}_{\sim}^A(\bar{x})| \leq 2$, and $\bar{x}$ is included in the path $\langle a_1, a_m|\widetilde{\mathcal{T}}^A\rangle = \langle a_1, a_m|G_{\sim}^A\rangle$.
Let $\langle a_1, a_m|G_{\sim}^A\rangle = (x_1, \dots, x_v)$ denote the path in $G_{\sim}^A$ connecting $a_1$ and $a_m$.
Thus, $\bar{x} = x_{\bar{k}}$ for some $\bar{k} \in \{1, \dots, v\}$.
Meanwhile, by $(a,b)$-semi-hybridness on $\mathcal{T}^A$,
we know that $G_{\sim}^A$ is a combination of the subtree $G_{\sim}^{A^{a \rightharpoonup b}}=\mathcal{T}^{A^{a \rightharpoonup b}}$,
the connected subgraph $G_{\sim}^{\langle a, b|\mathcal{T}^A\rangle}$ and
the subtree $G_{\sim}^{A^{b \rightharpoonup a}} = \mathcal{T}^{A^{b \rightharpoonup a}}$,
denoted by $G_{\sim}^A = G_{\sim}^{A^{a \rightharpoonup b}}\cup G_{\sim}^{\langle a, b|\mathcal{T}^A\rangle} \cup  G_{\sim}^{A^{b \rightharpoonup a}}$.
Then, $G_{\sim}^A = \widetilde{\mathcal{T}}^A$ implies that $G_{\sim}^{\langle a, b|\mathcal{T}^A\rangle}$ is a tree as well, and
$a_1 \in A^{a \rightharpoonup b}$ and $a_m \in A^{b \rightharpoonup a}$ imply $a = x_s$ and $b = x_t$ for some $1 \leq s < t \leq v$.
Thus, we have five cases: (1) $1 \leq \bar{k}<s$, (2) $t< \bar{k} \leq v$, (3) $\bar{k} =s$, (4) $\bar{k} =t$ and (5) $s < \bar{k} < t$.
In each case, we induce a contradiction.

The first two cases are symmetric.
We focus on the verification of case (1).
Thus, $\bar{x} \in A^{a \rightharpoonup b}\backslash \{a\}$.
We show $\mathbb{D} \subseteq \mathbb{D}_{\textrm{SSP}}(G_{\sim}^A, a)$.
First, given arbitrary $P_i \in \mathbb{D}$ with $r_1(P_i) \in A^{a \rightharpoonup b}\backslash \{a\}$, we show that $P_i$ is semi-single-peaked on the tree $G_{\sim}^A = G_{\sim}^{A^{a \rightharpoonup b}}\cup \big[G_{\sim}^{\langle a, b|\mathcal{T}^A\rangle}
\cup G_{\sim}^{A^{b \rightharpoonup a}}\big]$ w.r.t.~$a$,
which consists of two parts: (i)
semi-single-peakedness on $G_{\sim}^{A^{a \rightharpoonup b}}=\mathcal{T}^{A^{b \rightharpoonup a}}$ w.r.t.~$a$, and
(ii) $a\mathrel{P_i}x$ for all $x \in A\backslash A^{a \rightharpoonup b}$.
Since $r_1(P_i) \in A^{a \rightharpoonup b}\backslash \{a\}$,
both parts follow from $(a,b)$-semi-hybridness on $\mathcal{T}^A$.
Second, given $P_i \in \mathbb{D}$ with $r_1(P_i) \in \langle a, b|\mathcal{T}^A\rangle \cup A^{b \rightharpoonup a}$, we show that $P_i$ is semi-single-peaked on the tree $G_{\sim}^A = G_{\sim}^{A^{a \rightharpoonup b}}\cup \big[G_{\sim}^{\langle a, b|\mathcal{T}^A\rangle}
\cup G_{\sim}^{A^{b \rightharpoonup a}}\big]$ w.r.t.~$a$,
which consists of two parts:
(i) $a\mathrel{P_i}x$ for all $x \in A^{a \rightharpoonup b}\backslash \{a\}$, and
(ii) semi-single-peakedness on $G_{\sim}^{\langle a, b|\mathcal{T}^A\rangle}
\cup G_{\sim}^{A^{b \rightharpoonup a}}$ w.r.t.~$a$.
Since $r_1(P_i) \in \langle a, b|\mathcal{T}^A\rangle \cup A^{b \rightharpoonup a}$,
part (i) follows from $(a,b)$-semi-hybridness on $\mathcal{T}^A$.
For part (ii), recall the contradictory hypothesis that $P_i$ is semi-single-peaked on $G_{\sim}^A$ w.r.t.~$x_{\bar{k}}$.
Since $a \in \langle x, x_{\bar{k}}|G_{\sim}^A\rangle$ for all $x \in \langle a, b|\mathcal{T}^A\rangle\cup A^{b \rightharpoonup a}$,
it is also true that
$P_i$ is semi-single-peaked on $G_{\sim}^{\langle a, b|\mathcal{T}^A\rangle}
\cup G_{\sim}^{A^{b \rightharpoonup a}}$ w.r.t.~$a$, as required.
Therefore, $\mathbb{D} \subseteq \mathbb{D}_{\textrm{SSP}}(G_{\sim}^A, a)$.
Now, by statement (ii) of Clarification \ref{cla:ssp},
according to semi-single-peakedness on $G_{\sim}^A$ w.r.t.~$a$ and diversity,
we know $|\mathcal{N}_{\sim}^A(a)| \leq 2$, which implies $\mathcal{N}_{\sim}^A(a) = \mathcal{N}_{\sim}^A(x_s) = \{x_{s-1}, x_{s+1}\}$.
It is clear that $x_{s-1} \in  A^{a \rightharpoonup b}\backslash \{a\}$ and $x_{s+1} \in \langle a, b|\mathcal{T}^A\rangle$.
Then, $\mathcal{N}_{\sim}^A(a)\cap \langle a, b|\mathcal{T}^A\rangle = \{x_{s+1}\}$ implies $a \in \textrm{Leaf}\big(G_{\sim}^{\langle a, b|\mathcal{T}^A\rangle}\big)$.
Thus, we know that all preferences of $\mathbb{D}$ are semi-single-peaked on the tree $G_{\sim}^A$ w.r.t.~$a$ which is a leaf of $G_{\sim}^{\langle a, b|\mathcal{T}^A\rangle}$. This contradicts condition (iii) of Definition \ref{def:asspsh}.

Cases (3) and (4) are symmetric.
We focus on the verification of case (3).
Thus, we have $\mathbb{D} \subseteq \mathbb{D}_{\textrm{SSP}}(G_{\sim}^A, a)$.
Then, by the verification in case (1), we induce the same contradiction.

Last, let case (5) occur.
Recall that $x_{\bar{k}}$ has at most two neighbors in $G_{\sim}^A$, which is mentioned in the beginning of the proof.
Thus, we know $\mathcal{N}_{\sim}^A(x_{\bar{k}}) = \{x_{\bar{k-1}}, x_{\bar{k}+1}\}$.
We cut the tree $G_{\sim}^A$ at the edge $(x_{\bar{k}-1}, x_{\bar{k}})$, and obtain the subset $B = \big\{x \in A: x_{\bar{k}} \in \langle x, x_{\bar{k}-1}|G_{\sim}^A\rangle\big\}$ and the subtree $G_{\sim}^{B}$.
It is clear that the subtrees $G_{\sim}^{A^{a \rightharpoonup b}}$ and $G_{\sim}^{B}$ are separated, i.e.,
$A^{a \rightharpoonup b} \cap B = \emptyset$.
Let $\widehat{B} = \big[A\backslash [A^{a \rightharpoonup b} \cup B]\big]\cup \{a, x_{\bar{k}}\}$.
Clearly, $\widehat{B}  \subset \langle a, b|\mathcal{T}^A\rangle$.
Furthermore, we construct a line $\pi = (z_1, \dots, z_q)$ over all alternatives of $\widehat{B} $ such that $q = |\widehat{B} |$, $z_1 = a$, $z_q = x_{\bar{k}}$, and all alternatives of $\widehat{B} \backslash \{a, x_{\bar{k}}\}$ are arbitrarily arranged in the interior of the line $\pi$.
By combining the subtree $G_{\sim}^{A^{a \rightharpoonup b}}$, the line $\pi$ and the subtree $G_{\sim}^{B}$, we generate a tree $\widehat{\mathcal{T}}^A$.
Clearly, $a$ and $x_{\bar{k}}$ are dual-thresholds in $\widehat{\mathcal{T}}^A$.
Let $\widehat{A}^{a \rightharpoonup x_{\bar{k}}} = \big\{x \in A: a \in \langle x, x_{\bar{k}}|\widehat{\mathcal{T}}^A\rangle \big\}$ and
$\widehat{A}^{x_{\bar{k}} \rightharpoonup a} = \big\{x \in A: x_{\bar{k}} \in \langle x, a|\widehat{\mathcal{T}}^A\rangle \big\}$.
Note that $\widehat{A}^{a \rightharpoonup x_{\bar{k}}} = A^{a \rightharpoonup b}$ and $\widehat{A}^{x_{\bar{k}} \rightharpoonup a} = B$.
We next show $\mathbb{D} \subseteq \mathbb{D}_{\textrm{SH}}(\widehat{\mathcal{T}}^A, a, x_{\bar{k}})$.
First, given $P_i \in \mathbb{D}$ with $r_1(P_i) \in \widehat{A}^{a \rightharpoonup x_{\bar{k}}}\backslash \{a\}$,
we show that $P_i$ is semi-single-peaked on $\widehat{\mathcal{T}}^A$ w.r.t.~$a$, and $\max^{P_i}\big(\widehat{A}^{x_{\bar{k}} \rightharpoonup a}\big) = x_{\bar{k}}$.
Since $\widehat{\mathcal{T}}^A = G_{\sim}^{A^{a \rightharpoonup b}}\cup \big[\pi \cup G_{\sim}^{B}\big]$, the semi-single-peakedness requirement on $\widehat{\mathcal{T}}^A$ w.r.t.~$a$ consists of the following two parts:
(i) semi-single-peakedness on $G_{\sim}^{A^{a \rightharpoonup b}} = \mathcal{T}^{A^{b \rightharpoonup a}}$ w.r.t.~$a$, and
(ii) $a\mathrel{P_i}x$ for all $x \in A\backslash A^{a \rightharpoonup b}$.
Since $r_1(P_i) \in \widehat{A}^{a \rightharpoonup x_{\bar{k}}}\backslash \{a\} = A^{a \rightharpoonup b}\backslash \{a\}$,
both parts follow from $(a,b)$-semi-hybridness on $\mathcal{T}^A$.
Furthermore, since $r_1(P_i) \in \widehat{A}^{a \rightharpoonup x_{\bar{k}}} = A^{a \rightharpoonup b}$,
we know that for all $x \in \widehat{A}^{x_{\bar{k}} \rightharpoonup a} = B$,
$x \notin \langle r_1(P_i), x_{\bar{k}}|G_{\sim}^A\rangle$ and
$\textrm{Proj}\big(x, \langle r_1(P_i), x_{\bar{k}}|G_{\sim}^A\rangle\big) = x_{\bar{k}}$.
Then, by the contradictory hypothesis that $P_i$ is semi-single-peaked on $G_{\sim}^A$ w.r.t.~$x_{\bar{k}}$, we have $\max^{P_i}\big(\widehat{A}^{x_{\bar{k}} \rightharpoonup a}\big) = x_{\bar{k}}$, as required.
Second, given $P_i \in \mathbb{D}$ with $r_1(P_i) \in \widehat{A}^{x_{\bar{k}} \rightharpoonup a}\backslash \{x_{\bar{k}}\}$,
we show that $P_i$ is semi-single-peaked on $\widehat{\mathcal{T}}^A$ w.r.t.~$x_{\bar{k}}$, and $\max^{P_i}\big(\widehat{A}^{a \rightharpoonup x_{\bar{k}}}\big) = a$.
Since $\widehat{\mathcal{T}}^A = \big[G_{\sim}^{A^{a \rightharpoonup b}}\cup \pi\big] \cup G_{\sim}^{B}$,
the semi-single-peakedness requirement on $\widehat{\mathcal{T}}^A$ w.r.t.~$x_{\bar{k}}$ consists of the following two parts:
(i) semi-single-peakedness on $G_{\sim}^{B}$ w.r.t.~$x_{\bar{k}}$, and
(ii) $x_{\bar{k}}\mathrel{P_i}x$ for all $x \in A\backslash B$.
Indeed, both parts follow from the contradictory hypothesis that $P_i$ is semi-single-peaked on $G_{\sim}^A$ w.r.t.~$x_{\bar{k}}$.
Furthermore, since $r_1(P_i) \in \widehat{A}^{x_{\bar{k}} \rightharpoonup a}\backslash \{x_{\bar{k}}\} = B\backslash \{x_{\bar{k}}\} \subseteq A\backslash A^{a \rightharpoonup b}$,
$(a, b)$-semi-hybridness on $\mathcal{T}^A$ implies $\max^{P_i}(\widehat{A}^{a \rightharpoonup x_{\bar{k}}}) =
\max^{P_i}(A^{a \rightharpoonup b}) = a$, as required.
Last, given $P_i \in \mathbb{D}$ with $r_1(P_i) \in \langle a, x_{\bar{k}}|\widehat{\mathcal{T}}^A\rangle$,
we show $\max^{P_i}\big(\widehat{A}^{a \rightharpoonup x_{\bar{k}}}\big) = a$ and $\max^{P_i}\big(\widehat{A}^{x_{\bar{k}} \rightharpoonup a}\big) = x_{\bar{k}}$.
Since $r_1(P_i) \in \langle a, x_{\bar{k}}|\widehat{\mathcal{T}}^A\rangle = \widehat{B} \subset \langle a, b|\mathcal{T}^A\rangle$,
$(a, b)$-semi-hybridness on $\mathcal{T}^A$ implies $\max^{P_i}(\widehat{A}^{a \rightharpoonup x_{\bar{k}}}) =
\max^{P_i}(A^{a \rightharpoonup b}) = a$, as required.
Note that for each $x \in \widehat{A}^{x_{\bar{k}} \rightharpoonup a}\backslash \{x_{\bar{k}}\} = B\backslash \{x_{\bar{k}}\}$,
$x \notin \langle r_1(P_i), x_{\bar{k}}|G_{\sim}^A\rangle$ and $\textrm{Proj}\big(x, \langle r_1(P_i), x_{\bar{k}}|G_{\sim}^A\rangle\big) = x_{\bar{k}}$.
Then, by the contradictory hypothesis that $P_i$ is semi-single-peaked on $G_{\sim}^A$ w.r.t.~$x_{\bar{k}}$,
we have $x_{\bar{k}}\mathrel{P_i}x$ for all $x \in \widehat{A}^{x_{\bar{k}} \rightharpoonup a}\backslash \{x_{\bar{k}}\}$, which implies
$\max^{P_i}\big(\widehat{A}^{x_{\bar{k}} \rightharpoonup a}\big) = x_{\bar{k}}$, as required.
In conclusion, we have $\mathbb{D} \subseteq \mathbb{D}_{\textrm{SH}}\big(\widehat{\mathcal{T}}^A, a, x_{\bar{k}}\big)$ and
$\langle a, x_{\bar{k}}|\widehat{\mathcal{T}}^A\rangle =\widehat{B} \subset \langle a, b|\mathcal{T}^A\rangle$
which contradict condition (ii) of Definition \ref{def:asspsh}.
This proves the lemma, and completes the verification of the ``if part'' of \textbf{Statement (ii)}.
\end{proof}

\noindent
Henceforth, we show the ``only if part'' of \textbf{Statement (ii)}:
``There exists no invariant, tops-only and strategy-proof rule.''
$\Rightarrow$ ``Domain $\mathbb{D}$ is a semi-hybrid domain satisfying the unique seconds property.'', and show that
every two-voter, tops-only and strategy-proof rule defined on $\mathbb{D}$ is a hybrid rule and behaves like a dictatorship on the free zone.

First, since $\mathbb{D}$ is a non-dictatorial, unidimensional domain, the Auxiliary Proposition implies that $\mathbb{D}$ satisfies the unique seconds property.\footnote{It is worth mentioning that this is the only place in the proof where leaf symmetry plays a role.}

Next, we show that $\mathbb{D}$ is a semi-hybrid domain.\footnote{We are grateful to an anonymous referee for suggesting the proof.}
Since $\mathbb{D}$ satisfies path-connectedness, it is clear that $G_{\sim}^A$ is a connected graph.
First, since $\mathbb{D} \subseteq \mathbb{P} = \mathbb{D}_{\textrm{SH}}(\mathcal{L}^A, a_1, a_m)$,
there must exist a tree $\mathcal{T}^A$ and dual-thresholds $a,b \in A$ such that $\mathbb{D} \subseteq \mathbb{D}_{\textrm{SH}}(\mathcal{T}^A, a, b)$. Hence, condition (i) of Definition \ref{def:asspsh} is satisfied.
Furthermore, since $A$ is finite, we can push that the searching of the tree $\mathcal{T}^A$ and the dual-thresholds $a,b \in A$
towards the limit that there exist no tree $\widehat{\mathcal{T}}^A$ and dual-thresholds $\hat{a}, \hat{b} \in A$ such that $\mathbb{D} \subseteq \mathbb{D}_{\textrm{SH}}(\widehat{\mathcal{T}}^A, \hat{a}, \hat{b})$ and $\langle \hat{a}, \hat{b}|\widehat{\mathcal{T}}^A\rangle \subset \langle a, b|\mathcal{T}^A\rangle$. Thus, condition (ii) of Definition \ref{def:asspsh} is met.
Last, we show that $\mathbb{D}$ satisfies condition (iii) of Definition \ref{def:asspsh}.
Let $G_{\sim}^A$ be a tree. Since $G_{\sim}^A$ is a combination of the subtree $G_{\sim}^{A^{a \rightharpoonup b}} =\mathcal{T}^{A^{a \rightharpoonup b}}$, the connected subgraph $G_{\sim}^{\langle a, b|\mathcal{T}^A\rangle}$ and the subtree $G_{\sim}^{A^{b \rightharpoonup a}} =\mathcal{T}^{A^{b \rightharpoonup a}}$,
it must be the case that $G_{\sim}^{\langle a, b|\mathcal{T}^A\rangle}$ is a tree as well.
Since by hypothesis there exists no invariant, tops-only and strategy-proof rule,
\textbf{Statement (i)} implies that $\mathbb{D}$ is not a semi-single-peaked domain,
which immediately implies that for each $x \in \textrm{Leaf}\big(G_{\sim}^{\langle a, b|\mathcal{T}^A\rangle}\big)$,
there exists a preference $P_i \in \mathbb{D}$ such that $P_i$ is not semi-single-peaked on $G_{\sim}^A$ w.r.t.~$x$, as required.
In conclusion, $\mathbb{D}$ is an $(a, b)$-semi-hybrid domain on a tree $\mathcal{T}^A$.
Recall the two completely reversed preference $\underline{P}_i$ and $\overline{P}_i$ included in $\mathbb{D}$ and their peaks
$r_1(\underline{P}_i) = a_1$ and $r_1(\overline{P}_i) = a_m$.
By Lemma \ref{lem:wlog}, we henceforth assume w.l.o.g.~that $a_1 \in A^{a \rightharpoonup b}$ and $a_m \in A^{b \rightharpoonup a}$.

Last, we fix an arbitrary two-voter, tops-only and strategy-proof rule $f: \mathbb{D}^2 \rightarrow A$,\footnote{Such a rule always exists, e.g., a dictatorship.} and show that $f$ is a hybrid rule and behaves like a dictatorship on the free zone $\langle a, b|\mathcal{T}^A\rangle$.
Since there exists no invariant, tops-only and strategy-proof rule,
we know $f(a_1, a_m) = f(\underline{P}_1, \overline{P}_2) \neq f(\overline{P}_1, \underline{P}_2) = f(a_m, a_1)$.
The proof consists of the following three steps:\medskip

\noindent
\textbf{Step 1.} We construct a line over all alternatives involved in the path(s) connecting $a_1$ and $a_m$ in $G_{\sim}^A$
(see all proofs before Lemma \ref{lem:RFBR2}), and
partially characterize $f$ according to the constructed line (see Lemma \ref{lem:RFBR2}).

\noindent
\textbf{Step 2.} We construct a tree $\mathcal{T}_f^A$ where the two social outcomes $f(a_1, a_m)$ and $f(a_m, a_1)$ are dual-thresholds, using the adjacency graph $G_{\sim}^A$ and the partial characterization of $f$ (see Lemmas \ref{lem:sides} and \ref{lem:graphpartition}).
Then, we completely characterize $f$ to be a hybrid rule on $\mathcal{T}_f^A$ w.r.t.~the dual-thresholds $f(a_1, a_m)$ and $f(a_m, a_1)$ (see Lemmas \ref{lem:dictatorship5} and \ref{lem:fullcharacterization}),
which indicates that $f$ behaves like a dictatorship on the interval between $f(a_1, a_m)$ and $f(a_m, a_1)$ in $\mathcal{T}_f^A$.
Furthermore, we elicit some preference restriction embedded in $\mathbb{D}$ via strategy-proofness of $f$ (see Lemma \ref{lem:sh}),
which may be different from the aforementioned preference restriction of $(a,b)$-semi-hybridness on $\mathcal{T}^A$.

\noindent
\textbf{Step 3.}
Note that the preference restriction elicited in \textbf{Step 2} and the preference restriction of $(a,b)$-semi-hybridness on $\mathcal{T}^A$ must be compatible with each other as they are both embedded in $\mathbb{D}$.
We use the compatibility of these two preference restrictions to show that the interval between $f(a_1, a_m)$ and $f(a_m, a_1)$ in $\mathcal{T}_f^A$ is a superset of the interval $\langle a, b|\mathcal{T}^A\rangle$,
which of course implies that $f$ behaves like a dictatorship on $\langle a, b|\mathcal{T}^A\rangle$ (see Lemma \ref{lem:superset}).\medskip

Let $\Pi(a_1, a_m)$ denote the set of paths in $G_{\sim}^A$ connecting $a_1$ and $a_m$.
Clearly, $\Pi(a_1, a_m) \neq \emptyset$.
Let $B = \big\{a \in A: a \in \pi\; \textrm{for some}\; \pi \in \Pi(a_1, a_m)\big\}$.
Note that $B$ may not include all alternatives of $A$, and all paths of $\Pi(a_1, a_m)$ are included in $G_{\sim}^{B}$ which hence implies that $G_{\sim}^B$ is connected.
There are two cases: $|\Pi(a_1, a_m)|=1$ and $|\Pi(a_1, a_m)| > 1$.
In the first case, let $\mathcal{L}^{B} = (x_1, \dots, x_v)$ be the path in $G_{\sim}^A$ that connects $a_1$ and $a_m$, where $v = |B|$.

Next, assume $|\Pi(a_1, a_m)| > 1$.
Fix an arbitrary path $\pi = (z_1, \dots, z_{\omega}) \in \Pi(a_1, a_m)$.
Since $|\Pi(a_1, a_m)| > 1$ and all paths of $\Pi(a_1, a_m)$ start from $a_1$ and end at $a_m$,
we can identify $x,y \in \pi$, say $x = z_s$ and $y = z_{\omega-p}$, where $1 \leq s < \omega-p \leq \omega$, satisfying the following two conditions:
(i) $s$ is the maximum index agreed by all paths of $\Pi(a_1, a_m)$ in the direction from $a_1$ to $a_m$,
i.e., for each path $\pi'= (z_1', \dots, z_{\sigma}') \in \Pi(a_1, a_m)\backslash \{\pi\}$,
$z_k = z_k'$ for all $k \in \{1, \dots, s\}$, and
for \emph{some} path $\pi''= (z_1'', \dots, z_q'') \in \Pi(a_1, a_m)\backslash \{\pi\}$, $z_{s+1} \neq z_{s+1}''$, and
(ii) $\omega-p$ is the minimum index agreed by all paths of $\Pi(a_1, a_m)$ in the direction from $a_m$ to $a_1$,
i.e., for each path $\pi'= (z_1', \dots, z_{\sigma}') \in \Pi(a_1, a_m)\backslash \{\pi\}$,
$z_{\omega-k} = z_{\sigma-k}'$ for all $k \in \{0,1, \dots, p\}$, and
for \emph{some} path $\pi''= (z_1'', \dots, z_q'') \in \Pi(a_1, a_m)\backslash \{\pi\}$, $z_{\omega-p-1} \neq z_{q-p-1}''$.
Thus, we can pin down the adjacency graph $G_{\sim}^{B}$:
$\pi^L = (z_1, \dots, z_s)$ is the unique path in $G_{\sim}^B$ connecting $a_1$ and $x$,
$\pi^R = (z_{\omega-p}, \dots, z_{\omega})$ is the unique path in $G_{\sim}^{B}$ connecting $y$ and $a_m$, the set $\mathcal{O} = \{a \in B: a \notin \pi^L\cup \pi^R\}\cup \{x, y\}$ contains at least three alternatives\footnote{It is clear that $|\mathcal{O}| \geq 2$. If $|\mathcal{O}| = 2$, all paths of $\Pi(a_1, a_m)$ degenerate to an identical path.}, and
$G_{\sim}^{\mathcal{O}}$ is a connected graph and has no leaf (see the first diagram of Figure \ref{fig:multiplepaths}).
We next construct a line $(x, \dots, y)$ over all alternatives of $\mathcal{O}$ such that $x$ and $y$ are the two leaves,
and all alternatives of $\mathcal{O}\backslash \{x,y\}$ are arbitrarily arranged in the interior of the line.
Then, by combining $\pi^L$, $(x, \dots, y)$ and $\pi^R$, we construct a line $\mathcal{L}^{B} = (x_1, \dots, x_s, \dots, x_t, \dots, x_v)$,
where $v = |B|$, $1\leq s< t \leq v$, $t-s>1$, $x_1 = a_1$, $x_s = x$, $x_t = y$, $x_v = a_m$,
$(x_1, \dots, x_s) = \pi^L$, $(x_s, \dots, x_t) = (x, \dots, y)$ and $(x_t, \dots, x_v) = \pi^R$ (see the second diagram of Figure \ref{fig:multiplepaths}).
We intentionally let the notation of the constructed line $\mathcal{L}^{B}$ here be identical to the line $\mathcal{L}^B$ in the case $|\Pi(a_1, a_m)| = 1$. This helps us unify the henceforth proof for both cases, and does not create any loss of generality.

\bigskip

\hspace{4.5em}
\begin{figurehere}
\begin{tikzpicture}

\begin{scope}[thin, dashed]
\draw (4.5,0) ellipse (1.5 and 0.4);
\end{scope}
\node[label=-120:{$G_{\sim}^{\mathcal{O}}$}] at (5,0.5) {};
\draw (0,0) -- (3,0);
\draw (6,0) -- (9,0);

\node at (0,0) {{\footnotesize$\bullet$}};
\node at (3,0) {{\footnotesize$\bullet$}};
\node at (6,0) {{\footnotesize$\bullet$}};
\node at (9,0) {{\footnotesize$\bullet$}};

\node at (0,0.3) {$a_1$};
\node at (2.9,0.3) {$x$};
\node at (6.1,0.3) {$y$};
\node at (9,0.3) {$a_m$};

\node at (1.5,-0.2) {$\underbrace{\rule[0mm]{3cm}{0cm}}$};
\node at (1.55,-0.5) {$\pi^{L}$};
\node at (7.53,-0.2) {$\underbrace{\rule[0mm]{3cm}{0cm}}$};
\node at (7.55,-0.5) {$\pi^{R}$};

\node at (-1.5,0) {$G_{\sim}^{B}$:};

\draw (0,-1.4) -- (9,-1.4);

\node at (0,-1.4) {{\footnotesize$\bullet$}};
\node at (3,-1.4) {{\footnotesize$\bullet$}};
\node at (6,-1.4) {{\footnotesize$\bullet$}};
\node at (9,-1.4) {{\footnotesize$\bullet$}};

\node at (0,-1.1) {$a_1 = x_1$};
\node at (3.1,-1.1) {$x=x_s$};
\node at (5.9,-1.1) {$x_t=y$};
\node at (9,-1.1) {$x_v = a_m$};

\node at (1.5,-1.6) {$\underbrace{\rule[0mm]{3cm}{0cm}}$};
\node at (1.55,-2) {$\pi^{L}$};
\node at (4.5,-1.6) {$\underbrace{\rule[0mm]{3cm}{0cm}}$};
\node at (4.5,-2) {$(x, \dots, y)$};
\node at (7.5,-1.6) {$\underbrace{\rule[0mm]{3cm}{0cm}}$};
\node at (7.55,-2) {$\pi^{R}$};
\node at (-1.5,-1.4) {$\mathcal{L}^B$:};

\node at (1.3,-1.4) {{\footnotesize$\bullet$}};
\node at (1.3,-1.1) {$x_{\underline{k}}$};

\node at (7.7,-1.4) {{\footnotesize$\bullet$}};
\node at (7.7,-1.1) {$x_{\overline{k}}$};

\end{tikzpicture}
\vspace{-0.7em}
\caption{Adjacency graph $G_{\sim}^{B}$ and the constructed line $\mathcal{L}^B$}\label{fig:multiplepaths}
\end{figurehere}

\bigskip

In the case $|\Pi(a_1, a_m)| = 1$, statement (i) of Lemma \ref{lem:transitivity} and the violation of invariance imply $f(x_1, x_v) = x_{\underline{k}}$ and $f(x_v, x_1) = x_{\overline{k}}$ for some distinct $\underline{k}, \overline{k}\in \{1, \dots, v\}$.
We assume w.l.o.g. that $\underline{k}< \overline{k}$, which by statement (i) of Lemma \ref{lem:RFBR} implies that voter $1$ dictates on
$\langle x_{\underline{k}}, x_{\overline{k}}|\mathcal{L}^B\rangle$.
In the case $|\Pi(a_1, a_m)| > 1$,
since $G_{\sim}^{\mathcal{O}}$ is a connected graph and has no leaf,
statement (i) of Lemma \ref{lem:union2} implies that $f$ behaves like a dictatorship on $\mathcal{O} = \langle x_s, x_t|\mathcal{L}^{B}\rangle$.
We assume w.l.o.g.~that voter $1$ dictates on $\langle x_s, x_t|\mathcal{L}^B\rangle$, i.e., $f(x_k, x_{k'}) = x_k$ for all $k, k' \in \{s, \dots,  t\}$. This helps us unify the henceforth proof for both cases.
Thus, $f(x_s, x_t) = x_s$ and $f(x_t, x_s) = x_t$.
Then, according to the paths $(x_t, \dots, x_v)$ and $(x_1, \dots, x_s)$,
statement (ii) of Lemma \ref{lem:transitivity} implies $f(x_s, x_v) = x_s$ and $f(x_t, x_1) = x_t$.
Furthermore, according to $(x_1, \dots, x_s)$ and $(x_t, \dots, x_v)$,
by statement (iii) of Lemma \ref{lem:transitivity},
$f(x_s, x_v) = x_s$ implies $f(x_1, x_v) = x_{\underline{k}}$ for some $\underline{k} \in \{1, \dots,s\}$, and
$f(x_t, x_1) = x_t$ implies $f(x_v, x_1) = x_{\overline{k}}$ for some $\overline{k} \in \{t, \dots, v\}$.
\medskip

The next lemma provide a unified characterization of $f$ on the line $\mathcal{L}^B$.

\begin{lemma}\label{lem:RFBR2}
According to $\mathcal{L}^B=(x_1, \dots, x_v)$,
we have
$f(x_k, x_{k'}) = \left\{
\begin{array}{ll}
x_k & \emph{if}\; \underline{k}\leq k \leq \overline{k},\\
x_{\mathop{\emph{med}}(k,\,k',\,\underline{k})} & \emph{if}\; k< \underline{k},\; \textrm{and}\\
x_{\mathop{\emph{med}}(k,\,k',\,\overline{k})} & \emph{if}\; k> \overline{k}.
\end{array}
\right.$
\end{lemma}

\begin{proof}
If $|\Pi(a_1, a_m)| = 1$, the lemma follows from statement (i) of Lemma \ref{lem:RFBR}.
Next, we assume $|\Pi(a_1, a_m)| > 1$.
Note that for each path $\pi \in \Pi(a_1, a_m)$,
since $f(x_1, x_v) = x_{\underline{k}} \in \langle a_1, x|\pi\rangle$ and
$f(x_v, x_1) = x_{\overline{k}} \in \langle y, a_m|\pi\rangle$,
statement (i) of Lemma \ref{lem:RFBR} holds on $\pi$.
To prove the lemma, we fix an arbitrary profile $(x_k, x_{k'})$.

First, let $\underline{k}\leq k \leq \overline{k}$.
If $\underline{k} \leq k' \leq \overline{k}$, then by voter 1's dictatorship on $\langle x_{\underline{k}}, x_{\overline{k}}|\mathcal{L}^B\rangle$, we have $f(x_k, x_{k'}) = x_k$.
If $k'< \underline{k}$ or $k'> \overline{k}$, we know that there exists a path $\pi \in \Pi(a_1, a_m)$ which includes both $x_k$ and $x_{k'}$.
Clearly, $x_k \in \langle x_{\underline{k}}, x_{\overline{k}}|\pi\rangle$.
Then, statement (i) of Lemma \ref{lem:RFBR} on $\pi$ implies $f(x_k, x_{k'}) = x_k$.
Second, let $k < \underline{k}$.
Then, there exists a path $\pi \in \Pi(a_1, a_m)$ which includes both $x_k$ and $x_{k'}$.
Clearly, $x_k \in \langle x_1, x_{\underline{k}}|\pi\rangle\backslash \{x_{\underline{k}}\}$.
Then, statement (i) of Lemma \ref{lem:RFBR} on $\pi$ implies $f(x_k, x_{k'}) = x_{\textrm{med}(k, \,k', \,\underline{k})}$.
Symmetrically, if $k > \overline{k}$, statement (i) of Lemma \ref{lem:RFBR} implies $f(x_k, x_{k'}) = x_{\textrm{med}(k,\, k', \,\overline{k})}$.
This proves the lemma, and completes \textbf{Step 1} of the proof.
\end{proof}

\begin{lemma}\label{lem:sides}
Fixing an alternative $z \in A$ and a path $\pi = (z_1, \dots, z_s)$ in $G_{\sim}^A$, where $z_1 = z$ and $s \geq 2$,
the following two statements hold:
\begin{itemize}
\item[\rm (i)] if $z_{s-1} = x_{\underline{k}-1}$ and $z_s = x_{\underline{k}}$,
then $\pi$ is the unique path in $G_{\sim}^A$ connecting $z$ and $x_{\underline{k}}$, and

\item[\rm (ii)] if $z_{s-1} = x_{\overline{k}+1}$ and $z_s = x_{\overline{k}}$,
then $\pi$ is the unique path in $G_{\sim}^A$ connecting $z$ and $x_{\overline{k}}$.
\end{itemize}
\end{lemma}

\begin{proof}
The two statements are symmetric, and we hence focus on the verification of the first one.
Suppose that there exists another path $\pi' = (y_1, \dots, y_t)$ in $G_{\sim}^A$ connecting $z$ and $x_{\underline{k}}$.
Then, we can identify a cycle $\mathcal{C}$ in $G_{\sim}^A$ such that (i) $\mathcal{C} \subseteq \pi\cup \pi'$, (ii) $\pi \cap \mathcal{C} \neq \emptyset$ and
(iii) each edge in $\mathcal{C}$ belong to $\pi$ or $\pi'$.
Clearly, by Observation \ref{obs:cycle}, $f$ behaves like a dictatorship on $\mathcal{C}$.
We further identify the alternative
$z_{k^{\ast}} \in \pi \cap\mathcal{C}$ such that $z_{k} \notin \mathcal{C}$ for all $k\in \{k^{\ast}+1, \dots, s\}$.
We consider two cases: $k^{\ast} = s$ and $k^{\ast}< s$.

In the first case, we show that $x_{\underline{k}-1}$ is also included in $\mathcal{C}$.
On the one hand, as included in $ \mathcal{C}$, $x_{\underline{k}}$ has two distinct neighbors in $\mathcal{C}$.
On the other hand, $x_{\underline{k}}$ has a unique neighbor in $\pi$, which is $x_{\underline{k}-1}$, and a unique neighbor in $\pi'$.
Therefore, it must be true that $x_{\underline{k}-1}$ is included in $\mathcal{C}$.
Hence, $x_{\underline{k}-1}, x_{\underline{k}} \in \mathcal{C}$.
In the second case, we have the path $(z_{k^{\ast}}, \dots, z_s)$, which contains both $x_{\underline{k}-1}$ and $x_{\underline{k}}$.
Recall by Lemma \ref{lem:RFBR2} that voter 1 dictates on $\langle x_{\underline{k}}, x_{\overline{k}}|\mathcal{L}^B\rangle$.
In the first case, $\mathcal{C}\cap \langle x_{\underline{k}}, x_{\overline{k}}|\mathcal{L}^B\rangle \neq \emptyset$, while in the second case, the cycle $\mathcal{C}$ and the adjacency graph over $\langle x_{\underline{k}}, x_{\overline{k}}|\mathcal{L}^B\rangle$, which both are connected graphs and contain at least two alternatives, are linked via the path $(z_{k^{\ast}}, \dots, z_s)$.
Therefore, Lemma \ref{lem:union} implies that voter 1 dictates on $\{z_{s-1}, z_s\}$.
Consequently, we have $f(x_{\underline{k}-1}, x_{\underline{k}}) =f(z_{s-1}, z_s) = z_{s-1} = x_{\underline{k}-1} \neq x_{\textrm{med}(\underline{k}-1, \,\underline{k}, \,\underline{k})}$,
which contradicts Lemma \ref{lem:RFBR2}.
Hence, $\pi$ is the unique path in $G_{\sim}^A$ connecting $z$ and $x_{\underline{k}}$.
\end{proof}

We construct the following five sets:
\begin{align*}
& \underline{B} = \left\{z \in A:
\textrm{there exists a path}\; (z_1, \dots, z_s)\;\textrm{in}\; G_{\sim}^A\;
\textrm{connecting}\; z\; \textrm{and}\; x_{\underline{k}}\; \textrm{such that}\; z_{s-1} = x_{\underline{k}-1}
\right\},\\
& \overline{B} =\left\{z \in A:
\textrm{there exists a path}\; (z_1, \dots, z_s)\;\textrm{in}\; G_{\sim}^A\;
\textrm{connecting}\; z\; \textrm{and}\; x_{\overline{k}}\; \textrm{such that}\; z_{s-1} = x_{\overline{k}+1}
\right\},\\
& \underline{A} = \underline{B}\cup \{x_{\underline{k}}\},\;
\overline{A} =  \overline{B}\cup \{x_{\overline{k}}\},\; \textrm{and}\;
M =  \big\{z \in A: z \notin \underline{B}\cup \overline{B}\big\}.
\end{align*}

The next lemma shows that $G_{\sim}^A$ is a combination of three adjacency graphs
$G_{\sim}^{\underline{A}}$, $G_{\sim}^M$ and $G_{\sim}^{\overline{A}}$, denoted by $G_{\sim}^A = G_{\sim}^{\underline{A}} \cup G_{\sim}^{M} \cup G_{\sim}^{\overline{A}}$.

\begin{lemma}\label{lem:graphpartition}
We have $G_{\sim}^A = G_{\sim}^{\underline{A}} \cup G_{\sim}^{M} \cup G_{\sim}^{\overline{A}}$.
\end{lemma}

\begin{proof}
First, it is clear that $A = \underline{A}\cup M \cup \overline{A}$ and $\mathcal{E}_{\sim}^A \supseteq \mathcal{E}_{\sim}^{\underline{A}}\cup \mathcal{E}_{\sim}^{M}\cup \mathcal{E}_{\sim}^{\overline{A}}$.
To prove $\mathcal{E}_{\sim}^A = \mathcal{E}_{\sim}^{\underline{A}}\cup \mathcal{E}_{\sim}^{M}\cup \mathcal{E}_{\sim}^{\overline{A}}$,
it suffices to show that in $G_{\sim}^A$, no alternative of $\underline{B}$ is adjacent to any alternative not in $\underline{A}$,
and no alternative of $\overline{B}$ is adjacent to any alternative not in $\overline{A}$.

Given $z \in \underline{B}$ and $y \notin \underline{A}$,
suppose by contradiction that $(z, y) \in \mathcal{E}_{\sim}^A$.
Let $(z_1, \dots, z_s)$ be the unique path in $G_{\sim}^{A}$ that connects $z$ and $x_{\underline{k}}$, where $z_{s-1} = x_{\underline{k}-1}$.
It is evident that $y \notin \underline{A}$ implies $y \neq x_{\underline{k}} = z_s$.
Moreover, since $z_1, \dots, z_{s-1} \in \underline{B}$ by definition,
it is true that $y \notin \{z_1, \dots, z_{s-1}\}$.
Consequently, we have a path $(y, z_1, \dots, z_s)$ in $G_{\sim}^A$ that connects $y$ and $x_{\underline{k}}$,
which by definition implies $y \in \underline{B} \subset \underline{A}$ - a contradiction.
Therefore, $(z, y) \notin \mathcal{E}_{\sim}^A$.
Symmetrically, given $z' \in \overline{B}$ and $y' \notin \overline{A}$,
we have $(z', y') \notin \mathcal{E}_{\sim}^A$.
\end{proof}

By Lemma \ref{lem:sides} and the construction of $\underline{A}$ and $\overline{A}$,
one can easily infer that $G_{\sim}^{\underline{A}}$ and $G_{\sim}^{\overline{A}}$ are two trees.
Immediately, since $G_{\sim}^A$ is a connected graph, Lemma \ref{lem:graphpartition} implies that $G_{\sim}^{M}$ is a connected graph as well.
In particular, when $\underline{A} \neq \{x_{\underline{k}}\}$,
we know that $x_{\underline{k}-1}$ must be contained in $\underline{A}$, and more importantly,
for every $z \in \underline{B}$,
$x_{\underline{k}-1}$ is included in the path in $G_{\sim}^{\underline{A}}$ that connects $z$ and $x_{\underline{k}}$.
Therefore, it must be the case that $x_{\underline{k}-1}$ is the unique neighbor of $x_{\underline{k}}$ in $G_{\sim}^{\underline{A}}$, and hence $x_{\underline{k}} \in \textrm{Leaf}(G_{\sim}^{\underline{A}})$.
Symmetrically, if $\overline{A} \neq \{x_{\overline{k}}\}$,
$x_{\overline{k}+1}$ is the unique neighbor of $x_{\overline{k}}$ in $G_{\sim}^{\overline{A}}$ and $x_{\overline{k}} \in \textrm{Leaf}(G_{\sim}^{\overline{A}})$.

Now, we construct a line $\mathcal{L}^M$ over all alternatives of $M$ such that $x_{\underline{k}}$ and $x_{\overline{k}}$ are the two leaves of the line, and all alternatives of $M\backslash \{x_{\underline{k}},x_{\overline{k}}\}$ are arbitrarily arranged in the interior of the line.
By combining $G_{\sim}^{\underline{A}}$, $\mathcal{L}^M$ and $G_{\sim}^{\overline{A}}$,
we generate a tree $\mathcal{T}_f^{A}$.
Clearly, by construction, $\mathcal{T}_f^{\underline{A}} = G_{\sim}^{\underline{A}}$ and
$\mathcal{T}_f^{\overline{A}} = G_{\sim}^{\overline{A}}$.
By construction, $x_{\underline{k}}$ and $x_{\overline{k}}$ are dual-thresholds in $\mathcal{T}_f^A$.
Thus, according to $\mathcal{T}_f^A$ and $G_{\sim}^A$, we have $A^{x_{\underline{k}}\rightharpoonup x_{\overline{k}}} =
\big\{x \in A: x_{\underline{k}} \in \langle x, x_{\overline{k}}|\mathcal{T}_f^A\rangle\big\} = \underline{A}$,
$\langle x_{\underline{k}}, x_{\overline{k}}|\mathcal{T}_f^A\rangle = M$, and
$A^{x_{\overline{k}}\,\rightharpoonup x_{\underline{k}}} =
\big\{x \in A: x_{\overline{k}} \in \langle x, x_{\underline{k}}|\mathcal{T}_f^A\rangle\big\} = \overline{A}$.
In the rest of proof, for notational convenience, we use the notation $\underline{A}$, $M$ and $\overline{A}$, instead of $A^{x_{\underline{k}}\rightharpoonup x_{\overline{k}}}$,
$\langle x_{\underline{k}}, x_{\overline{k}}|\mathcal{T}^A\rangle$ and $A^{x_{\overline{k}}\,\rightharpoonup x_{\underline{k}}}$.\medskip

The next lemma shows that voter 1 dictates on $M$.

\begin{lemma}\label{lem:dictatorship5}
We have $f(z, z') = z$ for all $z, z' \in M$.
\end{lemma}

\begin{proof}
By the definition of $\underline{B}$ and $\overline{B}$, one can easily notice
$\langle x_{\underline{k}}, x_{\overline{k}}|\mathcal{L}^B\rangle \subseteq M$.
According to the connected graph $G_{\sim}^{M}$,
we know that either $\textrm{Leaf}(G_{\sim}^{M}) = \emptyset$ or $\textrm{Leaf}(G_{\sim}^{M}) \neq \emptyset$ holds.
If $\textrm{Leaf}(G_{\sim}^{M}) = \emptyset$, the lemma follows from statement (i) of Lemma \ref{lem:union2} and the hypothesis that voter $1$ dictates on $\langle x_{\underline{k}}, x_{\overline{k}}|\mathcal{L}^B\rangle \subseteq M$.
Henceforth, let $\textrm{Leaf}(G_{\sim}^{M}) \neq \emptyset$.
To complete the verification,
by statement (ii) of Lemma \ref{lem:union2},
we show that given an arbitrary $x \in \textrm{Leaf}(G_{\sim}^{M})$ and $(x,y) \in \mathcal{E}_{\sim}^M$,
$f$ behaves likes dictatorship on $\{x, y\}$.
We have two cases: $x \in \{x_{\underline{k}}, x_{\overline{k}}\}$ and $x \in M\backslash \{x_{\underline{k}}, x_{\overline{k}}\}$.
In the first case, since
$G_{\sim}^{\langle x_{\underline{k}}, x_{\overline{k}}|\mathcal{L}^B\rangle}$ is a connected graph nested in $G_{\sim}^M$ (implied by $\langle x_{\underline{k}}, x_{\overline{k}}|\mathcal{L}^B\rangle \subseteq M$),
we know that $x_{\underline{k}}$ has a neighbor that is in $\langle x_{\underline{k}}, x_{\overline{k}}|\mathcal{L}^B\rangle \subseteq M$.
Then, $x \in \textrm{Leaf}(G_{\sim}^{M})$ and $(x,y) \in \mathcal{E}_{\sim}^M$ imply $y \in \langle x_{\underline{k}}, x_{\overline{k}}|\mathcal{L}^B\rangle$.
Consequently, by Lemma \ref{lem:RFBR2}, $f$ behaves likes dictatorship on $\{x, y\}$, as required.
Henceforth, let the second case hold: $x \in M\backslash \{x_{\underline{k}}, x_{\overline{k}}\}$.

Since the adjacency graph over $\langle x_{\underline{k}}, x_{\overline{k}}|\mathcal{L}^B\rangle$
is a connected graph,
we fix a path $\pi = (y_1, \dots, y_t)$ such that $y_1, \dots, y_t \in \langle x_{\underline{k}}, x_{\overline{k}}|\mathcal{L}^B\rangle$, $y_1 = x_{\underline{k}}$ and $y_t = x_{\overline{k}}$.
Since $x \notin  \{x_{\underline{k}}, x_{\overline{x}}\}$ and $x \in \textrm{Leaf}(G_{\sim}^{M})$,
it is evident that $x \notin \pi$.
Moreover, since $G_{\sim}^M$ is a connected graph,
we can identify an alternative $y_s \in \pi$ and a path $\pi'=(z_1, \dots, z_p)$ in $G_{\sim}^M$ connecting $y_s$ and $x$, such that $z_2, \dots, z_p \notin \pi$.
Clearly, since $x \in \textrm{Leaf}(G_{\sim}^{M})$ and $(x,y) \in \mathcal{E}_{\sim}^M$,
we have $z_{p-1} = y$ and hence $p \geq 2$.
By Lemma \ref{lem:RFBR2}, we know that $f$ behaves likes dictatorship on $\pi$.
Since $z_1 = y_s \in \pi$, $f$ behaves likes dictatorship on the set $\pi\cup \{z_1\}$.
We pick an arbitrary $l \in \{2, \dots, p\}$, and provide an induction hypothesis:
SCF $f$ behaves like a dictatorship on $\pi\cup \{z_1, \dots, z_{l-1}\}$.
We show that $f$ behaves like a dictatorship on $\pi\cup \{z_1, \dots, z_l\}$.
Furthermore, by Lemma \ref{lem:union}, it suffices to show that voter 1 dictates on $\{ z_{l-1}, z_l\}$, i.e.,
$f(z_{l-1}, z_l) =z_{l-1}$ and $f(z_{l}, z_{l-1})=z_l$.

Since $z_{l-1} \sim z_l$, it is evident that $f(z_{l-1}, z_l) \in \{z_{l-1}, z_l\}$ and
$f(z_{l}, z_{l-1}) \in \{z_{l-1}, z_l\}$.
Suppose $f(z_{l-1}, z_l) = z_{l}$.
Then, according to the path $(z_{l-1}, \dots, z_1)$, statement (ii) of Lemma \ref{lem:transitivity} implies $f(z_{1}, z_l) = z_l$.
Given $\underline{P}_2 = \underline{P}_i$ and $\overline{P}_2 = \overline{P}_i$,
since $z_1 = y_s \in \langle x_{\underline{k}}, x_{\overline{k}}|\mathcal{L}^B\rangle$,
Lemma \ref{lem:RFBR2} implies $f(z_1, \underline{P}_2) = f(z_1, x_1)=z_1$ and
$f(z_1, \overline{P}_2) = f(z_1, x_v) = z_1$.
Then, strategy-proofness implies $z_1\mathrel{\underline{P}_2}z_l$ and $z_1\mathrel{\overline{P}_2}z_l$ which contradicts the fact that $\underline{P}_2$ and $\overline{P}_2$ are complete reversals.
Therefore, $f(z_{l-1}, z_l) = z_{l-1}$, as required.

Last, we show $f(z_{l}, z_{l-1}) = z_{l}$.
Suppose not, i.e., $f(z_{l}, z_{l-1}) = z_{l-1}$.
First, since $z_{l-1}, \dots, z_2 \in M\backslash \pi$,
we know $z_{l-1}, \dots, z_2 \notin \underline{A}\cup \overline{A}$.
Then, we can construct two concatenated paths
$\bar{\pi} = (z_{l-1}, \dots, z_1=y_s, \dots, y_1=x_{\underline{k}}, \dots, x_1)$ and
$\hat{\pi} = (z_{l-1}, \dots, z_1=y_s, \dots, y_t=x_{\overline{k}}, \dots, x_v)$.
Then, according to $\bar{\pi}$ and $\hat{\pi}$,
statement (iii) of Lemma \ref{lem:transitivity} implies
$f(z_l, x_1) \in \bar{\pi} $ and $f(z_l, x_v) \in \hat{\pi} $.
Next, we show the following claim.\medskip

\noindent
\textsc{Claim 1}: We have
$f(z_{l-1}, x_1) =z_{l-1}$ and
$f(z_{l-1}, x_v) = z_{l-1}$.\medskip

There are two cases:
$l=2$ and $l>2$.
In the first case, $z_{l-1} = y_s \in \langle x_{\underline{k}}, x_{\overline{k}}|\mathcal{L}^B\rangle$.
Then, Lemma \ref{lem:RFBR2} implies $f(z_{l-1}, x_1) =z_{l-1}$ and $f(z_{l-1}, x_v) = z_{l-1}$.
In the second case, the induction hypothesis implies $f(z_{l-1}, y_1) = z_{l-1}$ and
$f(z_{l-1}, y_t) = z_{l-1}$.
Note that $z_{l-1} \notin \pi$ implies $z_{l-1} \neq y_1$ and $z_{l-1} \neq y_t$.
Then, according to the paths
$(y_1 = x_{\underline{k}}, \dots, x_1)$ and
$(y_t = x_{\overline{k}}, \dots, x_v)$ which both clearly exclude $z_{l-1}$,
by statement (ii) of Lemma \ref{lem:transitivity},
$f(z_{l-1}, y_1) = z_{l-1}$ implies $f(z_{l-1}, x_1) =z_{l-1}$, and
$f(z_{l-1}, y_t) = z_{l-1}$ implies $f(z_{l-1}, x_v) = z_{l-1}$.
This completes the verification of the claim.\medskip

Furthermore,
since $z_l \sim z_{l-1}$, by statement (iii) of Lemma \ref{lem:transitivity},
Claim 1 implies $f(z_{l}, x_1) \in \{z_{l-1}, z_l\}$ and
$f(z_{l}, x_v) \in \{z_{l-1}, z_l\}$.
Therefore, we have $f(z_{l}, x_1) \in \bar{\pi} \cap \{z_{l-1}, z_l\} = \{z_{l-1}\}$ and
$f(z_{l}, x_v) \in \hat{\pi} \cap \{z_{l-1}, z_l\} = \{z_{l-1}\}$,
which respectively imply $f(z_{l}, x_1) = z_{l-1}$ and $f(z_{l}, x_v) = z_{l-1}$.
Thus, given $\underline{P}_2 = \underline{P}_i$ and $\overline{P}_2 = \overline{P}_i$,
we have $f(z_{l}, \underline{P}_2) = z_{l-1}$ and $f(z_l, \overline{P}_2) = z_{l-1}$.
Consequently, given $f(z_l, z_l) = z_l$ by unanimity,
strategy-proofness implies $z_{l-1}\mathrel{\underline{P}_2}z_l$ and $z_{l-1}\mathrel{\overline{P}_2}z_l$,
which contradicts the fact that $\underline{P}_2$ and $\overline{P}_2$ are complete reversals.
Hence, $f(z_{l}, z_{l-1})=z_l$, as required.
This completes the verification of the induction hypothesis.
Therefore, $f$ behaves like a dictatorship on $\{x, y\}$, as required.
This proves the lemma.
\end{proof}

\begin{lemma}\label{lem:fullcharacterization}
SCF $f$ is a hybrid rule on $\mathcal{T}_f^A$ w.r.t.~$x_{\underline{k}}$ and $x_{\overline{k}}$, i.e., for all $x, y\in A$,
\begin{align*}
f(x,y) = \left\{
\begin{array}{ll}
x & \emph{if}\; x \in M,\\
\mathop{\emph{Proj}}\big(x_{\underline{k}}, \langle x, y|\mathcal{T}_f^A\rangle\big) & \emph{if}\; x \in \underline{A}\backslash \{x_{\underline{k}}\},\; \textrm{and}\\
\mathop{\emph{Proj}}\big(x_{\overline{k}}, \langle x, y|\mathcal{T}_f^A\rangle\big) & \emph{if}\; x \in \overline{A}\backslash \{x_{\overline{k}}\}.
\end{array}
\right.
\end{align*}
\end{lemma}

\begin{proof}
We first know that voter $1$ dictates on $M$ by Lemma \ref{lem:dictatorship5}.
Next, given $x \in M$ and $y \in \underline{A}\backslash \{x_{\underline{k}}\}$, we show $f(x, y) = x$.
By a symmetric argument, we can show $f(x, y) = x$ for all $x \in M$ and $y \in \overline{A}\backslash \{x_{\overline{k}}\}$.
Since $y \in \underline{A}\backslash \{x_{\underline{k}}\}$,
we know $\underline{A}\neq \{x_{\underline{k}}\} $ and
hence $x_{\underline{k}-1} \in \underline{A}\backslash \{x_{\underline{k}}\}$.
Then, we have $f(x_{\underline{k}}, x_{\underline{k}-1}) = x_{\underline{k}}$ by Lemma \ref{lem:RFBR2}.
In $G_{\sim}^{M}$,
there exists a path $(z_1, \dots, z_s)$ connecting $x_{\underline{k}}$ and $x$.
Then, by statement (iii) of Lemma \ref{lem:transitivity}, $f(x_{\underline{k}}, x_{\underline{k}-1}) = x_{\underline{k}}$ implies
$f(x, x_{\underline{k}-1}) \in \{z_1, \dots, z_s\}$.
Suppose $f(x, x_{\underline{k}-1}) = z_k$ for some $k \in \{1, \dots, s-1\}$.
Then, strategy-proofness implies $f(x, z_k) = z_k$ which contradicts Lemma \ref{lem:RFBR2}.
Therefore, $f(x, x_{\underline{k}-1}) = z_s=x$.
Furthermore, in $G_{\sim}^{\underline{A}}$,
we have a path $(y_1, \dots, y_t)$ connecting $x_{\underline{k}-1}$ and $y$.
Clearly, $x$ is not included in $(y_1, \dots, y_t)$.
Then, by statement (ii) of Lemma \ref{lem:transitivity}, $f(x, x_{\underline{k}-1}) =x$ implies $f(x, y) = x$, as required.

Second, given $x \in \underline{A}\backslash \{x_{\underline{k}}\}$ and $y \in A$,
we show $f(x, y) = \mathop{\textrm{Proj}}\big(x_{\underline{k}}, \langle x, y|\mathcal{T}_f^A\rangle\big)$.
Since $x \in \underline{A}\backslash \{x_{\underline{k}}\}$,
we know $\underline{A}\neq \{x_{\underline{k}}\}$ and
$x_{\underline{k}-1} \in \underline{A}\backslash \{x_{\underline{k}}\}$.
Then, Lemma \ref{lem:RFBR2} implies $f(x_{\underline{k}-1}, x_{\underline{k}}) = x_{\underline{k}}$.
We consider two cases: $y \in \langle x_{\underline{k}}, x_{\overline{k}}|\mathcal{T}_f^A\rangle\cup \overline{A}$ and
$y \in \underline{A}\backslash \{x_{\underline{k}}\}$.
In the first case,
we have a path $(y_1, \dots, y_s)$ in the connected graph $G_{\sim}^M\cup G_{\sim}^{\overline{A}}$ connecting $x_{\underline{k}}$ and $y$.
Then, by statement (iii) of Lemma \ref{lem:transitivity}, $f(x_{\underline{k}-1}, x_{\underline{k}}) = x_{\underline{k}}$ implies $f(x_{\underline{k}-1}, y) \in \{y_1, \dots, y_s\}$.
Meanwhile, since $f(x_{\underline{k}}, y) = x_{\underline{k}}$ and $x_{\underline{k}-1} \sim x_{\underline{k}}$,
statement (iii) of Lemma \ref{lem:transitivity} implies $f(x_{\underline{k}-1}, y) \in \{x_{\underline{k}}, x_{\underline{k}-1}\}$.
Therefore, $f(x_{\underline{k}-1}, y) \in \{y_1, \dots, y_s\}\cap \{x_{\underline{k}}, x_{\underline{k}-1}\} = \{x_{\underline{k}}\}$, and hence
$f(x_{\underline{k}-1}, y) = x_{\underline{k}}$.
Furthermore,
in $G_{\sim}^{\underline{A}}$, we have the path $(z_1, \dots, z_t)$ connecting $x_{\underline{k}-1}$ and $x$. Clearly, $x_{\underline{k}}$ is not included in $(z_1, \dots, z_t)$.
Then, by statement (ii) of Lemma \ref{lem:transitivity},
$f(x_{\underline{k}-1}, y) = x_{\underline{k}}$ implies
$f(x, y) = x_{\underline{k}} = \mathop{\textrm{Proj}}\big(x_{\underline{k}}, \langle x, y|\mathcal{T}_f^A\rangle\big)$, as required.
In the second case,
we have a path $\pi$ in $G_{\sim }^{\underline{A}}$ connecting $x$ and $y$.
Then, statement (i) of Lemma \ref{lem:transitivity} implies $f(x, y) \in \pi$.
Meanwhile, we have $f(x_{\underline{k}}, y) = x_{\underline{k}}$ by the verification in the first paragraph and
$f(x, x_{\underline{k}}) = x_{\underline{k}}$ by the verification of the first case.
Note that in $G_{\sim}^{\underline{A}}$,
there exist a path $\pi'$ connecting $x_{\underline{k}}$ and $x$, and a path $\pi''$ connecting $x_{\underline{k}}$ and $y$.
Then, by statement (iii) of Lemma \ref{lem:transitivity},
$f(x_{\underline{k}}, y) = x_{\underline{k}}$ implies $f(x, y) \in \pi'$, and $f(x, x_{\underline{k}}) = x_{\underline{k}}$ implies $f(x, y) \in \pi''$.
Therefore, $f(x, y) \in \pi\cap \pi' \cap \pi''$.
Last, since $G_{\sim}^{\underline{A}}$ is a tree,
it is true that $\pi\cap \pi' \cap \pi'' = \big\{\mathop{\textrm{Proj}}(x_{\underline{k}}, \langle x, y|\mathcal{T}_f^A\rangle)\big\}$.
Hence, we have $f(x, y) = \mathop{\textrm{Proj}}\big(x_{\underline{k}}, \langle x, y|\mathcal{T}_f^A\rangle\big)$, as required.
Symmetrically,
when $x \in \overline{A}\backslash \{x_{\overline{k}}\}$, we can show $f(x, y) = \mathop{\textrm{Proj}}\big(x_{\overline{k}}, \langle x, y|\mathcal{T}_f^A\rangle\big)$ for all $y \in A$.
This proves the lemma.
\end{proof}

\begin{lemma}\label{lem:sh}
We have $\mathbb{D} \subseteq \mathbb{D}_{\emph{SH}}\big(\mathcal{T}_f^A, x_{\underline{k}}, x_{\overline{k}}\big)$.
\end{lemma}

\begin{proof}
First, given $P_i \in \mathbb{D}$ with $r_1(P_i) \in \underline{A}\backslash \{x_{\underline{k}}\}$,
we show that $P_i$ is semi-single-peaked on $\mathcal{T}_f^A$ w.r.t.~$x_{\underline{k}}$ and $\max^{P_i}(\overline{A}) = x_{\overline{k}}$.
Let $P_1 = P_i$.
Given distinct $x, y \in \langle r_1(P_1), x_{\underline{k}}|\mathcal{T}_f^A\rangle$ such that $x \in \langle r_1(P_1), y|\mathcal{T}_f^A\rangle$, we by Lemma \ref{lem:fullcharacterization} have $f(P_1, x) = \textrm{Proj}\big(x_{\underline{k}}, \langle r_1(P_1), x|\mathcal{T}_f^A\rangle\big) = x$ and $f(y, x) = \textrm{Proj}\big(x_{\underline{k}}, \langle y, x|\mathcal{T}_f^A\rangle\big) = y$.
Then, strategy-proofness implies $x\mathrel{P_1}y$, which meets condition (i) of Definition \ref{def:ssp}.
Given $x \notin \langle r_1(P_1), x_{\underline{k}}|\mathcal{T}_f^A\rangle$ and $\hat{x} = \textrm{Proj}\big(x, \langle r_1(P_1), x_{\underline{k}}|\mathcal{T}_f^A\rangle\big)$, we by Lemma \ref{lem:fullcharacterization} have
$f(P_1, x) = \textrm{Proj}\big(x_{\underline{k}}, \langle r_1(P_1), x|\mathcal{T}_f^A\rangle\big)
= \textrm{Proj}\big(x, \langle r_1(P_1), x_{\underline{k}}|\mathcal{T}_f^A\rangle\big) = \hat{x}$ and
$f(x, x) = x$.
Immediately, strategy-proofness implies $\hat{x}\mathrel{P_1}x$, which meets condition (ii) of Definition \ref{def:ssp}.
Therefore, $P_i$ is semi-single-peaked on $\mathcal{T}_f^A$ w.r.t.~$x_{\underline{k}}$, as required.
Next, to show $\max^{P_i}(\overline{A}) = x_{\overline{k}}$, we fix $P_2 = P_i$ and an arbitrary $x \in \overline{A}\backslash \{x_{\overline{k}}\}$, and show $x_{\overline{k}}\mathrel{P_2}x$.
Since
$f(x, P_2) = \textrm{Proj}\big(x_{\overline{k}}, \langle x, r_1(P_2)|\mathcal{T}_f^A\rangle\big) = x_{\overline{k}}$ by Lemma \ref{lem:fullcharacterization} and $f(x, x) = x$ by unanimity, strategy-proofness implies $x_{\overline{k}}\mathrel{P_2}x$, as required.

Symmetrically, given $P_i \in \mathbb{D}$ with $r_1(P_i) \in \overline{A}\backslash \{x_{\overline{k}}\}$,
we can show that $P_i$ is semi-single-peaked on $\mathcal{T}_f^A$ w.r.t.~$x_{\overline{k}}$ and $\max^{P_i}(\underline{A}) = x_{\underline{k}}$.

Last, given $P_i \in \mathbb{D}$ with $r_1(P_i) \in \langle x_{\underline{k}}, x_{\overline{k}}|\mathcal{T}_f^A\rangle = M$,
we show $\max^{P_i}(\underline{A}) = x_{\underline{k}}$ and $\max^{P_i}(\overline{A}) = x_{\overline{k}}$.
Fixing $P_2 = P_i$, an arbitrary $x \in \underline{A}\backslash \{x_{\underline{k}}\}$ and an arbitrary $y \in \overline{A}\backslash \{x_{\overline{k}}\}$, it suffices to show $x_{\underline{k}}\mathrel{P_2}x$ and $x_{\overline{k}}\mathrel{P_2}y$.
By Lemma \ref{lem:fullcharacterization}, we have $f(x, P_2) = \textrm{Proj}\big(x_{\underline{k}}, \langle x, r_1(P_2)|\mathcal{T}_f^A\rangle\big) = x_{\underline{k}}$ and $f(x, x) = x$. Then, strategy-proofness implies $x_{\underline{k}}\mathrel{P_2}x$, as required. Symmetrically,
$f(y, P_2) = \textrm{Proj}\big(x_{\overline{k}}, \langle y, r_1(P_2)|\mathcal{T}_f^A\rangle\big) = x_{\overline{k}}$ and $f(y, y) = y$,
which by strategy-proofness imply $x_{\overline{k}}\mathrel{P_2}y$, as required.
This proves the lemma, and completes \textbf{Step 2} of the proof.
\end{proof}

Now, we start the verification of \textbf{Step 3}.
We first make one observation that henceforth will be repeatedly applied according to the line $\mathcal{L}^B=(x_1, \dots, x_v)$ constructed in \textbf{Step 1} (note that $\mathcal{L}^B$ may not be a path in $G_{\sim}^A$) and $(x_{\underline{k}}, x_{\overline{k}})$-semi-hybridness on $\mathcal{T}_f^A$ established in Lemma \ref{lem:sh}.

\begin{observation}\label{obs:2}\rm
If $\underline{k}>1$, then $(x_1, \dots, x_{\underline{k}})$ is the unique path in $G_{\sim}^A$ (also in $G_{\sim}^{\underline{A}}$) connecting $a_1$ and $a$, and $\mathcal{N}_{\sim}^A(x_{\underline{k}})\cap \underline{A} = \{x_{\underline{k}-1}\}$;
if $\overline{k}<v$, then $(x_v, \dots, x_{\overline{k}})$ is the unique path in $G_{\sim}^A$ (also in $G_{\sim}^{\overline{A}}$) connecting $a_m$ and $b$, and $\mathcal{N}_{\sim}^A(x_{\overline{k}})\cap \overline{A} = \{x_{\overline{k}+1}\}$. Moreover, $x_{\underline{k}}, \dots, x_{\overline{k}} \in M$.
\end{observation}

Recall that $\mathbb{D}$ is an $(a,b)$-semi-hybrid domain on $\mathcal{T}^A$.
Hence, the adjacency graph $G_{\sim}^A$ is a combination of the subtree $G_{\sim}^{A^{a \rightharpoonup b}} = \mathcal{T}^{A^{a \rightharpoonup b}}$, the connected subgraph $G_{\sim}^{\langle a,b|\mathcal{T}^A\rangle}$ and the subtree $G_{\sim}^{A^{b \rightharpoonup a}} = \mathcal{T}^{A^{b \rightharpoonup a}}$, i.e., $G_{\sim}^A = G_{\sim}^{A^{a \rightharpoonup b}}\cup G_{\sim}^{\langle a,b|\mathcal{T}^A\rangle} \cup G_{\sim}^{A^{b \rightharpoonup a}}$.
Moreover, since $a_1 \in A^{a \rightharpoonup b}$ and $a_m \in A^{b \rightharpoonup a}$,
by $(a,b)$-semi-hybridness on $\mathcal{T}^A$,
we know that in every path of $\Pi(a_1, a_m)$, $a$ and $b$ must be included, and $a$ is located closer to $a_1$ than $b$.
Then, according to the line $\mathcal{L}^B = (x_1, \dots, x_v)$ constructed in \textbf{Step 1},
it must be the case that $a = x_p$ and $b = x_q$ for some $1 \leq p < q \leq v$.

\begin{lemma}\label{lem:observation}
According to the line $\mathcal{L}^B=(x_1, \dots, x_v)$ and $(a, b)$-semi-hybridness on $\mathcal{T}^A$,
if $p>1$, then $(x_1, \dots, x_p)$ is the unique path in $G_{\sim}^A$ (also in $G_{\sim}^{A^{a \rightharpoonup b}}$) connecting $a_1$ and $a$, and $\mathcal{N}_{\sim}^A(x_{p})\cap A^{a \rightharpoonup b} = \{x_{p-1}\}$;
if $q< v$, then $(x_v, \dots, x_q)$ is the unique path in $G_{\sim}^A$ (also in $G_{\sim}^{A^{b \rightharpoonup a}}$) connecting $a_m$ and $b$, and
$\mathcal{N}_{\sim}^A(x_{q})\cap A^{b \rightharpoonup a} = \{x_{q+1}\}$.
Moreover, $x_p, \dots, x_q \in \langle a, b|\mathcal{T}^A\rangle$.
\end{lemma}

\begin{proof}
First, let $p>1$.
Note that if $(x_1, \dots, x_p)$ is a path in $G_{\sim}^A$ connecting $a_1$ and $a$,
then by $(a,b)$-semi-hybridness on $\mathcal{T}^A$,
it is true that $(x_1, \dots, x_p)$ is the unique path in $G_{\sim}^A$ (also in $G_{\sim}^{A^{a \rightharpoonup b}}$) connecting $a_1$ and $a$, and $a \in \textrm{Leaf}(\mathcal{T}^{A^{a \rightharpoonup b}})$ which by statement (ii) of Clarification \ref{cla:sh} implies
$\mathcal{N}_{\sim}^A(x_{p})\cap A^{a \rightharpoonup b} = \{x_{p-1}\}$.
Henceforth, we focus on showing that $(x_1, \dots, x_p)$ is a path in $G_{\sim}^A$ connecting $a_1$ and $a$.
Recall the set $\Pi(a_1,a_m)$ specified in \textbf{Step 1} that contains all paths connecting $a_1$ and $a_m$ in $G_{\sim}^A$.
If $|\Pi(a_1,a_m)| = 1$, then the line $\mathcal{L}^B$ is a path in $G_{\sim}^A$ connecting $a_1$ and $a_m$,
which immediately implies that $(x_1, \dots, x_p)$ is a path in $G_{\sim}^A$ connecting $a_1$ and $a$, as required.
Next, let $|\Pi(a_1,a_m)| > 1$.
By $(a,b)$-semi-hybridness on $\mathcal{T}^A$ and Clarification \ref{cla:hybridrule}, we first know that
the following $(a,b)$-hybrid rule on $\mathcal{T}^A$ is tops-only and strategy-proof on $\mathbb{D}$: for all $P_1, P_2 \in \mathbb{D}$,
\begin{align*}
g(P_1, P_2) = \left\{
\begin{array}{ll}
r_1(P_1) & \textrm{if}\; r_1(P_1) \in \langle a, b|\mathcal{T}^A\rangle,\\
\mathop{\textrm{Proj}}\big(a, \langle r_1(P_1), r_1(P_2)|\mathcal{T}^A\rangle\big) & \textrm{if}\; r_1(P_1) \in A^{a \rightharpoonup b}\backslash \{a\},\; \textrm{and}\\
\mathop{\textrm{Proj}}\big(b, \langle r_1(P_1), r_1(P_2)|\mathcal{T}^A\rangle\big) & \textrm{if}\; r_1(P_1) \in A^{b \rightharpoonup a}\backslash \{b\}.
\end{array}
\right.
\end{align*}
Clearly, $g$ behaves like a dictatorship on $\langle a, b|\mathcal{T}^A\rangle$.
Meanwhile, recall the adjacency graph $G_{\sim}^B$ in Figure \ref{fig:multiplepaths} and the fact that
every two-voter, tops-only and strategy-proof rule behaves like a dictatorship on the set $\mathcal{O}$.
Therefore, it must be the case that $\mathcal{O} \subseteq \langle a, b|\mathcal{T}^A\rangle$.
Thus, in Figure \ref{fig:multiplepaths}, we have $g(x_s,x_t) = x_s$.
Note that $(x_1, \dots, x_s)$ and $(x_t, \dots, x_v)$ are two paths in $G_{\sim}^A$ according to $G_{\sim}^{\mathcal{O}}$.
Then, according to the paths $(x_t, \dots, x_v)$,
by statement (ii) of Lemma \ref{lem:transitivity}, $g(x_s,x_t) = x_s$ implies $g(x_s, x_v) = x_s$.
Furthermore, according to $(x_1, \dots, x_s)$,
by statement (iii) of Lemma \ref{lem:transitivity},
$g(x_s, x_v) = x_s$ implies $x_p = g(x_1, x_v) \in \{x_1, \dots, x_s\}$.
Therefore, we have $1 \leq p \leq s$ which implies that $(x_1, \dots, x_p)$ is a path in $G_{\sim}^A$ connecting $a_1$ and $a$, as required.
Symmetrically, if $q< v$, then $(x_v, \dots, x_q)$ is the unique path in $G_{\sim}^A$ (also in $G_{\sim}^{A^{b \rightharpoonup a}}$) connecting $a_m$ and $b$, and
$\mathcal{N}_{\sim}^A(x_{q})\cap A^{b \rightharpoonup a} = \{x_{q+1}\}$.

Last, we show $x_p, \dots, x_q \in \langle a, b|\mathcal{T}^A\rangle$.
Suppose not, i.e., we have $x_k \notin \langle a, b|\mathcal{T}^A\rangle$ for some $k \in \{p, \dots, q\}$.
Thus, either $x_k \in A^{a \rightharpoonup b}\backslash \{a\}$ or $x_k \in A^{b \rightharpoonup a}\backslash \{b\}$ holds.
We assume w.l.o.g.~that $x_k \in A^{a \rightharpoonup b}\backslash \{a\}$. Thus, we have $1 \leq k < p$, which implies $\mathcal{N}_{\sim}^A(x_{p})\cap A^{a \rightharpoonup b} = \{x_{p-1}\}$.
Note that by $(a,b)$-semi-hybridness on $\mathcal{T}^A$, according to the tree $G_{\sim}^{A^{a \rightharpoonup b}} = \mathcal{T}^{A^{a \rightharpoonup b}}$,
since $x_k \in A^{a \rightharpoonup b}\backslash \{a\}$,
$\langle x_k, a|G_{\sim}^{A^{a \rightharpoonup b}}\rangle$ is the unique path in $G_{\sim}^A$ connecting $x_k$ and $a$.
Furthermore, since $\mathcal{N}_{\sim}^A(x_{p})\cap A^{a \rightharpoonup b} = \{x_{p-1}\}$,
$x_{p-1}$ must be included in $\langle x_k, a|G_{\sim}^{A^{a \rightharpoonup b}}\rangle$.
However, according to the adjacency graph $G_{\sim}^{B}$ in Figure \ref{fig:multiplepaths},
we have a path in $G_{\sim}^A$ that connects $x_k$ and $a$, and excludes $x_{p-1}$ - a contradiction.
This proves the lemma.
\end{proof}

\begin{lemma}\label{lem:superset}
We have $\langle x_{\underline{k}}, x_{\overline{k}}|\mathcal{T}_f^A\rangle \supseteq \langle a, b|\mathcal{T}^A\rangle$.
\end{lemma}

\begin{proof}
\textsc{Claim 1}:
We have $\big[\underline{k} \leq p\big] \Rightarrow \big[\underline{A} \subseteq A^{a \rightharpoonup b}\big]$ and
$\big[\,\overline{k} \geq q\big] \Rightarrow \big[\,\overline{A} \subseteq A^{b \rightharpoonup a}\big]$.\medskip

Given $\underline{k} \leq p$, we show $\underline{A} \subseteq A^{a \rightharpoonup b}$.
There are two cases: $p = 1$ and $p>1$.

First, let $p=1$. Then, $1 \leq \underline{k} \leq p$ implies $\underline{k} =1$.
Thus, $a=x_{\underline{k}} = a_1$.
We first show $A^{a \rightharpoonup b} = \{a\}$. Suppose not, i.e., we have some alternative $c \in A^{a \backslash b}\backslash \{a\}$.
Then, by $(a,b)$-semi-hybridness on $\mathcal{T}^A$, $r_1(\overline{P}_i) = a_m \notin A^{a \rightharpoonup b}$ implies $a_1\mathrel{\overline{P}_i}c$, which contradicts the fact that $a_1$ is bottom-ranked in $\overline{P}_i$.
Hence, $A^{a \rightharpoonup b} = \{a\}$.
Similarly, by $(x_{\underline{k}},x_{\overline{k}})$-semi-hybridness on $\mathcal{T}_f^A$,
since $r_1(\overline{P}_i) = a_m \notin \underline{A}$ and $a_1$ is bottom-ranked in $\overline{P}_i$,
it is true that $\underline{A} = \{x_1\}$.
Therefore, we have $\underline{A} = A^{a \rightharpoonup b}$.

Next, let $p>1$.
Suppose by contradiction that we have an alternative $z \in \underline{A}\backslash A^{a \rightharpoonup b}$.
Since $\underline{k} \leq p$, Lemma \ref{lem:observation} implies $x_{\underline{k}}\in A^{a \rightharpoonup b}$.
Thus, $z \notin A^{a \rightharpoonup b}$ implies $z \neq x_{\underline{k}}$ and hence $z \in \underline{A}\backslash \{x_{\underline{k}}\}$.
Note that $\pi = \langle z, x_{\underline{k}}|G_{\sim}^{\underline{A}}\rangle$ is the unique path in $G_{\sim}^A$ connecting $z$ and $x_{\underline{k}}$.
Moreover, by $(x_{\underline{k}},x_{\overline{k}})$-semi-hybridness on $\mathcal{T}_f^A$, $r_1(\overline{P}_i) = a_m \notin \underline{A}$ implies $x_{\underline{k}}\mathrel{\overline{P}_i}z$.
Therefore, $x_{\underline{k}}$ is distinct to the bottom-ranked alternative in $\overline{P}_i$ which is $a_1 = x_1$, and hence $\underline{k}>1$.
Then, Observation \ref{obs:2} implies that $x_{\underline{k}-1}$ is the unique neighbor of $x_{\underline{k}}$ in $G_{\sim}^{\underline{A}}$, and therefore we have $x_{\underline{k}-1} \in \pi$.
Meanwhile, since $z \notin A^{a \rightharpoonup b}$,
we have a path $(z_1, \dots, z_s)$ in $G_{\sim}^A$ connecting $z$ and $a = x_p$ such that $z_k \notin A^{a \rightharpoonup b}$ for all $k \in \{1, \dots, s-1\}$.
Since $x_{\underline{k}}, \dots, x_p \in A^{a \rightharpoonup b}$ by Lemma \ref{lem:observation}, we have a concatenated path $\pi' = (z_1, \dots, z_s = x_p, \dots, x_{\underline{k}})$ in $G_{\sim}^A$ connecting $z$ and $x_{\underline{k}}$.
Clearly, $\pi'$ does not include $x_{\underline{k}-1}$, and hence is distinct to $\pi$ - a contradiction.
Therefore, we have $\underline{A} \subseteq A^{a \rightharpoonup b}$.

Overall, we have $\big[\underline{k} \leq p\big] \Rightarrow \big[\underline{A} \subseteq A^{a \rightharpoonup b}\big]$.
By a symmetric argument, we can show $\big[\,\overline{k} \geq q\big] \Rightarrow \big[\,\overline{A} \subseteq A^{b \rightharpoonup a}\big]$.
This completes the verification of the claim.

\medskip
Note that if $\underline{k} \leq p$ and $\overline{k} \geq q$ hold,
Claim 1 implies $\langle x_{\underline{k}}, x_{\overline{k}}|\mathcal{T}_f^A\rangle \supseteq \langle a, b|\mathcal{T}^A\rangle$.
Henceforth, to complete the proof, we show $\underline{k} \leq p$ and $\overline{k} \geq q$.
By symmetry, we focus on showing $\underline{k} \leq p$.
Suppose not, i.e., $\underline{k} > p$.
We have two cases: $p< \underline{k} < q$ and $\underline{k} \geq q$.
In each case, we induce a contradiction.

Let $p< \underline{k} < q$.
Symmetric to Claim 1, $p< \underline{k}$ implies $A^{a \rightharpoonup b} \subseteq \underline{A}$.
\medskip

\noindent
\textsc{Claim 2}: We have $\underline{A}\cap A^{b \rightharpoonup a} =\emptyset$.\medskip

Suppose not, i.e., we have some $z \in \underline{A}\cap A^{b \rightharpoonup a}$.
Since $\underline{k}< q$, we know $b=x_q \notin \underline{A}$ by Observation \ref{obs:2}.
Therefore, $z \neq  b$ and hence $z \in A^{b \rightharpoonup a}\backslash \{b\}$.
On the one hand, according to $(a,b)$-semi-hybridness on $\mathcal{T}^A$ and $z \in A^{b \rightharpoonup a}\backslash \{b\}$, we know that every path in $G_{\sim}^A$ connecting $z$ and $a$ must include $b$.
On the other hand, according to $(x_{\underline{k}}, x_{\overline{k}})$-semi-hybridness on $\mathcal{T}_f^A$,
since $z \in \underline{A}$ and $a \in \underline{A}$,
we know that $\langle z, a|G_{\sim}^{\underline{A}}\rangle$ is the unique path in $G_{\sim}^A$ that connects $z$ and $a$.
However, since $b \notin \underline{A}$, the path $\langle z, a|G_{\sim}^{\underline{A}}\rangle$ does not include $b$ - a contradiction.
This completes the verification of the claim.
\medskip

Now, let $\widehat{M} = \big[A\backslash [\underline{A}\cup A^{b \rightharpoonup a}]\big]\cup \{x_{\underline{k}}, b\}$.
We construct a line $(z_1, \dots, z_t)$ over all alternatives of $\widehat{M}$ such that $t = |\widehat{M}|$, $z_1 = x_{\underline{k}}$,
$z_t = b$, and all alternatives of $\widehat{M}\backslash \{x_{\underline{k}},b\}$ are arbitrarily arranged in the interior of the line.
By combining the subtree $\mathcal{T}_f^{\underline{A}}$, the line $(z_1, \dots, z_t)$ and the subtree $\mathcal{T}^{A^{b \rightharpoonup a}}$,
we generate a tree $\widehat{\mathcal{T}}^A$. Clearly, $x_{\underline{k}}$ and $b$ are dual-thresholds in $\widehat{\mathcal{T}}^A$.
Thus, we have $\widehat{A}^{x_{\underline{k}} \rightharpoonup b} = \big\{x \in A: x_{\underline{k}} \in \langle x, b|\widehat{\mathcal{T}}^A\rangle\big\} = \underline{A}$ and
$\widehat{A}^{ b\rightharpoonup x_{\underline{k}}} = \big\{x \in A: b\in \langle x, x_{\underline{k}}|\widehat{\mathcal{T}}^A\rangle\big\} = A^{b \rightharpoonup a}$.\medskip

\noindent
\textsc{Claim 3}: We have $\mathbb{D} \subseteq \mathbb{D}_{\textrm{SH}}\big(\widehat{\mathcal{T}}^A, x_{\underline{k}}, b\big)$.\medskip

First, given $P_i \in \mathbb{D}$ with $r_1(P_i) \in \widehat{A}^{x_{\underline{k}} \rightharpoonup b}\backslash \{x_{\underline{k}}\}$,
we show that $P_i$ is semi-single-peaked on $\widehat{\mathcal{T}}^A$ w.r.t.~$x_{\underline{k}}$, and $\max^{P_i}(\widehat{A}^{ b\rightharpoonup x_{\underline{k}}} ) = b$.
More specifically, by the construction of $\widehat{\mathcal{T}}^A$,
the semi-single-peakedness requirement on $\widehat{\mathcal{T}}^A$ w.r.t.~$x_{\underline{k}}$ consists of the following two parts: (i) semi-single-peakedness on $\widehat{\mathcal{T}}^{\widehat{A}^{x_{\underline{k}} \rightharpoonup b}} = \mathcal{T}_f^{\underline{A}}$ w.r.t.~$x_{\underline{k}}$, and (ii) $x_{\underline{k}}\mathrel{P_i}x$ for all $x \in A\backslash \widehat{A}^{x_{\underline{k}} \rightharpoonup b} = A\backslash \underline{A}$.
Indeed, since $r_1(P_i) \in \widehat{A}^{x_{\underline{k}} \rightharpoonup b}\backslash \{x_{\underline{k}}\} = \underline{A}\backslash \{x_{\underline{k}}\}$,
both parts follow from $(x_{\underline{k}}, x_{\overline{k}})$-semi-hybridness on $\mathcal{T}_f^A$.
Next, since $P_i$ is $(a,b)$-semi-hybrid on $\mathcal{T}^A$,
$r_1(P_i) \in \underline{A} = \widehat{A}^{x_{\underline{k}}\rightharpoonup b} \subseteq A\backslash \widehat{A}^{b\rightharpoonup x_{\underline{k}}} = A\backslash A^{b\rightharpoonup a}$ implies $\max^{P_i}(\widehat{A}^{ b\rightharpoonup x_{\underline{k}}} ) = \max^{P_i}(A^{b \rightharpoonup a}) =  b$, as required.

Second, given $P_i \in \mathbb{D}$ with $r_1(P_i) \in \widehat{A}^{b \rightharpoonup x_{\underline{k}}}\backslash \{b\}$,
we show that $P_i$ is semi-single-peaked on $\widehat{\mathcal{T}}^A$ w.r.t.~$b$, and
$\max^{P_i}(\widehat{A}^{ x_{\underline{k}}\rightharpoonup b} ) = x_{\underline{k}}$.
More specifically, by the construction of $\widehat{\mathcal{T}}^A$,
the semi-single-peakedness requirement on $\widehat{\mathcal{T}}^A$ w.r.t.~$b$ consists of the following two parts:
(i) semi-single-peakedness on $\widehat{\mathcal{T}}^{\widehat{A}^{ b\rightharpoonup x_{\underline{k}}}}=\mathcal{T}^{A^{b \rightharpoonup a}}$ w.r.t.~$b$, and (ii) $b\mathrel{P_i}x$ for all $x \in A\backslash \widehat{A}^{ b\rightharpoonup x_{\underline{k}}} = A\backslash A^{b\rightharpoonup a}$.
Indeed, since $r_1(P_i) \in \widehat{A}^{b \rightharpoonup x_{\underline{k}}}\backslash \{b\} = A^{b \rightharpoonup a}\backslash \{b\}$,
both parts follow from $(a, b)$-semi-hybridness on $\mathcal{T}^A$.
Next, since $r_1(P_i) \in \widehat{A}^{b \rightharpoonup x_{\underline{k}}} = A^{b \rightharpoonup a}$, Claim 2 implies $r_1(P_i) \in A\backslash \underline{A}$.
Then, by $(x_{\underline{k}},x_{\overline{k}})$-semi-hybridness on $\mathcal{T}_f^A$,
we have $\max^{P_i}(\widehat{A}^{ x_{\underline{k}}\rightharpoonup b} ) = \max^{P_i}(\underline{A}) =  x_{\underline{k}}$, as required.

Last, given $P_i \in \mathbb{D}$ with $r_1(P_i) \in \langle x_{\underline{k}},b|\widehat{\mathcal{T}}^A\rangle$,
we show $\max^{P_i}(\widehat{A}^{ x_{\underline{k}}\rightharpoonup b} ) = x_{\underline{k}}$ and
$\max^{P_i}(\widehat{A}^{ b\rightharpoonup x_{\underline{k}}} ) = b$.
By the construction of $\widehat{\mathcal{T}}^A$ and $r_1(P_i) \in \langle x_{\underline{k}},b|\widehat{\mathcal{T}}^A\rangle$,
we know $r_1(P_i) \in [A\backslash \underline{A}]\cup \{x_{\underline{k}}\}$ and $r_1(P_i) \in [A\backslash A^{b \rightharpoonup a}]\cup \{b\}$,
which respectively imply $\max^{P_i}(\widehat{A}^{ x_{\underline{k}}\rightharpoonup b}) = \max^{P_i}(\underline{A}) = x_{\underline{k}}$
by $(x_{\underline{k}},x_{\overline{k}})$-semi-hybridness on $\mathcal{T}_f^A$, and
$\max^{P_i}(\widehat{A}^{ b\rightharpoonup x_{\underline{k}}}) = \max^{P_i}(A^{b \rightharpoonup a}) = b$
by $(a,b)$-semi-hybridness on $\mathcal{T}^A$, as required.
This completes the verification of the claim.\medskip

Thus, we have $\mathbb{D} \subseteq \mathbb{D}_{\textrm{SH}}\big(\widehat{\mathcal{T}}^A, x_{\underline{k}}, b\big)$ and $\langle x_{\underline{k}},b|\widehat{\mathcal{T}}^A\rangle = \widehat{M} \subset \langle a, b|\mathcal{T}^A\rangle$,
which contradict condition (ii) of Definition \ref{def:asspsh}.\medskip

Last, let $\underline{k} \geq q$. Thus, $\underline{k}>1$ and hence Observation \ref{obs:2} implies that $(x_1, \dots, x_{\underline{k}})$ is the unique path in $G_{\sim}^A$ (also in $G_{\sim}^{\underline{A}}$) connecting $a_1$ and $x_{\underline{k}}$, and $\mathcal{N}_{\sim}^A(x_{\underline{k}})\cap \underline{A} = \{x_{\underline{k}-1}\}$.\medskip

\noindent
\textsc{Claim 4}: We have $\mathcal{N}_{\sim}^A(b)\cap \langle a, b|\mathcal{T}^A\rangle = \{x_{q-1}\}$.\medskip

Since $(x_1, \dots, x_{\underline{k}})$ is a path in $G_{\sim}^{A}$,
we have $x_{q-1} \sim x_q$ and hence $x_{q-1} \in \mathcal{N}_{\sim}^A(x_q) = \mathcal{N}_{\sim}^A(b)$.
By Lemma \ref{lem:observation}, we know $x_{q-1} \in \langle a, b|\mathcal{T}^A\rangle$.
Thus, $x_{q-1} \in \mathcal{N}_{\sim}^A(b)\cap \langle a, b|\mathcal{T}^A\rangle$.
Suppose that there exists $z \in \mathcal{N}_{\sim}^A(b)\cap \langle a, b|\mathcal{T}^A\rangle$ such that $z \neq x_{q-1}$.
On the one hand,
by $(a, b)$-semi-hybridness on $\mathcal{T}^A$,
$r_1(\overline{P}_i) = a_m \in A^{b \rightharpoonup a}$ and $z \in \langle a, b|\mathcal{T}^A\rangle$ imply $b\mathrel{\overline{P}_i}z$.
On the other hand, we show $b\mathrel{\underline{P}_i}z$, which contradicts the fact that $\underline{P}_i$ and $\overline{P}_i$ are complete reversals.
There are two cases: $q = \underline{k}$ and $q < \underline{k}$.
We first assume $q = \underline{k}$.
Since $\mathcal{N}_{\sim}^A(x_{\underline{k}})\cap \underline{A} = \{x_{\underline{k}-1}\}$, $z \in \mathcal{N}_{\sim}^A(b) = \mathcal{N}_{\sim}^A(x_q) = \mathcal{N}_{\sim}^A(x_{\underline{k}})$ and $z \neq x_{q-1} = x_{\underline{k}-1}$, it is true that $z \notin \underline{A}$.
Then, by $(x_{\underline{k}}, x_{\overline{k}})$-semi-hybridness on $\mathcal{T}_f^A$,
$r_1(\underline{P}_i) = a_1 \in \underline{A}$ implies $x_{\underline{k}}\mathrel{\underline{P}_i}z$ (equivalently, $b\mathrel{\underline{P}_i}z$), as required.
Last, let $q < \underline{k}$.
Thus, $q < v$, and hence Lemma \ref{lem:observation} implies $x_{q+1}, \dots, x_{\underline{k}} \in A^{b \rightharpoonup a}\backslash \{b\}$.
Moreover, since $z \neq x_q$ and $z \in \langle a, b|\mathcal{T}^A\rangle$,
we have a concatenated path $(z, x_q, \dots, x_{\underline{k}})$ in $G_{\sim}^A$ connecting $z$ and $x_{\underline{k}}$,
which of course includes $x_{\underline{k}-1}$.
Therefore, by definition, we know $z \in \underline{A}$.
Furthermore, since $z \sim x_q$ and $z \notin \{x_{q-1}, x_{q+1}\}$, according to the tree $G_{\sim}^{\underline{A}}$, it is true that $\textrm{Proj}\big(z, (x_1, \dots, x_{\underline{k}})\big) = x_q$.
Then, by $(x_{\underline{k}}, x_{\overline{k}})$-semi-hybridness on $\mathcal{T}_f^A$, $r_1(\underline{P}_i) = a_1 = x_1 \in \underline{A}$ implies  $x_q\mathrel{\underline{P}_i}z$ (equivalently, $b\mathrel{\underline{P}_i}z$), as required.
This completes the verification of the claim.
\medskip

\noindent
\textsc{Claim 5}: We have $A^{a \rightharpoonup b}\cup \langle a, b|\mathcal{T}^A\rangle \subseteq \underline{A}$.\medskip

We have already known $A^{a \rightharpoonup b} \subseteq \underline{A}$ and $b = x_q \in \underline{A}$.
To complete the verification, we show $z \in \underline{A}$ for all $z \in \langle a, b|\mathcal{T}^A\rangle \backslash \{a,b\}$.
Suppose that there exists $z \in \langle a, b|\mathcal{T}^A\rangle \backslash \{a,b\}$ such that $z \notin \underline{A}$.
On the one hand, since $G_{\sim}^{\langle a, b|\mathcal{T}^A\rangle}$ is a connected graph and $b \in \textrm{Leaf}\big(G_{\sim}^{\langle a, b|\mathcal{T}^A\rangle}\big)$ by Claim 4, we have a path in $G_{\sim}^{A}$ that connects $z$ and $a$, and excludes $b$.
On the other hand, since $z \notin \underline{A}$ and $a \in \underline{A}\backslash \{x_{\underline{k}}\}$, we know by Lemma \ref{lem:graphpartition} that $x_{\underline{k}}$ must be included in every path in $G_{\sim}^A$ that connects $z$ and $a$.
Moreover, since Observation \ref{obs:2} implies that $(x_p, \dots, x_{\underline{k}})$, which of course implies $b = x_q$,
is the unique path in $G_{\sim}^A$ connecting $a$ and $x_{\underline{k}}$,
it is true that $b$ must be included in every path in $G_{\sim}^A$ that connects $z$ and $a$ - a contradiction.
This completes the verification of the claim.
\medskip

Claim 5 immediately implies that $G_{\sim}^{A^{a \rightharpoonup b}}\cup G_{\sim}^{\langle a, b|\mathcal{T}^A\rangle}$ is included in $G_{\sim}^{\underline{A}}$.
More importantly, since $G_{\sim}^{\underline{A}}$ is a tree, $G_{\sim}^{A^{a \rightharpoonup b}}$ is a tree, and $G_{\sim}^{\langle a, b|\mathcal{T}^A\rangle}$ is a connected graph,
it must be the case that $G_{\sim}^{\langle a, b|\mathcal{T}^A\rangle}$ is a tree as well.
Consequently, $G_{\sim}^A = G_{\sim}^{A^{a \rightharpoonup b}}\cup G_{\sim}^{\langle a, b|\mathcal{T}^A\rangle}\cup G_{\sim}^{A^{b \rightharpoonup a}}$ is a tree.

\medskip

\noindent
\textsc{Claim 6}: We have $\mathbb{D} \subseteq \mathbb{D}_{\textrm{SSP}}(G_{\sim}^A, b)$.\medskip

Note that the tree $G_{\sim}^{A^{a \rightharpoonup b}}\cup G_{\sim}^{\langle a, b|\mathcal{T}^A\rangle}$ is nested in $G_{\sim}^{\underline{A}}$,
$x_q = b \in \textrm{Leaf}\big(G_{\sim}^{\langle a, b|\mathcal{T}^A\rangle}\big)$ by Claim 4 and
$x_{\underline{k}} \in \textrm{Leaf}\big(G_{\sim}^{\underline{A}}\big)$ by Observation \ref{obs:2}.
Then, it must be the case that
$b \in \langle x, x_{\underline{k}}|G_{\sim}^{\underline{A}}\rangle$ for all $x \in A^{a \rightharpoonup b}\cup \langle a, b|\mathcal{T}^A\rangle$.
Furthermore, since $G_{\sim}^{\underline{A}}$ is nested in the tree $G_{\sim}^A$,
we know $\langle x, x_{\underline{k}}|G_{\sim}^{\underline{A}}\rangle = \langle x, x_{\underline{k}}|G_{\sim}^{A}\rangle$ for $x \in A^{a \rightharpoonup b}\cup \langle a, b|\mathcal{T}^A\rangle$.
Therefore, $b \in \langle x, x_{\underline{k}}|G_{\sim}^{A}\rangle$ for all $x \in A^{a \rightharpoonup b}\cup \langle a, b|\mathcal{T}^A\rangle$.

First, given $P_i \in \mathbb{D}$ with $r_1(P_i) \in A^{a \rightharpoonup b}\cup \langle a, b|\mathcal{T}^A\rangle$,
we show that $P_i$ is semi-single-peaked on $G_{\sim}^A$ w.r.t.~$b$,
which consists of the following two parts:
(i) semi-single-peakedness on $G_{\sim}^{A^{a \rightharpoonup b}}\cup G_{\sim}^{\langle a, b|\mathcal{T}^A\rangle}$ w.r.t.~$b$, and
(ii) $\max^{P_i}(A^{b \rightharpoonup a}) = b$.
Note that by $(x_{\underline{k}}, x_{\overline{k}})$-semi-hybridness on $\mathcal{T}_f^A$, $P_i$ is semi-single-peaked on $\mathcal{T}_f^A$ w.r.t.~$x_{\underline{k}}$, which implies that $P_i$ is semi-single-peaked on $\mathcal{T}_f^{\underline{A}} = G_{\sim}^{\underline{A}}$ w.r.t.~$x_{\underline{k}}$.
Consequently, $G_{\sim}^{A^{a \rightharpoonup b}}\cup G_{\sim}^{\langle a, b|\mathcal{T}^A\rangle} \subseteq G_{\sim}^{\underline{A}}$ and
the fact $b \in \langle x, x_{\underline{k}}|G_{\sim}^{\underline{A}}\rangle$ for all $x \in A^{a \rightharpoonup b}\cup \langle a, b|\mathcal{T}^A\rangle$ together imply $P_i$ is semi-single-peaked on $G_{\sim}^{A^{a \rightharpoonup b}}\cup G_{\sim}^{\langle a, b|\mathcal{T}^A\rangle}$ w.r.t.~$b$. This confirms part (i).
Next, since $r_1(P_i) \in A^{a \rightharpoonup b}\cup \langle a, b|\mathcal{T}^A\rangle$, part (ii) immediately follows from $(a, b)$-semi-hybridness on $\mathcal{T}^A$.

Second, given $P_i \in \mathbb{D}$ with $r_1(P_i) \in A^{b \rightharpoonup a}$,
we show that $P_i$ is semi-single-peaked on $G_{\sim}^A$ w.r.t.~$b$,
which consists of the following two parts:
(i) semi-single-peakedness on $G_{\sim}^{A^{b \rightharpoonup a}} = \mathcal{T}^{A^{b \rightharpoonup a}}$ w.r.t.~$b$, and
(ii) $\max^{P_i}(A^{a \rightharpoonup b}\cup \langle a, b|\mathcal{T}^A\rangle) = b$.
Since $r_1(P_i) \in A^{b \rightharpoonup a}$, both parts immediately follow from $(a, b)$-semi-hybridness on $\mathcal{T}^A$.
This completes the verification of the claim.
\medskip

Now, we know that all preferences of $\mathbb{D}$ are semi-single-peaked on the tree $G_{\sim}^A$ w.r.t.~the threshold $b$ which by Claim 4 is a leaf of $G_{\sim}^{\langle a, b|\mathcal{T}^A\rangle}$.
This contradicts condition (iii) of Definition \ref{def:asspsh}.

In conclusion, we must have $\underline{k} \leq p$, as required.
By a symmetric argument, we can also show $\overline{k} \geq q$.
This proves the lemma, and hence completes the verification of the ``only if part'' of \textbf{Statement (ii)} of Theorem \ref{thm:invariance}.
\end{proof}

\section{Proof of Corollary \ref{cor:anonymity}}\label{app:anonymity}

We first introduce a Ramification Theorem that will be repeatedly applied in the verification.
\begin{RT}
Let $\mathbb{D}$ be an $(a,b)$-semi-hybrid domain on a tree $\mathcal{T}^A$, and
satisfy diversity.
The following two statements are equivalent:
\begin{itemize}
\item[\rm (a)] every two-voter, tops-only and strategy-proof rule $f: \mathbb{D}^2 \rightarrow A$ behaves likes a dictatorship on $\langle a,b|\mathcal{T}^A\rangle$, and

\item[\rm (b)] every tops-only and strategy-proof rule $f: \mathbb{D}^n \rightarrow A$, $n\geq 2$, behaves likes a dictatorship on $\langle a,b|\mathcal{T}^A\rangle$.
\end{itemize}
\end{RT}

The proof of the Ramification Theorem is lengthy and is relegated to the \href{https://drive.google.com/file/d/1rxXDpH_pPUDtSHEGX-YgD6bbo_p9v0gj/view}{\textcolor[rgb]{0.00,0.00,1.00}{Supplementary Material}} of \citet{CZ2022SM}.

Now, we start to prove the Corollary.
Let $\mathbb{D}$ be a non-dictatorial, unidimensional domain.

To prove statement (i), we show the equivalence of the following three sub-statements:
\begin{itemize}
\item[\rm (1)] there exists an anonymous, tops-only and strategy-proof rule,

\item[\rm (2)] domain $\mathbb{D}$ is a semi-single-peaked domain, and

\item[\rm (3)] domain $\mathbb{D}$ admits a strategy-proof projection rule.
\end{itemize}

It is clear that the direction (2) $\Rightarrow$ (3) $\Rightarrow$ (1) follows from the proof of the sufficiency part of the Theorem of \citet{CSS2013}.
We focus on showing (1) $\Rightarrow$ (2).
Now, assume that there exists an anonymous, tops-only and strategy-proof rule.
We show that $\mathbb{D}$ is a semi-single-peaked domain.
Suppose by contradiction that $\mathbb{D}$ is not a semi-single-peaked domain.
Since Theorem \ref{thm:invariance} exclusively and exhaustively classifies all non-dictatorial, unidimensional domains into the class of semi-single-peaked domains and the class of semi-hybrid domains,
$\mathbb{D}$ must be a semi-hybrid domain, more specifically, say an $(a,b)$-semi-hybrid domain on a tree $\mathcal{T}^A$.
Then, by the proof of Statement (ii) of Theorem \ref{thm:invariance}, we know that
every two-voter, tops-only and strategy-proof rule behaves like a dictatorship on $\langle a, b|\mathcal{T}^A\rangle$.
Furthermore, by the Ramification Theorem, we know that every tops-only and strategy-proof rule with an arbitrary number of voters behaves like a dictatorship on $\langle a, b|\mathcal{T}^A\rangle$.
This contradicts the hypothesis that there exists an anonymous, tops-only and strategy-proof rule.
This completes the verification of the direct (1) $\Rightarrow$ (2), and hence proves statement (i).

Next, we show statement (ii).
Note that the ``if part'' of statement (ii) follows exactly from Lemmas \ref{lem:wlog} and \ref{lem:nssp} in Appendix \ref{app:Theorem} and statement (i).
We focus on showing the ``only if part'', and furthermore showing that
every tops-only and strategy-proof rule behaves like a dictatorship on a weak superset of the free zone.
Since there exists no anonymous, tops-only and strategy-proof rule,
statement (i) immediately implies that $\mathbb{D}$ is not a semi-single-peaked domain.
Then, by the proof of \textbf{Statement (ii)} of Theorem \ref{thm:invariance},
we know that $\mathbb{D}$ is an $(a,b)$-semi-hybrid domain on a tree $\mathcal{T}^A$ satisfying the unique seconds property, and furthermore
every two-voter, tops-only and strategy-proof rule behaves like a dictatorship on a weak superset of $\langle a, b|\mathcal{T}^A\rangle$.
Last, by the Ramification Theorem, we know that every tops-only and strategy-proof rule with an arbitrary number of voters also behaves likes a dictatorship on $\langle a,b|\mathcal{T}^A\rangle$.
This proves statement (ii).

\section{Proof of Fact \ref{fact:PNTrule}}\label{app:Fact2}

Fix arbitrary $n \geq 2$ and the PNT rule $f: \mathbb{D}^n \rightarrow A$ on  $\mathcal{T}^A$ w.r.t.~$(x, y)$.
For notational convenience, let $i = 1$ and $j = 2$ in Definition \ref{def:PNT}.

We first show that given $n > 2$, if the PNT rule $f$ is strategy-proof,
then conditions (i) and (ii) of Fact \ref{fact:PNTrule} hold, and furthermore if $f$ violates the tops-only property, then condition (iii) of Fact \ref{fact:PNTrule} holds.
We fix $P_1 \in \mathbb{D}$ with $r_1(P_1) =z \in A^{x\rightharpoonup y}$ and
show condition (i) of Fact \ref{fact:PNTrule} in two steps.
In the first step, we show that
for all distinct $a, b\in \langle z, y|\mathcal{T}^A\rangle$,
$\big[a \in \langle z, b|\mathcal{T}^A\rangle\big] \Rightarrow [a\mathrel{P_i}b]$.
Since $\langle z, y|\mathcal{T}^A\rangle$ is a combination of $\langle z, x|\mathcal{T}^A\rangle$ and the edge $(x,y)$, by transitivity of $P_1$, it suffices to show
(i) given distinct $a, b\in \langle z, x|\mathcal{T}^A\rangle$,
$[a \in \langle z, b|\mathcal{T}^A\rangle] \Rightarrow [a\mathrel{P_1}b]$ and (ii) $x\mathrel{P_1}y$.
We fix two profiles $P=(P_1, P_{-1})$ and $P'=(P_1', P_{-1})$, where $P_1' \in \mathbb{D}^b$ and $P_{\nu} \in \mathbb{D}^{a}$ for all $\nu \in \{2, \dots, n\}$.
By construction, $f(P) =\mathop{\textrm{Proj}}\big(x, \langle z, a|\mathcal{T}^A\rangle\big) = a$ and
$f(P') = \mathop{\textrm{Proj}}(x, \langle b, a|\mathcal{T}^A\rangle) = b$.
Then, strategy-proofness implies $a\mathrel{P_1}b$, as required.
Next, we construct two profiles $P=(P_1, P_2, P_{-\{1,2\}})$ and $P'=(P_1', P_2, P_{-\{1,2\}})$,
where $P_1' \in \mathbb{D}^y$, $P_2 \in \mathbb{D}^x$ and $r_1(P_{\nu}) \in A^{x\rightharpoonup y}$ for all $\nu \in \{3, \dots, n\}$.
By construction, $f(P) = \mathop{\textrm{Proj}}\big(x, \mathcal{T}^{\Gamma(P)}\big) = x$ and
$f(P') = y$.
Then, strategy-proofness implies $x\mathrel{P_1}y$, as required.
In the second step, given $a \notin \langle z, y|\mathcal{T}^A\rangle$ and $a' = \mathop{\textrm{Proj}}(a, \langle z, y|\mathcal{T}^A\rangle)$,
we show $a'\mathrel{P_1}a$.
Clearly, $a' \in \langle z, y|\mathcal{T}^A\rangle$.
If $a' = z$, $a'\mathrel{P_1}a$ holds immediately.
If $a'= y$, we know $a \in A^{y \rightharpoonup x}\backslash \{y\}$.
Then, we construct two profiles $P=(P_1, P_2, P_{-\{1,2\}})$ and $P'=(P_1', P_2, P_{-\{1,2\}})$,
where $P_1' \in \mathbb{D}^a$, $P_2 \in \mathbb{D}^y$ and $P_{\nu}$ is arbitrary for all $\nu \in \{3, \dots, n\}$.
By construction, $f(P) = \max^{P_2}(\{x, y\}) = y = a'$ and $f(P') = a$.
Then, strategy-proofness implies $a'\mathrel{P_1}a$, as required.
Last, let $a' \in \langle z, y|\mathcal{T}^A\rangle\backslash \{z, y\}$.
Thus, $a \in A^{x \rightharpoonup y}$ and $a' = \mathop{\textrm{Proj}}(a, \langle z, y|\mathcal{T}^A\rangle)= \mathop{\textrm{Proj}}(a, \langle z, x|\mathcal{T}^A\rangle)
= \mathop{\textrm{Proj}}(x, \langle z, a|\mathcal{T}^A\rangle)$.
Accordingly, we construct two profiles $P=(P_1, P_{-1})$ and $P'=(P_1', P_{-1})$, where $P_1' \in \mathbb{D}^{a}$ and $P_{\nu} \in \mathbb{D}^{a}$ for all $\nu \in \{2, \dots, n\}$.
We then have $f(P) = \mathop{\textrm{Proj}}\big(x, \langle z, a|\mathcal{T}^A\rangle\big)
= a'$ and
$f(P') = a$ by construction, and $a'\mathrel{P_1}a$ by strategy-proofness, as required.
This proves condition (i) of Fact \ref{fact:PNTrule}.

To verify condition (ii) of Fact \ref{fact:PNTrule},
we fix arbitrary $v \in N\backslash \{1,2\}$ and  $P_v \in \mathbb{D}$ with $r_1(P_v) = z\in A^{y\rightharpoonup x}$, and show $\max^{P_v}(A^{x \rightharpoonup y}) = x$.
Given arbitrary $a \in A^{x\rightharpoonup y}\backslash \{x\}$,
we construct two profiles $P = (P_1, P_2, P_v, P_{-\{1,2,v\}})$ and $P' = (P_1, P_2, P_v', P_{-\{1,2,v\}})$,
where $P_v' \in \mathbb{D}^a$ and $P_{\ell} \in \mathbb{D}^a$ for all $\ell \in N\backslash \{v\}$.
By construction, we have $f(P) = \mathop{\textrm{Proj}}\big(x, \langle z, a |\mathcal{T}^A\rangle\big) = x$ and $f(P') = a$.
Then, strategy-proofness implies $x\mathrel{P_v}a$. Hence, we have $\max^{P_v}(A^{x \rightharpoonup y}) = x$, as required.

To verify condition (iii) of Fact \ref{fact:PNTrule},, let $f$ violate the tops-only property.
By construction, there must exist
$P =(P_1, P_2, P_{-\{1,2\}})$ and $P'=(P_1, P_2', P_{-\{1,2\}})$,
where $r_1(P_1) \in A^{x\rightharpoonup y}$ and $r_1(P_2) = r_1(P_2') \in A^{y\rightharpoonup x}$
such that $f(P) = \max^{P_2}(\{x,y\}) = x \neq y= \max^{P_2'}(\{x, y\}) = f(P')$.
Therefore, we have $x\mathrel{P_2}y$ and $y\mathrel{P_2'}x$, as required.
\medskip

Conversely, we show that if domain $\mathbb{D}$ satisfy conditions (i), (ii) and (iii) of Fact \ref{fact:PNTrule},
then the PNT rule $f$ satisfies strategy-proofness and violates the tops-only property.

Clearly, condition (iii) of Fact \ref{fact:PNTrule} implies that $f$ violates the tops-only property.
In the rest of proof, we show strategy-proofness of $f$.
We first consider an arbitrary voter $i \in N\backslash \{1, 2\}$ in the case $n > 2$.
Given profiles $P =(P_1, P_2, P_i, P_{-\{1,2, i\}})$ and $P' =(P_1, P_2, P_i', P_{-\{1,2, i\}})$,
let $f(P) \neq f(P')$.
Then, the construction implies $r_1(P_1) \in A^{x \rightharpoonup y}$, $r_1(P_2) \in A^{x \rightharpoonup y}$,
$f(P)=\mathop{\textrm{Proj}}(x, \mathcal{T}^{\Gamma(P)})$ and $f(P')=\mathop{\textrm{Proj}}(x, \mathcal{T}^{\Gamma(P')})$.
Furthermore, we can infer $r_1(P_j) \in A^{x \rightharpoonup y}$ for all $j \in N\backslash \{1, 2, i\}$.
Otherwise, both $\Gamma(P)$ and $\Gamma(P')$ contain $x$, and $f(P) = x = f(P')$.
Similarly, we can infer either $r_1(P_i) \notin A^{y \rightharpoonup x}$ or $r_1(P_i') \notin A^{y \rightharpoonup x}$.
Thus, there are three cases:
(1) $r_1(P_i) \in A^{x \rightharpoonup y}$ and $r_1(P_i') \in A^{x \rightharpoonup y}$,
(2) $r_1(P_i) \in A^{x \rightharpoonup y}$ and $r_1(P_i') \in A^{y \rightharpoonup x}$, and
(3) $r_1(P_i) \in A^{y \rightharpoonup x}$ and $r_1(P_i') \in A^{x \rightharpoonup y}$.
In both cases (1) and (2), since $r_1(P_i) \in A^{x \rightharpoonup y}$,
condition (i) of Fact \ref{fact:PNTrule} implies that $P_i$ is also semi-single-peaked on $\mathcal{T}^A$ w.r.t.~$x$.
Consequently, given $f(P)=\mathop{\textrm{Proj}}(x, \mathcal{T}^{\Gamma(P)}) \neq \mathop{\textrm{Proj}}(x, \mathcal{T}^{\Gamma(P')}) =f(P')$,
the proof of the sufficiency part of the Theorem of \citet{CSS2013} implies $f(P)\mathrel{P_i}f(P')$.
In case (3),
$f(P)= \mathop{\textrm{Proj}}(x, \mathcal{T}^{\Gamma(P)}) = x$ and $f(P') = \mathop{\textrm{Proj}}(x, \mathcal{T}^{\Gamma(P')}) \in A^{x \rightharpoonup y}$.
Then, condition (ii) of Fact \ref{fact:PNTrule} implies $f(P)\mathrel{P_i}f(P')$.
Overall, voter $i$ has no incentive to manipulate.
Henceforth, we focus on the possible manipulations of voters 1 and 2.

First, given two profiles $P = (P_1, P_2, P_{-\{1,2\}})$ and $P' = (P_1', P_2, P_{-\{1,2\}})$
with $f(P) \neq f(P')$,
there are three possible manipulations of voter 1:\\
(1) $f(P) = \mathop{\textrm{Proj}}\big(x, \mathcal{T}^{\Gamma(P)}\big)$ and
$f(P') = \mathop{\textrm{Proj}}\big(x, \mathcal{T}^{\Gamma(P')}\big)$,\\
(2) $f(P) = \mathop{\textrm{Proj}}\big(x, \mathcal{T}^{\Gamma(P)}\big)$ and $f(P') = r_1(P_1') \in A^{y \rightharpoonup x}$, and \\
(3) $f(P) =\max^{P_2}(\{x,y\})$ and $f(P') = r_1(P_1') \in A^{y \rightharpoonup x}$.\\
In each case, we show $f(P) \mathrel{P_1} f(P')$.
In case (1), we know $r_1(P_1) \in A^{x \rightharpoonup y}$.
Then, condition (i) of Fact \ref{fact:PNTrule} implies that $P_1$ is also semi-single-peaked on $\mathcal{T}^A$ w.r.t.~$x$.
Consequently, given $f(P)=\mathop{\textrm{Proj}}(x, \mathcal{T}^{\Gamma(P)}) \neq \mathop{\textrm{Proj}}(x, \mathcal{T}^{\Gamma(P')}) = f(P')$,
the proof of the sufficiency part of the Theorem of \citet{CSS2013} implies $f(P)\mathrel{P_1}f(P')$, as required.
In case (2), we know $r_1(P_1) \in A^{x \rightharpoonup y}$,
which
implies $\min^{P_1}\big(\langle r_1(P_1), x|\mathcal{T}^A\rangle\big) =x$, $x\mathrel{P_1}y$ and
$y=\max^{P_1}(A^{y \rightharpoonup x})$ by condition (i) of Fact \ref{fact:PNTrule}.
Then,
$f(P) = \mathop{\textrm{Proj}}\big(x, \mathcal{T}^{\Gamma(P)}\big) \in \langle r_1(P_1), x|\mathcal{T}^A\rangle$ and $f(P') = r_1(P_1') \in A^{y \rightharpoonup x}$ imply
$f(P)\mathrel{P_1}f(P')$, as required.
In case (3), we know $r_1(P_1) \in A^{x \rightharpoonup y}$ which implies $x\mathrel{P_1}y$ and $y=\max^{P_1}(A^{y \rightharpoonup x})$ by condition (i) of Fact \ref{fact:PNTrule}.
Then, $f(P) =\max^{P_2}(\{x,y\}) \in \{x, y\}$, $f(P') \in A^{y \rightharpoonup x}$ and $f(P) \neq f(P')$ imply $f(P)\mathrel{P_1}f(P')$, as required.

Last, given two profiles $P = (P_1, P_2, P_{-\{1,2\}})$ and $P' =(P_1, P_2', P_{-\{1,2\}})$
with $f(P) \neq f(P')$,
there are three possible manipulations of voter 2:\\
(1) $f(P) = \mathop{\textrm{Proj}}\big(x, \mathcal{T}^{\Gamma(P)}\big)$ and
$f(P') =\mathop{\textrm{Proj}}\big(x, \mathcal{T}^{\Gamma(P')}\big)$, \\
(2) $f(P) = \mathop{\textrm{Proj}}\big(x, \mathcal{T}^{\Gamma(P)}\big)$ and $f(P') = \max^{P_2'}(\{x, y\})$, and\\
(3) $f(P) =\max^{P_2}(\{x,y\})$ and $f(P') = \mathop{\textrm{Proj}}\big(x, \mathcal{T}^{\Gamma(P')}\big)$.\\
In each case, we show $f(P) \mathrel{P_2} f(P')$ holds.
Clearly, the verification of case (1) is similar to that of case (1) for voter 1.
In case (2), we know $r_1(P_2) \in A^{x \rightharpoonup y}$, which implies $\min^{P_2}\big(\langle r_1(P_2), x|\mathcal{T}^A\rangle\big) = x$ and $x\mathrel{P_2}y$ by condition (i) of Fact \ref{fact:PNTrule}.
Consequently, $f(P) = \mathop{\textrm{Proj}}\big(x, \mathcal{T}^{\Gamma(P)}\big) \in \langle r_1(P_2), x|\mathcal{T}^A\rangle$,
$f(P') = \max^{P_2'}(\{x, y\}) \in \{x, y\}$ and $f(P) \neq f(P')$ imply $f(P)\mathrel{P_2}f(P')$, as required.
In case (3), we know $r_1(P_2) \in A^{y \rightharpoonup x}$
which implies $\max^{P_2}(A^{x \rightharpoonup y}) = x$ by condition (ii) of Fact \ref{fact:PNTrule}, and
$r_1(P_2') \in A^{x \rightharpoonup y}$ which implies $f(P')
= \mathop{\textrm{Proj}}\big(x, \mathcal{T}^{\Gamma(P')}\big) \in \langle r_1(P_2'), x|\mathcal{T}^A\rangle \subseteq A^{x \rightharpoonup y}$.
Then, $f(P) =\max^{P_2}(\{x,y\})$, $f(P') \in A^{x \rightharpoonup y}$ and $f(P) \neq f(P')$ imply $f(P)\mathrel{P_2}f(P')$, as required.
In conclusion, the PNT rule $f$ is strategy-proof.

\section{Proof of Proposition \ref{prop:nontopsonly}}\label{app:Proposition1}
We first show statement (i) of Proposition \ref{prop:nontopsonly}.
Let $\mathbb{D}$ be a semi-single-peaked domain on a tree $\mathcal{T}^A$.
Clearly, the existence of a critical spot ensures $\mathbb{D} \nsubseteq \mathbb{D}_{\textrm{SP}}(\mathcal{T}^A)$.
Henceforth, let $\mathbb{D} \nsubseteq \mathbb{D}_{\textrm{SP}}(\mathcal{T}^A)$, and we show the existence of a critical spot.

Since $\mathbb{D}$ is a semi-single-peaked domain on $\mathcal{T}^A$,
we identify the set $\mathcal{Z} \subseteq A$ such that $[z \in \mathcal{Z}] \Rightarrow [\mathbb{D} \subseteq \mathbb{D}_{\textrm{SSP}}(\mathcal{T}^A, z)]$.
Clearly, either $\mathcal{Z}$ is a singleton set, or $\mathcal{Z}$ is not a singleton set and $\mathcal{T}^{\mathcal{Z}}$ is a subtree nested in $\mathcal{T}^A$.\footnote{Given a tree $\mathcal{T}^A$ and two distinct alternatives $a, b \in A$, if $\mathbb{D}\subseteq \mathbb{D}_{\textrm{SSP}}(\mathcal{T}^A, a)$ and $\mathbb{D}\subseteq \mathbb{D}_{\textrm{SSP}}(\mathcal{T}^A, b)$, then we have $\mathbb{D} \subseteq \mathbb{D}_{\textrm{SSP}}(\mathcal{T}^A, x)$ for all $x \in \langle a, b|\mathcal{T}^A\rangle$.}
Since $\mathbb{D} \nsubseteq \mathbb{D}_{\textrm{SP}}(\mathcal{T}^A)$,
it must be the case that some leaf of $\mathcal{T}^A$ and its unique neighbor are not contained in $\mathcal{Z}$.
Then, we can identify a threshold $x' \in \mathcal{Z}$ with
$[\mathcal{Z}\; \textrm{is not a singleton set}] \Rightarrow [x' \in \textrm{Leaf}(\mathcal{T}^{\mathcal{Z}})]$
and an edge $(x,y) \in \mathcal{E}^A$ with $x,y \notin \mathcal{Z}$ and $y \in \langle x, x'|\mathcal{T}^A\rangle$, such that
the following two conditions are satisfied:
\begin{itemize}
\item[\rm (1)] \emph{every} preference $P_i \in \mathbb{D}$ with $r_1(P_i) \in A^{y \rightharpoonup x}$ is single-peaked on
$\mathcal{T}^{A^{x \rightharpoonup y}}$, i.e.,
for all distinct $a, b \in A^{x \rightharpoonup y}$, $\big[a\in \langle r_1(P_i), b|\mathcal{T}^A\rangle\big] \Rightarrow [a\mathrel{P_i}b]$, and

\item[\rm (2)] \emph{some} preference $P_i^{\ast} \in \mathbb{D}$ with $r_1(P_i^{\ast} ) \in A^{y \rightharpoonup x}$ is not single-peaked on the subtree
$\mathcal{T}^{A^{x \rightharpoonup y} \cup \{y\}}$.\footnote{The subtree $\mathcal{T}^{A^{x \rightharpoonup y} \cup \{y\}}$ is a combination of the subtree $\mathcal{T}^{A^{x \rightharpoonup y}}$ and the edge $(x, y)$. We adopt the instance of the line $\mathcal{L}^A$ to exemplify how conditions (1) and (2) are specified. Fixing a semi-single-peaked domain $\mathbb{D}$ on the line $\mathcal{L}^A$, let $\mathbb{D}\nsubseteq \mathbb{D}_{\textrm{SP}}(\mathcal{L}^A)$.
Then, we have $\mathcal{Z} = \langle a_p, a_q|\mathcal{L}^A\rangle$ for some $2< p \leq q \leq m$ or $1 \leq p \leq q < m-1$.
We assume w.l.o.g.~that $p>2$.
First, according to $a_1$, it is natural that every preference with the peak located in $\langle a_2, a_m |\mathcal{L}^A\rangle$ is single-peaked on $\langle a_1, a_1|\mathcal{L}^A\rangle$.
Second, since $\mathbb{D} \subseteq \mathbb{D}_{\textrm{SSP}}(\mathcal{L}^A, a_p)$ and
$\mathbb{D} \nsubseteq \mathbb{D}_{\textrm{SSP}}(\mathcal{L}^A, a_{p-1})$,
some preference with the peak located in $\langle a_{p-1}, a_m|\mathcal{L}^A\rangle$ must not be single-peaked on $\langle a_1, a_{p-1}|\mathcal{L}^A\rangle$.
Searching from $a_1$ to $a_{p-1}$,
we can identify $1 \leq s< p-1$ such that
(i) every preference with the peak located in $\langle a_{s+1}, a_m |\mathcal{L}^A\rangle$ is single-peaked on $\langle a_1, a_s|\mathcal{L}^A\rangle$, and
(ii) some preference with the peak located in $\langle a_{s+1}, a_m |\mathcal{L}^A\rangle$ is not single-peaked on $\langle a_1, a_{s+1}|\mathcal{L}^A\rangle$.
These two conditions are analogous to conditions (1) and (2) above respectively.}
\end{itemize}
Since condition (1) implies $\max^{P_i}(A^{x \rightharpoonup y}) = x$ for all $P_i \in \mathbb{D}$ with $r_1(P_i) \in A^{y \rightharpoonup x}$,
condition (ii) of Fact \ref{fact:PNTrule} is satisfied.
We next show condition (i) of Fact \ref{fact:PNTrule}: given $P_i \in \mathbb{D}$ with $r_1(P_i) \in A^{x \rightharpoonup y}$, $P_i$ is semi-single-peaked on $\mathcal{T}^{A}$ w.r.t.~$y$.
Since $y \in \langle r_1(P_i) , x'|\mathcal{T}^A\rangle$,
semi-single-peakedness on $\mathcal{T}^A$ w.r.t.~$x'$ implies that $P_i$ is also semi-single-peaked on $\mathcal{T}^A$ w.r.t.~$y$, as required.
Last, we show condition (iii) of Fact \ref{fact:PNTrule}:
there exist $P_i,P_i' \in \mathbb{D}$ such that $r_1(P_i) = r_1(P_i') \in A^{y \rightharpoonup x}$, $y\mathrel{P_i}x$ and $x\mathrel{P_i'}y$.
First, since $G_{\sim}^A = \mathcal{T}^A$, it is true that $G_{\sim}^{A^{y\rightharpoonup x}} = \mathcal{T}^{A^{y \rightharpoonup x}}$ is a subtree.
Second, recall preference $P_i^{\ast}$ in condition (2) above.
Since $P_i^{\ast}$ is single-peaked on $\mathcal{T}^{A^{x \rightharpoonup y} }$ by condition (1) and violates single-peakedness on $\mathcal{T}^{A^{x \rightharpoonup y} \cup \{y\}}$ by condition (2), it must be the case that $\max^{P_i^{\ast}}(A^{x \rightharpoonup y}) = x$ and $x\mathrel{P_i^{\ast}}y$.
Given subdomains
$\mathbb{D}_1 = \{P_i \in \mathbb{D}: r_1(P_i) \in A^{y\rightharpoonup x}\; \textrm{and}\; y\mathrel{P_i}x\}$ and
$\mathbb{D}_2 = \{P_i \in \mathbb{D}: r_1(P_i) \in A^{y\rightharpoonup x}\; \textrm{and}\; x\mathrel{P_i}y\}$,
we know $\mathbb{D}_1 \neq \emptyset$ by minimal richness and $P_i^{\ast} \in \mathbb{D}_2$.
Now, suppose that condition (iii) of Fact \ref{fact:PNTrule} is not true.
Thus, all $P_i,P_i' \in \mathbb{D}$ with $r_1(P_i) = r_1(P_i') \in A^{y \rightharpoonup x}$
agree on the relative ranking of $x$ and $y$.
Therefore, for all $z \in A^{y \rightharpoonup x}$, either $\mathbb{D}^z \subseteq \mathbb{D}_1$
or $\mathbb{D}^z \subseteq \mathbb{D}_2$ holds.
Consequently,
given arbitrary distinct $z, z' \in A^{y \rightharpoonup x}$ with $\mathbb{D}^z \subseteq \mathbb{D}_1$ and $\mathbb{D}^{z'} \subseteq \mathbb{D}_2$,
we have $z\mathrel{P_i}z'$ and $y\mathrel{P_i}x$ for all $P_i \in \mathbb{D}^z$, and $z'\mathrel{P_i'}z$ and $x\mathrel{P_i'}y$ for all $P_i' \in \mathbb{D}^{z'}$,
which imply that $z$ and $z'$ are never adjacent.
This contradicts the fact that $G_{\sim}^{A^{y \rightharpoonup x}}$ is a subtree.
Therefore, we identify a critical spot $(x,y)$. This proves statement (i).\medskip

Next, we show statement (ii) of Proposition \ref{prop:nontopsonly}.
Let $\mathbb{D}$ be an $(a, b)$-semi-hybrid domain on a tree $\mathcal{T}^A$.
Clearly, the existence of a critical spot $(x, y)$ in $\mathcal{T}^A$ such that either $x, y \in A^{a \rightharpoonup b}\backslash \{a\}$ or $x, y \in A^{b \rightharpoonup a}\backslash \{b\}$ implies $\mathbb{D} \nsubseteq \mathbb{D}_{\textrm{H}}(\mathcal{T}^A, a, b)$.
Henceforth, let $\mathbb{D} \nsubseteq \mathbb{D}_{\textrm{H}}(\mathcal{T}^A, a, b)$, and we show that there exists a critical spot $(x, y)$ in $\mathcal{T}^A$ such that either $x, y \in A^{a \rightharpoonup b}\backslash \{a\}$ or $x, y \in A^{b \rightharpoonup a}\backslash \{b\}$.
Since $\mathbb{D} \nsubseteq \mathbb{D}_{\textrm{H}}(\mathcal{T}^A, a, b)$,
there must exists $\hat{P}_i \in \mathbb{D}$ that is either not single-peaked on $\mathcal{T}^{A^{a \rightharpoonup b}}$ or not single-peaked on $\mathcal{T}^{A^{b \rightharpoonup a}}$.
We assume w.l.o.g. that $\hat{P}_i$ is not single-peaked on $\mathcal{T}^{A^{a \rightharpoonup b}}$.
This implies $|A^{a \rightharpoonup b}| \geq 3$.
Given arbitrary $P_i \in \mathbb{D}$, let $P_i|_{A^{a \rightharpoonup b}}$ be the induced preference over $A^{a \rightharpoonup b}$ by removing all alternatives not in $A^{a \rightharpoonup b}$.
Accordingly, let $\mathbb{D}|_{A^{a \rightharpoonup b}} = \{P_i|_{A^{a \rightharpoonup b}}: P_i \in \mathbb{D}\}$.
We claim that $\mathbb{D}|_{A^{a \rightharpoonup b}}$ is semi-single-peaked on $\mathcal{T}^{A^{a \rightharpoonup b}}$ w.r.t.~$a$.
Given arbitrary $P_i \in \mathbb{D}$, either $r_1(P_i) \in A^{a \rightharpoonup b}$ or $r_1(P_i) \in A\backslash A^{a \rightharpoonup b}$ holds.
If $r_1(P_i) \in A^{a \rightharpoonup b}$, we know that $P_i$ is semi-single-peaked on $\mathcal{T}^A$ w.r.t.~$a$ by $(a,b)$-semi-hybridness on $\mathcal{T}^A$.
Then, it is evident that $P_i|_{A^{a \rightharpoonup b}}$ is semi-single-peaked on $\mathcal{T}^{A^{a \rightharpoonup b}}$ w.r.t.~$a$, as required.
If $r_1(P_i) \in A\backslash A^{a \rightharpoonup b}$, $(a,b)$-semi-hybridness on $\mathcal{T}^A$ implies $\max^{P_i}(A^{a \rightharpoonup b}) = a$.
Then, we have $r_1(P_i|_{A^{a \rightharpoonup b}}) = a$ which implies that $P_i|_{A^{a \rightharpoonup b}}$ is semi-single-peaked on $\mathcal{T}^{A^{a \rightharpoonup b}}$ w.r.t.~$a$, as required.
Moreover, since $\hat{P}_i$ is not single-peaked on $\mathcal{T}^{A^{a \rightharpoonup b}}$,
it is clear that $\hat{P}_i|_{A^{a \rightharpoonup b}}$ is not single-peaked on $\mathcal{T}^{A^{a \rightharpoonup b}}$.
Consequently, $\mathbb{D}|_{A^{a \rightharpoonup b}}$ is not single-peaked on $\mathcal{T}^{A^{a \rightharpoonup b}}$.
Then, according to $\mathbb{D}|_{A^{a \rightharpoonup b}}$,
we can adopt the verification of statement (i) to identify a critical spot $(x, y) $ in $\mathcal{T}^{A^{a \rightharpoonup b}}$
such that $x, y \in A^{a \rightharpoonup b}\backslash \{a\}$.
Last, back to $\mathbb{D}$, by $(a,b)$-semi-hybridness on $\mathcal{T}^A$,
it is easy to show that $(x,y)$ is also a critical spot in $\mathcal{T}^A$.
This completes the verification of statement (ii).

\section{Proof of Corollary \ref{cor:topsonly}}\label{app:topsonly}

We first show statement (i).
The ``if part'' holds evidently. We focus on showing the ``only if part''.
Fix a non-dictatorial, tops-only, unidimensional domain $\mathbb{D}$,
and let $\mathbb{D}$ admit an anonymous and strategy-proof rule.
Clearly, the admissible rule also satisfies the tops-only property.
Then, by statement (i) of Corollary \ref{cor:anonymity}, $\mathbb{D}$ must be a semi-single-peaked domain on a tree $\mathcal{T}^A$ w.r.t.~some threshold $\bar{x} \in A$.
Furthermore, since $\mathbb{D}$ is a tops-only domain, Fact \ref{fact:PNTrule} and statement (i) of Proposition \ref{prop:nontopsonly} together imply that $\mathbb{D}$ must be a single-peaked domain on $\mathcal{T}^A$.
Last, since $\mathbb{D}$ includes the completely reversed preferences $\underline{P}_i$ and $\overline{P}_i$,
according to the labeling of alternatives in them,
it must be the case that $\mathcal{T}^A= \mathcal{L}^A$.
Therefore, $\mathbb{D}$ is a single-peaked domain on $\mathcal{L}^A$.
This proves statement (i) of Corollary \ref{cor:topsonly}.

We next show statement (ii).
Fix a non-dictatorial, tops-only, unidimensional domain $\mathbb{D}$.
To show the ``if part'' of statement (ii), let $\mathbb{D}$ further be a non-degenerate hybrid domain on $\mathcal{L}^A$.
Thus, $\mathbb{D} \subseteq \mathbb{D}_{\textrm{H}}(\mathcal{L}^A, a_p, a_q)$ for some $1 < q-p < m-1$.
According to statement (i), to show non-existence of an anonymous and strategy-proof rule, it suffices to show $\mathbb{D}\nsubseteq \mathbb{D}_{\textrm{SP}}(\mathcal{L}^A)$.
Suppose by contradiction that $\mathbb{D} \subseteq \mathbb{D}_{\textrm{SP}}(\mathcal{L}^A)$.
Consequently, we have $\mathbb{D} \subseteq \mathbb{D}_{\textrm{SP}}(\mathcal{L}^A) \subseteq \mathbb{D}_{\textrm{H}}(\mathcal{L}^A, a_p, a_{q-1})$,
which contradiction condition (ii) of Definition \ref{def:asph}
This proves the ``if part'' of statement (ii).

To prove the ``only if part'' of statement (ii), let $\mathbb{D}$ be a non-dictatorial, tops-only, unidimensional domain, and admit no anonymous and strategy-proof rule.
Thus, $\mathbb{D}$ admits no anonymous, tops-only and strategy-proof rule.
Then, by statement (ii) of Corollary \ref{cor:anonymity}, we know that $\mathbb{D}$ is an $(a, b)$-semi-hybrid domain on a tree $\mathcal{T}^A$.
By Lemma \ref{lem:wlog}, we can assume w.l.o.g.~that $a_1 \in A^{a \rightharpoonup b}$ and $a_m \in A^{b \rightharpoonup a}$.
We show that $\mathbb{D}$ is a non-degenerate hybrid domain on $\mathcal{L}^A$ in the next lemma.

\begin{lemma}
Domain $\mathbb{D}$ is an $(a_p, a_q)$-hybrid domain on $\mathcal{L}^A$ for some $1< q-p< m-1$.
\end{lemma}

\begin{proof}
Since $\mathbb{D} \subseteq \mathbb{D}_{\textrm{SH}}(\mathcal{T}^A, a, b)$ is a tops-only domain,
Fact \ref{fact:PNTrule} and statement (ii) of Proposition \ref{prop:nontopsonly} together imply $\mathbb{D} \subseteq \mathbb{D}_{\textrm{H}}(\mathcal{T}^A, a, b)$.
Furthermore, the presence of the completely reversed preferences $\underline{P}_i$ and $\overline{P}_i$ implies that $\mathcal{T}^A$ must be a line.
We next show that the two leaves of $\mathcal{T}^A$ are $a_1$ and $a_m$.
Suppose by contradiction that $a_1$ is not a leaf of the line $\mathcal{T}^A$.
Thus, $a_1$ has two neighbors in $\mathcal{T}^A$.
Consequently, we must have some $c \in \mathcal{N}^A(a_1)$ such that $c \neq a_m$ and $a_1 \in \langle c, a_m|\mathcal{T}^A\rangle$.
Since $a_1 \in A^{a \rightharpoonup b}$ and $a_m \in A^{b \rightharpoonup a }$, $a_1 \in \langle c, a_m|\mathcal{T}^A\rangle$ implies $c \in A^{a  \rightharpoonup b}$.
Then, by $(a,b)$-hybridness on $\mathcal{T}^A$, $r_1(\overline{P}_i) = a_m \in A^{b \rightharpoonup a }$ implies $c\mathrel{\overline{P}_i}a_1$,
which contradicts the fact that $a_1$ is bottom-ranked in $\overline{P}_i$.
Therefore, $a_1$ is not a leaf of the line $\mathcal{T}^A$.
Symmetrically, $a_m$ is the other leaf of the line $\mathcal{T}^A$.
Now, let $\langle a_1, a|\mathcal{T}^A\rangle = (x_1, \dots, x_{p})$ and
$\langle b, a_m|\mathcal{T}^A\rangle = (x_q, \dots, x_m)$,
where $p = |\langle a_1, a|\mathcal{T}^A\rangle|$,
$q-p+1 = |\langle a, b|\mathcal{T}^A\rangle|$, and
$m-q+1 = |\langle b, a_m|\mathcal{T}^A\rangle|$.
Since $\mathbb{D}$ is an $(a, b)$-semi-hybrid domain on $\mathcal{T}^A$,
we have $|\langle a_p, a_q|\mathcal{L}^A\rangle| = |\langle a, b|\mathcal{T}^A\rangle| \geq 3$ and hence $q-p>1$.
This confirms condition (iii) of Definition \ref{def:asph}.
Furthermore, according to the labeling of alternatives in $\underline{P}_i$ and $\overline{P}_i$,
we infer that $x_k = a_k$ for all $k \in \{1,\dots, p\}\cup \{q, \dots, m\}$.
Thus, $\mathcal{T}^A$ is a combination of the line $(a_1, \dots, a_p)$,
the interval $\langle a, b|\mathcal{T}^A\rangle$ and the line $(a_q, \dots, a_m)$, and
hence $\mathbb{D}_{\textrm{H}}(\mathcal{T}^A, a, b) = \mathbb{D}_{\textrm{H}}(\mathcal{L}^A, a_p, a_q)$.
Therefore, we have
$\mathbb{D} \subseteq \mathbb{D}_{\textrm{H}}(\mathcal{T}^A, a, b) = \mathbb{D}_{\textrm{H}}(\mathcal{L}^A, a_p, a_q)$, which, in conjunction with the fact that $G_{\sim}^A$ is a connected graph, confirms condition (i) of Definition \ref{def:asph}.
Next, since $\mathbb{D}$ is an $(a, b)$-semi-hybrid domain on $\mathcal{T}^A$,
by condition (ii) of Definition \ref{def:asspsh},
there exist no tree $\widehat{\mathcal{T}}^A$ and dual-thresholds $\hat{a}$ and $\hat{b}$ such that
$\mathbb{D} \subseteq \mathbb{D}_{\textrm{SH}}(\widehat{\mathcal{T}}^A, \hat{a}, \hat{b})$ and $\langle \hat{a}, \hat{b}|\widehat{\mathcal{T}}^A\rangle \subset \langle a, b|\mathcal{T}^A\rangle = \langle a_p, a_q|\mathcal{L}^A\rangle$.
This implies that there exist no tree $\widehat{\mathcal{T}}^A$ and dual-thresholds $\hat{a}$ and $\hat{b}$ such that
$\mathbb{D} \subseteq \mathbb{D}_{\textrm{H}}(\widehat{\mathcal{T}}^A, \hat{a}, \hat{b})$ and $\langle \hat{a}, \hat{b}|\widehat{\mathcal{T}}^A\rangle \subset \langle a_p, a_q|\mathcal{L}^A\rangle$.
This confirms condition (ii) of Definition \ref{def:asph}.
In conclusion, $\mathbb{D}$ is an $(a_p, a_q)$-hybrid domain on $\mathcal{L}^A$.

We last show that $\mathbb{D}$ is a non-degenerate hybrid domain, i.e., $q-p < m-1$.
Suppose by contradiction $q - p = m-1$.
Thus, $\mathbb{D}$ is an $(a_1, a_m)$-hybrid domain on $\mathcal{L}^A$, which is
a degenerate hybrid domain.
Consequently, by the proof of \textbf{Statement (ii)} of Theorem \ref{thm:invariance},
we know that every two-voter, tops-only and strategy-proof rule behaves like a dictatorship on $\langle a_1, a_m|\mathcal{L}^A\rangle = A$.
Hence, every two-voter, tops-only and strategy-proof rule is a dictatorship.
Consequently, by the tops-only-domain hypothesis, every strategy-proof rule is a dictatorship, which further by the ramification theorem of \citet{ACS2003} implies that $\mathbb{D}$ is a dictatorial domain - a contradiction.
This proves the lemma and completes the verification of the ``only if part'' of statement (ii).
\end{proof}

\section{Some Clarifications}\label{app:clarification}

\begin{clarification}\label{cla:ssp}
Fixing a tree $\mathcal{T}^A$, a threshold $\bar{x} \in A$ and the semi-single-peaked domain \\
$\mathbb{D}_{\emph{SSP}}(\mathcal{T}^A, \bar{x})$,
the following two statements hold:
\begin{itemize}
\item[\rm (i)] the adjacency graph of $\mathbb{D}_{\emph{SSP}}(\mathcal{T}^A, \bar{x})$ coincides with $\mathcal{T}^A$, and

\item[\rm (ii)]
if a domain $\mathbb{D} \subseteq \mathbb{D}_{\emph{SSP}}(\mathcal{T}^A, \bar{x})$ satisfies diversity, then $|\mathcal{N}^A(\bar{x})| \leq 2$.
Conversely, if $|\mathcal{N}^A(\bar{x})| \leq 2$, then $\mathbb{D}_{\emph{SSP}}(\mathcal{T}^A, \bar{x})$ satisfies diversity.
\end{itemize}
\end{clarification}

It is clear that any two alternatives form an edge in $\mathcal{T}^A$ are adjacent in the semi-single-peaked domain $\mathbb{D}_{\textrm{SSP}}(\mathcal{T}^A, \bar{x})$, and hence form an edge in the corresponding adjacency graph.
Therefore, $\mathcal{T}^A$ is a subgraph of the adjacency graph of the semi-single-peaked domain $\mathbb{D}_{\textrm{SSP}}(\mathcal{T}^A, \bar{x})$.
To prove statement (i), it suffices to show that any two alternatives that never form an edge in $\mathcal{T}^A$, are not adjacent.
Suppose not, i.e., there exist $a,b \in A$ such that $(a,b) \notin \mathcal{E}^A$ and $a \sim b$.
Thus, $v \geq 3$ and we have $\hat{P}_i, \tilde{P}_i \in \mathbb{D}_{\textrm{SSP}}(\mathcal{T}^A, \bar{x})$ such that $r_1(\hat{P}_i) = r_2(\tilde{P}_i) = a$ and
$r_1(\tilde{P}_i) = r_2(\hat{P}_i) = b$.
Let $\langle a, b|\mathcal{T}^A\rangle = (x_1, \dots, x_v)$ denote path in $\mathcal{T}^A$ that connects $a$ and $b$.
Moreover, let $\textrm{Proj}(\bar{x}, \langle a, b|\mathcal{T}^A\rangle) = x_{\bar{k}}$ for some $\bar{k} \in \{1, \dots, v\}$.
Note that either $1 < \bar{k} < v$ or $\bar{k} \in \{1,v\}$ holds.
If $1 < \bar{k} < v$, then $\textrm{Proj}(b, \langle a, \bar{x}|\mathcal{T}^A\rangle) = \textrm{Proj}(\bar{x}, \langle a, b|\mathcal{T}^A\rangle) = x_{\bar{k}}$ implies $x_{\bar{k}}\mathrel{P_i} b$ for all $P_i \in \mathbb{D}_{\textrm{SSP}}(\mathcal{T}^A, \bar{x})$ with $r_1(P_i) = a$.
This contradicts the presence of the aforementioned preference $\hat{P}_i$.
If $\bar{k} = 1$, then $a, x_2 \in \langle b, \bar{x}|\mathcal{T}^A\rangle$ and $x_2 \in \langle b, a|\mathcal{T}^A\rangle$,
which imply $x_{2}\mathrel{P_i} a$ for all $P_i \in \mathbb{D}_{\textrm{SSP}}(\mathcal{T}^A, \bar{x})$ with $r_1(P_i) = b$.
This contradicts the presence of the aforementioned preference $\tilde{P}_i$.
Symmetrically, if $\bar{k} = v$, we can induce a contradiction to the presence of the preference $\hat{P}_i$.
This proves statement (i).\medskip

For statement (ii), we first show that given a domain $\mathbb{D} \subseteq \mathbb{D}_{\textrm{SSP}}(\mathcal{T}^A, \bar{x})$,
if it satisfies diversity, then $|\mathcal{N}^A(\bar{x})| \leq 2$.
Let $\pi = (x_1, \dots, x_v)$ denote the path in $\mathcal{T}^A$ connecting $a_1$ and $a_m$.
We first claim $\bar{x} \in \pi$. Suppose not, i.e., $\bar{x} \notin \pi$. Then, we have $\mathop{\textrm{Proj}}(\bar{x}, \pi) = x_k$ for some $1 \leq k \leq v$. Furthermore, if $k = 1$, semi-single-peakedness on $\mathcal{T}^A$ w.r.t.~$\bar{x}$ implies
$a_1\mathrel{\overline{P}_i}\bar{x}$, which contradicts the fact that $a_1$ is bottom-ranked in $\overline{P}_i$;
if $k = v$,
we have $a_m\mathrel{\underline{P}_i}\bar{x}$ by semi-single-peakedness on $\mathcal{T}^A$ w.r.t.~$\bar{x}$, which contradicts the fact that $a_m$ is bottom-ranked in $\underline{P}_i$;
if $1< k < v$, we have $x_k\mathrel{\underline{P}_i}\bar{x}$ and $x_k\mathrel{\overline{P}_i}\bar{x}$ by
by semi-single-peakedness on $\mathcal{T}^A$ w.r.t.~$\bar{x}$, which contradicts the fact that $\underline{P}_i$ and $\overline{P}_i$ are complete reversals.
Therefore, $\bar{x} \in \pi$.
Now, we show $|\mathcal{N}^A(\bar{x})| \leq 2$.
Suppose not, i.e., $|\mathcal{N}^A(\bar{x})| \geq 3$.
Consequently, we have $x \in \mathcal{N}^A(\bar{x})$ such that
$x \notin \pi$. Clearly, $x \notin \langle a_1, \bar{x}|\mathcal{T}^A\rangle$,
$\mathop{\textrm{Proj}}(x, \langle a_1, \bar{x}|\mathcal{T}^A\rangle) = \bar{x}$,
$x \notin \langle a_m, \bar{x}|\mathcal{T}^A\rangle$ and
$\mathop{\textrm{Proj}}(x, \langle a_m, \bar{x}|\mathcal{T}^A\rangle) = \bar{x}$.
Consequently, we have
$\bar{x}\mathrel{\underline{P}_i}x$ and $\bar{x}\mathrel{\overline{P}_i}x$ by
by semi-single-peakedness on $\mathcal{T}^A$ w.r.t.~$\bar{x}$, which contradicts the fact that $\underline{P}_i$ and $\overline{P}_i$ are complete reversals.

Next, we show that if
$|\mathcal{N}^A(\bar{x})| \leq 2$, then the semi-single-peaked domain $\mathbb{D}_{\textrm{SSP}}(\mathcal{T}^A, \bar{x})$ satisfies diversity.
Since $|\mathcal{N}^A(\bar{x})| \leq 2$,
we have a path $(x_1, \dots, x_v)$ in $\mathcal{T}^A$ such that
(i) $\bar{x} = x_{\bar{k}}$ for some $1 \leq \bar{k} \leq v$,
(ii) $x_1, x_v \in \textrm{Leaf}(\mathcal{T}^A)$ and
(iii) $[1< \bar{k} <v] \Rightarrow [\mathcal{N}^A(\bar{x}) = \{x_{\bar{k}-1}, x_{\bar{k}+1}\}]$.
First, let $\bar{k} = 1$.
Since $\bar{x} = x_1 \in \textrm{Leaf}(\mathcal{T}^A)$,
we by definition have a preference $P_i \in \mathbb{D}_{\textrm{SSP}}(\mathcal{T}^A, x_1)$ such that $r_1(P_i) = x_v$ and $r_m(P_i) = x_1$.
Next, we construct a linear order $P_i'$ that is a complete reversal of $P_i$.
Clearly, $P_i' \in \mathbb{D}_{\textrm{SSP}}(\mathcal{T}^A, x_1)$.
Symmetrically, if $\bar{k} = v$, we also have two completely reversed preferences in $\mathbb{D}_{\textrm{SSP}}(\mathcal{T}^A, \bar{x})$.
Last, let $1< \bar{k} <v$. Thus, $\mathcal{N}^A(\bar{x}) = \{x_{\bar{k}-1}, x_{\bar{k}+1}\}$.
It is natural that $x_{\bar{k}-1}$ and $x_{\bar{k}}$ are dual-thresholds in $\mathcal{T}^A$, and
$x_{\bar{k}}$ and $x_{\bar{k}+1}$ are dual-thresholds in $\mathcal{T}^A$.
Then, we identify the two subsets $A^{x_{\bar{k}-1} \rightharpoonup x_{\bar{k}}}$ and $A^{x_{\bar{k}+1} \rightharpoonup x_{\bar{k}}}$.
Note that $x_{\bar{k}} \notin A^{x_{\bar{k}-1} \rightharpoonup x_{\bar{k}}}$, $x_{\bar{k}} \notin A^{x_{\bar{k}+1} \rightharpoonup x_{\bar{k}}}$,
$A^{x_{\bar{k}-1} \rightharpoonup x_{\bar{k}}}\cap A^{x_{\bar{k}+1} \rightharpoonup x_{\bar{k}}} = \emptyset$ and
$A^{x_{\bar{k}-1} \rightharpoonup x_{\bar{k}}}\cup \{x_{\bar{k}}\}\cup A^{x_{\bar{k}+1} \rightharpoonup x_{\bar{k}}} = A$.
Now, pick arbitrary $P_i, P_i' \in \mathbb{D}_{\textrm{SSP}}(\mathcal{T}^A, x_{\bar{k}})$ such that $r_1(P_i) = x_1$ and $r_1(P_i') = x_v$.
Then, we construct two linear orders: $\hat{P}_i$ and $\hat{P}_i'$ over $A$ such that
(i) for all $x \in A^{x_{\bar{k}-1} \rightharpoonup x_{\bar{k}}}$ and $y \in A^{x_{\bar{k}+1} \rightharpoonup x_{\bar{k}}}$,
$x\mathrel{\hat{P}_i} x_{\bar{k}}$, $x_{\bar{k}}\mathrel{\hat{P}_i}y$,
$y\mathrel{\hat{P}_i'} x_{\bar{k}}$ and $x_{\bar{k}}\mathrel{\hat{P}_i'}x$,
(ii) $\hat{P}_i$ and $P_i$ agree on the relative rankings over $A^{x_{\bar{k}-1} \rightharpoonup x_{\bar{k}}}$, i.e., for all $x, y \in A^{x_{\bar{k}-1} \rightharpoonup x_{\bar{k}}}$, $[x\mathrel{\hat{P}_i}y]\Leftrightarrow [x\mathrel{P_i}y]$, and
$\hat{P}_i$ and $P_i'$ completely disagree on the relative rankings over $A^{x_{\bar{k}+1} \rightharpoonup x_{\bar{k}}}$, i.e., for all $x, y \in A^{x_{\bar{k}+1} \rightharpoonup x_{\bar{k}}}$, $[x\mathrel{\hat{P}_i}y]\Leftrightarrow [y\mathrel{P_i'}x]$, and
(iii) $\hat{P}_i'$ and $P_i'$ agree on the relative rankings over $A^{x_{\bar{k}+1} \rightharpoonup x_{\bar{k}}}$, i.e., for all  $x, y \in A^{x_{\bar{k}+1} \rightharpoonup x_{\bar{k}}}$, $[x\mathrel{\hat{P}_i'}y]\Leftrightarrow [x\mathrel{P_i'}y]$, and
    $\hat{P}_i'$ and $P_i$ completely disagree on the relative rankings over $A^{x_{\bar{k}-1} \rightharpoonup x_{\bar{k}}}$, i.e., for all $x, y \in A^{x_{\bar{k}-1} \rightharpoonup x_{\bar{k}}}$, $[x\mathrel{\hat{P}_i'}y]\Leftrightarrow [y\mathrel{P_i}x]$.
It is easy to show that $\hat{P}_i$ and $\hat{P}_i'$ are complete reversals, and both are semi-single-peaked on $\mathcal{T}^A$ w.r.t. $x_{\bar{k}}$.
This proves statement (ii).

\begin{clarification}\label{cla:sh}
Fixing a tree $\mathcal{T}^A$, dual-thresholds $a,b \in A$ and the semi-hybrid domain\\
$\mathbb{D}_{\emph{SH}}(\mathcal{T}^A, a, b)$,
the following two statements hold:
\begin{itemize}
\item[\rm (i)] the adjacency graph of $\mathbb{D}_{\emph{SH}}(\mathcal{T}^A, a,b)$
is a combination of the adjacency subgraph $G_{\sim}^{A^{a \rightharpoonup b}}$, which coincides with the subtree $\mathcal{T}^{A^{a \rightharpoonup b}}$, the adjacency subgraph $G_{\sim}^{\langle a, b|\mathcal{T}^A\rangle}$,
which is a complete subgraph, and
the adjacency subgraph $G_{\sim}^{A^{b \rightharpoonup a}}$, which coincides with $\mathcal{T}^{A^{b \rightharpoonup a}}$, and

\item[\rm (ii)]
if a domain $\mathbb{D} \subseteq \mathbb{D}_{\emph{SH}}(\mathcal{T}^A, a,b)$ satisfies diversity, then we have
$[A^{a \rightharpoonup b} \neq \{a\}] \Rightarrow [a \in \emph{Leaf}(\mathcal{T}^{A^{a \rightharpoonup b}})]$ and
$[A^{b \rightharpoonup a} \neq \{b\}] \Rightarrow [b \in \emph{Leaf}(\mathcal{T}^{A^{b \rightharpoonup a}})]$.
Conversely, if we have
$[A^{a \rightharpoonup b} \neq \{a\}] \Rightarrow [a \in \emph{Leaf}(\mathcal{T}^{A^{a \rightharpoonup b}})]$ and
$[A^{b \rightharpoonup a} \neq \{b\}] \Rightarrow [b \in \emph{Leaf}(\mathcal{T}^{A^{b \rightharpoonup a}})]$,
the semi-hybrid domain $\mathbb{D}_{\emph{SH}}(\mathcal{T}^A, a, b)$ satisfies diversity.
\end{itemize}
\end{clarification}

Similar to the verification of statement (i) of Clarification \ref{cla:ssp},
we know that the adjacency subgraph $G_{\sim}^{A^{a \rightharpoonup b}}$ coincides with the subtree $\mathcal{T}^{A^{a \rightharpoonup b}}$, and
the adjacency subgraph $G_{\sim}^{A^{b \rightharpoonup a}}$ coincides with the subtree $\mathcal{T}^{A^{b \rightharpoonup a}}$.
Furthermore, since any two distinct alternatives in the set $\langle a, b|\mathcal{T}^A\rangle$ are adjacent in the the semi-hybrid domain $\mathbb{D}_{\textrm{SH}}(\mathcal{T}^A, a, b)$, the adjacency subgraph $G_{\sim}^{\langle a, b|\mathcal{T}^A\rangle}$ is a complete subgraph.
To complete the verification, we show that the adjacency graph of the semi-hybrid domain $\mathbb{D}_{\textrm{SH}}(\mathcal{T}^A, a, b)$ is simply a combination of the subtree $G_{\sim}^{A^{a \rightharpoonup b}}=\mathcal{T}^{A^{a \rightharpoonup b}}$,
the complete subgraph $G_{\sim}^{\langle a, b|\mathcal{T}^A\rangle}$ and the subtree $G_{\sim}^{A^{b \rightharpoonup a}}=\mathcal{T}^{A^{b \rightharpoonup a}}$.
Since all these three subgraphs are nested in the adjacency graph of $\mathbb{D}_{\textrm{SH}}(\mathcal{T}^A, a, b)$,
it suffices to show that for all $x,y \in A$ with $x \sim y$,
either $(x,y) \in \mathcal{E}_{\sim}^{A^{a \rightharpoonup b}}$, or $(x,y) \in \mathcal{E}_{\sim}^{\langle a, b|\mathcal{T}^{A}\rangle}$, or $(x,y) \in \mathcal{E}_{\sim}^{A^{b \rightharpoonup a}}$ holds.
Given distinct $x,y\in A$ with $x \sim y$, suppose by contradiction that $(x,y) \notin \mathcal{E}_{\sim}^{A^{a \rightharpoonup b}} \cup \mathcal{E}_{\sim}^{\langle a, b|\mathcal{T}^{A}\rangle} \cup \mathcal{E}_{\sim}^{A^{b \rightharpoonup a}}$.
Since $x \sim y$, we have a preference $P_i^{\ast} \in \mathbb{D}_{\textrm{SH}}(\mathcal{T}^A, a, b)$ such that
$r_1(P_i^{\ast}) = x$ and $r_2(P_i^{\ast}) = y$.
We know that one of the following five cases must occur:
(1) $x \in A^{a \rightharpoonup b}\backslash \{a\}$,
(2) $x \in A^{b \rightharpoonup a}\backslash \{b\}$,
(3) $x \in \langle a, b|\mathcal{T}^A\rangle \backslash \{a,b\}$,
(4) $x = a$, or
(5) $x = b$.
In case (1), by the contradictory hypothesis, it must be true that $y \notin A^{a\rightharpoonup b}$.\footnote{Otherwise, $y \in A^{a\rightharpoonup b}$. Thus, $x, y \in A^{a\rightharpoonup b}$. Then, $x \sim y$ implies $(x, y) \in \mathcal{E}_{\sim}^{A^{a \rightharpoonup b}}$, which contrasts the contradictory hypothesis.}
Then, by $(a,b)$-semi-hybridness on $\mathcal{T}^A$,
we have $a\mathrel{P_i}y$ for all $P_i \in \mathbb{D}_{\textrm{SH}}(\mathcal{T}^A, a, b)$ with $r_1(P_i) = x$,
which contradicts the presence of the aforementioned preference $P_i^{\ast}$.
Similarly, in case (2), we can show $b\mathrel{P_i}y$ for all $P_i \in \mathbb{D}_{\textrm{SH}}(\mathcal{T}^A, a, b)$ with $r_1(P_i) = x$,
which contradicts the presence of the aforementioned preference $P_i^{\ast}$ as well.
In case (3), by the contradictory hypothesis, it must be true that $y \notin \langle a,b|\mathcal{T}^A\rangle$,\footnote{Otherwise,
$y \in \langle a,b|\mathcal{T}^A\rangle$. Thus, $x, y \in \langle a,b|\mathcal{T}^A\rangle$. Then, $x \sim y$ implies $(x, y) \in \mathcal{E}_{\sim}^{\langle a, b|\mathcal{T}^A\rangle}$, which contrasts the contradictory hypothesis.} which further implies
$y \in A^{a \rightharpoonup b}\backslash \{a\}$ or $y \in A^{b \rightharpoonup a}\backslash \{b\}$.
Then, by $(a,b)$-semi-hybridness on $\mathcal{T}^A$,
we know that
if $y \in A^{a \rightharpoonup b}\backslash \{a\}$, then $a\mathrel{P_i}y$ for all $P_i \in \mathbb{D}_{\textrm{SH}}(\mathcal{T}^A, a, b)$ with $r_1(P_i) = x\big]$;
if $y \in A^{b \rightharpoonup a}\backslash \{b\}$, then $b\mathrel{P_i}y$ for all $P_i \in \mathbb{D}_{\textrm{SH}}(\mathcal{T}^A, a, b)$ with $r_1(P_i) = x\big]$.
This contradicts the presence of the aforementioned preference $P_i^{\ast}$.
In case (4), by the contradictory hypothesis, it must be true that $y \in A^{b \rightharpoonup a}\backslash \{b\}$.\footnote{Otherwise, $y \in A^{a \rightharpoonup b}$ or $y \in \langle a, b|\mathcal{T}^A\rangle$.
Then, $x \sim y$ implies either $(x, y) \in \mathcal{E}_{\sim}^{A^{a \rightharpoonup b}}$ or $(x, y) \in \mathcal{E}_{\sim}^{\langle a, b|\mathcal{T}^A\rangle}$, which contrasts the contradictory hypothesis.}
Then, by $(a,b)$-semi-hybridness on $\mathcal{T}^A$,
we have $b\mathrel{P_i}y$ for all $P_i \in \mathbb{D}_{\textrm{SH}}(\mathcal{T}^A, a, b)$ with $r_1(P_i) = x$.
This contradicts the presence of the aforementioned preference $P_i^{\ast}$.
Similarly, in case (5), we can show $a\mathrel{P_i}y$ for all $P_i \in \mathbb{D}_{\textrm{SH}}(\mathcal{T}^A, a, b)$ with $r_1(P_i) = x$,
which contradicts the presence of the preference $P_i^{\ast}$ as well.
This proves statement (i).
\medskip

For statement (ii),
we first notice that if $|\langle a, b|\mathcal{T}^A\rangle| = 2$, $\mathbb{D}_{\textrm{SH}}(\mathcal{T}^A, a, b) = \mathbb{D}_{\textrm{SSP}}(\mathcal{T}^A, a) \cap \mathbb{D}_{\textrm{SH}}(\mathcal{T}^A, b)$, and then statement (ii) of Clarification \ref{cla:ssp} implies the result.
Henceforth, we assume $|\langle a, b|\mathcal{T}^A\rangle| \geq 3$.
We first show that given a domain $\mathbb{D} \subseteq \mathbb{D}_{\textrm{SH}}(\mathcal{T}^A, a, b)$,
if $\mathbb{D}$ satisfies diversity, then we have $[A^{a \rightharpoonup b} \neq \{a\}] \Rightarrow [a \in \textrm{Leaf}(\mathcal{T}^{A^{a \rightharpoonup b}})]$ and
$[A^{b \rightharpoonup a} \neq \{b\}] \Rightarrow [b \in \textrm{Leaf}(\mathcal{T}^{A^{b \rightharpoonup a}})]$.
The verifications for both cases $A^{a \rightharpoonup b} \neq \{a\}$ and
$A^{b \rightharpoonup a} \neq \{b\}$ are symmetric. We hence assume w.l.o.g.~that $A^{a \rightharpoonup b} \neq \{a\}$.
By diversity, we have $\underline{P}_i, \overline{P}_i \in \mathbb{D}$.
Furthermore, since $A^{a \rightharpoonup b} \neq \{a\}$,
$(a,b)$-semi-hybridness on $\mathcal{T}^A$ implies either $a_1 \in A^{a\rightharpoonup b}\backslash \{a\}$ and $a_m \in A\backslash A^{a \rightharpoonup b}$,
or $a_1 \in A\backslash A^{a \rightharpoonup b}$ and $a_m \in A^{a\rightharpoonup b}\backslash \{a\}$.
We can assume w.l.o.g.~that $a_1 \in A^{a\rightharpoonup b}\backslash \{a\}$ and $a_m \in \langle a, b|\mathcal{T}^A\rangle\cup A^{b \rightharpoonup a}$.
Suppose that $a \notin \textrm{Leaf}(\mathcal{T}^{A^{a \rightharpoonup b}})$.
Thus, we have distinct $x, y\in \mathcal{N}^{A^{a \rightharpoonup b}}(a)$.
Note that at least one of $\{x, y\}$ is not located in $\langle a_1, a|\mathcal{T}^{A^{a \rightharpoonup b}}\rangle$.
We assume w.l.o.g.~$x \notin \langle a_1, a|\mathcal{T}^{A^{a \rightharpoonup b}}\rangle$.
Consequently, given $a_1 \in A^{a\rightharpoonup b}\backslash \{a\}$ and $a_m \in A\backslash A^{a \rightharpoonup b}$,
by $(a,b)$-semi-hybridness on $\mathcal{T}^A$, we have
$a\mathrel{\underline{P}_i}x$ and $a\mathrel{\overline{P}_i}x$,
which contradict the fact that $\underline{P}_i$ and $\overline{P}_i$ are complete reversals.

Last, we show that if we have
$[A^{a \rightharpoonup b} \neq \{a\}] \Rightarrow [a \in \textrm{Leaf}(\mathcal{T}^{A^{a \rightharpoonup b}})]$ and
$[A^{b \rightharpoonup a} \neq \{b\}] \Rightarrow [b \in \textrm{Leaf}(\mathcal{T}^{A^{b \rightharpoonup a}})]$,
then the semi-hybrid domain $\mathbb{D}_{\textrm{SH}}(\mathcal{T}^A, a, b)$ satisfies diversity.
If $A^{a \rightharpoonup b} = \{a\}$ and $A^{b \rightharpoonup a} = \{b\}$,
then $\mathbb{D}_{\textrm{SH}}(\mathcal{T}^A, a, b) = \mathbb{P}$ includes two completely reversed preference.
Next, let $A^{a \rightharpoonup b} = \{a\}$ and $A^{b \rightharpoonup a} \neq \{b\}$.
It is evident that $b$ has a unique neighbor in the line $\langle a, b|\mathcal{T}^A\rangle$, and
the hypothesis implies that $b$ has a unique neighbor in the subtree $\mathcal{T}^{A^{b \rightharpoonup a}}$.
Since $\mathcal{T}^A$ is a union of the line $\langle a, b|\mathcal{T}^A\rangle$ and the subtree $\mathcal{T}^{A^{b \rightharpoonup a}}$,
we have $|\mathcal{N}^A(b)| = 2$.
Then, Clarification 2 implies that $\mathbb{D}_{\textrm{SSP}}(\mathcal{T}^A, b)$ contains two completely reversed preferences.
Moreover, since $A^{a \rightharpoonup b} = \{a\}$, it is true that $\mathbb{D}_{\textrm{SSP}}(\mathcal{T}^A, b) \subset \mathbb{D}_{\textrm{SH}}(\mathcal{T}^A, a, b)$,
which implies that $\mathbb{D}_{\textrm{SH}}(\mathcal{T}^A, a, b)$ satisfies diversity.
Symmetrically, if $A^{a \rightharpoonup b} \neq \{a\}$ and $A^{b \rightharpoonup a} = \{b\}$,
$\mathbb{D}_{\textrm{SH}}(\mathcal{T}^A, a, b)$ satisfies diveristy.
Last, we consider the situation $A^{a \rightharpoonup b} \neq \{a\}$ and $A^{b \rightharpoonup a} \neq \{b\}$.
Thus, the hypothesis implies $a \in \textrm{Leaf}(\mathcal{T}^{A^{a \rightharpoonup b}})$ and $b \in \textrm{Leaf}(\mathcal{T}^{A^{b \rightharpoonup a}})$.
Let $\bar{a}$ be the unique neighbor of $a$ in $\mathcal{T}^{A^{a \rightharpoonup b}}$, and
$\bar{b}$ be the unique neighbor of $b$ in $\mathcal{T}^{A^{b \rightharpoonup a}}$.
Note that $\mathbb{D}_{\textrm{SP}}(\mathcal{T}^A) \subset \mathbb{D}_{\textrm{SH}}(\mathcal{T}^A, a, b)$.
We fix arbitrary $x,y \in \textrm{Leaf}(\mathcal{T}^A)$ such that $x \in A^{a \rightharpoonup b}$ and $y \in A^{b \rightharpoonup a}$.
Clearly, $x \notin \{a,b\}$ and $y \notin \{a,b\}$.
According to $\mathbb{D}_{\textrm{SP}}(\mathcal{T}^A)$, we fix two single-peaked preferences $P_i$ and $P_i'$ such that
$r_1(P_i) = x$ and $r_1(P_i') = y$.
Clearly, $P_i$ and $P_i'$ completely disagree on the relative rankings over $\langle a, b|\mathcal{T}^A\rangle$, i.e.,
for all $z, z' \in \langle a, b|\mathcal{T}^A\rangle$, $[z\mathrel{P_i}z'] \Leftrightarrow [z'\mathrel{P_i'}z]$.
Now, we construct two linear orders $\hat{P}_i$ and $\hat{P}_i'$ over $A$ satisfying the following three conditions:
(i) for all $z \in A^{a \rightharpoonup b}\backslash \{a\}$, $z' \in \langle a, b|\mathcal{T}^A\rangle$ and $z'' \in A^{b\rightharpoonup a}\backslash \{b\}$,
$z\mathrel{\hat{P}_i}z'$, $z'\mathrel{\hat{P}_i}z''$, $z''\mathrel{\hat{P}_i'}z'$ and $z'\mathrel{\hat{P}_i'}z$,
(ii) $\hat{P}_i$ and $P_i$ agree on the relative rankings over $A^{a \rightharpoonup b}\cup \langle a, b|\mathcal{T}^A\rangle$, and
$\hat{P}_i$ and $P_i'$ completely disagree on the relative rankings over $A^{b \rightharpoonup a}\backslash \{b\}$, and
(iii) $\hat{P}_i'$ and $P_i'$ agree on the relative rankings over $\langle a, b|\mathcal{T}^A\rangle\cup A^{b \rightharpoonup a}$, and
$\hat{P}_i'$ and $P_i$ completely disagree on the relative rankings over $A^{a \rightharpoonup b}\backslash \{a\}$.
By construction, it is easy to show that $\hat{P}_i$ and $\hat{P}_i'$ are complete reversals and $(a, b)$-semi-hybrid on $\mathcal{T}^A$.
This proves statement (ii).

\begin{clarification}\label{cla:hybridrule}
Given a tree $\mathcal{T}^A$ and dual-thresholds $a,b\in A$,
an $(a,b)$-hybrid rule on $\mathcal{T}^A$ is a tops-only and strategy-proof rule on a domain
$\mathbb{D} \subseteq \mathbb{D}_{\emph{SH}}(\mathcal{T}^A, a, b)$.
\end{clarification}

Fixing a voter $i \in N$, we consider the following $(a,b)$-hybrid rule on $\mathcal{T}^A$:
for all $P \in \mathbb{D}^n$,
\begin{align*}
f(P) = \left\{
\begin{array}{ll}
r_1(P_i) & \textrm{if}\; r_1(P_i) \in \langle a, b|\mathcal{T}^A\rangle,\\
\mathop{\textrm{Proj}}\big(a, \mathcal{T}^{\Gamma(P)}\big) & \textrm{if}\; r_1(P_i) \in A^{a \rightharpoonup b}\backslash \{a\},\; \textrm{and}\\
\mathop{\textrm{Proj}}\big(b, \mathcal{T}^{\Gamma(P)}\big) & \textrm{if}\; r_1(P_i) \in A^{b \rightharpoonup a}\backslash \{b\}.
\end{array}
\right.
\end{align*}
By definition, it is clear that the hybrid rule $f$ satisfies unanimity and the tops-only property.
We hence focus on showing its strategy-proofness.

First, we consider the possible manipulation of a voter other than $i$, say $j \in N\backslash \{i\}$.
Given two profiles $P=(P_j, P_{-j})$ and $P'=(P_j', P_{-j})$,
let $f(P) \neq f(P')$. We show $f(P)\mathrel{P_j}f(P')$.
Since $f(P) \neq f(P')$, by the definition of $f$,
it must be true that either $r_1(P_i) \in A^{a \rightharpoonup b}\backslash \{a\}$ or $r_1(P_i) \in A^{b \rightharpoonup a}\backslash \{b\}$ holds.
The two situations are symmetric, and we hence focus on verifying the first one.
Let $r_1(P_i) \in A^{a \rightharpoonup b}\backslash \{a\}$.
Thus, by definition, $f(P) = \textrm{Proj}\big(a, \mathcal{T}^{\Gamma(P)}\big)$ and $f(P') = \textrm{Proj}\big(a, \mathcal{T}^{\Gamma(P')}\big)$.
We consider two cases: $r_1(P_j) \in \langle a, b|\mathcal{T}^A\rangle \cup A^{b \rightharpoonup a}$ and
$r_1(P_j) \in A^{a \rightharpoonup b}\backslash \{a\}$.
In the first case, we know that $a \in \Gamma(P)$ and hence $f(P)=a$, and moreover
$(a, b)$-semi-hybridness on $\mathcal{T}^A$ implies $\max^{P_j}(A^{a \rightharpoonup b}) =a$.
Meanwhile, since $f(P') = \mathop{\textrm{Proj}}\big(a, \mathcal{T}^{\Gamma(P')}\big) \in \langle r_1(P_i), a|\mathcal{T}^A\rangle \subseteq
A^{a \rightharpoonup b}$ and $f(P) \neq f(P')$,
it must be the case that $f(P)\mathrel{P_j} f(P')$, as required.
In the second case, $P_2$ is semi-single-peaked on $\mathcal{T}^{A}$ w.r.t.~$a$.
Then, by the proof of the sufficiency part of the Theorem of \citet{CSS2013},
we immediately have $f(P)\mathrel{P_j} f(P')$, as required.
In conclusion, voter $j \in N\backslash \{i\}$ has no incentive to manipulate.

Next, we consider the possible manipulation of voter $i$.
Given two profiles $P=(P_i,  P_{-i})$ and $P'=(P_i', P_{-i})$, let $f(P) \neq f(P')$.
We show $f(P)\mathrel{P_i}f(P')$.
Note that if $r_1(P_i) \in \langle a, b|\mathcal{T}^A\rangle$, then $f(P) = r_1(P_i)$, which immediately implies $f(P)\mathrel{P_i}f(P')$.
Henceforth, we assume that either $r_1(P_i) \in A^{a \rightharpoonup b}\backslash \{a\}$ or $r_1(P_i) \in A^{b \rightharpoonup a}\backslash \{b\}$ holds.
There are six cases to consider:\\
(1) $r_1(P_i) \in A^{a \rightharpoonup b}\backslash \{a\}$ and $r_1(P_i') \in A^{a \rightharpoonup b}\backslash \{a\}$,~~~~
(2) $r_1(P_i) \in A^{b \rightharpoonup a}\backslash \{b\}$ and $r_1(P_i') \in A^{b \rightharpoonup a}\backslash \{b\}$,\\
(3) $r_1(P_i) \in A^{a \rightharpoonup b}\backslash \{a\}$ and $r_1(P_i') \in A^{b \rightharpoonup a}\backslash \{b\}$,~~~~\,
(4) $r_1(P_i) \in A^{b \rightharpoonup a}\backslash \{b\}$ and $r_1(P_i') \in A^{a \rightharpoonup b}\backslash \{a\}$,\\
(5) $r_1(P_i) \in A^{a \rightharpoonup b}\backslash \{a\}$ and $r_1(P_i') \in \langle a, b|\mathcal{T}^A\rangle$, and
(6) $r_1(P_i) \in A^{b \rightharpoonup a}\backslash \{b\}$ and $r_1(P_i') \in \langle a, b|\mathcal{T}^A\rangle$.\\
The first two cases are symmetric, and we hence focus on the verification of case (1).
In case (1), $P_i$ is semi-single-peaked on $\mathcal{T}^A$ w.r.t.~$a$,
$f(P) = \mathop{\textrm{Proj}}\big(a, \mathcal{T}^{\Gamma(P)}\big)$ and $f(P') = \mathop{\textrm{Proj}}\big(a, \mathcal{T}^{\Gamma(P')}\big)$.
Then, by the proof of the sufficiency part of the Theorem of \citet{CSS2013},
we immediately have $f(P)\mathrel{P_i} f(P')$, as required.
Cases (3) and (4) are symmetric, and we hence focus on the verification of case (3).
In case (3), we know $f(P) = \mathop{\textrm{Proj}}\big(a, \mathcal{T}^{\Gamma(P)}\big) \in \langle r_1(P_i), a|\mathcal{T}^A\rangle$ and
$f(P') = \mathop{\textrm{Proj}}\big(b, \mathcal{T}^{\Gamma(P')}\big) \in \langle r_1(P_i'), b|\mathcal{T}^A\rangle \subseteq A^{b \rightharpoonup a}$.
Moreover, since $r_1(P_i) \in A^{a \rightharpoonup b}\backslash \{a\}$,
by $(a,b)$-semi-hybridness on $\mathcal{T}^A$,
we know $P_i$ is semi-single-peaked on $\mathcal{T}^A$ w.r.t.~$a$ which implies $\min^{P_i}(\langle r_1(P_i), a|\mathcal{T}^A\rangle) = a$ and $a\mathrel{P_i} b$, and $\max^{P_i}(A^{b \rightharpoonup a}) = b$.
Therefore, by transitivity of $P_i$, it must be the case that $f(P)\mathrel{P_i}f(P')$, as required.
Cases (5) and (6) are symmetric, and we hence focus on the verification of case (5).
In case (5), we know $f(P) = \mathop{\textrm{Proj}}\big(a, \mathcal{T}^{\Gamma(P)}\big) \in \langle r_1(P_i), a|\mathcal{T}^A\rangle$ and
$f(P') = r_1(P_i') \in \langle a,b|\mathcal{T}^A\rangle$.
Moreover, since $r_1(P_i) \in A^{a \rightharpoonup b}\backslash \{a\}$,
we know $P_i$ is semi-single-peaked on $\mathcal{T}^A$ w.r.t.~$a$ which implies $\min^{P_i}(\langle r_1(P_i), a|\mathcal{T}^A\rangle) = a$ and  and $\max^{P_i}(\langle a, b|\mathcal{T}^A\rangle) = a$.
Therefore, by transitivity of $P_i$, it must be the case that $f(P)\mathrel{P_i}f(P')$, as required.
In conclusion, voter $i$ has no incentive to manipulate.
Therefore, the hybrid rule $f$ is strategy-proof.

\begin{clarification}\label{cla:leafsymmetry}
Fixing a tree $\mathcal{T}^A$ and a threshold $\bar{x} \in A$,
a semi-single-peaked domain $\mathbb{D}$ on $\mathcal{T}^A$ w.r.t.~$\bar{x}$
satisfies leaf symmetry if and only if either $\bar{x} \notin \emph{Leaf}(\mathcal{T}^A)$,
or $\bar{x} \in \emph{Leaf}(\mathcal{T}^A)$ and $\mathbb{D}\subseteq \mathbb{D}_{\emph{SSP}}(\mathcal{T}^A, \bar{x})\cap \mathbb{D}_{\emph{SSP}}(\mathcal{T}^A, x)$, where $\mathcal{N}^A(\bar{x}) = \{x\}$.
\end{clarification}

First, if $\bar{x} \notin \textrm{Leaf}(\mathcal{T}^A)$,
semi-single-peakedness on $\mathcal{T}^A$ w.r.t.~$\bar{x}$ immediately implies $|\mathcal{S}(\mathbb{D}^z)|=1$ for all $z \in \textrm{Leaf}(\mathcal{T}^A) = \textrm{Leaf}(G_{\sim}^A)$.
Similarly, given $\bar{x} \in \textrm{Leaf}(\mathcal{T}^A)$ and $\mathbb{D}\subseteq \mathbb{D}_{\textrm{SSP}}(\mathcal{T}^A, \bar{x})\cap \mathbb{D}_{\textrm{SSP}}(\mathcal{T}^A, x)$,
where $\mathcal{N}^A(\bar{x}) = \{x\}$,
we know that $x \notin \textrm{Leaf}(\mathcal{T}^A)$ and then
semi-single-peakedness on $\mathcal{T}^A$ w.r.t.~$x$ immediately implies $|\mathcal{S}(\mathbb{D}^z)|=1$ for all $z \in \textrm{Leaf}(\mathcal{T}^A)=\textrm{Leaf}(G_{\sim}^A)$.
Conversely, let $\mathbb{D}$ satisfies leaf symmetry. We show that either $\bar{x} \notin \textrm{Leaf}(\mathcal{T}^A)$,
or $\bar{x} \in \textrm{Leaf}(\mathcal{T}^A)$ and $\mathbb{D}\subseteq \mathbb{D}_{\textrm{SSP}}(\mathcal{T}^A, \bar{x})\cap \mathbb{D}_{\textrm{SSP}}(\mathcal{T}^A, x)$, where $\mathcal{N}^A(\bar{x}) = \{x\}$.
Suppose not, i.e., $\bar{x} \in \textrm{Leaf}(\mathcal{T}^A)$ and $\mathbb{D}\nsubseteq \mathbb{D}_{\textrm{SSP}}(\mathcal{T}^A, x)$.
Since $\bar{x} \in \textrm{Leaf}(\mathcal{T}^A)$ and $\mathcal{N}^A(\bar{x}) = \{x\}$,
we know that for all $z \in A\backslash \{\bar{x}\}$,
$x \in \langle z, \bar{x}|\mathcal{T}^A\rangle$.
Then, by semi-single-peakedness on $\mathcal{T}^A$ w.r.t.~$\bar{x}$,
we know that all preferences with the peak other than $\bar{x}$ are also semi-single-peaked on $\mathcal{T}^A$ w.r.t.~$x$.
Consequently, $\mathbb{D}\nsubseteq \mathbb{D}_{\textrm{SSP}}(\mathcal{T}^A, x)$ implies that
some $P_i^{\ast} \in \mathbb{D}^{\bar{x}}$ is not semi-single-peaked on $\mathcal{T}^A$ w.r.t.~$x$.
Note that given $\bar{x} \in \textrm{Leaf}(\mathcal{T}^A)$ and $\mathcal{N}^A(\bar{x}) = \{x\}$,
in any preference with the peak $\bar{x}$ that is semi-single-peaked on $\mathcal{T}^A$ w.r.t.~$x$, the alternative $x$ must be the second best. Therefore, we must have $r_2(P_i^{\ast}) \neq x$.
Furthermore, since $\mathcal{T}^A = G_{\sim}^A$, $\bar{x} \in \textrm{Leaf}(\mathcal{T}^A) = \textrm{Leaf}(G_{\sim}^A)$ and $\mathcal{N}^A(\bar{x}) = \{x\}$, we know
$(\bar{x}, x) \in \mathcal{E}_{\sim}^A$.
Then, by leaf symmetry, we have a preference $P_i' \in \mathbb{D}$ such that
$r_1(P_i') \notin \{\bar{x}, x\}$ and $r_2(P_i') = \bar{x}$, which contradicts semi-single-peakedness on $\mathcal{T}^A$ w.r.t.~$\bar{x}$.

\begin{clarification}\label{cla:example}
An example of a degenerate semi-hybrid domain.
\end{clarification}

\begin{table}[t]
\centering
\begin{tabular}{ccccccccc}
$P_1$ & $P_2$ & $P_3$ & $P_4$ & $P_5$ & $P_6$ & $P_7$ & $P_8$ & $P_9$ \\
$a_1$ & $a_1$ & $a_1$ & $a_2$ & $a_2$ & $a_3$ & $a_3$ & $a_3$ & $a_4$ \\[-0.2em]
$a_2$ & $a_2$ & $a_3$ & $a_1$ & $a_3$ & $a_1$ & $a_2$ & $a_4$ & $a_3$ \\[-0.2em]
$a_3$ & $a_4$ & $a_2$ & $a_3$ & $a_1$ & $a_3$ & $a_1$ & $a_2$ & $a_2$ \\[-0.2em]
$a_4$ & $a_3$ & $a_4$ & $a_4$ & $a_4$ & $a_4$ & $a_4$ & $a_1$ & $a_1$
\end{tabular}
\caption{Domain $\mathbb{D}$}\label{tab:gap}
\end{table}

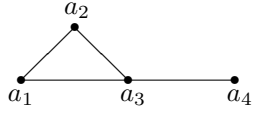
\begin{figure}[t]
\hspace{1.2em}
\begin{tikzpicture}
        \put(175, -30){\circle*{3}}
        \put(195, -10){\circle*{3}}
        \put(215, -30){\circle*{3}}
        \put(255, -30){\circle*{3}}
		\put(170, -38){\footnotesize{$a_1$}}
        \put(191, -5){\footnotesize{$a_2$}}
		\put(212, -38){\footnotesize{$a_3$}}
		\put(252, -38){\footnotesize{$a_4$}}
        \put(195, -10){\line(-1,-1){20}}
        \put(195, -10){\line(1,-1){20}}
        \put(175,-30){\line(1,0){80}}
\end{tikzpicture}
\vspace{2.5em}
\caption{The adjacency graph $G_{\sim}^A$}\label{fig:gap}
\end{figure}

Let $A = \{a_1, a_2, a_3, a_4\}$. All 9 preferences of the domain $\mathbb{D}$ and the adjacency graph $G_{\sim}^A$ are respectively specified in Table \ref{tab:gap} and Figure \ref{fig:gap}.
First, domain $\mathbb{D}$ satisfies path-connectedness according to Figure \ref{fig:gap},
diversity according to preferences $P_1$ and $P_9$ in Table \ref{tab:gap}, and
leaf symmetry, i.e., $\textrm{Leaf}(G_{\sim}^A) = \{a_4\}$ and $|\mathcal{S}(\mathbb{D}^{a_4})| = 1$.
Second, since $|\mathcal{S}(\mathbb{D}^{a_4})| = 1$, the unique seconds property is also satisfied.
Last, we observe that $G_{\sim}^A$ contains a cycle, coincides with the adjacency graph of the semi-hybrid domain $\mathbb{D}_{\textrm{SH}}(\mathcal{L}^A, a_1, a_3)$, and is strictly included in the adjacency graph of $\mathbb{D}_{\textrm{SH}}(\mathcal{L}^A, a_1, a_4)$ which is a complete graph. Moreover,
since $\mathbb{D} \subset \mathbb{P} = \mathbb{D}_{\textrm{SH}}(\mathcal{L}^A, a_1, a_4)$ and
$P_2$ in Table \ref{tab:gap} is not $(a_1, a_3)$-semi-hybrid on $\mathcal{L}^A$, we infer that $\mathbb{D}$ is an $(a_1,a_4)$-semi-hybrid domain on $\mathcal{L}^A$, and hence a degenerate semi-hybrid domain.

\begin{clarification}\label{cla:RS2019}
The relation to \citet{RS2019}.
\end{clarification}

First, recall the notion of weak adjacency introduced in footnote \ref{footnote:weakadjacency}.
A rich domain $\mathbb{D}$ of \citet{RS2019} satisfies the following three conditions:
\begin{enumerate}
\item domain $\mathbb{D}$ is weakly path-connected,

\item $[a \in \mathcal{S}(\mathbb{D}^b)] \Leftrightarrow [b \in \mathcal{S}(\mathbb{D}^a)]$ for all $a,b \in A$, and

\item given $a, b, c \in A$ with $a \nsim b$ and $b\nsim c$, there exist $P_i \in \mathbb{D}^a$ and $P_i' \in \mathbb{D}^c$ such that
given $d \in A$ with $d\mathrel{P_i}c$ and $d\mathrel{P_i'}a$, we have $d\mathrel{P_i''} c$ or $d\mathrel{P_i''}a$ at some $P_i'' \in \mathbb{D}^b$.
\end{enumerate}

It is clear that path-connectedness implies the first condition.
Next, if we strengthen the notion of weak adjacency in the hypothesis of the third condition by the notion of adjacency,
it is easy to identify preferences in the domain satisfying the requirement of the third condition.
Given $a, b, c \in A$ such that $a \sim b$ and $b \sim c$,
we have preferences $P_i, \hat{P}_i, P_i', \hat{P}_i' \in \mathbb{D}$ such that
$r_1(P_i) = r_2(\hat{P}_i) = a$, $r_1(\hat{P}_i) = r_2(P_i) = b$ and $r_k(P_i) = r_k(\hat{P}_i)$ for all $k \in \{3, \dots, m\}$, and
$r_1(P_i') = r_2(\hat{P}_i') = c$, $r_1(\hat{P}_i') = r_2(P_i') = b$ and $r_k(P_i') = r_k(\hat{P}_i')$ for all $k \in \{3, \dots, m\}$.
Thus, we have $P_i \in \mathbb{D}^a$, $P_i'\in \mathbb{D}^c$ and $\hat{P}_i, \hat{P}_i' \in \mathbb{D}^b$.
Furthermore, given $d \in A$ with $d\mathrel{P_i}c$ and $d\mathrel{P_i'}a$, we have $d\mathrel{\hat{P}_i}c$ and $d\mathrel{\hat{P}_i'}a$, either of which meets the requirement of the third condition.

Compared to leaf symmetry, the second condition is too demanding.
For instance, recall the semi-single-peaked domain $\mathbb{D}_{\textrm{SSP}}(\mathcal{L}^A, a_{3})$,
where $|A| =m \geq 4$.
Clearly, a preference with the peak $a_{3}$ has no restriction on the rankings of other alternatives.
Therefore, any alternative other than $a_{3}$ can be second ranked in a semi-single-peaked preference with the peak $a_{3}$, i.e.,
$\mathcal{S}\big(\mathbb{D}_{\textrm{SSP}}^{a_{3}}(\mathcal{L}^A, a_{3})\big) = \{a_1, a_2, a_4, \dots, a_m\}$, and hence $\big|\mathcal{S}\big(\mathbb{D}_{\textrm{SSP}}^{a_{3}}(\mathcal{L}^A, a_{3})\big)\big|>1$.
Meanwhile, by semi-single-peakedness on $\mathcal{L}^A$ w.r.t.~$a_{3}$, we know $a_2\mathrel{P_i}a_{3}$ for all preference $P_i \in \mathbb{D}_{\textrm{SSP}}(\mathcal{L}^A, a_{3})$ with $r_1(P_i) = a_1$,
which implies $a_{3} \notin \mathcal{S}\big(\mathbb{D}_{\textrm{SSP}}^{a_1}(\mathcal{L}^A, a_{3})\big)$.
Therefore, the second condition above is violated here.
This however does not contrast leaf symmetry since $a_{3}$ is not a leaf in the adjacency graph $G_{\sim}^A = \mathcal{L}^A$.

\begin{clarification}\label{cla:violatepath-connectedness}
Verify strategy-proofness of the SCF $f$ constructed in Example \ref{exm:violatepath-connectedness}.
\end{clarification}

Fix a profile $(P_i, P_j)$.
We first verify all possible manipulations of voter $i$.
First, let $r_1(P_j) \in A\backslash B$, Then, by construction, we know $f(P_i', P_j) =r_1(P_j)$ for all $P_i' \in \mathbb{D}_2$.
Hence, no manipulation of voter $i$ is profitable.
Next, let $r_1(P_j) \in B$. Then, by construction, we know $f(P_i', P_j) \in B$ for all $P_i' \in \mathbb{D}_2$.
Hence, $f(P_i, P_j) = \max^{P_i}(B)$ implies that no manipulation of voter $i$ is profitable.
Therefore, voter $i$ has no incentive to manipulate.
We next verify all possible manipulation of voter $j$.
First, let $r_1(P_j) \in A\backslash B$. Then, $f(P_i, P_j) =r_1(P_j)$ implies that no manipulation of voter $j$ is profitable.
Next, let $r_1(P_j) \in B$. Thus, $f(P_i, P_j) = \max^{P_i}(B) \in B$.
Consider a possible manipulation $P_j' \in \mathbb{D}_2$.
If $r_1(P_j') \in B$, we have $f(P_i, P_j') = \max^{P_i}(B)=f(P_i, P_j)$, which implies that the manipulation is not profitable.
If $r_1(P_j') \in A\backslash B$, we have $f(P_i, P_j') = r_1(P_i') \in A\backslash B$.
Since $r_1(P_j) \in B$, we know that $P_j$ must be one of the first six preferences in Table \ref{tab:violatepath-connectedness}, and hence
voter $j$ prefers each alternative of $B$ to each alternative of $A\backslash B$.
Hence, we have $f(P_i, P_j)\mathrel{P_j}f(P_i, P_j')$, as required.
Therefore, voter $j$ has no incentive to manipulate.
In conclusion, $f$ is strategy-proof.

\begin{clarification}\label{cla:strategy-proofness}
Verify strategy-proofness of the SCF $f$ constructed in Example \ref{exm:indispensability}.
\end{clarification}

Since $f$ is anonymous, it suffices to focus on all possible manipulations of voter $i$.
Given two distinct profiles $(P_i, P_j)$ and $(P_i', P_j)$,
it is clear that if $f(P_i, P_j) = r_1(P_i)$ or $f(P_i, P_j) = f(P_i', P_j)$,
voter $i$ has no incentive to manipulate at $(P_i, P_j)$ via $P_i'$.
We henceforth assume $f(P_i, P_j) \neq r_1(P_i)$ and $f(P_i, P_j) \neq f(P_i', P_j)$.
There are four cases:
(1) $f(P_i, P_j) = d \neq r_1(P_i)$ and $f(P_i', P_j) \neq d$,
(2) $f(P_i, P_j) = a \neq r_1(P_i)$ and $f(P_i', P_j) \neq a$,
(3) $f(P_i, P_j) = c \neq r_1(P_i)$ and $f(P_i', P_j) \neq c$, and
(4) $f(P_i, P_j) = b \neq r_1(P_i)$ and $f(P_i', P_j) \neq b$.
Note that the first three cases are symmetric. We hence focus on the verification of cases (1) and (4).

First, we consider case (1).
By construction, $f(P_i, P_j) = d$ implies one of the following subcases:
(i) $P_i = P_7$ and $P_j \in \mathbb{D}_3^d$, or (ii) $P_i \in \mathbb{D}_3^d$ and $P_j = P_7$,
or (iii) $f(P_i, P_j) =\textrm{Proj}\big(b, \langle r_1(P_i), r_1(P_j)|\mathcal{T}^A\rangle\big)$.
Immediately, we rule out subcase (ii) since $r_1(P_i) \neq d$, and also rule out subcase (iii) since
$d = f(P_i, P_j) =\textrm{Proj}\big(b, \langle r_1(P_i), r_1(P_j)|\mathcal{T}^A\rangle\big)$ implies $r_1(P_i) = r_1(P_j) = d$,
which contrasts the hypothesis $r_1(P_i) \neq d$.
Therefore, subcase (i) holds.
Thus, we have $f(P_i, P_j) = d$, $f(P_i', P_j) \neq d$, $P_i = P_7$ and $P_j \in \mathbb{D}_3^d$.
By the construction of $f$, we observe that given $P_j \in \mathbb{D}_3^d$,
there exists no $P_i' \in \mathbb{D}_3$ such that $f(P_i', P_j) = a$.
Therefore, $r_1(P_i) = r_1(P_7)=a \neq f(P_i', P_j)$, $r_2(P_i) = r_2(P_7)=d = f(P_i, P_j)$ and $f(P_i, P_j)\neq f(P_i', P_j)$ imply $f(P_i, P_j)\mathrel{P_i}f(P_i', P_j)$.
In conclusion, in case (1), voter $i$ has no incentive to manipulate at $(P_i, P_j)$ via $P_i'$.

Last, we consider case (4).
Since $r_1(P_i) \neq b$, one of the following three situations must occur: $P_i \in \mathbb{D}_3^a$, or $P_i \in \mathbb{D}_3^c$,
or $P_i \in \mathbb{D}_3^a$.
These three situations are symmetric, and we hence assume w.l.o.g.~that $P_i \in \mathbb{D}_3^a=\{P_1, P_7\}$.
Since $f(P_i, P_j) = b$, the construction of $f$ implies $b = f(P_i, P_j) = \textrm{Proj}\big(b,\langle r_1(P_i),r_1(P_j)|\mathcal{T}^A\rangle\big) = \textrm{Proj}\big(b,\langle a,r_1(P_j)|\mathcal{T}^A\rangle\big)$,
which further implies $r_1(P_j) \neq a$.
Furthermore, if $r_1(P_j) = b$, it is clear that $f(P_i', P_j) =
\textrm{Proj}\big(b,\langle r_1(P_i),b|\mathcal{T}^A\rangle\big)=b$ by construction, which contradicts the hypothesis $f(P_i', P_j) \neq b$.
Therefore, $P_j \in \mathbb{D}_3^c \cup \mathbb{D}_3^d = \{P_5, P_8, P_6, P_9\}$.
Moreover, since $f(P_i, P_j) = b \neq a$ and $P_i \in \mathbb{D}_3^a$, the construction of $f$ further implies $P_j \neq P_8$.
Therefore, $P_j \in \{P_5, P_6, P_9\}$.
Now, according to $P_i \in \{P_1, P_7\}$, we consider two subcases: (i) $P_i = P_1$ and (ii) $P_i = P_7$.
In subcase (i), suppose by contradiction that $f(P_i', P_j)\mathrel{P_i}f(P_i, P_j)$.
This, by $f(P_i, P_j) = b$ and $P_i = P_1$, implies $f(P_i', P_j) = a$.
Since $P_j \notin \mathbb{D}_3^a$ and $P_j \neq P_8$,
by the construction of $f$, we must have
$a = f(P_i', P_j) = \textrm{Proj}\big(b, \langle r_1(P_i'), r_1(P_j)|\mathcal{T}^A\rangle \big)
=\textrm{Proj}\big(r_1(P_i'), \langle b, r_1(P_j)|\mathcal{T}^A\rangle \big) \in \langle b, r_1(P_j)|\mathcal{T}^A\rangle$
This implies $r_1(P_j) = a$, which contradicts the fact $P_j \notin \mathbb{D}_3^a$.
In subcase (ii), since $P_i = P_7$ and $f(P_i, P_j) = b \neq d$, the construction of $f$ implies $P_j \notin \mathbb{D}_3^d = \{P_6, P_9\}$.
Hence, we have $P_j = P_5 \in \mathbb{D}_3^c$.
Thus, by the construction of $f$, we know that either $P_i' = P_9$ and $f(P_i', P_j) = c$,
or $f(P_i', P_j) = \textrm{Proj}\big(b, \langle r_1(P_i'), r_1(P_j)|\mathcal{T}^A\rangle\big)
= \textrm{Proj}\big(b, \langle r_1(P_i'), c|\mathcal{T}^A\rangle\big) = \textrm{Proj}\big(r_1(P_i'), \langle b, c|\mathcal{T}^A\rangle\big)
\in \langle b, c|\mathcal{T}^A\rangle = \{b,c\}$ holds.
Furthermore, since $f(P_i', P_j) \neq b$, we have $f(P_i', P_j) = c$.
Consequently, according to $P_i = P_7$, we have $f(P_i,P_j)\mathrel{P_i}f(P_i',P_j)$.
Overall, in both subcases, voter $i$ has no incentive to manipulate.
In conclusion, $f$ is strategy-proof.

\section{A brief extension to multidimensional models}\label{app:extension}

In this section,
we provide an application of Theorem \ref{thm:invariance} to a multidimensional setting.
Let a finite alternative set $A$ be decomposed as a Cartesian product, i.e.,
$A = \times_{s \in M}A^s$, where $M = \{1, \dots, \ell\}$, $\ell>1$, and
$|A^s| \geq 2$ for each $s \in M$.
Thus, an alternative is uniquely assembled by $\ell$ elements of $A^1, \dots, A^\ell$, i.e.,
$x = (x^1, \dots, x^{\ell}) = (x^s, x^{-s})$.
Given $a,b\in A$, let $M(a,b)=\{s \in M: a^s \neq b^s\}$ denote the set of components on which $a$ and $b$ disagree. Thus, two alternatives $a,b\in A$ are said \textbf{similar} if $|M(a,b)| = 1$.
Fix a domain $\mathbb{D}$ of preferences over $A$.
Given a preference $P_i \in \mathbb{D}$ and a nonempty subset $B \subset A$, let $P_i|_B$ denote the preference over $B$ induced via removing all alternatives out of $B$ in $P_i$.
Given $s \in M$ and $a^{-s} \in A^{-s}$, let $(A^s, a^{-s}) = \{(a^s, a^{-s}) \in A: a^s \in A^s\}$ and $\mathbb{D}^{(A^s, a^{-s})} = \{P_i \in \mathbb{D}: r_1(P_i) \in (A^s, a^{-s})\}$.
In particular, a preference $P_i \in \mathbb{D}$ is \textbf{separable} if there exists a \emph{marginal preference} (or a linear order) $P_i^s$ over $A^s$ for each $s \in M$ such that for all similar alternatives $a,b \in A$, say $M(a,b) = \{s\}$, we have $[a^s\mathrel{P_i^s}b^s] \Rightarrow [a\mathrel{P_i}b]$.

As noted earlier, the domain of all separable preferences is excluded from unidimensional domains due to the violation of adjacency. We introduce a new notion between two alternatives which generalizes the notion of adjacency to cover domains that contain separable preferences.
Fixing a domain $\mathbb{D}$,
two alternatives $a,b \in A$ are \textbf{adjacent\textsuperscript{+}}, denoted $a \sim^{+} b$, if they are similar, say $M(a,b) = \{s\}$, and there exist two separable preferences $P_i \in \mathbb{D}^a$ and $P_i' \in \mathbb{D}^b$ satisfying the following two conditions:
\begin{itemize}
\item[\rm (i)] given $x^{-s} \in A^{-s}$, $(a^s, x^{-s}) = r_k(P_i) = r_{k+1}(P_i')$ and
$(b^s, x^{-s}) = r_k(P_i') = r_{k+1}(P_i)$ for some $1 \leq k < |A|$, and

\item[\rm (ii)] given $c \in A$, $\big[c^s \notin \{a^s, b^s\}\big]\Rightarrow \big[c = r_k(P_i) = r_k(P_i')\;
\textrm{for some}\; 2< k \leq |A|\big]$.
\end{itemize}
Accordingly, given $s \in M$ and $a^{-s} \in A^{-s}$, we induce a graph $G_{\sim^+}^{(A^s, a^{-s})}$, where the vertex set is $(A^s, a^{-s})$,
and two vertices of $(A^s, a^{-s})$ form an edge if and only if they are adjacent\textsuperscript{+}.
A domain $\mathbb{D}$ satisfies \textbf{path-connectedness\textsuperscript{$+$}} if
$G_{\sim^+}^{(A^s,\, a^{-s})}$ is a connected graph for each $s \in M$ and $a^{-s}\in A^{-s}$.
Moreover, a domain $\mathbb{D}$ satisfies \textbf{diversity\textsuperscript{$+$}} if there exists a pair of separable preferences $\underline{P}_i, \overline{P}_i \in \mathbb{D}$ that are complete reversals, and furthermore for each $s \in M$, there exist $a^{-s} \in A^{-s}$ and $P_i, P_i' \in \mathbb{D}^{(A^s, a^{-s})}$ such that $P_i|_{(A^s, a^{-s})}$ and $P_i'|_{(A^s, a^{-s})}$ are complete reversals.
Henceforth, a domain $\mathbb{D}$ is called a \textbf{multidimensional domain} if
it satisfies both path-connectedness\textsuperscript{$+$} and diversity\textsuperscript{$+$}.\footnote{The notion of adjacency\textsuperscript{$+$} was originally introduced by \citet{CZ2019}. The multidimensional domains introduced by \citet{CZ2019} are restricted domains:
all preferences are required to be top-separable \citep[introduced by][]{LW1999}.
Moreover, their investigation focuses on a subclass of multidimensional domains where
all preferences are well organized in the manner that is analogous to no-restoration.
The class of multidimensional domains introduced here is significantly less demanding,
as it is only concerned with the richness of separable preferences included in the domain,
and consequentially the universal domain is covered as a special case.}

Now fix a multidimensional domain $\mathbb{D}$ and a two-voter, tops-only and strategy-proof rule $f: \mathbb{D}^2 \rightarrow A$. As a first step towards characterizing such a rule,
we fix an arbitrary component $s \in M$.
By diversity\textsuperscript{$+$}, we have $a^{-s} \in A^{-s}$ and preferences $P_i, P_i' \in \mathbb{D}^{(A^s, a^{-s})}$ such that $P_i|_{(A^s, a^{-s})}$ and $P_i'|_{(A^s, a^{-s})}$ are complete reversals.
We then restrict attention to $f$ on the subdomain $\mathbb{D}^{(A^s, a^{-s})}$.

Let $\mathbb{D}|_{(A^s, a^{-s})} = \big\{P_i|_{(A^s, a^{-s})}: P_i \in \mathbb{D}^{(A^s, a^{-s})}\big\}$.
Note that $\mathbb{D}|_{(A^s, a^{-s})}$ satisfies path-connectedness and diversity.\footnote{Given distinct $a^s, b^s \in A^s$, let $P_i \in \mathbb{D}^{(a^s, a^{-s})}$ and $P_i' \in \mathbb{D}^{(b^s, a^{-s})}$ be two separable preferences that indicate $(a^s, a^{-s}) \sim^+ (b^s, a^{-s})$.
Then, the induced preferences $P_i|_{(A^s, a^{-s})}$ and $P_i'|_{(A^s, a^{-s})}$ satisfy the following conditions: (i) $r_1\big(P_i|_{(A^s, a^{-s})}\big) = r_2\big(P_i'|_{(A^s, a^{-s})}\big) = (a^s, a^{-s})$ and $r_1\big(P_i'|_{(A^s, a^{-s})}\big) = r_2\big(P_i|_{(A^s, za^{-s})}\big) = (b^s, a^{-s})$, and
(ii) $r_k\big(P_i|_{(A^s, a^{-s})}\big) = r_k\big(P_i'|_{(A^s, a^{-s})}\big)$ for all $k \in \{3, \dots, |A^s|\}$.
This indicates that in the induced domain $\mathbb{D}|_{(A^s, a^{-s})}$, alternatives $(a^s, a^{-s})$ and $(b^s, a^{-s})$ are adjacent.
Consequently, the adjacency graph $G_{\sim}^{(A^s, a^{-s})}$ induced according to $\mathbb{D}|_{(A^{s}, a^{-s})}$ must contain the connected graph $G_{\sim^+}^{(A^s, a^{-s})}$.
Therefore, $\mathbb{D}|_{(A^{s}, a^{-s})}$ is a path-connected domain.}
Also, note that given $P_1, P_2, P_1', P_2' \in \mathbb{D}^{(A^s, a^{-s})}$,
if $P_1|_{(A^s, a^{-s})} = P_1'|_{(A^s, a^{-s})}$ and $P_2|_{(A^s, a^{-s})} = P_2'|_{(A^s, a^{-s})}$,
then $r_1(P_1) = r_1(P_1')$, $r_1(P_2) = r_1(P_2')$, and hence $f(P_1,P_2) = f(P_1', P_2')$ by the tops-only property.
Similar to statement (i) of Lemma \ref{lem:transitivity},
we can easily show $f(P_1, P_2) \in (A^s, a^{-s})$ for all $P_1, P_2 \in \mathbb{D}^{(A^s, a^{-s})}$.
Then, we can induce a tops-only and strategy-proof rule $g: \big[\mathbb{D}|_{(A^s, a^{-s})}\big]^2 \rightarrow (A^s, a^{-s})$ such that
\begin{align*}
g\big(P_1|_{(A^s, a^{-s})}, P_2|_{(A^s, a^{-s})}\big) =f(P_1, P_2)\;
\textrm{for all} \; P_1, P_2 \in \mathbb{D}^{(A^s, a^{-s})}.
\end{align*}

According to the two completely reversed preferences in $\mathbb{D}|_{(A^s, a^{-s})}$,
we know that $g$ is either invariant or not.
If $g$ is invariant, then by \textbf{Statement (i)} of Theorem \ref{thm:invariance},
$g$ is a projection rule, and hence $f$ \textbf{behaves like a projection rule on $(A^s, a^{-s})$}),
i.e., there exist a tree $\mathcal{T}^{(A^s, a^{-s})}$ and an alternative $(\hat{x}^s, a^{-s}) \in (A^s, a^{-s})$ such that for all $a, b \in (A^s, a^{-s})$,
\begin{align*}
f(a, b) = g(a, b)
=  \textrm{Proj}\big((\hat{x}^s, a^{-s}), \langle a, b|\mathcal{T}^{(A^s, a^{-s})}\rangle \big).
\end{align*}
If $g$ is not invariant, then by the proof of \textbf{Statement (ii)} of Theorem \ref{thm:invariance},
$g$ is a hybrid rule, and hence $f$ \textbf{behaves like a hybrid rule on $(A^s, a^{-s})$}),
i.e., there exist a tree $\mathcal{T}^{(A^s, a^{-s})}$ and dual-thresholds $(\underline{x}^s, a^{-s})$ and $(\overline{x}^s, a^{-s})$ in $\mathcal{T}^{(A^s, a^{-s})}$,
which induce the following three exhaustive subsets of $(A^s, a^{-s})$:
\begin{align*}
\underline{(A^s, a^{-s})}
= & \big\{a \in (A^s, a^{-s}): (\underline{x}^s, a^{-s}) \in \langle a, (\overline{x}^s, a^{-s})|\mathcal{T}^{(A^s, a^{-s})}\rangle\big\},\\
\widehat{(A^s, a^{-s})} = & \big\{a \in (A^s, a^{-s}): a \in \langle (\underline{x}^s, a^{-s}), (\overline{x}^s, a^{-s}) |\mathcal{T}^{(A^s, a^{-s})}\rangle\big\},\; \textrm{and}\\
\overline{(A^s, a^{-s})}
= & \big\{a \in (A^s, a^{-s}): (\overline{x}^s, a^{-s}) \in \langle a, (\underline{x}^s, a^{-s})|\mathcal{T}^{(A^s, a^{-s})}\rangle\big\},
\end{align*}
such that one of the following two cases must occur:
(i) for all $a,b \in (A^s, a^{-s})$,
\begin{align*}
f(a, b )
= g(a, b)
=
\left\{
\begin{array}{ll}
a & \textrm{if}\; a \in \widehat{(A^s, a^{-s})},\\[0.4em]
\textrm{Proj}\big((\underline{x}^s, a^{-s}), \langle a, b|\mathcal{T}^{(A^s, a^{-s})}\rangle\big) & \textrm{if}\; a \in \underline{(A^s, a^{-s})}\big\backslash \{(\underline{x}^s, a^{-s})\},\; \textrm{and}\\[0.4em]
\textrm{Proj}\big((\overline{x}^s, a^{-s}), \langle a, b|\mathcal{T}^{(A^s, a^{-s})}\rangle\big) & \textrm{if}\; a \in \overline{(A^s, a^{-s})}\big\backslash \{(\overline{x}^s, a^{-s})\},
\end{array}
\right.
\end{align*}
or (ii) for all $a,b \in (A^s, a^{-s})$,
\begin{align*}
f(a, b )
= g(a, b)
=
\left\{
\begin{array}{ll}
b & \textrm{if}\; b \in \widehat{(A^s, a^{-s})},\\[0.4em]
\textrm{Proj}\big((\underline{x}^s, a^{-s}), \langle a, b|\mathcal{T}^{(A^s, a^{-s})}\rangle\big) & \textrm{if}\; b \in \underline{(A^s, a^{-s})}\big\backslash \{(\underline{x}^s, a^{-s})\},\;\textrm{and}\\[0.4em]
\textrm{Proj}\big((\overline{x}^s, a^{-s}), \langle a, b|\mathcal{T}^{(A^s, a^{-s})}\rangle\big) & \textrm{if}\; b \in \overline{(A^s, a^{-s})}\big\backslash \{(\overline{x}^s, a^{-s})\}.
\end{array}
\right.
\end{align*}

The characterization of $g$ provided above constitutes an essential ingredient of a full characterization of $f$.
We leave the task of fully characterizing a rule on a multidimensional domain and extracting the domain implications of the strategy-proofness of such a rule for future work.\footnote{All our analysis so far is based on strategy-proof SCFs which are only concerned with unilateral deviations at a preference profile. The investigation on multidimensional domains can be enriched when we expand to SCFs that are immune to group deviations. It is known from \citet{BBM2016} that on a multidimensional domain, the number of voters and dimensions affect the elicitation of preference restriction when the SCF is required to satisfy some stronger notion of incentive compatibility that is related to group deviations, like group-strategy-proofness or immunity to credible deviations.}

}

\begin{thebibliography}{39}
\newcommand{\enquote}[1]{``#1''}
\expandafter\ifx\csname natexlab\endcsname\relax\def\natexlab#1{#1}\fi

\bibitem[\protect\citeauthoryear{Achuthankutty and Roy}{Achuthankutty and
  Roy}{2020}]{AR2020}
\textsc{Achuthankutty, G. and S.~Roy} (2020): \enquote{Strategy-proof rules on
  partially single-peaked domains,} \emph{Indian Statistical Institute (Kolkata) Working Paper}.

\bibitem[\protect\citeauthoryear{Aswal, Chatterji, and Sen}{Aswal
  et~al.}{2003}]{ACS2003}
\textsc{Aswal, N., S.~Chatterji, and A.~Sen} (2003): \enquote{Dictatorial
  domains,} \emph{Economic Theory}, 22, 45--62.

\bibitem[\protect\citeauthoryear{Barber{\`a}}{Barber{\`a}}{2011}]{B2011}
\textsc{Barber{\`a}, S.} (2011): \enquote{Strategy-proof social choice,} in
  \emph{Handbook of Social Choice and Welfare}, vol.~2, 731--831.

\bibitem[\protect\citeauthoryear{Barber{\`a}, Berga, and Moreno}{Barber{\`a}
  et~al.}{2016}]{BBM2016}
\textsc{Barber{\`a}, S., D.~Berga, and B.~Moreno} (2016): \enquote{Immunity to credible deviations from the truth,}
\emph{Mathematical Social Sciences}, 90, 129--140.

\bibitem[\protect\citeauthoryear{Barber{\`a}, Berga, and Moreno}{Barber{\`a}
  et~al.}{2020}]{BBM2020}
\textsc{Barber{\`a}, S., D.~Berga, and B.~Moreno} (2020): \enquote{Arrow on
  domain conditions: A fruitful road to travel,} \emph{Social Choice and
  Welfare}, 54, 237--258.

\bibitem[\protect\citeauthoryear{Barber{\`a}, Gul, and Stacchetti}{Barber{\`a}
  et~al.}{1993}]{BGS1993}
\textsc{Barber{\`a}, S., F.~Gul, and E.~Stacchetti} (1993):
  \enquote{Generalized median voter schemes and committees,} \emph{Journal of
  Economic Theory}, 61, 262--289.

\bibitem[\protect\citeauthoryear{Barber\`{a} and Jackson}{Barber\`{a} and
  Jackson}{1994}]{BJ1994}
\textsc{Barber\`{a}, S. and M.~Jackson} (1994): \enquote{A characterization of
  strategy-proof social choice functions for economies with pure public goods,}
  \emph{Social Choice and Welfare}, 11, 241--252.


\bibitem[\protect\citeauthoryear{Barber\`{a}, Mass\'{o}, and Neme}{Barber\`{a} et~al.}{1997}]{BMN1997}
\textsc{Barber\`{a}, S., J.~Mass\'{o}, and A.~Neme} (1997): \enquote{Voting under constraints,}
  \emph{Journal of Economic Theory}, 76, 298--321.

\bibitem[\protect\citeauthoryear{Barber\`{a}, and Moreno}{Barber\`{a} and Moreno}{2011}]{BM2011}
\textsc{Barber\`{a}, S. and B.~Moreno} (2011): \enquote{Top monotonicity: A common root for single peakedness, single crossing and the median voter result,}
  \emph{Games and Economic Behavior}, 73, 345--359.

\bibitem[\protect\citeauthoryear{Barber{\`a}, Sonnenschein, and
  Zhou}{Barber{\`a} et~al.}{1991}]{BSZ1991}
\textsc{Barber{\`a}, S., H.~Sonnenschein, and L.~Zhou} (1991): \enquote{Voting
  by committees,} \emph{Econometrica}, 595--609.


\bibitem[\protect\citeauthoryear{Black}{Black}{1948}]{B1948}
\textsc{Black, D.} (1948): \enquote{On the rationale of group decision-making,}
  \emph{the Journal of Political Economy}, 56, 23.

\bibitem[\protect\citeauthoryear{Bonifacio and Mass\'{o}}{Bonifacio and
  Mass\'{o}}{2020}]{BM2020}
\textsc{Bonifacio, A. and J.~Mass\'{o}} (2020): \enquote{On strategy-proofness
  and semilattice single-peakedness,} \emph{Games and Economic Behavior}, 124,
  219--238.

\bibitem[\protect\citeauthoryear{Border and Jordan}{Border and
  Jordan}{1983}]{BJ1983}
\textsc{Border, K. and J.~Jordan} (1983): \enquote{Straightforward elections,
  unanimity and phantom voters,} \emph{Review of Economic Studies}, 153--170.

\bibitem[\protect\citeauthoryear{Chatterji and Mass{\'o}}{Chatterji and
  Mass{\'o}}{2018}]{CM2018}
\textsc{Chatterji, S. and J.~Mass{\'o}} (2018): \enquote{On strategy-proofness
  and the salience of single-peakedness,} \emph{International Economic Review},
  59, 163--189.

\bibitem[\protect\citeauthoryear{Chatterji, Roy, Sadnhukhan, Sen, and
  Zeng}{Chatterji et~al.}{2022}]{CRSSZ2022}
\textsc{Chatterji, S., S.~Roy, S.~Sadnhukhan, A.~Sen, and H.~Zeng} (2022):
\enquote{Probabilistic fixed ballot rules and hybrid domains,}
\emph{Journal of Mathematical Economics},
  100: 102656.

\bibitem[\protect\citeauthoryear{Chatterji, Sanver, and Sen}{Chatterji
  et~al.}{2013}]{CSS2013}
\textsc{Chatterji, S., R.~Sanver, and A.~Sen} (2013): \enquote{On domains that
  admit well-behaved strategy-proof social choice functions,} \emph{Journal of
  Economic Theory}, 148, 1050--1073.

\bibitem[\protect\citeauthoryear{Chatterji and Sen}{Chatterji and
  Sen}{2011}]{CS2011}
\textsc{Chatterji, S. and A.~Sen} (2011): \enquote{Tops-only domains,}
  \emph{Economic Theory}, 46, 255--282.

\bibitem[\protect\citeauthoryear{Chatterji and Zeng}{Chatterji
  and Zeng}{2019}]{CZ2019}
\textsc{Chatterji, S. and H.~Zeng} (2019): \enquote{Random mechanism design on multidimensional domains,}
\emph{Journal of Economic Theory}, 182, 25--105.

\bibitem[\protect\citeauthoryear{Chatterji and Zeng}{Chatterji and
  Zeng}{2020}]{CZ2020}
\textsc{Chatterji, S. and H.~Zeng} (2020): \enquote{\href{https://ink.library.smu.edu.sg/cgi/viewcontent.cgi?article=3407&context=soe_research}{\textcolor[rgb]{0.00,0.00,1.00}{A taxonomy of non-dictatorial domains}},} \emph{Singapore Management University Working Paper}.

\bibitem[\protect\citeauthoryear{Chatterji and Zeng}{Chatterji and
  Zeng}{2022}]{CZ2022SM}
\textsc{Chatterji, S. and H.~Zeng} (2022): \enquote{\href{https://drive.google.com/file/d/1UAaABo8N0c-TYFIw_NJLxqLP8WRECZu0/view}{\textcolor[rgb]{0.00,0.00,1.00}{Supplementary Material for ``A taxonomy of non-dictatorial domains''}},} \emph{mimeo}.




\bibitem[\protect\citeauthoryear{Ching}{Ching}{1997}]{Ching1997}
\textsc{Ching, S.} (1997): \enquote{Strategy-proofness and ``median voters",}
  \emph{International Journal of Game Theory}, 26, 473--490.

\bibitem[\protect\citeauthoryear{Demange}{Demange}{1982}]{D1982}
\textsc{Demange, G.} (1982): \enquote{Single-peaked orders on a tree,}
  \emph{Mathematical Social Sciences}, 3, 389--396.

\bibitem[\protect\citeauthoryear{Gibbard}{Gibbard}{1973}]{G1973}
\textsc{Gibbard, A.} (1973): \enquote{Manipulation of voting schemes: A general
  result,} \emph{Econometrica}, 587--601.

\bibitem[\protect\citeauthoryear{Grandmont}{Grandmont}{1978}]{G1978}
\textsc{Grandmont, J.} (1978): \enquote{Intermediate preferences and the
  majority rule,} \emph{Econometrica}, 317--330.

\bibitem[\protect\citeauthoryear{Kalai and Ritz}{Kalai and Ritz}{1980}]{KR1980}
\textsc{Kalai, E. and Z.~Ritz} (1980): \enquote{Characterization of the private
  alternatives domains admitting Arrow social welfare functions,} \emph{Journal
  of Economic Theory}, 22, 23--36.

\bibitem[\protect\citeauthoryear{Le~Breton and Sen}{Le~Breton and
  Sen}{1999}]{LS1999}
\textsc{Le~Breton, M. and A.~Sen} (1999): \enquote{Separable preferences,
  strategyproofness, and decomposability,} \emph{Econometrica}, 67, 605--628.

\bibitem[\protect\citeauthoryear{Le~Breton and Weymark}{Le~Breton and
  Weymark}{1999}]{LW1999}
\textsc{Le~Breton, M. and J.~Weymark} (1999): \enquote{Strategy-proof social
  choice with continuous separable preferences,} \emph{Journal of Mathematical
  Economics}, 32, 47--85.

\bibitem[\protect\citeauthoryear{Monjardet}{Monjardet}{2009}]{M2009}
\textsc{Monjardet, B.} (2009): \enquote{Acyclic domains of linear orders: A
  survey,} in \emph{The Mathematics of Preference, Choice and Order}, Springer,
  139--160.

\bibitem[\protect\citeauthoryear{Moulin}{Moulin}{1980}]{M1980}
\textsc{Moulin, H.} (1980): \enquote{On strategy-proofness and single
  peakedness,} \emph{Public Choice}, 35, 437--455.

\bibitem[\protect\citeauthoryear{Nehring and Puppe}{Nehring and
  Puppe}{2007}]{NP2007}
\textsc{Nehring, K. and C.~Puppe} (2007): \enquote{The structure of
  strategy-proof social choice: Part {I}: General characterization and
  possibility results on median spaces,} \emph{Journal of Economic Theory},
  135, 269--305.

\bibitem[\protect\citeauthoryear{Pramanik}{Pramanik}{2015}]{P2015}
\textsc{Pramanik, A.} (2015): \enquote{Further results on dictatorial domains,}
  \emph{Social Choice and Welfare}, 45, 379--398.

\bibitem[\protect\citeauthoryear{Puppe}{Puppe}{2018}]{P2018}
\textsc{Puppe, C.} (2018): \enquote{The single-peaked domain revisited: A
  simple global characterization,} \emph{Journal of Economic Theory}, 176,
  55--80.

\bibitem[\protect\citeauthoryear{Reffgen}{Reffgen}{2015}]{R2015}
\textsc{Reffgen, A.} (2015): \enquote{Strategy-proof social choice on multiple
  and multi-dimensional single-peaked domains,} \emph{Journal of Economic
  Theory}, 157, 349--383.

\bibitem[\protect\citeauthoryear{Roberts}{Roberts}{1979}]{R1979}
\textsc{Roberts, K.} (1979): \enquote{The characterization of implementable
  choice rules,} \emph{Aggregation and Revelation of Preferences}, 12,
  321--348.

\bibitem[\protect\citeauthoryear{Roy and Storcken}{Roy and
  Storcken}{2019}]{RS2019}
\textsc{Roy, S. and T.~Storcken} (2019): \enquote{A characterization of
  possibility domains in strategic voting,} \emph{Journal of Mathematical
  Economics}, 84, 46--55.

\bibitem[\protect\citeauthoryear{Saporiti}{Saporiti}{2009}]{S2009}
\textsc{Saporiti, A.} (2009): \enquote{Strategy-proofness and single-crossing,}
\emph{Theoretical Economics}, 4(2), 127--163.

\bibitem[\protect\citeauthoryear{Sato}{Sato}{2010}]{S2010}
\textsc{Sato, S.} (2010): \enquote{Circular domains,} \emph{Review of Economic
  Design}, 14, 331--342.

\bibitem[\protect\citeauthoryear{Sato}{Sato}{2013}]{S2013}
\textsc{Sato, S.} (2013): \enquote{A sufficient condition
  for the equivalence of strategy-proofness and nonmanipulability by
  preferences adjacent to the sincere one,} \emph{Journal of Economic Theory},
  148, 259--278.

\bibitem[\protect\citeauthoryear{Satterthwaite}{Satterthwaite}{1975}]{S1975}
\textsc{Satterthwaite, M.~A.} (1975): \enquote{Strategy-proofness and Arrow's
  conditions: Existence and correspondence theorems for voting procedures and
  social welfare functions,} \emph{Journal of Economic Theory}, 10, 187--217.

\bibitem[\protect\citeauthoryear{Schummer and Vohra}{Schummer and
  Vohra}{2002}]{SV2002}
\textsc{Schummer, J. and R.~V. Vohra} (2002): \enquote{Strategy-proof location
  on a network,} \emph{Journal of Economic Theory}, 104, 405--428.

\bibitem[\protect\citeauthoryear{Sen}{Sen}{2001}]{S2001}
\textsc{Sen, A.} (2001): \enquote{Another direct proof of the
  Gibbard-Satterthwaite Theorem,} \emph{Economics Letters}, 70, 381--385.

\bibitem[\protect\citeauthoryear{Sprumont}{Sprumont}{1995}]{S1995}
\textsc{Sprumont, Y.} (1995): \enquote{Strategyproof collective choice in
  economic and political environments,} \emph{Canadian Journal of Economics},
  68--107.

\bibitem[\protect\citeauthoryear{Thomson}{Thomson}{1993}]{T1993}
\textsc{Thomson, W.} (1993): \enquote{The replacement principle in public good
  economies with single-peaked preferences,} \emph{Economics Letters}, 42,
  31--36.

\bibitem[\protect\citeauthoryear{Vohra}{Vohra}{1999}]{V1999}
\textsc{Vohra, R.} (1999): \enquote{The replacement principle and tree
  structured preferences,} \emph{Economics Letters}, 63, 175--180.

\bibitem[\protect\citeauthoryear{Weymark}{Weymark}{2008}]{W2008}
\textsc{Weymark, J.} (2008): \enquote{Strategy-proofness and the tops-only
  property,} \emph{Journal of Public Economic Theory}, 10, 7--26.

\bibitem[\protect\citeauthoryear{Weymark}{Weymark}{2011}]{W2011}
\textsc{Weymark, J.} (2011): \enquote{A unified approach to
  strategy-proofness for single-peaked preferences,} \emph{Journal of the
  Spanish Economic Association}, 2, 529--550.

\end{thebibliography}
\end{document}